\documentclass[11pt,letterpaper]{article}
\usepackage[english]{babel}
\usepackage{enumitem}
\usepackage{amsmath,amssymb,amsthm,mathrsfs}
\usepackage[noadjust]{cite}
\usepackage{caption}

\usepackage[letterpaper,left=2.55cm,right=2.55cm,top=2.3cm,bottom=2.55cm]{geometry}

\usepackage{xcolor}
\usepackage{hyperref}
\hypersetup{colorlinks=true,linkcolor=blue,citecolor=blue}

\newtheorem{lemma}{Lemma}[section]
\newtheorem{proposition}{Proposition}[section]
\newtheorem{theorem}{Theorem}[section]
\newtheorem{corollary}{Corollary}[section]

\theoremstyle{definition}
\newtheorem{definition}{Definition}[section]
\newtheorem{fact}{Fact}[section]
\newtheorem{notation}{Notation}[section]
\newtheorem{example}{Example}[section]

\theoremstyle{remark}
\newtheorem{remark}{Remark}[section]

\DeclareMathOperator{\injtp}{\check{\otimes}\hspace{0.3mm}}

\DeclareMathOperator{\altp}{\breve{\otimes}\hspace{0.2mm}}

\DeclareMathOperator{\prtp}{\widehat{\otimes}\hspace{0.3mm}}

\DeclareMathOperator{\intp}{\overset{\Gamma}{\otimes}\hspace{0.3mm}}

\newcommand{\tre}{\hspace{0.3mm}}
\newcommand{\quattro}{\hspace{0.4mm}}
\newcommand{\cinque}{\hspace{0.5mm}}
\newcommand{\sei}{\hspace{0.6mm}}

\newcommand{\otto}{\hspace{0.8mm}}

\newcommand{\dieci}{\hspace{1mm}}
\newcommand{\dodici}{\hspace{1.2mm}}

\newcommand{\msei}{\hspace{-0.6mm}}

\newcommand{\mdodici}{\hspace{-1.2mm}}

\newcommand{\fin}{\hspace{0.3mm}}

\newcommand{\defi}{\mathrel{\mathop:}=}
\newcommand{\ifed}{=\mathrel{\mathop:}}

\newcommand{\afhy}{\mathsf{H}_{[\tr=1]}}

\newcommand{\ltrc}{\mathscr{L}_1(\hh)}

\newcommand{\btrc}{\mathcal{B}(\trc)}
\newcommand{\noro}{\|_{[1]}}
\newcommand{\normo}{\|\cdot\|_{[1]}}

\newcommand{\ptph}{\mathfrak{E}\hspace{0.3mm}}
\newcommand{\ptpj}{\mathfrak{G}\hspace{0.3mm}}
\newcommand{\ptphj}{\ptph\prtp\ptpj}

\newcommand{\st}{\mathrm{st}}
\newcommand{\stbhj}{\st(\bophj)}
\newcommand{\stbhjs}{\st(\bophj)_{\mathrm{s}}}
\newcommand{\stal}{\mho}
\newcommand{\stals}{\mho_{\mathrm{s}}}
\newcommand{\stald}{\mho_{D}}

\newcommand{\prot}{\mathfrak{Q}}

\newcommand{\indm}{\mathtt{I}}
\newcommand{\indn}{\mathtt{J}}

\newcommand{\dom}{\mathrm{dom}}

\newcommand{\vbf}{\vartheta}
\newcommand{\svbf}{\breve{\vbf}}
\newcommand{\vlm}{\Theta}

\newcommand{\aih}{\mathfrak{J}_{\mathcal{H}}}
\newcommand{\aij}{\mathfrak{J}_{\hspace{-0.3mm}\mathcal{J}}}

\newcommand{\isot}{\mathfrak{I}\hspace{0.3mm}}
\newcommand{\isoh}{\hspace{0.5mm}\mathsf{h}\hspace{0.2mm}}
\newcommand{\isoj}{\hspace{0.5mm}\mathsf{j}\hspace{0.3mm}}
\newcommand{\ops}{\mathtt{S}}
\newcommand{\opt}{\mathtt{T}}
\newcommand{\opc}{\mathtt{C}}

\newcommand{\ent}{\mathscr{E}}
\newcommand{\entvw}{\mathscr{E}_{\vv}^{\ww}\hspace{-0.5mm}}
\newcommand{\entex}{\mathscr{E}_{\mathrm{e}}}
\newcommand{\enthe}{\mathscr{H}}

\newcommand{\hjo}{(\hh\otimes\jj)_{\star}}
\newcommand{\pustahjo}{\mathcal{P}\hjo}

\newcommand{\repa}{\Re\mathrm{e}}
\newcommand{\kre}{K_{\mbox{\tiny $\Re$}}}
\newcommand{\fre}{F_{\mbox{\tiny $\Re$}}}

\newcommand{\erre}{\mathbb{R}}
\newcommand{\errep}{\mathbb{R}^{\mbox{\tiny $+$}}}

\newcommand{\erreps}{\mathbb{R}_{\hspace{0.3mm}\ast}^{\mbox{\tiny $+$}}}
\newcommand{\errems}{\mathbb{R}_{\hspace{0.3mm}\ast}^{\mbox{\tiny $-$}}}
\newcommand{\erres}{\mathbb{R}_\ast^{\phantom{\mbox{\tiny $i$}}}}
\newcommand{\toro}{\mathbb{T}}
\newcommand{\rati}{\mathbb{Q}}
\newcommand{\itg}{\mathbb{Z}}

\newcommand{\trsp}{\mathsf{T}}

\newcommand{\balp}{\ball\big(\thptj\big)_{\hspace{-0.3mm}\mbox{\tiny $+$}}}

\newcommand{\nep}{{\mathtt{N}_{\epsilon}}}

\newcommand{\xjl}{x_{jl}}
\newcommand{\xjln}{x_{j(n)\hspace{0.4mm}l(n)}}
\newcommand{\jln}{(j(n),l(n))}
\newcommand{\jlm}{(j(m),l(m))}
\newcommand{\kmn}{(k(n),m(n))}

\newcommand{\iset}{\mathcal{I}}
\newcommand{\sba}{\mathfrak{e}}
\newcommand{\cf}{\mathfrak{c}}

\newcommand{\tr}{\mathrm{tr}}
\newcommand{\card}{\mathrm{card}}

\newcommand{\ran}{\mathrm{ran}\hspace{0.3mm}}
\newcommand{\srank}{\mathrm{srank}\hspace{0.3mm}}
\newcommand{\sra}{\mathsf{s}\hspace{0.3mm}}
\newcommand{\scfs}{\mathrm{scfs}\hspace{0.3mm}}

\newcommand{\opd}{\mathsf{Opd}\hspace{0.3mm}\big(\thptj\big)}
\newcommand{\opds}{\mathsf{Opd}\hspace{0.3mm}\big(\thsptjs\big)}
\newcommand{\hopd}{\mathsf{HOpd}\hspace{0.3mm}\big(\thsptjs\big)}

\newcommand{\optrn}{\mbox{\scriptsize $\|\hspace{0.3mm}\cdot\hspace{0.3mm}\ptrn$-}\hspace{0.4mm}\mathtt{Osd}\hspace{0.3mm}}
\newcommand{\optrntr}{\mbox{\scriptsize $\ptrntr{\hspace{0.5mm}\cdot\hspace{0.5mm}}$-}\hspace{0.4mm}\mathtt{Osd}\hspace{0.3mm}}
\newcommand{\hoptrn}{\mbox{\scriptsize $\|\hspace{0.3mm}\cdot\hspace{0.3mm}\ptrn$-}\hspace{0.4mm}\mathtt{OHsd}\hspace{0.3mm}}

\newcommand{\paih}{\mathscr{I}_1(\hh)}
\newcommand{\paij}{\mathscr{I}_1(\jj)}

\newcommand{\kac}{\varkappa_C}
\newcommand{\kacn}{\varkappa_{C_n}}

\newcommand{\bothptj}{\bo\big(\trc,\trcjj;\thptj\big)}

\newcommand{\con}{\mathscr{C}}
\newcommand{\ves}{\mathcal{V}}

\newcommand{\pode}{\mathsf{Pod}\hspace{0.3mm}\big(\thsptjs\big)}

\newcommand{\act}{\mathsf{a}_{V}}

\newcommand{\eik}{\eee^{\ima\tre 2\pi k\argo}}
\newcommand{\eikx}{\eee^{-\ima\tre 2\pi kx}}
\newcommand{\eit}{\eee^{\ima\tre 2\pi\vartheta}}
\newcommand{\eix}{\eee^{\ima\tre 2\pi x}}

\newcommand{\PP}{R}

\newcommand{\bfo}{\omega}

\newcommand{\elluno}{\ell^{\mbox{\tiny $1$}}}
\newcommand{\elldue}{\ell^{\mbox{\tiny $2$}}}

\newcommand{\hsjth}{\mathcal{B}_2(\jj;\hh)}
\newcommand{\trjth}{\mathcal{B}_1(\jj;\hh)}

\newcommand{\pnh}[1]{{\left\vert\kern-0.25ex\left\vert #1\right\vert\kern-0.25ex\right\vert}^{\hspace{-0.3mm}\sotimp}}

\newcommand{\Dl}{\Delta_l}
\newcommand{\Den}{D^{\mbox{\tiny $(\enne)$}}\hspace{-0.3mm}}
\newcommand{\pien}{\varpi^{\mbox{\tiny $(\enne)$}}\hspace{-0.3mm}}

\newcommand{\neek}{\phi_k}
\newcommand{\neel}{\phi_l}
\newcommand{\nffk}{\psi_k}
\newcommand{\nffl}{\psi_l}

\newcommand{\nhekl}{\widehat{\hspace{-0.7mm}\phantom{.}\neek\neel}}

\newcommand{\nhfkl}{\widehat{\nffk\nffl}}
\newcommand{\nhekk}{\widehat{\neek\neek}}
\newcommand{\nhfll}{\widehat{\hspace{-0.5mm}\phantom{.}\nffl\nffl}}

\newcommand{\cast}{\mathrm{C}^\ast}
\newcommand{\wst}{w^\ast\hspace{-0.5mm}}
\newcommand{\clcowst}{\overline{\mathrm{co}}^{\hspace{0.4mm}w^\ast}\hspace{-0.8mm}}
\newcommand{\wstcl}{\mbox{$\wst$-$\hspace{0.3mm}\cl$}}
\newcommand{\wstli}{\mbox{$\wst$--}\lim}
\newcommand{\awst}{\mathcal{A}w^\ast\hspace{-0.5mm}}
\newcommand{\awstcl}{\mbox{$\awst$-$\hspace{0.3mm}\cl$}}
\newcommand{\awstli}{\mbox{$\awst$--}\lim}
\newcommand{\ewst}{ew^\ast\hspace{-0.5mm}}
\newcommand{\ewstcl}{\mbox{$\ewst$-$\hspace{0.3mm}\cl$}}

\newcommand{\cwst}{cw^\ast\hspace{-0.5mm}}
\newcommand{\clcocws}{\overline{\mathrm{co}}^{\hspace{0.4mm}cw^\ast}\hspace{-0.8mm}}
\newcommand{\cwstcl}{\mbox{$\cwst$-$\hspace{0.3mm}\cl$}}
\newcommand{\cwstclb}{\overline{\shjb}^{\hspace{0.4mm}cw^\ast}\hspace{-0.8mm}}
\newcommand{\cwstli}{\mbox{$\cwst$--}\lim}

\newcommand{\hj}{\hh_j}
\newcommand{\jl}{\jj_l}

\newcommand{\pij}{\pi_j}
\newcommand{\pil}{\varpi_l}
\newcommand{\pim}{\pi_m}
\newcommand{\pin}{\varpi_n}
\newcommand{\pijl}{\pij\otimes\pil}
\newcommand{\pimn}{\pim\otimes\pin}
\newcommand{\cjlmn}{C_{jl}^{mn}}
\newcommand{\cjl}{C_{jl}^{jl}}

\newcommand{\nuno}{{\mathsf{n}_1}}
\newcommand{\ndue}{{\mathsf{n}_2}}
\newcommand{\nn}{\mathsf{n}}

\newcommand{\idpil}{I\otimes\pil}

\newcommand{\trnsum}{\mbox{\scriptsize $\|\hspace{0.2mm}\cdot\hspace{0.2mm}\trn$--}\hspace{-0.5mm}\sum}
\newcommand{\ttrnsum}{\mbox{\scriptsize $\|\hspace{0.2mm}\cdot\hspace{0.2mm}\ttrn$--}\hspace{-0.5mm}\sum}
\newcommand{\ptrnsum}{\mbox{\scriptsize $\|\hspace{0.2mm}\cdot\hspace{0.2mm}\ptrn$--}\hspace{-0.5mm}\sum}
\newcommand{\ptrnsumb}{\mbox{\scriptsize $\|\hspace{0.2mm}\cdot\hspace{0.2mm}\ptrn$--}\hspace{-2.3mm}\sum}
\newcommand{\ptrntrsum}{\mbox{\scriptsize $\ptrntr{\hspace{0.5mm}\cdot\hspace{0.5mm}}$--}\hspace{-0.5mm}\sum}

\newcommand{\vv}{\mathcal{V}}
\newcommand{\ww}{\mathcal{W}}
\newcommand{\zz}{\mathcal{Z}}
\newcommand{\tzz}{\widetilde{\mathcal{Z}}}

\newcommand{\lath}{\Lambda_{\theta}}
\newcommand{\linv}{\lin(\vv)}
\newcommand{\linvw}{\lin(\vv,\ww;\ccc)}
\newcommand{\linvwad}{\lin(\vv,\ww;\ccc)^\prime}
\newcommand{\linvwz}{\lin(\vv,\ww;\zz)}
\newcommand{\linvwc}{\lin(\vv,\ww;\ccc)}
\newcommand{\linvawz}{\lin\big(\vv\altp\ww;\zz\big)}

\newcommand{\vvad}{\mathcal{V}^\prime}
\newcommand{\wwad}{\mathcal{W}^\prime}

\newcommand{\vvww}{\vv\otimes\ww}
\newcommand{\trcvv}{\mathcal{B}_1(\vv)}
\newcommand{\trcww}{\mathcal{B}_1(\ww)}
\newcommand{\trcvvs}{\mathcal{B}_1(\vv)_{\hspace{-0.3mm}\mbox{\tiny $\mathbb{R}$}}}
\newcommand{\trcwws}{\mathcal{B}_1(\ww)_{\hspace{-0.3mm}\mbox{\tiny $\mathbb{R}$}}}
\newcommand{\trcvw}{\mathcal{B}_1(\vv\otimes\ww)}
\newcommand{\trcvwsa}{\mathcal{B}_1(\vv\otimes\ww)_{\hspace{-0.3mm}\mbox{\tiny $\mathbb{R}$}}}
\newcommand{\vevw}{{}_{\mbox{\tiny $\hspace{0.8mm}\vv$}}^{\mbox{\tiny $\ww$}}\hspace{-0.3mm}\|}
\newcommand{\vwptrnor}{\vevw\cdot\ptrn}

\newcommand{\jjm}{\jj^{\mbox{\tiny $(\emme)$}}\hspace{-0.2mm}}
\newcommand{\trcjjm}{\mathcal{B}_1(\jjm)}
\newcommand{\trchjm}{\mathcal{B}_1(\hh\otimes\jjm)}
\newcommand{\vem}{{}^{\mbox{\tiny $\emme$}}\hspace{-0.2mm}\|}
\newcommand{\mptrnor}{\vem\cdot\ptrn}

\newcommand{\trcjl}{\mathcal{B}_1(\jl)}
\newcommand{\trchjl}{\mathcal{B}_1(\hh\otimes\jl)}
\newcommand{\stahjl}{\mathcal{D}(\hh\otimes\jl)}

\newcommand{\identi}{\mathfrak{f}\hspace{0.5mm}}
\newcommand{\iden}{\widehat{\mathfrak{f}}\hspace{0.5mm}}

\newcommand{\sinft}{\mbox{\tiny $\infty$}}

\newcommand{\clcow}{\overline{\mathrm{co}}^{\hspace{0.4mm}ew^\ast}\hspace{-0.8mm}}
\newcommand{\clcoaw}{\overline{\mathrm{co}}^{\hspace{0.4mm}\mathcal{A}w^\ast}\hspace{-0.8mm}}

\newcommand{\clwsub}{\overline{\subh}^{\hspace{0.4mm}ew^\ast}\hspace{-1.2mm}}
\newcommand{\clawsub}{\overline{\subh}^{\hspace{0.4mm}\mathcal{A}w^\ast}\hspace{-1.2mm}}
\newcommand{\clawsubid}{\overline{\identi(\subh)}^{\hspace{0.4mm}\mathcal{A}w^\ast}\hspace{-1.2mm}}

\newcommand{\ext}{\mathrm{ext}}

\newcommand{\pail}{[\hspace{-0.5mm}[}
\newcommand{\pair}{\hspace{0.3mm}]\hspace{-0.5mm}]}

\newcommand{\pax}{\pail K,\hspace{0.3mm}\cdot\hspace{0.7mm}\pair}

\newcommand{\paxa}{\pail\hspace{0.7mm}\cdot\hspace{0.3mm},\xia\pair}
\newcommand{\padigu}{\pail\hspace{0.7mm}\cdot\hspace{0.3mm},\digu\pair}

\newcommand{\fu}{\Gamma_{\hspace{-0.4mm}L}}
\newcommand{\fk}{\Gamma_{\hspace{-0.4mm}K}}
\newcommand{\bif}{\gamma_{\hspace{-0.2mm}L}}
\newcommand{\digu}{\xi_{C}}
\newcommand{\digub}{\Xi_{C}}
\newcommand{\digud}{\Xi_{D}}
\newcommand{\xid}{\xi_{D}}
\newcommand{\xidi}{\xi_{D_{i}}}
\newcommand{\xia}{\xi_{A}}

\newcommand{\nof}{\|_{\gamma}}
\newcommand{\norf}{\|\cdot\nof}

\newcommand{\nog}{\|_{\Gamma}}
\newcommand{\norg}{\|\cdot\nog}

\newcommand{\dunog}{\|^{\ast}_{\Gamma}}
\newcommand{\dunorg}{\|\cdot\dunog}

\newcommand{\nodg}{\|_{\digamma}}
\newcommand{\nordg}{\|\cdot\nodg}

\newcommand{\cph}{\mathcal{C}(\hh)}
\newcommand{\cpj}{\mathcal{C}(\jj)}
\newcommand{\cphj}{\mathcal{C}(\hhjj)}
\newcommand{\shj}{\mathcal{S}(\hhjj)}
\newcommand{\shjt}{\widetilde{\mathcal{S}}(\hhjj)}
\newcommand{\stahjb}{\stahj_{\sph}}
\newcommand{\shjb}{\stahj_{\ball}}

\newcommand{\alghj}{\mathcal{A}(\hhjj)}
\newcommand{\alghjd}{\mathcal{A}(\hhjj)^\ast}
\newcommand{\alghjg}{\alghj_{\Gamma}}
\newcommand{\alghjgd}{\alghj_{\Gamma}^{\ast}}
\newcommand{\halghj}{{\alghj\hspace{-0.4mm}\widehat{\phantom{x}}}}

\newcommand{\alghjgc}{\overline{\alghj_{\Gamma}}}
\newcommand{\alghjgcd}{\overline{\alghj_{\Gamma}}^{\ast}}

\newcommand{\ball}{\mathsf{B}}
\newcommand{\sph}{\mathsf{S}}

\newcommand{\ethtj}{\nb(\trc,\trcjj)}
\newcommand{\ethtjs}{\nb(\trcsa,\trcsajj)}
\newcommand{\eshsj}{\nb(\sta,\stajj)}

\newcommand{\ntr}[1]{{\left\vert\kern-0.25ex\left\vert\kern-0.25ex\left\vert #1
\right\vert\kern-0.25ex\right\vert\kern-0.25ex\right\vert}}

\newcommand{\ptrntr}[1]{{\left\vert\kern-0.25ex\left\vert\kern-0.25ex\left\vert #1
\right\vert\kern-0.25ex\right\vert\kern-0.25ex\right\vert}_{1}^{\hspace{-0.3mm}\sotimp}}

\newcommand{\ptrnortr}{\ptrntr{\hspace{0.6mm}\cdot\hspace{0.6mm}}}
\newcommand{\ptrnortrli}{\mbox{\scriptsize $\ptrntr{\hspace{0.5mm}\cdot\hspace{0.5mm}}$--}\lim}

\newcommand{\ima}{\mathrm{i}}
\newcommand{\ccc}{\mathbb{C}}

\newcommand{\sgn}{\mathrm{sgn}\hspace{0.2mm}}

\newcommand{\don}{D_1}
\newcommand{\dtw}{D_2}

\newcommand{\dpl}{D^{\mbox{\tiny $+$}}\hspace{-0.3mm}}
\newcommand{\dmi}{D^{\mbox{\tiny $-$}}\hspace{-0.3mm}}

\newcommand{\tkg}{\mbox{\small $t_{k}\hspace{-0.6mm}>\hspace{-0.5mm}0$}}
\newcommand{\tkl}{\mbox{\small $t_{k}\hspace{-0.6mm}<\hspace{-0.5mm}0$}}

\newcommand{\mesp}{\mathscr{X}}

\newcommand{\salga}{\mathscr{A}}
\newcommand{\mespy}{\mathscr{Y}}
\newcommand{\bora}{\mathscr{B}(\mathscr{X})}
\newcommand{\boray}{\mathscr{B}(\mathscr{Y})}
\newcommand{\cobora}{\overline{\mathscr{B}(\mathscr{X})}^{\hspace{0.3mm}\mu}}

\newcommand{\nset}{\mathcal{N}}
\newcommand{\mnset}{\mathcal{F}}
\newcommand{\meset}{\mathcal{G}}
\newcommand{\nmu}{\mathscr{N}_\mu}
\newcommand{\comu}{\overline{\mu}}
\newcommand{\bose}{\mathcal{E}}
\newcommand{\posp}{\mathfrak{P}}
\newcommand{\borap}{\mathscr{B}(\posp)}
\newcommand{\borapx}{\mathscr{B}_\posp(\mathscr{X})}
\newcommand{\borapxy}{\mathscr{B}_\mespy(f(\mathscr{X}))}

\newcommand{\cmp}{\mathsf{c}}

\newcommand{\suse}{\mathscr{S}}

\newcommand{\bhobj}{\bop\hspace{-0.2mm}\otimes\hspace{-0.2mm}\bopjj}
\newcommand{\bhabj}{\bop\altp\bopjj}

\newcommand{\thotj}{\trc\otimes\trcjj}
\newcommand{\thatj}{\trc\altp\trcjj}
\newcommand{\thptj}{\trc\prtp\trcjj}
\newcommand{\thptjd}{\big(\thptj\big)^\ast}
\newcommand{\thsotjs}{\trcsa\otimes\trcsajj}
\newcommand{\thsatjs}{\trcsa\altp\trcsajj}
\newcommand{\thsptjs}{\trcsa\prtp\trcsajj}
\newcommand{\thptjsa}{\big(\trc\prtp\trcjj\big)_{\hspace{-0.3mm}\mbox{\tiny $\mathbb{R}$}}}

\newcommand{\tjpth}{\trcjj\prtp\trc}
\newcommand{\tjspths}{\trcsajj\prtp\trcsa}

\newcommand{\neigh}{\mathcal{O}}
\newcommand{\neighx}{\mathcal{O}_{\xi}}
\newcommand{\xp}{\xi^\prime}
\newcommand{\neighxp}{\mathcal{O}_{\xp}}
\newcommand{\neighc}{{\neigh^{\mathsf{c}}}}
\newcommand{\neighp}{{\neigh^{\prime}}}
\newcommand{\chne}{\chi_\neigh}
\newcommand{\chnec}{\chi_\neighc}

\newcommand{\df}{\mathsf{d}_{\bsb,\hspace{0.2mm}\bsc}}

\newcommand{\rkn}{r_{k;\hspace{0.2mm}n}}
\newcommand{\Xkn}{X_{k;\hspace{0.2mm}n}}
\newcommand{\Ykn}{Y_{k;\hspace{0.2mm}n}}

\newcommand{\imme}{\mathfrak{j}\hspace{0.3mm}}
\newcommand{\immer}{\mathfrak{j}_{\hspace{0.3mm}\mbox{\tiny $\mathbb{R}$}}}
\newcommand{\immers}{\mathsf{k}\hspace{0.2mm}}
\newcommand{\nimm}{\mathsf{K}\hspace{0.2mm}}

\newcommand{\aekm}{\eta_m^k}
\newcommand{\afkm}{\chi_m^k}
\newcommand{\cekn}{\phi_n^k}
\newcommand{\cfkn}{\psi_n^k}

\newcommand{\tkm}{t_m^k}
\newcommand{\ukn}{u_n^k}

\newcommand{\heemn}{\widehat{\aekm\cekn}}
\newcommand{\hffmn}{\widehat{\afkm\cfkn}}

\newcommand{\ka}{\mathsf{K}}
\newcommand{\je}{\mathsf{J}}
\newcommand{\elle}{\mathsf{L}}
\newcommand{\emme}{\mathsf{M}}
\newcommand{\enne}{\mathsf{N}}

\newcommand{\optp}{\mathfrak{C}}
\newcommand{\opts}{\mathfrak{R}}

\newcommand{\clnor}{\cl_{\|\hspace{0.3mm}\cdot\hspace{0.3mm}\|}}
\newcommand{\cltrnor}{\cl_{\|\hspace{0.3mm}\cdot\hspace{0.3mm}\trn}}

\newcommand{\clispattrn}{\overline{\mathrm{span}}^{\hspace{0.3mm}\|\hspace{0.3mm}\cdot\hspace{0.3mm}\ttrn\hspace{-0.4mm}}}
\newcommand{\clispat}{\overline{\mathrm{span}}^{\hspace{0.3mm}\sotim}\hspace{-0.3mm}}
\newcommand{\clispaptrn}{\overline{\mathrm{span}}^{\hspace{0.3mm}\|\hspace{0.3mm}\cdot\hspace{0.3mm}\ptrn\hspace{-0.4mm}}}
\newcommand{\clispap}{\overline{\mathrm{span}}^{\hspace{0.3mm}\sotimp}\hspace{-0.3mm}}

\newcommand{\vp}{\varpi}
\newcommand{\va}{\mathfrak{a}}
\newcommand{\vc}{\mathfrak{c}}
\newcommand{\hvac}{\widehat{\hspace{0.3mm}\va\vc\hspace{0.2mm}}}
\newcommand{\vac}{|\va\rangle\hspace{-0.3mm}\langle\vc|}
\newcommand{\vak}{\mathfrak{a}_k}
\newcommand{\vck}{\mathfrak{c}_k}
\newcommand{\vek}{\mathfrak{e}_k}
\newcommand{\hvack}{\widehat{\hspace{-0.4mm}\phantom{.}\vak\vck}}

\newcommand{\vack}{|\vak\rangle\hspace{-0.3mm}\langle\vck|}

\newcommand{\hva}{\widehat{\hspace{0.3mm}\va\hspace{0.2mm}}}
\newcommand{\hvaa}{\widehat{\hspace{0.3mm}\va\va\hspace{0.2mm}}}
\newcommand{\vaa}{|\va\rangle\hspace{-0.3mm}\langle\va|}
\newcommand{\vaan}{|\va_n\rangle\hspace{-0.3mm}\langle\va_n|}
\newcommand{\veek}{|\vek\rangle\hspace{-0.3mm}\langle\vek|}

\newcommand{\vamn}{|\va_m\rangle\hspace{-0.3mm}\langle\va_n|}

\newcommand{\vccn}{|\vc_\enne\rangle\hspace{-0.3mm}\langle\vc_\enne|}

\newcommand{\fr}{\mathscr{F}}
\newcommand{\frh}{\fr(\hh)}
\newcommand{\frsah}{\fr(\hh)_{\mbox{\tiny $\mathbb{R}$}}}
\newcommand{\frj}{\fr(\jj)}
\newcommand{\frsaj}{\fr(\jj)_{\mbox{\tiny $\mathbb{R}$}}}
\newcommand{\frhj}{\fr(\hhjj)}

\newcommand{\bs}{\mathfrak{S}}
\newcommand{\bsd}{{\mathfrak{S}^\ast}}
\newcommand{\norbs}{\|\cdot\|_\mathfrak{S}}
\newcommand{\bsb}{\mathfrak{V}}
\newcommand{\bsc}{\mathfrak{Z}}
\newcommand{\bsbd}{{\mathfrak{V}^\ast}}

\newcommand{\sbs}{\mathfrak{B}}

\newcommand{\bim}{\lambda}
\newcommand{\sotim}{\mbox{\tiny $\otimes$}}
\newcommand{\sotima}{\mbox{\tiny $\breve{\otimes}$}}
\newcommand{\sotimp}{\mbox{\tiny $\widehat{\otimes}$}}
\newcommand{\Bm}{\Lambda}
\newcommand{\Bim}{\Lambda^{\hspace{-0.5mm}\sotim\hspace{-0.3mm}}}
\newcommand{\Bima}{\Lambda^{\hspace{-0.5mm}\sotima\hspace{-0.3mm}}}
\newcommand{\Bimp}{\Lambda^{\hspace{-0.5mm}\sotimp\hspace{-0.3mm}}}

\newcommand{\lispa}{\mathrm{span}\hspace{0.5mm}}
\newcommand{\clispa}{\overline{\mathrm{span}}\hspace{0.5mm}}

\newcommand{\lispar}{\mathrm{span}_{\hspace{0.3mm}\mbox{\tiny $\mathbb{R}$}}\hspace{-0.2mm}}
\newcommand{\clispar}{\overline{\mathrm{span}}_{\hspace{0.3mm}\mbox{\tiny $\mathbb{R}$}}\hspace{-0.2mm}}

\newcommand{\lispra}{\mathrm{span}_{\hspace{0.3mm}\mathbb{Q}}}
\newcommand{\lispranb}{\lispra\{\nb(\trco,\trcjjo)\}}

\newcommand{\co}{\mathrm{co}\hspace{0.3mm}}
\newcommand{\clco}{\overline{\mathrm{co}}\hspace{0.3mm}}
\newcommand{\clcon}{\overline{\mathrm{co}}^{\hspace{0.3mm}\sotim}\hspace{-0.3mm}}
\newcommand{\prm}{\mathscr{P}}

\newcommand{\clcop}{\overline{\mathrm{co}}^{\hspace{0.3mm}\sotimp}\hspace{-0.3mm}}

\newcommand{\rah}{\rangle_{\hh}}
\newcommand{\raj}{\rangle_{\jj}}
\newcommand{\rahaj}{\rangle_{\hh\altp\hspace{-0.3mm}\jj}}
\newcommand{\scap}{\langle\hspace{0.5mm}\cdot\hspace{0.6mm},\cdot\hspace{0.5mm}\rangle}
\newcommand{\scaph}{\langle\hspace{0.5mm}\cdot\hspace{0.6mm},\cdot\hspace{0.5mm}\rangle_{\hh}}
\newcommand{\scaphj}{\langle\hspace{0.5mm}\cdot\hspace{0.6mm},\cdot\hspace{0.5mm}\rangle_{\hh\otimes\jj}}
\newcommand{\scaphaj}{\langle\hspace{0.5mm}\cdot\hspace{0.6mm},\cdot\hspace{0.5mm}\rangle_{\hh\altp\hspace{-0.3mm}\jj}}

\newcommand{\no}{\|\cdot\|}

\newcommand{\notri}{\ntr{\hspace{0.6mm}\cdot\hspace{0.6mm}}}

\newcommand{\nops}{\|\cdot\|_{\mathcal{B}}}
\newcommand{\norps}{\|_{\mathcal{B}}}
\newcommand{\bnorps}{\big\|_{\mathcal{B}}}

\newcommand{\subc}{\mathcal{C}}
\newcommand{\subh}{\mathcal{X}}
\newcommand{\subj}{\mathcal{Y}}
\newcommand{\subhj}{\mathcal{Z}}

\newcommand{\nb}{\theta}
\newcommand{\snb}{\breve{\theta}}
\newcommand{\slma}{\breve{\Theta}}

\newcommand{\ub}{\widehat{\theta}}

\newcommand{\hph}{\widehat{\phi}}
\newcommand{\hphk}{\widehat{\phi_k}}
\newcommand{\hpph}{\widehat{\phi\phi}}
\newcommand{\hp}{\widehat{\psi}}
\newcommand{\hpu}{\widehat{\psi_1}}
\newcommand{\hpd}{\widehat{\psi_2}}
\newcommand{\hpp}{\widehat{\psi\psi}}

\newcommand{\hep}{\widehat{\eta\phi}}
\newcommand{\ep}{|\eta\rangle\hspace{-0.3mm}\langle\phi|}

\newcommand{\chp}{|\chi\rangle\hspace{-0.3mm}\langle\psi|}
\newcommand{\ecpp}{|\eta\otimes\chi\rangle\hspace{-0.3mm}\langle\phi\otimes\psi|}
\newcommand{\pck}{|\psi_k\rangle\hspace{-0.3mm}\langle\chi_k|}
\newcommand{\pjp}{|\phi\rangle\hspace{-0.3mm}\langle J\hspace{0.3mm}\psi|}
\newcommand{\hcp}{\widehat{\chi\psi}}

\newcommand{\ink}{{\mbox{\tiny $(k)$}}}
\newcommand{\inkn}{{\mbox{\tiny $(k(n))$}}}
\newcommand{\sjk}{s_j^\ink}
\newcommand{\tlk}{t_l^\ink}
\newcommand{\ejk}{\eta_j^\ink}
\newcommand{\pjk}{\phi_j^\ink}
\newcommand{\plk}{\psi_l^\ink}
\newcommand{\clk}{\chi_l^\ink}
\newcommand{\epjk}{\big|\eta_j^\ink\big\rangle\hspace{-0.3mm}\big\langle\phi_j^\ink\big|}
\newcommand{\pclk}{\big|\psi_l^\ink\big\rangle\hspace{-0.3mm}\big\langle\chi_l^\ink\big|}

\newcommand{\rknn}{r_{k(n)}}
\newcommand{\sjmk}{s_{j(m)}^\ink}
\newcommand{\tlmk}{t_{l(m)}^\ink}
\newcommand{\epjmk}{\big|\eta_{j(m)}^\ink\big\rangle\hspace{-0.3mm}\big\langle\phi_{j(m)}^\ink\big|}
\newcommand{\pclmk}{\big|\psi_{l(m)}^\ink\big\rangle\hspace{-0.3mm}\big\langle\chi_{l(m)}^\ink\big|}
\newcommand{\sjmkn}{s_{j(m(n))}^\inkn}
\newcommand{\tlmkn}{t_{l(m(n))}^\inkn}
\newcommand{\epjmkn}{\big|\eta_{j(m(n))}^\inkn\big\rangle\hspace{-0.3mm}\big\langle\phi_{j(m(n))}^\inkn\big|}
\newcommand{\pclmkn}{\big|\psi_{l(m(n))}^\inkn\big\rangle\hspace{-0.3mm}\big\langle\chi_{l(m(n))}^\inkn\big|}

\newcommand{\hhpu}{|\psi_1\rangle\hspace{-0.3mm}\langle\psi_1|}
\newcommand{\phk}{|\phi_k\rangle\hspace{-0.3mm}\langle\phi_k|}
\newcommand{\phikl}{|\phi_k\rangle\hspace{-0.3mm}\langle\phi_l|}

\newcommand{\psimn}{|\psi_m\rangle\hspace{-0.3mm}\langle\psi_n|}
\newcommand{\psikl}{|\psi_k\rangle\hspace{-0.3mm}\langle\psi_l|}

\newcommand{\epn}{|\eta_n\rangle\hspace{-0.3mm}\langle\phi_n|}
\newcommand{\pcn}{|\psi_n\rangle\hspace{-0.3mm}\langle\chi_n|}

\newcommand{\epk}{|\eta_k\rangle\hspace{-0.3mm}\langle\phi_k|}

\newcommand{\phipsikl}{|\phi_k\otimes\psi_k\rangle\hspace{-0.3mm}\langle\phi_l\otimes\psi_l|}

\newcommand{\cpl}{|\chi_l\rangle\hspace{-0.3mm}\langle\psi_l|}

\newcommand{\bo}{\mathcal{B}}

\newcommand{\hh}{\mathcal{H}}
\newcommand{\trhh}{\tr_{\hh}}

\newcommand{\hhd}{\mathcal{H}^\ast}

\newcommand{\trc}{\mathcal{B}_1(\hh)}
\newcommand{\trcsa}{\mathcal{B}_1(\hh)_{\mbox{\tiny $\mathbb{R}$}}}
\newcommand{\trcd}{\mathcal{B}_1(\hh)^\ast}

\newcommand{\trchjsa}{\mathcal{B}_1(\hh\otimes\jj)_{\mbox{\tiny $\mathbb{R}$}}}
\newcommand{\trcp}{\mathcal{B}_1(\hh)_{\mbox{\tiny $+$}}}
\newcommand{\trcjjp}{\mathcal{B}_1(\jj)_{\mbox{\tiny $+$}}}

\newcommand{\bop}{\mathcal{B}(\hh)}
\newcommand{\bopsa}{\mathcal{B}(\mathcal{H})_{\mbox{\tiny $\mathbb{R}$}}}
\newcommand{\bhh}{\bo(\hh;\hh)}
\newcommand{\sta}{\mathcal{D}(\hh)}
\newcommand{\stajj}{\mathcal{D}(\jj)}
\newcommand{\pusta}{\mathcal{P}(\hh)}
\newcommand{\pustajj}{\mathcal{P}(\jj)}
\newcommand{\pustahj}{\mathcal{P}(\hh\otimes\jj)}

\newcommand{\trco}{\mathcal{B}_1(\hh)_0}

\newcommand{\bophj}{\mathcal{B}(\hhjj)}
\newcommand{\bophjsa}{\mathcal{B}(\hhjj)_{\mbox{\tiny $\mathbb{R}$}}}
\newcommand{\bophjg}{\bophj_{\Gamma}}

\newcommand{\bopvw}{\mathcal{B}(\vv\otimes\ww)}

\newcommand{\dbop}{\mathcal{B}(\hh)^\ast}
\newcommand{\dbopjj}{\mathcal{B}(\jj)^\ast}
\newcommand{\dbophj}{\mathcal{B}(\hhjj)^\ast}

\newcommand{\jj}{\mathcal{J}}

\newcommand{\jjd}{\mathcal{J}^\ast}

\newcommand{\deta}{\langle\eta,\cdot\hspace{0.5mm}\rangle_{\hh}}
\newcommand{\dchi}{\langle\chi,\cdot\hspace{0.5mm}\rangle_{\hspace{-0.4mm}\jj}}
\newcommand{\etap}{\eta^\ast}
\newcommand{\chip}{\chi^\ast}

\newcommand{\argo}{(\hspace{0.3mm}\cdot\hspace{0.3mm})}

\newcommand{\rhh}{\rangle_{\hh}}

\newcommand{\rjj}{\rangle_{\hspace{-0.4mm}\jj}}

\newcommand{\rhj}{\rangle_{\hh\otimes\jj}}

\newcommand{\trcjj}{\mathcal{B}_1(\jj)}
\newcommand{\trcsajj}{\mathcal{B}_1(\jj)_{\mbox{\tiny $\mathbb{R}$}}}
\newcommand{\hhjj}{\hh\otimes\jj}
\newcommand{\trjj}{\tr_{\hspace{-0.3mm}\jj}}
\newcommand{\trchj}{\mathcal{B}_1(\hh\otimes\jj)}
\newcommand{\trchjd}{\mathcal{B}_1(\hh\otimes\jj)^\ast}
\newcommand{\trchjp}{\mathcal{B}_1(\hh\otimes\jj)_{\mbox{\tiny $+$}}}
\newcommand{\bopjj}{\mathcal{B}(\jj)}
\newcommand{\trcjjo}{\mathcal{B}_1(\jj)_0}

\newcommand{\htrchj}{{\mathcal{B}_1(\hh\otimes\jj)\hspace{-0.4mm}\widehat{\phantom{x}}}}
\newcommand{\htrchjs}{{\mathcal{B}_1(\hh\otimes\jj)\hspace{-0.4mm}\widehat{\phantom{x}}_{\hspace{-1.7mm}\mbox{\tiny $\mathbb{R}$}}}}

\newcommand{\htrcjh}{{\mathcal{B}_1(\jj\otimes\hh)\hspace{-0.4mm}\widehat{\phantom{x}}}}
\newcommand{\htrcjhs}{{\mathcal{B}_1(\jj\otimes\hh)\hspace{-0.4mm}\widehat{\phantom{x}}_{\hspace{-1.7mm}\mbox{\tiny $\mathbb{R}$}}}}

\newcommand{\stahj}{\mathcal{D}(\hh\otimes\jj)}
\newcommand{\stahjr}{\mathcal{D}(\hh\otimes\jj)_r}
\newcommand{\stavw}{\mathcal{D}(\vv\otimes\ww)}

\newcommand{\hstahj}{{\mathcal{D}(\hh\otimes\jj)\hspace{-0.4mm}\widehat{\phantom{x}}}}
\newcommand{\hpustahj}{{\mathcal{P}(\hh\otimes\jj)\hspace{-0.4mm}\widehat{\phantom{x}}}}

\newcommand{\bopj}{\mathcal{B}(\jj)}

\newcommand{\ltrcjjau}{\mathscr{L}_1(\hh,\jj;\hau)}
\newcommand{\ftrcjj}{\mathscr{F}_1(\hh,\jj)}

\newcommand{\lthtjthj}{\mathscr{L}_1(\hh,\jj;\hhjj)}

\newcommand{\lin}{\mathcal{L}}
\newcommand{\linhh}{\lin(\hh)}
\newcommand{\linjj}{\lin(\jj)}
\newcommand{\linhj}{\lin(\hh,\jj;\ccc)}
\newcommand{\linhjd}{\lin(\hhd,\jjd;\ccc)}

\newcommand{\linhjad}{\lin(\hh,\jj;\ccc)^\prime}

\newcommand{\bfshj}{\bo(\hh,\jj;\ccc)}

\newcommand{\linhau}{\lin(\hau)}

\newcommand{\sesta}{{\mathcal{D}(\hh\otimes\jj)_{\hspace{-0.1mm}\mbox{\tiny $\mathsf{se}$}\hspace{-0.2mm}}}}
\newcommand{\sestap}{{\mathcal{P}(\hh\otimes\jj)_{\hspace{-0.1mm}\mbox{\tiny $\mathsf{se}$}\hspace{-0.2mm}}}}

\newcommand{\dsta}{{\mathcal{D}(\hh\otimes\jj)_{\hspace{-0.1mm}\mbox{\tiny $\mathsf{ds}$}\hspace{-0.2mm}}}}
\newcommand{\fsta}{{\mathcal{D}(\hh\otimes\jj)_{\hspace{-0.1mm}\mbox{\tiny $\mathsf{f\hspace{0mm}s}$}\hspace{-0.2mm}}}}
\newcommand{\ndsta}{{\mathcal{D}(\hh\otimes\jj)\hspace{-0.4mm}\widehat{\hspace{0.2mm}\mbox{\tiny $\pm$}}}}
\newcommand{\csta}{{\mathcal{D}(\hh\otimes\jj)_{\hspace{-0.1mm}\mbox{\tiny $\mathsf{cs}$}\hspace{-0.2mm}}}}

\newcommand{\trn}{\|_1}
\newcommand{\trnor}{\|\cdot\trn}

\newcommand{\tcn}{\|_{\mathrm{tr}}}
\newcommand{\trano}{\|\cdot\tcn}

\newcommand{\trf}{\tr\argo}

\newcommand{\norb}{\|_{(1)}}
\newcommand{\norbil}{\|\cdot\norb}

\newcommand{\tnori}{\|_{\hspace{-0.3mm}\sinft}^{\hspace{-0.3mm}\sotim}}
\newcommand{\btnori}{{\big\|}_{\hspace{-0.3mm}\sinft}^{\hspace{-0.3mm}\sotim}}
\newcommand{\tnormi}{\|\cdot\tnori}

\newcommand{\ttrn}{\|_{1}^{\hspace{-0.3mm}\sotim}}
\newcommand{\ttrnor}{\|\cdot\ttrn}
\newcommand{\ttrnorli}{\mbox{\scriptsize $\|\hspace{0.2mm}\cdot\hspace{0.2mm}\ttrn$--}\lim}
\newcommand{\trnorli}{\mbox{\scriptsize $\|\hspace{0.2mm}\cdot\hspace{0.2mm}\trn$--}\lim}

\newcommand{\cl}{\mathrm{cl}\hspace{0.2mm}}
\newcommand{\clttrnor}{\cl_{\|\hspace{0.3mm}\cdot\hspace{0.3mm}\ttrn}}

\newcommand{\ptrn}{\|_{1}^{\hspace{-0.3mm}\sotimp}}
\newcommand{\ptrnor}{\|\cdot\ptrn}
\newcommand{\ptrnorli}{\mbox{\scriptsize $\|\hspace{0.2mm}\cdot\hspace{0.2mm}\ptrn$--}\lim}
\newcommand{\clptrnor}{\cl_{\|\hspace{0.3mm}\cdot\hspace{0.3mm}\ptrn}}

\newcommand{\bptrn}{\big\|_{1}^{\hspace{-0.3mm}\sotimp}}

\newcommand{\trnb}{\big\|_{1}}

\newcommand{\elledue}{\mathrm{L\hspace{-0.2mm}}^{\mbox{\tiny $2$}}}
\newcommand{\elleds}{\elledue([0,1])}
\newcommand{\lem}{\beta}
\newcommand{\elledsc}{\elledue([0,1],\lem;\ccc)}

\newcommand{\de}{\mathrm{d}\hspace{0.2mm}}
\newcommand{\eee}{\mathrm{e}}

\newcommand{\hau}{\mathcal{K}}

\newcommand{\bopp}{\mathcal{B}(\hh)_{\mbox{\tiny $+$}}}

\newcommand{\trcau}{\mathcal{B}_1(\hau)}

\newcommand{\nori}{\|_{\sinft}}
\newcommand{\normi}{\|\cdot\|_{\sinft}}
\newcommand{\sotinj}{\mbox{\tiny $\injtp$}}
\newcommand{\innori}{\|_{\hspace{-0.3mm}\sinft}^{\hspace{-0.3mm}\sotinj}}
\newcommand{\binnori}{{\big\|}_{\hspace{-0.3mm}\sinft}^{\hspace{-0.3mm}\sotinj}}
\newcommand{\innormi}{\|\cdot\innori}

\newcommand{\pp}{\mathsf{p}}

\newcommand{\normp}{\|\cdot\|_{\pp}}

\newcommand{\funs}{\hspace{0.3mm}\mathfrak{s}\hspace{0.3mm}}
\newcommand{\funt}{\hspace{0.3mm}\mathfrak{t}\hspace{0.3mm}}

\newcommand{\xu}{X_{\mbox{\tiny $\Re$}}}
\newcommand{\xd}{X_{\mbox{\tiny $\Im$}}}
\newcommand{\yu}{Y_{\mbox{\tiny $\Re$}}}
\newcommand{\yd}{Y_{\mbox{\tiny $\Im$}}}
\newcommand{\xuk}{X_{k;\hspace{0.2mm}\mbox{\tiny $\Re$}}}
\newcommand{\xdk}{X_{k;\hspace{0.2mm}\mbox{\tiny $\Im$}}}
\newcommand{\yuk}{Y_{k;\hspace{0.2mm}\mbox{\tiny $\Re$}}}
\newcommand{\ydk}{Y_{k;\hspace{0.2mm}\mbox{\tiny $\Im$}}}

\newcommand{\rko}{r_{k;\hspace{0.2mm}1}^\epsilon}
\newcommand{\rlt}{r_{l;\hspace{0.2mm}2}^\epsilon}
\newcommand{\rkj}{r_{k;\hspace{0.2mm}j}^\epsilon}
\newcommand{\xko}{X_{k;\hspace{0.2mm}1}^\epsilon}
\newcommand{\xlt}{X_{l;\hspace{0.2mm}2}^\epsilon}
\newcommand{\xkj}{X_{k;\hspace{0.2mm}j}^\epsilon}
\newcommand{\yko}{Y_{k;\hspace{0.2mm}1}^\epsilon}
\newcommand{\ylt}{Y_{l;\hspace{0.2mm}2}^\epsilon}
\newcommand{\ykj}{Y_{k;\hspace{0.2mm}j}^\epsilon}

\newcommand{\pkm}{p_{k;\hspace{0.2mm}m}}
\newcommand{\qkn}{q_{k;\hspace{0.2mm}n}}
\newcommand{\pikm}{\pi_{k;\hspace{0.2mm}m}}
\newcommand{\varpikn}{\varpi_{k;\hspace{0.2mm}n}}

\newcommand{\apl}{A_{\mbox{\tiny $+$}}}
\newcommand{\aplu}{a_{\mbox{\tiny $+$}}}
\newcommand{\depl}{D_{\mbox{\tiny $+$}}}
\newcommand{\ami}{A_{\mbox{\tiny $-$}}}
\newcommand{\demi}{D_{\mbox{\tiny $-$}}}
\newcommand{\amin}{a_{\mbox{\tiny $-$}}}

\newcommand{\xjp}{x_{j;\mbox{\tiny $+$}}}
\newcommand{\xjm}{x_{j;\mbox{\tiny $-$}}}
\newcommand{\yjp}{y_{j;\mbox{\tiny $+$}}}
\newcommand{\yjm}{y_{j;\mbox{\tiny $-$}}}
\newcommand{\rjp}{\rho_{j;\mbox{\tiny $+$}}}
\newcommand{\rjm}{\rho_{j;\mbox{\tiny $-$}}}
\newcommand{\sjp}{\sigma_{j;\mbox{\tiny $+$}}}
\newcommand{\sjm}{\sigma_{j;\mbox{\tiny $-$}}}

\newcommand{\tjpp}{s_{j;\mbox{\tiny $++$}}}
\newcommand{\tjpm}{s_{j;\mbox{\tiny $+-$}}}
\newcommand{\tjmp}{s_{j;\mbox{\tiny $-+$}}}
\newcommand{\tjmm}{s_{j;\mbox{\tiny $--$}}}

\newcommand{\nat}{\mathbb{N}}

\newcommand{\trdd}{\mathsf{d}_{1,1}}

\newcommand{\supp}{\mathrm{supp}\hspace{0.2mm}}

\begin{document}

\title{The cross states of a composite quantum system: separability and entanglement in any Hilbert space dimension}

\author{Paolo Aniello$^{1,2}$\thanks{Email: paolo.aniello@na.infn.it}   \vspace{2mm}
\\
\small \it   $^1$Dipartimento di Fisica ``Ettore Pancini'', Universit\`a di Napoli ``Federico II'',  \\
\small \it   Complesso Universitario di Monte S.\ Angelo, via Cintia, I-80126 Napoli, Italy  \vspace{1mm}
\\
\small \it   $^2$Istituto Nazionale di Fisica Nucleare, Sezione di Napoli,  \\
\small \it   Complesso Universitario di Monte S.\ Angelo, via Cintia, I-80126 Napoli, Italy}

\date{}

\maketitle


\begin{abstract}
\noindent
We introduce a class of states of a composite quantum system --- the so-called \emph{cross states}
--- that turn to play a major role in the theory of entanglement for a genuinely infinite-dimensional
bipartite system. In the case where at least one of the Hilbert spaces of the bipartition is
finite-dimensional, all states are cross states, whereas, in the genuinely infinite-dimensional setting
where the dimension of \emph{both} Hilbert spaces is \emph{not} finite, the cross states form a
trace-norm dense, convex, proper subset of the set of all states. In the latter case, the cross states
can be regarded as those physical states that possess a finite amount of entanglement; accordingly, all
separable states are of this kind. We prove that, for any Hilbert space dimension, the separable states
can be characterized as those cross states that minimize a suitable norm, i.e., the \emph{projective norm}
associated with the projective tensor product of two trace classes; all other cross states are
density operators belonging to the projective tensor product space. This is a generalization of the
classical \emph{cross norm criterion} of separability. Finally, we define an extended real-valued
\emph{entanglement function} and study its main properties. Coherently with the interpretation of cross
states as finitely entangled states, this function is finite, and coincides with the projective norm,
precisely on the cross states of the system.
\end{abstract}


\section{Introduction}
\label{intro}

Entanglement is one of the most distinguishing features of quantum versus classical mechanics~\cite{Horodecki,Guhne}.
Among the manifold consequences of this very fundamental aspect of quantum theory, it is worth mentioning
the applications in the context of quantum information science~\cite{Nielsen}. However, to capture the intimate
essence of entanglement is a highly nontrivial task, and it turns out that actually, in order to obtain a profound
understanding of this exquisitely quantum phenomenon, a rigorous mathematical approach cannot be dispensed with.
This fact, on the one hand, probably explains why a finite-dimensional setting --- where several intricacies
associated with the machinery of infinite-dimensional Hilbert spaces become much milder, or simply disappear
--- is often assumed in the literature~\cite{Aubrun,Dariano}. On the other hand, whereas certain quantum systems
are sufficiently well described by resorting to an effective finite-dimensional model, the understanding of entanglement
of physical states in \emph{any} Hilbert space dimension is of central importance, e.g., in the context of quantum
measurement theory or in the study of open quantum systems~\cite{Holevo,Heino,Busch}.

Recall that, in the standard formulation of quantum mechanics~\cite{Holevo,Heino,Busch,Prugovecki,Emch,Moretti},
the two main families of fundamental physical entities --- i.e., \emph{states} and \emph{observables}
--- can be realized as suitable classes of linear operators on a separable complex Hilbert space $\hh$.
Specifically, the (normal, or $\sigma$-additive) quantum states on $\hh$ form a distinguished subset of the complex
Banach space $\trc$ of trace class operators --- namely, the convex set $\sta$ of \emph{density operators} ---
whereas the elements of the selfadjoint part of the $\cast$-algebra $\bop$ of all bounded operators on $\hh$
play the role of (bounded) observables. The pairing between states (density operators) and observables (bounded
selfadjoint operators), which provides the relevant measurement probability distributions and expectation values,
is implemented by the \emph{trace} functional.

This picture becomes somewhat more complicated as soon as one considers the mathematical modeling of \emph{two mutually
interacting quantum systems}. In this setting, according to one of the fundamental axioms of quantum mechanics, the
description of the resulting \emph{bipartite composite quantum system} involves the \emph{tensor product} $\hhjj$ of two
(separable, complex) Hilbert spaces $\hh$ and $\jj$~\cite{Holevo,Heino,Busch,Prugovecki,Moretti}, and it is precisely at
this stage that entanglement enters into the game. If one restricts to considering \emph{pure states} (i.e., rank-one
projections) only, since these states can be described in terms of \emph{state vectors} in the Hilbert space $\hhjj$,
then the intricacies of quantum entanglement are mainly \emph{conceptual} rather than \emph{mathematical}, because
the notion of tensor product of Hilbert spaces is fairly simple and straightforward~\cite{Prugovecki,Busch,Moretti}.
The \emph{separable pure states} are precisely the rank-one projections associated with (normalized nonzero) vectors
in $\hhjj$ of the elementary factorized form $\phi\otimes\psi$, whereas all other pure states are \emph{entangled}.

If, instead, one is interested in considering the \emph{mixed states} of the composite system as well (e.g.,
in the case where the \emph{decoherence}~\cite{Breuer} effects cannot be neglected), then the relevant operator spaces
associated with the product Hilbert space $\hhjj$ must be properly defined and studied. It is an interesting fact,
incidentally, that a precise and complete definition of separability/entanglement, in the general case of mixed states,
seems to have appeared rather late in the literature~\cite{Werner}. According to this (physically grounded) definition,
the \emph{separable mixed states} are either product states of the form $\rho\otimes\sigma$ --- where $\rho$, $\sigma$
are density operators --- or suitable (generalized) statistical \emph{mixtures} of such states; all other states are
entangled. This definition is also coherent with the central role played by \emph{convexity}~\cite{Gudder} in the
context of quantum theory.

Once again, one may avoid most mathematical complications by simply assuming that the Hilbert spaces $\hh$ and
$\jj$ are finite-dimensional. However, as previously observed, there are cases where one is forced to leave
this relatively safe playground and face the most general situation. The main aim of the present contribution
is to achieve a complete characterization of the convex set of all \emph{separable} (i.e., non-entangled) states
of a bipartite quantum system --- for any Hilbert space dimension of the subsystems --- in terms of the
\emph{most fundamental structures} involved in the mathematical formalization of a composite quantum
system, i.e., the tensor product of the relevant operator (Banach) spaces and the canonically associated
norms~\cite{Grothendieck-res,Grothendieck-book,Defant,Ryan,DFS,Kubrusly}. As a byproduct, we will also obtain
a natural `measure of entanglement'.

Let us briefly outline the fundamental problem and the main ideas at the basis of our work.

Our primary scope is to characterize the convex set $\sesta$ of all separable states of a bipartite quantum system,
whose carrier Hilbert space $\hhjj$ is the tensor product of the `local' Hilbert spaces $\hh$ and $\jj$.
This set is defined as the \emph{closed convex hull} of the set of product states, i.e.,
\begin{equation}
\sesta\defi\clco\{\rho\otimes\sigma\colon\rho\in\sta,\ \sigma\in\stajj\} \fin .
\end{equation}
Here, the convex hull is supposed to be closed wrt the trace norm $\ttrnor$ of the ambient Banach space
$\trchj$ --- the trace class (or Schatten 1-class) on $\hhjj$ --- the bipartite states $\stahj$ belong to.

At this point, it is natural to consider, at first, the \emph{algebraic tensor product} $\thatj$ of
the `local' (or factor) trace classes $\trc$ and $\trcjj$, that can be identified with the complex
linear span of the elementary tensor products $\rho\otimes\sigma$  of `local' density operators.
It can be easily shown that, by taking the completion, wrt the norm $\ttrnor$, of this linear space
one obtains precisely the Banach space $\trchj$, i.e., the \emph{bipartite trace class}.

It should be noted, however, that, if one regards the factor trace classes $\trc$ and $\trcjj$ as
\emph{abstract} Banach spaces, then completing the algebraic tensor product $\thatj$ wrt the trace norm
$\ttrnor$ is not a very natural choice. In fact, in order to obtain a Banach space out of the algebraic tensor
product of two abstract Banach spaces the most natural choices turn out to be the \emph{projective tensor product}
and the \emph{injective tensor product}~\cite{Grothendieck-res,Grothendieck-book,Defant,Ryan,DFS,Kubrusly}.
Both types of tensor products will play some role in the following; however, in some sense, the first of these
two options is a particularly natural choice when dealing with trace class operators, because the trace class
$\trc$ itself can be regarded as --- i.e., it is isomorphic to --- the projective tensor product of two copies
of the Hilbert space $\hh$, that we will denote by $\hh\prtp\hh$.

Let us then focus on the projective tensor product $\thptj$, that we will call the \emph{cross trace class};
i.e., the Banach space completion of the algebraic tensor product $\thatj$ wrt the \emph{projective norm}
--- or \emph{greatest cross norm} --- $\ptrnor$ (not to be confused with the trace norm $\ttrnor$ of $\trchj$)
associated with the trace norm $\trnor$ of the `local' trace classes $\trc$ and $\trcjj$. It has been first
discovered by Rudolph~\cite{Rudolph,Rudolph-further}, who considered the case where both the Hilbert spaces
of the bipartition $\hhjj$ are finite-dimensional, that a density operator $D\in\stahj$ is separable iff
$\|D\ptrn=1$; otherwise stated, that the convex set $\sesta$ is the intersection of the unit sphere in
$\thptj$ with the set $\stahj$ of all density operators on $\hhjj$. This important discovery has
attracted a certain interest in the characterization of entanglement by means of the projective
norm~\cite{Sokoli,Jivulescu}. A fundamental step forward in the understanding of the central role of the
projective norm in the study of entanglement is due to Arveson~\cite{Arveson}, who, using elegant methods
of convex analysis and of the theory of operator algebras,  has proved --- among other things --- that
Rudolph's characterization of separability remains true in the more general case where at least one of
the Hilbert spaces of the bipartition $\hhjj$ (or, more generally, of a multipartite Hilbert space) is
finite-dimensional. It is then natural to wonder what happens in the most general case where no restriction
on the Hilbert space dimension is assumed. This is a nontrivial issue, because it turns out that, actually,
there is a sharp difference between the case where the at least one of the Hilbert spaces of the bipartition
is finite-dimensional and the \emph{genuinely infinite-dimensional case}.

In fact, the following observation is of central importance: Whereas, if at least one of the Hilbert spaces
$\hh$ and $\jj$ is finite-dimensional, the \emph{cross} trace class $\thptj$ can be identified --- \emph{as a set},
or as a linear space --- with the trace class $\trchj$, in the case where $\dim(\hh)=\dim(\jj)=\infty$, instead,
$\thptj$ can be realized \emph{as a proper subset} (or, more precisely, as a proper linear subspace) of $\trchj$.

Specifically, one can prove (and we will actually provide proofs of) the following facts:

\begin{enumerate}

\item In the case where $\dim(\hh)<\infty$ and/or $\dim(\jj)<\infty$, the norms $\ttrnor$ (trace norm) and
$\ptrnor$ (projective norm) are \emph{mutually equivalent} on the linear space $\trchj=\thptj$ (set equality,
or isomorphism of linear spaces, only). Otherwise stated, the Banach space $\thptj$ can be regarded as a
\emph{renorming} of the trace class $\trchj$. However, whereas `all bipartite states are equal' as far as the
trace norm is concerned --- namely, for every density operator $D\in\stahj$, $\|D\ttrn=1$ --- the projective norm
is able to \emph{efficiently detect entanglement}; precisely, $\|D\ptrn\ge 1$, and $D$ is entangled iff $\|D\ptrn>1$.

\item In the case where $\dim(\hh)=\dim(\jj)=\infty$ (i.e., in the \emph{genuinely infinite-dimensional} setting),
instead, even if the projective tensor product $\thptj$ admits a natural continuous embedding in $\trchj$ ---
actually, it is a \emph{dense} (wrt the trace norm) linear subspace of $\trchj$ --- $\thptj\subsetneq\trchj$.
Accordingly, the projective norm $\ptrnor$ and the restriction of the trace norm $\ttrnor$ to the cross trace
class $\thptj$ are \emph{not} equivalent (the latter being simply pointwise dominated by the former). In particular,
focusing on the set $\hstahj$ of all \emph{cross states} --- namely, the convex set of bipartite states defined
by $\hstahj\defi\stahj\cap\thptj$ --- one can prove that, in the genuinely infinite-dimensional setting,
$\hstahj\subsetneq\stahj$ (precisely, the cross states form a $\ttrnor$-dense proper subset of $\stahj$),
and the projective norm $\ptrnor$ is \emph{unbounded} on cross states.

\end{enumerate}

In spite of our last remark, and of the consequent technical subtleties, the projective norm turns
out to be still a venerable tool in the genuinely infinite-dimensional case, because:

\begin{enumerate}
\setcounter{enumi}{2}

\item The set $\sesta$ of all separable states is (strictly) contained in the convex set $\hstahj$
of all cross states, of which, actually, it is a distinguished convex subset.

\item In fact, the projective norm $\ptrnor$ is able to characterize the separability of bipartite states,
for any dimension of the `local' Hilbert spaces $\hh$ and $\jj$; i.e., $D\in\sesta$ iff $\|D\ptrn=1$.
This is the core fact of the \emph{Extended Cross Norm Criterion} of separability (ECNC), which is one of
the main results of our present contribution. As a byproduct of our investigation, we also get several
insights into the relations between the \emph{bipartite} trace class $\trchj$ and the \emph{cross} trace
class $\thptj$, and we obtain some important duality relations involving the Banach space $\thptj$.

\item The norm $\ptrnor$ (as well as the strictly related \emph{Hermitian projective norm} $\ptrnortr$,
see below) can be regarded as a \emph{measure of entanglement}. Accordingly, the cross states $\hstahj$
can be regarded as those bipartite states that possess a \emph{finite amount of entanglement}. In the
case where $\dim(\hh)<\infty$ and/or $\dim(\jj)<\infty$, actually all states enjoy this property ---
i.e., $\hstahj=\stahj$ --- whereas, in the genuinely infinite-dimensional setting, this is no longer true.

\end{enumerate}

The picture of cross states as finitely entangled states is supported by the following further important
points:

\begin{enumerate}
\setcounter{enumi}{5}

\item We can define an extended real-valued \emph{entanglement function} $\ent\colon\stahj\rightarrow[1,+\infty]$.
It turns out that, if restricted to the convex set $\hstahj$ of all cross states, it coincides with the projective
norm $\ptrnor$; hence, it is able to characterize quantum entanglement. We will also prove that --- in the genuinely
infinite-dimensional case where $\dim(\hh)=\dim(\jj)=\infty$, and \emph{in that case only} --- the value $+\infty$
is actually attained by $\ent$. Precisely, $\ent(D)=+\infty$ iff $D\in\stahj$ is \emph{not} a cross state. We also
show that, as one should expect, the entanglement function $\ent$ has a regular behavior wrt to restriction to
subspaces, is invariant wrt `local' unitary transformations and is well-behaved wrt the action of suitable
`local quantum maps'; i.e., such a map, by acting on a cross state, cannot increase the entanglement of this state.

\item We next prove that if $\ent(D)<+\infty$ --- i.e., in the case where $D$ is a cross state --- then the value
$\ent(D)$ of the entanglement function may be estimated, in principle with arbitrary precision, by measuring a suitable
kind of quantum observable, a so-called \emph{entanglement witness} (in this regard, however, we warn the reader
that our definition of an entanglement witness differs, even if in a nonessential way, from the convention usually
adopted in the literature~\cite{Guhne,Horodecki-bis,Terhal}).

\end{enumerate}

In order to prove our main results, after deriving several important properties of the cross trace class
(e.g., we prove that $\thptj$ is a Banach algebra), we exploit two different kinds of mathematical tools:
certain types of \emph{barycentric decompositions} of separable density operators, and some methods of
functional analysis and of the theory of operator algebras. In the first case, we have been influenced
by ideas and results of Holevo, Shirokov and Werner~\cite{HoShiWe,HoShiWe-bis,Holevo-EBC}; in the second
case, by Arveson's fundamental paper~\cite{Arveson}. It is worth mentioning, however, that Grothendieck's
groundbreaking contributions --- in particular, his celebrated \emph{R\'esum\'e}~\cite{Grothendieck-res}
--- lie at the foundations of the modern metric theory of tensor products, hence, of the applications
we are studying here.

Hoping to have conveyed some of the key points --- or, at least, the \emph{flavor} --- of the issues
we are going to investigate, it is now worth spending a few words about the meaning of our work in the
context of our present understanding of quantum theory. Even if the characterization of
separability/entanglement in terms of the projective norm may not be of much practical use, because,
in many cases of interest, this quantity is not analytically computable --- and, accordingly, various
related computable quantities have been proposed that give rise to \emph{partial} separability
criteria~\cite{Rudolph-properties,Rudolph-computable,Chen,Aniello_CI} (namely, sufficient conditions
for a bipartite quantum state to be entangled or, equivalently, necessary conditions for separability)
--- we believe that our work may shed new light on the \emph{very essence} of entanglement. In fact,
to the best of our knowledge, the projective norm is the only mathematical tool stemming \emph{directly}
(and in a natural way) from the tensor product structure that is able to provide us with a simple
characterization of separable versus entangled states, for any dimension of the Hilbert spaces forming
a composite quantum system.

In our present contribution, mainly for the sake of simplicity, we will consider the bipartite entanglement
only; we will consider the general multipartite case elsewhere.

The paper is organized as follows. In Section~\ref{preliminaries}, we collect some of the main notations
and some basic facts that will be used in the rest of the paper. Next, in Section~\ref{tensor}, we introduce
some mathematical tools that will be fundamental for the characterization of separable states; in particular,
the projective tensor product of trace classes and two `canonical decompositions' of a cross trace class operator.
The following three sections contain our main results. In particular, in Section~\ref{sepsta}, we focus on the notion
of separability of quantum states, and discuss the barycentric decompositions that characterize separable states.
In Section~\ref{further}, we complete the characterization of separable states in terms of the projective norms.
In particular, in Subsection~\ref{ecnc}, we condense in a unique statement --- the Extended Cross Norm Criterion
--- some of our main findings. Since the characterization of separability in terms of the projective norms crucially
depends on the dimension of the Hilbert spaces involved in the bipartition $\hhjj$, this summary should help the
reader to have a synopsis of the various (essentially three) possible cases: $\dim(\hh\otimes\jj)<\infty$, the case
where precisely one of the Hilbert spaces $\hh$ and $\jj$ is infinite-dimensional, and $\dim(\hh)=\dim(\jj)=\infty$.
Finally, in Section~\ref{entanglement}, we introduce the entanglement function and study its main properties.

The paper is addressed to a readership from diverse backgrounds, that ideally may include both physicists and
mathematicians. Therefore, throughout the paper, the reader will find some claims --- we refer to as \emph{facts},
because they are mainly well known --- that are particularly relevant for our purposes, and hence deserve to be
properly emphasized. These facts, that are stated without proofs (that can be found in standard references, or
easily reconstructed by the reader), may be more or less familiar to physicists, or to mathematicians, depending
on their backgrounds. For the sake of conciseness, we will also omit the proofs of all those results that can
reasonably be proved by the reader, without much effort, using the information provided in the paper.

\section{Preliminaries and main notations}
\label{preliminaries}

In this section, we will establish our main conventions and notations, and recall some
preliminary facts. All technical facts recalled in the following can be found by the
reader in standard references on operator theory and functional-analytic aspects of quantum
mechanics~\cite{Busch,Prugovecki,Emch,Moretti,Weidmann,Reed,Pedersen,Conway-bis,Conway-FA,Conway,Blanchard},
the theory of operator algebras~\cite{Moretti,Murphy,Kadison-EL,Kadison},
Banach space theory~\cite{Conway-FA,Conway,Lindenstrauss,Megginson,Fabian,Narici,Hytonen},
convex analysis~\cite{Rockafellar,Aliprantis,Simon_conv,Alfsen}, measure theory~\cite{Folland-RA,Bogachev,Cohn}
(especially on metric and Polish spaces~\cite{Parthasarathy,Kechris,Srivastava}, and the theory
of standard Borel spaces~\cite{Kechris,Srivastava,Raja}).
Everywhere in the paper, we will denote by $\erres$, $\errep$, $\erreps$, $\errems$ the sets of nonzero,
non-negative, strictly positive and strictly negative real numbers, respectively.

\subsection{Operators on Hilbert spaces}
\label{operators}

Coherently with the notation adopted in the introduction, we will denote by $\hh$
--- or by $\jj$, $\hau$ --- a separable complex Hilbert space. The scalar product
$\scap\equiv\scaph$ on $\hh$ is supposed to be conjugate-linear in its \emph{first}
argument. Given a subspace $\vv$ of $\hh$, we denote by $\vv^\perp$ its
\emph{orthogonal complement}, i.e., the closed subspace
$\{\phi\in\hh\colon\langle\phi,\psi\rangle=0 ,\ \forall\cinque\psi\in\vv\}$.
We adopt the symbol $I$ for the \emph{identity operator} on $\hh$ (or $\jj$, $\hau$).
We denote by $\trc$ the (separable) \emph{complex} Banach space of all \emph{trace class operators}
on $\hh$ and  by $\trcsa$ the \emph{real} Banach space of all \emph{selfadjoint} trace class operators.
We denote by $\trf\colon\trc\rightarrow\ccc$ the \emph{trace} functional, and by $\trnor\equiv\trano$
the \emph{trace norm} (on $\trc$ or $\trcsa$); i.e., for every $S\in\trc$,
$\|S\trn\defi\tr(|S|)=\|S^\ast\trn$, with $|S|\in\trc$ denoting the \emph{absolute value} of
$S$ ($|S|\equiv\sqrt{S^\ast S}$). The set $\trcp\subset\trcsa$ of all \emph{positive} trace class
operators is a norm-closed convex subset of $\trc$. The Banach space \emph{dual} $\trcd$ of $\trc$
is identified --- via the pairing $\trc\times\bop\ni(S,B)\mapsto\tr(SB)$ --- with $\bop$, i.e.,
the complex Banach space of all \emph{bounded operators} on $\hh$, endowed with the standard
\emph{operator norm} $\normi$. The convex set of all \emph{positive} bounded operators will be denoted
by $\bopp$. The trace class $\trc$ is a \emph{two-sided $\ast$-ideal in $\bop$} and, for any
$B\in\bop$ and $S\in\trc$, we have that $\|B\tre S\trn,\|SB\trn\le\|B\nori\|S\trn$; moreover,
for every $S\in\trc\subset\bop$, $\|S\nori\le\|S\trn$. Let us collect a few more facts concerning
these norms and the Banach space $\trc$.

\begin{fact} \label{fabop}
For the operator norm $\normi$ on $\bop$, we have:
\begin{equation}
\|B\nori\defi\sup\{\|B\tre\psi\|\colon\|\psi\|=1\}=\sup\{|\langle\phi,B\tre\psi\rangle|\colon
\|\phi\|=\|\psi\|=1\} \fin , \quad \forall\cinque B\in\bop \fin .
\end{equation}
In particular, if $B\in\bopsa$, then $\|B\nori=\sup\{|\langle\psi,B\tre\psi\rangle|\colon\|\psi\|=1\}$.
\end{fact}

\begin{fact} \label{trnsubm}
The trace norm $\trnor$ is \emph{submultiplicative}; i.e., for any $S,T\in\trc$,
$\|S\tre T\trn\le\|S\trn\tre\|T\trn$.
\end{fact}

\begin{fact} \label{prolico}
Every element $S$ of $\trc$ can be expressed as a linear combination of (at most) four positive
trace class operators $S_1,\ldots,S_4\in\trcp$: $S=S_1 - S_2 + \ima\sei(S_3 - S_4)$, where, for
$S\in\trcsa$, $S_3 - S_4=0$. In particular, the  operators $S_1,\ldots,S_4\in\trcp$ can be taken
of the form
\begin{equation} \label{nota}
S_1=\frac{1}{4}\sei(|S+S^\ast|+S+S^\ast) \fin , \quad
S_2=\frac{1}{4}\sei(|S+S^\ast|-S-S^\ast) \fin ,
\end{equation}
\begin{equation} \label{notb}
S_3=\frac{1}{4}\sei(|S-S^\ast|-\ima\sei (S-S^\ast)) \fin , \quad
S_4=\frac{1}{4}\sei(|S-S^\ast|+\ima\sei (S-S^\ast)) \fin ,
\end{equation}
and, in this case, we have that $S_1\tre S_2=S_2\tre S_1=0=S_3\tre S_4=S_4\tre S_3$. Thus, if
$S=S_1 - S_2\in\trcsa$, then, with such a choice of the positive trace class operators $S_1$
and $S_2$, we have that $|S|=S_1 + S_2$, $\|S\trn\defi\tr(|S|)=\tr(S_1)+\tr(S_2)$ and
$S_1\tre S_2=S_2\tre S_1=0$.
\end{fact}

\begin{fact} \label{sivadea}
Every trace class operator $S\in\trc$ admits a decomposition of the form
\begin{equation} \label{sivalde}
S=\sum_n s_n\tre\epn \fin , \quad (\epn)\tre\psi\defi\langle\phi_n,\psi\rangle\sei\eta_n \fin ,
\ \psi\in\hh \fin ,
\end{equation}
where $s_n\ge 0$, for every $n$, and $\{\eta_n\}$, $\{\phi_n\}$ are orthonormal systems in $\hh$;
if $\dim(\hh)=\infty$, the sum --- whenever not finite --- is absolutely convergent wrt the
trace norm $\trnor$, and, in fact, $\sum_n s_n=\tr(|S|)=\|S\trn$.
\end{fact}

The set of all trace class operators on $\hh$ admitting a \emph{finite} decomposition of
the form~\eqref{sivalde} is a dense linear subspace of $\trc$, that coincides with the set
$\frh$ of \emph{finite rank operators} on $\hh$, and $\frh=\lispa\{\ep\colon \eta,\phi\in\hh\}$;
see, e.g., Chapter~{6} of~\cite{Weidmann}. $\frh$ is a $\normi$-dense linear subspace of the
Banach space $\cph$ of \emph{compact operators} on $\hh$; endowed with the operator norm, $\cph$
is a $\cast$-subalgebra of $\bop$~\cite{Murphy}.

\begin{remark}[The singular value decomposition] \label{sivadeb}
By Fact~\ref{sivadea}, every $S\in\trc$, with $S\neq 0$, admits a \emph{singular value decomposition} (SVD)
of the form $S=\sum_k s_k\tre\epk$, where $s_k> 0$, for every value of the index $k$, and $\{\eta_k\}$,
$\{\phi_k\}$ are orthonormal systems in $\hh$; the sum, whenever not finite, is absolutely convergent
wrt the trace norm, and $\sum_k s_k=\tr(|S|)=\|S\trn$ (whereas, $\tr(S)=\sum_k s_k\tre\langle\phi_k,\eta_k\rangle$).
The set $\{s_k\}$ of all \emph{singular values} of $S$ coincides with the set of the \emph{nonzero eigenvalues}
of the \emph{positive operator} $|S|\neq 0$, each eigenvalue being counted according to its multiplicity.
We stress that, compared to decomposition~\eqref{sivalde}, the SVD is a \emph{sheer decomposition} of
the trace class operator $S\neq 0$; i.e., \emph{it does not contain any zero term}.
\end{remark}

\begin{fact} \label{norsub}
The linear spaces $\frh$, $\cph $ of finite rank operators and of compact operators on $\hh$ ---
$\frh\subset\cph\subset\bop=\trcd$ --- are \emph{norming subspaces} of $\bop=\trcd$ for $\trc$; i.e.,
for every $S\in\trc$,
\begin{equation} \label{renorsub}
\|S\trn=\sup\{|\tr(SF)|\colon F\in\frh, \ \|F\nori=1\}=\sup\{|\tr(SK)|\colon K\in\cph, \ \|K\nori=1\} \fin .
\end{equation}
\end{fact}

The set of \emph{density operators} on $\hh$ --- the unit trace, positive trace class operators,
here denoted by $\sta$ --- is a norm-closed convex subset of $\trc$, the intersection of $\trcp$
with the closed set of unit trace operators. We will identify $\sta$ with the set of all the
(normal, or $\sigma$-additive) \emph{states} of a quantum system with Hilbert space
$\hh$~\cite{Emch,Holevo,Heino,Moretti}. The convex set $\sta$ contains the set $\pusta$ of all rank-one
projections on $\hh$, i.e., the so-called \emph{pure} states. The dual space of $\trcsa$ --- i.e., $\bopsa$
the \emph{real} Banach space of all \emph{selfadjoint} bounded operators --- will be regarded as the space
of (bounded) quantum \emph{observables}.

\begin{fact} \label{fabop-bis}
For every $B\in\bopsa$, $\|B\nori=\sup\{|\tr(PB)|\colon P\in\pusta\}=\sup\{|\tr(DB)|\colon D\in\sta\}$.
\end{fact}

\begin{fact}[The standard topology of $\sta$~\cite{Aniello_CSP}] \label{statop}
The weak and the strong operator topologies on $\sta$ (inherited from the space of bounded operators $\bop$),
as well as the topologies induced on $\sta$ by the metrics associated with the Schatten $\pp$-norms
$\normp$, $1\le\pp\le\infty$, all coincide. This unique topology will be called the \emph{standard topology}
of $\sta$.
\end{fact}

\subsection{A few facts and notations concerning normed and Banach spaces}
\label{basps}

A norm on a, real or complex, vector space $\bsb$ --- in particular, a Banach or a Hilbert space ---
will be usually denoted by the symbols $\no$ or $\notri$ (with or without some complement, in
order to further specialize their meaning), and, whenever no risk of confusion may occur, the
normed space $(\bsb,\no)$ will be simply denoted by $\bsb$. The symbols $\ball(\bsb)$, $\sph(\bsb)$
will denote, respectively, the \emph{closed unit ball} and the \emph{unit sphere} centered at the
origin of the normed space $\bsb$; i.e., we put
$\ball(\bsb)\defi\{x\in\bsb\colon \|x\|\le 1\}$ and $\sph(\bsb)\defi\{x\in\bsb\colon \|x\|= 1\}$.
We say that $(\bsb,\notri)$ is a \emph{renorming} of $(\bsb,\no)$ if the norms $\no$ and $\notri$
are equivalent.

The (continuous) \emph{dual space} of a normed space $\bsb$ --- which is always a Banach space --- will be
denoted by $\bsb^\ast\equiv(\bsb,\no)^\ast$ (unless another symbol be more suitable for the occasion),
and, for the time being, we will denote by $\no_\bsbd$ the norm of the Banach space $\bsbd$; precisely,
setting $\|\xi\|_\bsbd\defi\sup\{\|\xi(x)\|\colon x\in\bsb,\ \|x\|\le 1\}$, the norm $\no_\bsbd$ is called
the \emph{dual norm}  of $\no$. If $\bsb$ is a Banach space itself, then $\bsbd$ is also called the
\emph{Banach space dual} of $\bsb$.

If $\bs$ is a linear subspace of a normed space $\bsb$, then, for every bounded functional
$\xi\in\bsb^\ast$, $\xi_0\equiv\xi\vert_\bs\in\bs^\ast$ (i.e., the restriction $\xi_0$ of $\xi$ to
$\bs$ is a bounded functional), and $\|\xi_0\|_\bsd\le\|\xi\|_\bsbd$. Here, keeping $\xi_0\in\bs^\ast$
fixed, the inequality is saturated by some element $\xi$ of $\bsb^\ast$ such that $\xi_0=\xi\vert_\bs$.
Indeed, conversely, by the Hahn-Banach (norm-preserving, continuous extension) theorem --- see, e.g.,
Theorem~{1.9.6} of~\cite{Megginson}, or Corollary~{2.3} of~\cite{Fabian} --- every element $\xi_0$ of
$\bs^\ast$ admits a \emph{Hahn-Banach extension} $\xi\in\bsb^\ast$; i.e., there is a functional
$\xi\in\bs^\ast$ such that $\xi\vert_\bs=\xi_0$ and $\|\xi\|_\bsbd=\|\xi_0\|_\bsd$. An element
of $\bs^\ast$ may admit more than one Hahn-Banach extension. However, if $\bs$ is a \emph{norm-dense}
linear subspace of $\bsb$, then every element of $\bs^\ast$ will admit \emph{exactly one} continuous
extension, so that $\bs^\ast$ is isomorphic to the Banach space $\bsb^\ast$, where the isomorphism
of Banach spaces is implemented by the map $\bsb^\ast\ni\xi\mapsto\xi\vert_\bs\in\bs^\ast$, and
we will simply write $\bs^\ast=\bsb^\ast$.

The \emph{$\wst$-topology} (namely, the weak$^{\ast}$ topology) on the dual $\bsb^\ast$ of a
normed space $\bsb$ --- also called the \emph{$\sigma(\bsb^\ast,\bsb)$-topology} --- is defined
as the \emph{initial topology} (or inverse image topology) induced by the family of maps
$\{\bsb^\ast\ni\xi\mapsto\xi(x)\in\ccc\colon x\in\bsb\}$; i.e., the weakest topology on $\bsb^\ast$
such that all these maps are continuous~\cite{Pedersen,Conway-bis,Conway-FA,Conway,Megginson,Fabian,Narici}.
Equivalently, the $\sigma(\bsb^\ast,\bsb)$-topology can be defined as the initial topology
induced by the family of semi-norms $\{\bsb^\ast\ni\xi\mapsto|\xi(x)|\in\errep\colon x\in\bsb\}$,
and $\bsb^\ast$ is a locally convex topological vector space wrt this topology (see Chapter~{2}
of~\cite{Megginson} and Chapter~{8} of~\cite{Narici}). The $\sigma(\bsb^\ast,\bsb)$-topology,
by Theorem~{2.6.2} of~\cite{Megginson}, is a completely regular --- i.e., Tychonoff or
$\mathsf{T}_{3\frac{1}{2}}$ --- and locally convex subtopology of the weak topology (hence, of
the norm topology) of $\bsb^\ast$, that is, of the $\sigma(\bsb^\ast,\bsb^{\ast\ast})$-topology.
A net $\{\xi_i\}\subset\bsb^\ast$ converges to some $\xi\in\bsb^\ast$ wrt the
$\sigma(\bsb^\ast,\bsb)$-topology iff $\lim_i \xi_i(x)=\xi(x)$, for all $x\in\bsb$.

\begin{fact} \label{wstatop}
Let $\bsb_0$ be a norm-dense subset of the normed space $\bsb$, or, also, a norm-dense subset of
$\sph(\bsb)$. Then, a \emph{norm-bounded} net $\{\xi_i\}\subset\bsb^\ast$ converges to $\xi\in\bsb^\ast$
wrt the $\sigma(\bsb^\ast,\bsb)$-topology iff $\lim_i \xi_i(x_0)=\xi(x_0)$, for all $x_0\in\bsb_0$.
\end{fact}

\begin{fact} \label{wstatop-bis}
Let $\bs$ be a norm-dense linear subspace of the normed space $\bsb$, and let $\sbs$ a \emph{norm-bounded}
subset of $\bs^\ast=\bsb^\ast$ (every element of $\bs^\ast$ being identified with its continuous extension
to $\bsb^\ast$). Then, the relative topology of $\sbs$ wrt the $\sigma(\bs^\ast,\bs)$-topology
coincides with the relative topology of $\sbs$ wrt the $\sigma(\bsb^\ast,\bsb)$-topology.
\end{fact}

Given norms $\no,\notri\colon\bsb\rightarrow\errep$ on the vector space $\bsb$, we will say that
the norm $\no$ is \emph{dominated} --- or \emph{majorized} --- by the norm $\notri$ if, for every
$x\in\bsb$, $\|x\|\le\ntr{x}$.

\begin{fact} \label{fanonotri}
Let $\no,\notri\colon\bsb\rightarrow\errep$ be norms on the vector space $\bsb$, and
suppose that the norm $\no$ is dominated by the norm $\notri$. Then, the dual space
$\bsb^\ast\equiv(\bsb^\ast,\no^\ast)=(\bsb,\no)^\ast$ of $(\bsb,\no)$, with
$\no^\ast\equiv\no_\bsbd$, is a linear subspace of
$\bsb^\circledast\equiv\big(\bsb^\circledast,\notri^\circledast\big)=(\bsb,\notri)^\ast$,
and the $\sigma(\bsb^\ast,\bsb)$-topology coincides with the subspace topology of
$\bsb^\ast\subset\bsb^\circledast$ wrt the $\sigma(\bsb^\circledast,\bsb)$-topology of
$\bsb^\circledast$.
\end{fact}

In the following, we will deal with Banach spaces, denoted by the symbols $\bsb$, $\bsc$
and $\bs$.

Given a real or complex Banach space $\bsb$, a countable --- i.e., at most denumerable --- and
ordered set of vectors $\{\sba_i\}_{i\in\iset}\subset\bsb$ is called a \emph{Schauder basis}
if, for every $x\in\bsb$, there is a \emph{unique} set of scalars $\{\cf_i(x)\}_{i\in\iset}$,
such that $x=\sum_{i\in\iset}\cf_i(x)\sei\sba_i$, where the sum, whenever not finite, converges
wrt the norm of $\bsb$. Every Banach space admitting a Schauder basis is \emph{separable},
because, for every countable subset $\bsb_0$ of $\bsb$, the norm-closed linear span of
$\bsb_0$  is separable (see, e.g., Proposition~{1.12.1} of~\cite{Megginson}).

Assuming that the reader is familiar with the basic facts concerning the \emph{absolutely convergent}
series in a Banach space (see, e.g., Chapter~{1} of~\cite{Megginson}), and the \emph{double} sequences
and series of real numbers (see, e.g., Chapter~{7} of~\cite{Ghorpade}), we recall a few facts about the
\emph{double} series in a Banach space $\bsb$. We say that a \emph{double sequence} $\{x_{jl}\}\subset\bsb$
is \emph{absolutely summable} --- or that the \emph{double series} $\sum_{jl}\xjl$ is \emph{absolutely convergent}
--- if any of the following \emph{equivalent conditions} is satisfied:
\begin{itemize}

\item every `row series' $\sum_l \xjl$ is absolutely convergent --- $\sum_l \|\xjl\|<\infty$ ---
and $\sum_j \sum_l \|\xjl\|<\infty$;

\item every `column series' $\sum_j \xjl$ is absolutely convergent  and $\sum_l \sum_j \|\xjl\|<\infty$;

\item the double (real) series $\sum_{jl} \|x_{jl}\|$ converges \textit{\`a la Pringsheim};

\item $\sum_n \|\xjln\|<\infty$, for \emph{some} bijection $\nat\ni n\mapsto\jln\in\nat\times\nat$;

\item $\sum_n \|\xjln\|<\infty$, for \emph{every} bijection $\nat\ni n\mapsto\jln\in\nat\times\nat$.

\end{itemize}

\begin{fact} \label{facdou}
Let the double sequence $\{x_{jl}\}\subset\bsb$ be absolutely summable. Then, the double series
$\sum_{jl} x_{jl}$ is convergent \textit{\`a la Pringsheim} to some vector $v\in\bsb$ --- in symbols,
$v=\sum_{jl} x_{jl}$ --- namely, for every $\epsilon>0$ there is some $\nep\in\nat$ such that
$m,n>\nep\implies\|v-\sum_{j=1}^m \sum_{l=1}^n x_{jl}\|<\epsilon$. Moreover, we have that
\begin{equation} \label{reldou}
v=\sum_{jl} \xjl=\sum_j\sum_l \xjl=\sum_l\sum_j \xjl=\sum_n \xjln \fin ,
\end{equation}
for \emph{every} bijection $\nat\ni n\mapsto\jln\in\nat\times\nat$.
\end{fact}

The previous notions and results about double series extend, in a straightforward way, to \emph{multiple}
series in a Banach space.

From this point onwards, the standard norm of bounded linear or bilinear maps between generic (real or complex)
Banach spaces will be usually denoted by $\nops$, unless another symbol be more appropriate. Clearly, this is
a Banach space norm itself, but using the symbol $\nops$ --- rather than the bare symbol $\no$ (or $\no^\ast$
for the associated dual norm) --- is a way to emphasize the fact that we are dealing precisely with the standard
norm of a Banach space of \emph{bounded linear maps}. The meaning, however, should be always clear from the context.

E.g., $\bo(\bsb;\bs)$ denotes the Banach space of all bounded linear maps from the Banach space
$\bsb$ into the Banach space $\bs$, endowed with the norm $\nops$ defined in the usual way.
Analogously, given (real or complex) Banach spaces $\bsb$, $\bsc$ and $\bs$, and a bounded bilinear
map $\bim\colon\bsb\times\bsc\rightarrow\bs$, we set:
\begin{equation}
\|\bim\norps\defi\sup\{\|\bim(x,y)\|\colon x\in\bsb,\ y\in\bsc,\ \|x\|,\|y\|\le 1\} \fin ,
\quad \bim\in\bo(\bsb,\bsc;\bs) \fin .
\end{equation}
Clearly, here the symbol $\bo(\bsb,\bsc;\bs)$ denotes the Banach space of all bounded bilinear maps
from $\bsb\times\bsc$ into $\bs$, endowed with the norm $\nops$ defined above. In particular, we will
consider the case where $\bsb=\hh$, $\bsc=\jj$ (Hilbert spaces) and $\bs=\ccc$, i.e., the complex Banach
space $\bfshj$ of all bounded bilinear forms on $\hh\times\jj$.

We say that a bilinear map $\bim\colon\bsb\times\bsc\rightarrow\bs$ is a \emph{bilinear isometry}
if $\|\bim(x,y)\|=\|x\|\tre\|y\|$, for all $x\in\bsb$ and $y\in\bsc$. Clearly, if $\bim$ is a
bilinear isometry, then it is bounded and $\|\bim\norps= 1$.

Apart from the previously introduced operator norm $\normi$ on $\bop\equiv\bhh$ (Subsection~\ref{operators}),
we will specialize the symbol $\nops$ only in the remarkable case where all Banach spaces involved are trace
classes of Hilbert space operators; see below.

Further notations, generalizing in a straightforward way the previously introduced symbols, are summarized
in the following Table~\ref{fursymb}.

\begin{center}
\begin{tabular}{|c|c|}
\hline
Symbol
&
Meaning
\\
\hline
$\ltrc\equiv\btrc\equiv\bo(\trc;\trc)$
&
bounded linear operators on $\trc$
\\
$\ftrcjj\equiv\bo(\trc,\trcjj;\ccc)$
&
bounded bilinear forms on $\trc\times\trcjj$
\\
$\ltrcjjau\equiv\bo(\trc,\trcjj;\trcau)$
&
bounded bil.\ maps from $\trc\times\trcjj$ into $\trcau$
\\
$\bo(\trc,\trc;\bs)$
&
bounded bilinear maps from $\trc\times\trc$ into $\bs$
\\
$\linv$, $\linvwc$
&
linear maps on $\vv$, bilinear forms on $\vv\times\ww$
\\
$\linvwz$
&
bilinear maps from $\vv\times\ww$ into $\zz$
\\
\hline
\end{tabular}
\captionsetup{type=table}
\captionof{table}{Further symbols used in the paper ($\vv$, $\ww$, $\zz$ are real or complex vector spaces).}
\label{fursymb}
\end{center}

The norm of the Banach space $\ltrc$ will denoted by $\normo$ (defined as the norm $\nops$).

The norms of the Banach spaces of \emph{bilinear forms} $\ftrcjj$ and of \emph{bilinear maps}
$\ltrcjjau$ will all be denoted by $\norbil$.

A (real or complex) Banach space $\bs$ is said to have the \emph{approximation property} if,
for every compact subset $\mathfrak{K}$ of $\bs$ and every $\epsilon>0$, there is a finite rank
operator $F\colon\bs\rightarrow\bs$ such that $\|x- Fx\|<\epsilon$, for all $x\in\mathfrak{K}$;
see, Chapter~{4} of~\cite{Ryan}. $\bs$ is said to have the \emph{Radon-Nikod\'ym property}
if the Radon-Nikod\'ym theorem holds for suitable $\bs$-valued measures; see Chapter~{5} of~\cite{Ryan}.

\begin{fact} \label{appran}
The Banach space $\trc$ has both the approximation property and the Radon-Nikod\'ym property.
\end{fact}

\subsection{Linear spans, convex cones, convex hulls}
\label{lin-con}

A subset $\con$ of a real or complex vector space $\ves$ is called a \emph{convex cone} if,
for any $p,q\in\erreps$ and $v,w\in\con$, $p\tre v+q\tre w\in\con$. A convex cone $\con$ is
said to be \emph{pointed} if $\con\cap(-\con)=\{0\}$. E.g., the set $\trcp$ of all positive trace
class operators on $\hh$ is a pointed convex cone in $\trc$; analogously, $\bopp$ is a pointed
convex cone in $\bop$.

Given a nonempty subset $\subh$ of a complex Banach space $(\bs,\no)$, we will denote by
$\co(\subh)$ the \emph{convex hull} of $\subh$, and by $\clco(\subh)$ the \emph{closed convex hull},
i.e., the smallest norm-closed convex set containing $\subh$. As is well known,
$\clco(\subh)=\cl(\co(\subh))\equiv\clnor(\co(\subh))$. Moreover, it turns out that
\begin{equation} \label{reclco}
\clco(\subh)=\clco\big(\clnor(\subh)\big) .
\end{equation}

\begin{remark} \label{recloco}
Given a dual Banach space $\bsd$, endowed with the $\wst$-topology, we will also consider
the $\wst$-closed convex hull of a nonempty subset $\subj$ of $\bsd$. Since $\bsd$,
endowed with the $\wst$-topology, is a \emph{topological vector space}~\cite{Megginson,Fabian,Narici},
then, as it happens for the norm-closed convex hull, we have that $\clcowst(\subj)=\wstcl(\co(\subj))$,
i.e., the $\wst$-closed convex hull of $\subj$ coincides with the $\wst$-closure of the convex hull of
$\subj$ (see, e.g., Theorem~{2.2.9\hspace{0.6mm}-(i)} of~\cite{Megginson}).
\end{remark}

Analogously, $\lispa(\subh)$ denotes the complex \emph{linear span}
of $\subh\subset\bs$, and $\clispa(\subh)$ the \emph{closed linear span}, namely, the smallest norm-closed
linear subspace of $\bs$ containing $\subh$; moreover, $\clispa(\subh)=\cl(\lispa(\subh))$, and
$\clispa(\subh)=\clispa(\cl(\subh))$. Further useful relations are the following:
$\lispa(\subh)=\co(\ccc\sei\subh)$; $\lispa(\co(\subh))=\lispa(\subh)=\co(\lispa(\subh))$, and hence
$\clispa(\clco(\subh))=\clispa(\co(\subh))=\clispa(\subh)=\clco(\lispa(\subh))=\clco(\clispa(\subh))$.

\begin{fact} \label{infconv}
If $\subc$ is a closed convex subset  of a complex Banach space $\bs$, then every norm-convergent,
countably infinite convex combination $\sum_{n=1}^\infty p_n \tre x_n$ of elements $\{x_n\}_{n\in\nat}$
of $\subc$ ($p_n>0$, $\sum_{n=1}^\infty p_n=1$) is contained in $\subc$. Hence, given a subset $\subh$
of $\bs$, every $\no$-convergent, countably infinite convex combination $\sum_{n=1}^\infty p_n \tre x_n$
of elements $\{x_n\}_{n\in\nat}$ of $\subh$ belongs to $\clco(\subh)$.
\end{fact}

\begin{fact} \label{propro}
$\sta=\clco(\pusta)$.
\end{fact}

Recall that an element $x$ of a convex subset $\con$ of a real or complex vector space
$\ves$ (in particular, of a Banach space $\bs$) is said to be an \emph{extreme point}
of $\con$ if there \emph{do not} exist $y,z\in\con$, with $y\neq z$, and $t\in(0,1)$
such that $x=ty+(1-t)z$; namely, if $x$ \emph{does not lie between any two distinct points
of} $\ves$. We will denote the set of all extreme points of $\con$ by $\ext(\con)$.

\begin{fact} \label{facextp}
If $\bsb$ is a, real or complex, normed space,
then the extreme points of its closed unit ball belong to its unit sphere:
$\ext(\ball(\bsb))\subset\sph(\bsb)\defi\{x\in\bsb\colon \|x\|= 1\}$.
\end{fact}

\begin{fact} \label{extdens}
The extreme points of the convex set $\sta$ form precisely the set $\pusta$ of all rank-one projections
on $\hh$ (pure states); i.e., $\ext(\sta)=\pusta$.
\end{fact}

\begin{fact} \label{factext}
Let $\subh,\subj$ be nonempty convex subsets of $\ves$. Then,
$\subh\subset\subj\implies\ext(\subj)\cap\subh\subset\ext(\subh)$.
\end{fact}

\subsection{The algebraic tensor product}
\label{algtens}

There are two main equivalent approaches to the definition of the algebraic tensor product
of two vector spaces; see~\cite{Ryan} and~\cite{Kubrusly}, respectively (also see Chapter~{4}
of~\cite{Kubrusly}, for an analysis of various approaches). In the approach adopted, e.g.,
in~\cite{Ryan} --- i.e., the \emph{linear-bilinear approach} which will now briefly sketched ---
the so-called \emph{universal property} is a consequence of the definition rather than part
of it (as it happens in the approach adopted in~\cite{Kubrusly}). For the sake of definiteness,
we consider the case of complex vector spaces; the case of real vector spaces is analogous.

Let $\vv$, $\ww$ be (complex) vector spaces, and let us denote by $\vvad$, $\wwad$ the \emph{algebraic}
dual spaces of $\vv$ and $\ww$, respectively. We will denote by $\linvw$ the vector space of all
bilinear forms on $\vv\times\ww$. Given any pair of vectors $v\in\vv$ and $w\in\ww$, we can
define a linear functional $v\altp w\colon\linvw\rightarrow\ccc$, by putting
$\big(v\altp w\big)(\bfo)\defi\bfo(v,w)$, for all $\bfo\in\linvw$.
Then, the (standard realization of the) \emph{algebraic tensor product} of $\vv$ and $\ww$
is the linear subspace $\vv\altp\ww$ of $\linvwad$ --- i.e., of the algebraic dual of
$\linvw$ --- defined as follows
\begin{equation}
\vv\altp\ww\defi\lispa\big\{v\altp w\big\}\equiv
\lispa\big\{v\altp w\in\linvwad\colon v\in\vv, \ w\in\ww\big\} \fin ,
\end{equation}
together with the bilinear map $\snb\colon\vv\times\ww\ni (v,w)\mapsto v\altp w\in\vv\altp\ww$,
the so-called (standard realization of the) \emph{natural bilinear map} on $\vv\times\ww$.
Taking into account that, for every $c\in\ccc$, $c\tre(v\altp w)=(c\tre v)\altp w$,
every element of $\vv\altp\ww$ can be expressed as a finite sum
$\sum_{k=1}^n v_k\altp w_k$ of \emph{elementary tensors}.
By the \emph{universal property of algebraic tensor products} --- see Proposition~{1.4}
of~\cite{Ryan} --- for every bilinear map $\theta\in\linvwz$ from $\vv\times\ww$ into a complex
vector space $\zz$, there is a unique linear map $\Theta\colon\vv\altp\ww\rightarrow\zz$
such that $\theta=\Theta\circ\snb$ (universality relation) --- i.e.,
\begin{equation} \label{relunip}
\theta(v,w)=\Theta\big(v\altp w\big) , \quad
\forall\cinque v\in\vv,\ \forall\cinque w\in\ww
\end{equation}
--- and, moreover, the one-to-one correspondence $\theta\leftrightarrow\Theta$ is an isomorphism
between the vector spaces $\linvwz$ and $\linvawz$.

If $\lispa(\ran(\theta))=\zz$ and, moreover, $\Theta$ is a linear isomorphism of $\vv\altp\ww$ onto
$\zz$ --- i.e., since, by~\eqref{relunip}, $\ran(\Theta)=\lispa(\ran(\theta))=\zz$, if $\Theta$
is injective --- then $\theta\in\linvwz$ automatically enjoys the crucial \emph{universal property};
i.e., given any bilinear map $\lambda\colon\vv\times\ww\rightarrow\tzz$ (where $\tzz$ is a
complex vector space), there is a unique linear map $\lath\colon\zz\rightarrow\tzz$ such that
$\lambda=\lath\circ\theta$. Precisely, denoting by $\Lambda\colon\vv\altp\ww\rightarrow\tzz$
the unique linear map such that $\lambda(v,w)=\Lambda\big(v\altp w\big)$, for all $v\in\vv$ and
$w\in\ww$, we have that $\lath=\Lambda\circ\Theta^{-1}\circ\theta$. Conversely, if $\theta\in\linvwz$
enjoys the universal property (including the uniqueness condition of the linear map associated
with a bilinear map on $\vv\times\ww$), then $\lispa(\ran(\theta))=\zz$ and, moreover, $\Theta$
is a linear isomorphism of $\vv\altp\ww$ onto $\zz$. In fact, in this case, we have that
\begin{equation}
\snb=\slma\circ\theta=\slma\circ\Theta\circ\snb \ \ \mbox{and} \ \
\theta=\Theta\circ\snb=\Theta\circ\slma\circ\theta \fin ,
\end{equation}
where $\Theta\colon\vv\altp\ww\rightarrow\zz$ and $\slma\colon\zz\rightarrow\vv\altp\ww$ are
uniquely determined linear maps such that
\begin{equation}
\dom(\Theta)=\vv\altp\ww=\lispa\big(\ran\big(\snb\big)\big)=\ran\big(\big(\slma\big)\big), \ \
\dom\big(\slma\big)=\zz\supset\lispa(\ran(\theta))=\ran(\Theta) \fin .
\end{equation}
Here, by the uniqueness of the linear map $\slma$, we conclude that, actually,
$\zz=\lispa(\ran(\theta))$. Hence, $\Theta$ is a linear isomorphism of $\vv\altp\ww$ onto $\zz$
and $\slma=\Theta^{-1}$.

Therefore, it is natural to set the following:

\begin{definition} \label{defreatp}
If $\theta\in\linvwz$ is such that $\lispa(\ran(\theta))=\zz$ and, moreover, the unique linear map
$\Theta\colon\vv\altp\ww\rightarrow\zz$ such that $\theta=\Theta\circ\snb$ is injective (hence, a
linear isomorphism of $\vv\altp\ww$ onto $\zz$), we say that the pair $(\zz,\theta\in\linvwz)$ is
\emph{(a realization of) the algebraic tensor product} of $\vv$ and $\ww$, and we will
usually put $v\otimes w\equiv\theta(v,w)$ and $\vv\altp\ww\equiv\zz$. The bilinear map
\begin{equation}
\nb\colon\vv\times\ww\ni (v,w)\mapsto v\otimes w\in\zz\equiv\vv\altp\ww
\end{equation}
is called \emph{(a realization of) the natural bilinear map} on $\vv\times\ww$.
\end{definition}

\begin{remark}
We will not adhere to a standard usage, in the mathematical literature, of denoting by
$\vv\otimes\ww$ the algebraic tensor product space, because, when considering the case
where $\vv$ and $\ww$ are Banach or Hilbert spaces, this notation will be used for indicating
the completion of (any realization of) the algebraic tensor product space $\vv\altp\ww$ wrt to
some natural norm. This is the case, e.g., of the tensor product $\hhjj$ of two Hilbert spaces.
\end{remark}

Another relevant fact, for our purposes, concerns a suitable realization of the algebraic tensor
product of spaces of linear operators.

\begin{fact} \label{facteplin}
Let $\hh$, $\jj$ and $\hau$ be complex vector spaces, and suppose that the pair
\begin{equation}
\big(\hau\equiv\hh\altp\jj,\hh\times\jj\ni(\phi,\psi)\mapsto\vbf(\phi,\psi)
\equiv\phi\otimes\psi\in\hau\big),
\end{equation}
--- with $\hau\equiv\hh\altp\jj=\lispa(\ran(\vbf))=\lispa\big(\big\{\phi\otimes\psi
\colon\phi\in\hh,\ \psi\in\jj\big\}\big)$ --- be a realization
of the algebraic tensor product of $\hh$ and $\jj$. Then, the following claims hold true:
\begin{itemize}

\item For every pair $(A\in\linhh,B\in\linjj)$, there is a unique linear map $A\otimes B\in\linhau$
such that $\big(A\otimes B\big)(\phi\otimes\psi)=(A\phi)\otimes(B\psi)$, for all $\phi\in\hh$ and
$\psi\in\jj$.

\item For every pair $(\vv\subset\linhh,\ww\subset\linjj)$ of linear subspaces $\vv$, $\ww$ of
$\linhh$ and $\linjj$, the pair
\begin{equation}
\big(\zz, \theta\colon\vv\times\ww\ni (A,B)\mapsto A\otimes B\in\zz\big)
\end{equation}
--- with $\zz=\lispa(\ran(\theta))=\lispa\{A\otimes B\in\linhau\colon
A\in\vv\subset\linhh,\ B\in\ww\subset\linjj\}$ ---
is a realization of the algebraic tensor product of $\vv$ and $\ww$.

\end{itemize}
\end{fact}

\subsection{Probability measures and standard Borel spaces}
\label{measures}

Let $(\mesp,\bora,\mu)$ be a \emph{Borel probability measure space}; i.e., let $\mesp$ be a topological space,
$\bora$ the associated Borel $\sigma$-algebra and $\mu$ a probability measure on $\bora$. We will denote by
$\prm(\mesp)$ the collection of all such measures.

Consider the set $\nmu\defi\big\{\nset\in\bora\colon\mu\big(\nset\big)=0\big\}$ of all
\emph{$\mu$-null} subsets in $\bora$. Then, the set $\cobora\supset\bora$ defined by
$\cobora\defi\big\{\meset=\bose\cup\mnset\colon\bose\in\bora,\
\mbox{$\mnset\subset\nset$ for some $\nset\in\nmu$}\big\}$
is a $\sigma$-algebra, and every element $\meset\in\cobora$ is said to be a \emph{$\mu$-measurable}
set. Moreover, there is a unique extension of $\mu$ to a probability measure $\comu$ on $\cobora$,
which is called the \emph{completion} of $\mu$; see Theorem~{1.9} of~\cite{Folland-RA} (also
see~\cite{Kechris}, Chapter~{II}, Section~{17.A}). We say that a $\mu$-measurable set
$\meset\in\cobora$ is \emph{of full measure} if $\comu(\meset)=1$.
The \emph{support} of the measure $\mu\in\prm(\mesp)$ is defined as the set
\begin{equation} \label{defsupp}
\supp(\mu)\defi\{\xi\in\mesp\colon\mbox{$\mu(\neigh)>0$, for every open neighborhood $\neigh$ of $\xi$}\}.
\end{equation}
Otherwise stated, $\supp(\mu)$ is the complement of the union of all points admitting a $\mu$-null open
neighborhood. It is a \emph{closed} subset of the topological space $\mesp$, as it is clear from the following:

\begin{fact} \label{faequi}
The set $\supp(\mu)^\cmp=\mesp\setminus\supp(\mu)$ is the union of all $\mu$-null open subsets of $\mesp$.
Therefore, $\supp(\mu)$ is the intersection of all closed subsets of $\mesp$ of full measure.
\end{fact}

\begin{remark} \label{remdir}
Let $\delta_\xi\in\prm(\mesp)$ be the Dirac (point mass) measure concentrated at $\xi\in\mesp$.
If the topological space $\mesp$ is $\mathsf{T}_2$ (Hausdorff), then, for every other point
$\xp\neq\xi$ in $\mesp$, there are open neighborhoods $\neighx$, $\neighxp$ of $\xi$ and $\xp$,
respectively, such that $\neighx\cap\neighxp=\varnothing$; hence, $\delta_\xi(\neighx)=1$ and
$\delta_\xi(\neighxp)=0$ ($\neighxp$ is contained in the complement of $\neighx$), so that
$\supp(\delta_\xi)=\{\xi\}$. Moreover, as is well known, a topological space $\mesp$ is $\mathsf{T}_1$
(Frech\'et) iff every singleton $\{\xi\}\subset\mesp$ is closed. Therefore, if, for every
$\xi\in\mesp$, there exists a Borel probability measure $\mu_\xi$ such that $\supp(\mu_\xi)=\{\xi\}$,
then the topological space $\mesp$ must be $\mathsf{T}_1$.
\end{remark}

Note that $\supp(\mu)$ may be empty. However, this will not be our case, because we will deal
with \emph{second countable} topological spaces.

\begin{fact} \label{famesp}
Let the topological space $\mesp$ be second countable. Then, the probability measure $\mu$ is
\emph{not locally measure zero}; i.e., $\supp(\mu)\neq\varnothing$. In fact, $\mu(\supp(\mu)^\cmp)=0$,
so that $\supp(\mu)\neq\varnothing$ and $\mu(\supp(\mu))=1$. Precisely, in this case $\supp(\mu)$ is
the smallest closed subset of $\mesp$ of full measure. In particular, if $\supp(\mu)$ is a singleton
--- i.e., if $\supp(\mu)=\{\xi\}$, for some $\xi\in\mesp$ --- then $\mu$ is the Dirac measure $\delta_\xi$
at $\xi$.
\end{fact}

Recall that a \emph{Polish space} $\posp$ is a separable, completely metrizable topological
space (see~\cite{Kechris}, Chapter~{I}, Section~{3}); namely, a topological space homeomorphic
to a complete metric space having a countable dense subset. Note that, in particular, a
Polish space is second countable. We will consider probability measures on the Borel space
$(\posp,\borap)$, where $\borap$ is the natural Borel structure on $\posp$ (i.e., the Borel
structure associated with the norm topology). More generally, a \emph{standard Borel space}
--- see~\cite{Kechris}, Chapter~{II}, Section~{12.B}, Chapter~{3} of~\cite{Srivastava},
or~\cite{Raja}, Chapter~{V}, Section~{1} --- is a measurable space $(\mesp,\salga)$ that is
isomorphic to $(\posp,\borap)$, for some Polish space $\posp$; equivalently, $(\mesp,\salga)$
is a standard Borel space if there is a Polish topology on $\mesp$ and $\salga=\bora$ is the Borel
$\sigma$-algebra associated with this topology. E.g., consider a Borel space of the form
$\big(\mesp\subset\posp,\borapx\defi\{\bose\cap\mesp\colon\bose\in\borap\}\big)$, where
$\posp$ is a Polish space. Note that $\borapx$ is precisely the Borel $\sigma$-algebra
associated with the subspace topology of $\mesp$ as a subset of $\posp$. It turns out that
$\big(\mesp,\bora\equiv\borapx\big)$ is a standard Borel space iff $\mesp$ is a \emph{Borel}
subset of $\posp$ (\cite{Raja}, Chapter~{V}, Section~{1}); actually, \emph{every} standard
Borel space is, up to isomorphisms, of this form (this is the so-called \emph{Ramsey-Mackey theorem};
see Theorem~{3.3.22} of~\cite{Srivastava}, where a standard Borel space is defined precisely
as a measurable space isomorphic to a Borel subset of a Polish space).
The motivation for using the term \emph{standard} rests on the following result:

\begin{fact} \label{faposp}
A \emph{standard} Borel space $(\mesp,\salga=\bora)$ is completely characterized --- up to
Borel isomorphisms --- by its cardinality; i.e., two standard Borel spaces are isomorphic iff
they have the same cardinality (Kuratowski's theorem). In particular, the set $\mesp$ is finite,
countable or has the power of the continuum; in the latter case, it is Borel isomorphic to the
unit interval $[0,1]$ (endowed with the Borel $\sigma$-algebra associated with the standard topology).

If $(\mesp,\bora)$, $(\mespy,\boray)$ are standard Borel spaces and $f\colon\mesp\rightarrow\mespy$
is an injective Borel map, then, $f$ maps Borel sets to Borel sets; i.e., for every $\bose\in\bora$,
$f(\bose)\in\boray$. Therefore, the mapping $\mesp\ni\xi\mapsto f(\xi)\in f(\mesp)$ --- where $f(\mesp)$
is endowed with the Borel $\sigma$-algebra $\borapxy$ associated with the subspace topology --- is a
Borel isomorphism.
\end{fact}

\section{Tools: tensor products, cross trace classes and cross states}
\label{tensor}

We will now recall some salient facts about the tensor products of Hilbert spaces and of Hilbert
space operators~\cite{Prugovecki,Moretti,Kubrusly}. Next, we will introduce the projective tensor
product~\cite{Grothendieck-res,Grothendieck-book,Defant,Ryan,DFS,Kubrusly} of trace classes, and
the related notions of cross trace class operator and of cross state of a bipartite quantum system.
The notion of cross state will allow us to obtain, in the subsequent sections, remarkable characterizations
of separability and entanglement for a bipartite quantum system.

\subsection{Introducing the relevant tensor products: Hilbert spaces and operators}
\label{introtens}

Let us first recall that the the tensor product $\hh\otimes\jj$ of two separable complex Hilbert
spaces $\hh$ and $\jj$ can be introduced as follows~\cite{Moretti,Kubrusly}. For every pair of vectors
$\phi\in\hh$ and $\psi\in\jj$, we can define a bilinear form $\phi\otimes\psi$ on $\hhd\times\jjd$,
where $\hhd$, $\jjd$ are the (topological) \emph{dual spaces} of $\hh$ and $\jj$, respectively; i.e.,
we set
\begin{equation} \label{deftepo}
\phi\otimes\psi\colon\hhd\times\jjd\ni\big(\deta,\dchi\big)\mapsto
\langle\eta,\phi\rhh\,\langle\chi,\psi\rjj\in\ccc \fin , \quad \eta\in\hh \fin ,\ \chi\in\jj .
\end{equation}
Clearly, the Hilbert space $\hh$ (or $\jj$) is identified with its dual space $\hhd$ (respectively,
with $\jjd$) via the usual antilinear isomorphism $\aih\colon\hh\ni\eta\mapsto\etap\in\hhd$, where
$\etap\equiv\deta\colon\hh\rightarrow\ccc$ (respectively, $\aij\colon\jj\ni\chi\mapsto\chip\equiv\dchi\in\jjd$),
and we have:
$\big(\phi\otimes\psi\big)(\etap,\chip)=\etap(\phi)\otto\chip(\psi)=\langle\eta,\phi\rhh\,\langle\chi,\psi\rjj$.

Note that putting $\lispa\big\{\phi\otimes\psi\big\}\equiv\lispa\big\{\phi\otimes\psi\in\linhjd\colon
\phi\in\hh,\ \psi\in\jj\big\}$, where $\linhjd$ is the complex vector space of all bilinear forms on
$\hhd\times\jjd$, the map
\begin{equation} \label{defvbf}
\vbf\colon\hh\times\jj\ni(\phi,\psi)\mapsto\phi\otimes\psi\in\lispa\big\{\phi\otimes\psi\big\}
\equiv\lispa\big\{\phi\otimes\psi\in\linhjd\colon\phi\in\hh,\ \psi\in\jj\big\}
\end{equation}
is a (bounded) bilinear map. This observation allows us identify the pair
$\big(\lispa\big\{\phi\otimes\psi\big\},\vbf\big)$ with the \emph{algebraic tensor product} of
the vector spaces $\hh$ and $\jj$, as defined in the theory of tensor products of abstract vector
spaces --- in particular, normed and Banach spaces --- recall Subsection~\ref{algtens}, and see,
e.g., Chapter~{1} of~\cite{Ryan}. In fact, for any $\phi\in\hh$ and $\psi\in\jj$, let
$\phi\altp\psi\colon\linhj\rightarrow\ccc$ be the linear form --- on the complex vector space
$\linhj$ of all bilinear forms on $\hh\times\jj$ --- defined by
\begin{equation}
\big(\phi\altp\psi\big)(\bfo)\defi\bfo(\phi,\psi) \fin , \quad \forall\cinque\bfo\in\linhj \fin .
\end{equation}
Then, the pair formed by the (complex) linear span
\begin{equation}
\lispa\big\{\phi\altp\psi\big\}\equiv\lispa\big\{\phi\altp\psi\in\linhjad\colon
\phi\in\hh,\ \psi\in\jj\big\}
\end{equation}
and by the bilinear map
$\svbf\colon\hh\times\jj\ni(\phi,\psi)\mapsto\phi\altp\psi\in\lispa\big\{\phi\altp\psi\big\}$
is the standard realization of the algebraic tensor product of $\hh$ and $\jj$. Then, recalling
Definition~\ref{defreatp}, we have:

\begin{proposition} \label{proatps}
The pair $\big(\lispa\big\{\phi\otimes\psi\big\},\vbf\colon\hh\times\jj\rightarrow
\lispa\big\{\phi\otimes\psi\big\}\big)$ is a realization
of the algebraic tensor product of $\hh$ and $\jj$. In fact, the mapping
\begin{equation}
\big\{\phi\altp\psi\colon\phi\in\hh,\ \psi\in\jj\big\}\ni\eta\altp\chi\mapsto
\eta\otimes\chi\in\big\{\phi\otimes\psi\colon\phi\in\hh,\ \psi\in\jj\big\}
\end{equation}
extends to a linear isomorphism $\vlm\colon\lispa\big(\ran\big(\svbf\big)\big)=\lispa\big\{\phi\altp\psi\big\}
\rightarrow\lispa\big\{\phi\otimes\psi\big\}=\lispa\big(\ran(\vbf)\big)$ that satisfies the universality
relation $\vbf=\vlm\circ\svbf$.
\end{proposition}

By Proposition~\ref{proatps}, we can define the algebraic tensor product of the Hilbert spaces
$\hh$ and $\jj$ as the pair
\begin{equation}
\big(\hh\altp\jj\defi\lispa\big\{\phi\otimes\psi\in\linhjd\colon \phi\in\hh,\ \psi\in\jj\big\},
\vbf\colon\hh\times\jj\rightarrow\hh\altp\jj\big) \fin .
\end{equation}

At this point, the complex vector space $\hh\altp\jj$ can be endowed with a (well defined) natural
scalar product by setting $\langle\phi_1\otimes\psi_1 , \phi_2\otimes\psi_2\rahaj=\langle\phi_1,
\phi_2\rah\sei\langle\psi_1, \psi_2\raj$,
and then extending the mapping $\scaphaj\colon(\hh\altp\jj)\times(\hh\altp\jj)\rightarrow\ccc$
--- by antilinearity in the first argument and by linearity in the second one --- so obtaining
a sesquilinear form on $\hh\altp\jj$. Finally, the tensor product Hilbert space
$\hh\otimes\jj\equiv(\hh\otimes\jj,\scaphj)$ is obtained as the Hilbert space completion of
this space; see Section~{10.2.1} of~\cite{Moretti} (or, Section~{9.2} of~\cite{Kubrusly})
for all technical details. Clearly, the application
$\vbf\colon\hh\times\jj\ni(\phi,\psi)\mapsto\phi\otimes\psi\in\hh\otimes\jj$
is a bounded bilinear map --- which, with a slight abuse, we denote by the same symbol adopted
for the map~\eqref{defvbf} --- and can be regarded, again with a slight abuse, as a (realization of)
the natural bilinear map on $\hh\times\jj$.

\begin{fact}[The Schmidt decomposition of a bipartite vector] \label{schmdec}
Every \emph{nonzero} vector $\va\in\hhjj$ can be expressed as a sum (norm-convergent, if infinite)
of the form
\begin{equation} \label{schmdecon}
\va=\sum_{n=1}^{\sra} a_n\tre (\phi_n\otimes\psi_n) \fin , \quad
1\le\sra\equiv\srank(\va)\le\emme\equiv\min\{\dim(\hh),\dim(\jj)\}\le\infty \fin ,
\end{equation}
where $\{\phi_n\}_{n=1}^\sra$, $\{\psi_n\}_{n=1}^\sra$ are orthonormal systems in the Hilbert spaces
$\hh$ and $\jj$, respectively, and $\{a_n\}_{n=1}^\sra\equiv\scfs(\va)$ is a set of \emph{strictly positive
real numbers} --- the so-called \emph{Schmidt coefficients} of $\va$ --- that is uniquely determined
when arranged in a weakly decreasing order $a_1\ge a_2\ge\cdots$. In particular, the cardinality
$\sra\equiv\srank(\va)\le\min\{\dim(\hh),\dim(\jj)\}$ of the set $\scfs(\va)$ --- the so-called
\emph{Schmidt rank} of $\va$ --- is uniquely determined.
\end{fact}

\begin{remark} \label{exschmdec}
By Fact~\ref{schmdec}, every vector $\va\in\hhjj$ admits an expansion of the form
\begin{equation} \label{exschmdecon}
\va=\sum_{m=1}^{\emme} a_m\tre (\phi_m\otimes\psi_m) \fin , \quad
\emme\equiv\min\{\dim(\hh),\dim(\jj)\}\le\infty \fin ,
\end{equation}
where $a_m\ge 0$, and $\{\phi_m\}_{m=1}^\emme$, $\{\psi_m\}_{m=1}^\emme$ are orthonormal
systems in $\hh$ and $\jj$, respectively. We will call any expansion of the form~\eqref{exschmdecon}
an \emph{extended Schmidt decomposition} of the vector $\va$.
\end{remark}

At this point, the tensor product of two bounded operators $A\in\bop$, $B\in\bopj$ can be defined
in the standard way (see Section~{10.2.2} of~\cite{Moretti}, Section~{9.2} of~\cite{Kubrusly} and
Section~{6.3} of~\cite{Murphy}), i.e., $A\otimes B$ is the unique bounded linear operator on the
Hilbert space $(\hh\otimes\jj,\scaphj)$ such that
$\big(A\otimes B\big)(\phi\otimes\psi)=(A\phi)\otimes(B\psi)$, for every pair of vectors
$\phi\in\hh$ and $\psi\in\jj$, and --- denoting by $\tnormi$ the norm of the Banach space $\bophj$
--- $\|A\otimes B\tnori=\|A\nori\tre \|B\nori$ (the operator norm $\tnormi$ is a \emph{cross norm});
moreover, $(A\otimes B)^\ast=A^\ast\otimes B^\ast$ and $(A\otimes B)(E\otimes F)= (AE)\otimes(BF)$.

\begin{remark} \label{reboplin}
Regarding $\bop$, $\bopj$ as linear subspaces of $\linhh$ and $\linjj$, respectively, then,
by Fact~\ref{facteplin}, putting
\begin{equation}
\bop\altp\bopjj\defi\lispa\{A\otimes B\colon A\in\bop,\ B\in\bopjj\} \fin ,
\end{equation}
the pair
\begin{equation}
\big(\bop\altp\bopjj, \bop\times\bopjj\ni (A,B)\mapsto A\otimes B\in\bop\altp\bopjj)
\end{equation}
is a realization of the algebraic tensor product of $\bop$ and $\bopjj$. Here, with a
slight abuse, we are identifying the linear operator $A\otimes B\in\linhh\altp\linjj$
--- uniquely defined on the algebraic tensor product $\hh\altp\jj$, according to the first
claim in Fact~\ref{facteplin} --- with its continuous extension to $\hhjj$ (see Section~{9.2}
of~\cite{Kubrusly}, and in particular Theorem~{9.8} therein, for a more precise, but also
more cumbersome, notation). By the second claim in Fact~\ref{facteplin}, the algebraic
tensor product of any pair of linear subspaces of $\linhh$ and $\linjj$, respectively,
can be realized analogously.
\end{remark}

Let us now focus on the tensor product of trace class operators. By the last assertion in
Remark~\ref{reboplin}, setting
\begin{equation} \label{lisp}
\thatj\defi\lispa\{S\otimes T\colon S\in\trc,\ T\in\trcjj\} \fin ,
\end{equation}
the pair $\big(\thatj,\trc\times\trcjj\ni (S,T)\mapsto S\otimes T\in\thatj\big)$ is (a realization of)
the \emph{algebraic tensor product} of the trace classes $\trc$ and $\trcjj$.

With a standard abuse, in the following we will usually denote the algebraic tensor product of
$\trc$ and $\trcjj$ by $\thatj$ \emph{tout court}; i.e., by the symbol of the tensor product space.

Analogous considerations will be understood for the algebraic tensor product
\begin{equation}
\thsatjs\defi\lispar\{S\otimes T\colon S\in\trcsa,\ T\in\trcsajj\}
\end{equation}
of the \emph{real} Banach spaces $\trcsa$ and $\trcsajj$ (\emph{selfadjoint} trace class operators).

\begin{fact} \label{factrc}
If $S\in\trc$ and $T\in\trcjj$, then $S\otimes T\in\trchj$ --- hence, $\thatj$ is a linear
subspace of $\trchj$ --- and $\tr(S\otimes T)=\tr(S)\sei\tr(T)$; moreover,
$|S\otimes T|=|S|\otimes |T|$ so that the trace norm $\ttrnor$ of the Banach space $\trchj$
is a cross norm, namely,
\begin{equation} \label{crono}
\|S\otimes T\ttrn\defi\tr(|S\otimes T|)=\tr(|S|\otimes |T|)=\|S\trn \tre \|T\trn \fin , \quad
\forall\cinque S\in\trc \fin ,\  \forall\cinque T\in\trcjj \fin .
\end{equation}
\end{fact}

\begin{fact} \label{factdenso-bis}
Given trace class operators $X\in\trc$ and $Y\in\trcjj$, if $X\otimes Y\in\trchjp$, then
$X\otimes Y=S\otimes T$, for some $S\in\trcp$ and $T\in\trcjjp$.
\end{fact}

\begin{fact} \label{factdenso}
Given nonzero positive trace class operators $S,X\in\trcp$ and $T,Y\in\trcjjp$, we have:
$S\otimes T=X\otimes Y$ iff $S=r\tre X$ and $T=r^{-1}\tre Y$, where $r=\tr(Y)/\tr(T)=\tr(S)/\tr(X)>0$.
Then, given any density operators $\rho,\sigma\in\sta$ and $\tau,\omega\in\stajj$,
$\rho\otimes\sigma=\tau\otimes\omega$ iff $\rho=\tau$ and $\sigma=\omega$.
\end{fact}

As previously noted, since, for $S\in\trc$ and $T\in\trcjj$, the bounded operator $S\otimes T$
is contained in the trace class $\trchj$ of the Hilbert space $\hh\otimes\jj$, then $\trc\altp\trcjj$
is (identified with) a linear subspace of $\trchj$, and then can be endowed with the norm $\ttrnor$ of this
space, i.e., $\|{\textstyle \sum_k S_k\otimes T_k}\ttrn\defi\tr(|{\textstyle \sum_k S_k\otimes T_k}|)$;
namely, with the \emph{natural} or \emph{spatial} norm of the algebraic tensor product $\thatj$ (these terms
refer to the fact that this is a restriction of the norm $\ttrnor$ of the \emph{natural ambient space}
$\trchj$ containing $\thatj$). Clearly, for $S\in\trcsa$ and $T\in\trcsajj$, $S\otimes T\in\trchjsa$,
and $\thsatjs$ is a linear subspace of $\trchjsa$.

\subsection{The natural tensor product of trace classes}
\label{natens}

By Fact~\ref{factrc}, the norm $\ttrnor$ of the Banach space $\trchj$ is a \emph{cross norm},
and the algebraic tensor product $\trc\altp\trcjj$ can be endowed with the natural
restriction of this norm (the natural, or spatial, norm); i.e.,
$\|S\otimes T\ttrn = \|S\trn \tre \|T\trn$, for all $S\in\trc$ and $T\in\trcjj$. We now
consider the Banach space completion of the normed space $(\thatj,\ttrnor)$; namely, we put
\begin{equation} \label{natp}
\trc\otimes\trcjj\defi\clispa\{S\otimes T\colon S\in\trc,\ T\in\trcjj\}
=\clttrnor\big(\trc\altp\trcjj\big) \fin ,
\end{equation}
where $\clispa\equiv\clispat\equiv\clispattrn$ ($\ttrnor$-closed linear span). Clearly,
$\trc\otimes\trcjj$ is a Banach subspace of the the \emph{bipartite trace class} $\trchj$,
i.e., with the trace class of the bipartite Hilbert space $\hhjj$ (but also see
Theorem~\ref{pronatp} below, for a more precise characterization). Analogously, we can define
the following subspace of the real Banach space $\trchjsa$:
\begin{equation} \label{rnatp}
\trcsa\otimes\trcsajj\defi\clispar\{S\otimes T\colon S\in\trcsa,\ T\in\trcsajj\}
=\clttrnor\big(\trcsa\altp\trcsajj\big) \fin .
\end{equation}

\begin{definition}
We call the Banach space $(\trc\otimes\trcjj,\ttrnor)$ defined by~\eqref{natp} the \emph{natural}
--- or \emph{spatial} --- tensor product of the trace classes $\trc$ and $\trcjj$. Analogously,
$\trcsa\otimes\trcsajj$, endowed with (the restriction of) the trace norm $\ttrnor$, is called
the \emph{natural} (or \emph{spatial}) tensor product of $\trcsa$ and $\trcsajj$.
\end{definition}

Given any nonempty subset $\subhj$ of $\trc\otimes\trcjj$, coherently with the previously introduced
notation, we denote by $\clispa(\subhj)\equiv\clispat(\subhj)$ the $\ttrnor$-closed linear span of
$\subhj$; i.e., the closed subspace $\clttrnor(\lispa(\subhj))$ of $\trc\otimes\trcjj$.

Note that the application
\begin{equation} \label{appnb}
\nb\colon\trc\times\trcjj\ni (S,T) \mapsto S\otimes T\in\trchj
\end{equation}
is a bounded bilinear map --- i.e., $\nb\in\lthtjthj\equiv\bo(\trc,\trcjj;\trchj)$ --- the
\emph{natural bilinear map} on $\trc\times\trcjj$; precisely, by relation~\eqref{crono},
$\nb$ is a bilinear isometry --- $\|\nb(S,T)\ttrn=\|S\otimes T\ttrn=\|S\trn\tre\|T\trn$,
for all $(S,T)\in\trc\times\trcjj$, so that
\begin{equation}
\|\nb\norb\defi\sup_{\substack{0\neq S\in\trc \\ 0\neq T\in\trcjj}}
\frac{\|\nb(S,T)\ttrn}{\|S\trn\tre\|T\trn}
=1 \fin .
\end{equation}

\begin{notation} \label{notsit}
Given any pair $\subh$, $\subj$ of nonempty subsets of $\trc$ and $\trcjj$, respectively, in
the following we set
\begin{equation}
\nb(\subh,\subj)\defi\{\nb(S,T)=S\otimes T\colon S\in\subh, \ T\in\subj\} \fin ,
\end{equation}
whereas we will \emph{refrain from} denoting this set by $\subh\otimes\subj$ (see Notation~\ref{notensa} below).
\end{notation}

\begin{remark}
Recall that by a well-known result --- see, e.g., Theorem~{6.2} of~\cite{Kubrusly} --- the boundedness
of a bilinear map $\bim\colon\bsb\times\bsc\rightarrow\bs$, where $\bsb$, $\bsc$, $\bs$ are complex
Banach spaces, is equivalent to its \emph{joint} continuity wrt the metric
\begin{equation}
\df((x,w),(y,z))\defi\max\{\|x-y\|,\|w-z\|\} \fin , \quad x,y\in\bsb \fin ,
\ w,z\in\bsc
\end{equation}
--- i.e., wrt the product topology --- or, also, to is \emph{separate} continuity. Therefore, the
natural bilinear map on $\trc\times\trcjj$ is continuous wrt the product topology; i.e., wrt the metric
\begin{equation}
\trdd((X,W),(Y,Z))\defi\max\{\|X-Y\trn,\|W-Z\trn\} \fin , \quad X,Y\in\trc  \fin ,
\ W,Z\in\trcjj \fin .
\end{equation}
\end{remark}

\begin{notation} \label{notensa}
Let $\subh$, $\subj$ (nonempty) \emph{subsets} of $\trc$ and $\trcjj$, respectively. We put
\begin{equation}
\subh\altp\subj\defi\co(\nb(\subh,\subj)) \fin , \quad
\subh\otimes\subj\defi\clco(\nb(\subh,\subj))=\clttrnor\big(\subh\altp\subj\big) \fin ;
\end{equation}
i.e., $\subh\otimes\subj$ is the $\ttrnor$-closed convex hull
$\clco(\nb(\subh,\subj))\equiv\clcon(\nb(\subh,\subj))$ of the set $\nb(\subh,\subj)$.
\end{notation}

Privileging the \emph{convex} structures wrt to the linear ones in Notation~\ref{notensa} is essentially
motivated by the applications we have in mind. However, it is worth observing that, in the case where
$\subh=\trc$ and $\subj=\trcjj$, Notation~\ref{notensa} is consistent with our previous definition
of the algebraic tensor product $\trc\altp\trcjj$ and of the natural tensor product $\trc\otimes\trcjj$
as $\lispa\big(\nb(\trc,\trcjj)\big)$ and $\clispa\big(\nb(\trc,\trcjj)\big)$, respectively; see
Corollary~\ref{corlispa} below.

\begin{proposition} \label{procoh}
For every pair $\subh$, $\subj$ of (nonempty) subsets of $\trc$ and $\trcjj$, respectively,
\begin{equation} \label{renb}
\lispa(\nb(\subh,\subj))=\lispa\big(\nb(\lispa(\subh),\lispa(\subj))\big) \fin ,
\end{equation}
\begin{equation} \label{renb-bis}
\lispa(\nb(\subh,\subj))=\co\big(\nb(\lispa(\subh),\lispa(\subj))\big)
\ifed\lispa(\subh)\altp\lispa(\subj)
\end{equation}
and
\begin{equation} \label{renb-tris}
\subh\altp\subj\defi\co(\nb(\subh,\subj))=\co\big(\nb(\co(\subh),\co(\subj))\big)
\ifed\co(\subh)\altp\co(\subj) \fin ,
\end{equation}
so that
\begin{equation} \label{recli}
\clispa(\nb(\subh,\subj))=\clispa\big(\nb(\lispa(\subh),\lispa(\subj))\big) \fin ,
\end{equation}
\begin{equation} \label{relisco}
\clispa(\nb(\subh,\subj))=\clco\big(\nb(\lispa(\subh),\lispa(\subj))\big)
\ifed\lispa(\subh)\otimes\lispa(\subj)
\end{equation}
and
\begin{align}
\subh\otimes\subj\defi\clco(\nb(\subh,\subj))
& =
\clttrnor\big(\co\big(\nb(\co(\subh),\co(\subj))\big)\big)
\nonumber \\ \label{relsua}
& =
\clco\big(\nb(\co(\subh),\co(\subj))\big)\ifed\co(\subh)\otimes\co(\subj) \fin .
\end{align}
Moreover, we have that
\begin{equation} \label{relsub}
\subh\otimes\subj\defi\clco(\nb(\subh,\subj))=\clco\big(\nb(\clco(\subh),\clco(\subj))\big)
\ifed\clco(\subh)\otimes\clco(\subj)
\end{equation}
--- where $\clco(\subh)=\cltrnor(\co(\subh))$, $\clco(\subj)=\cltrnor(\co(\subj))$ ---
and
\begin{equation} \label{relsuc}
\clispa\big(\nb(\subh,\subj)\big)=\clispa\big(\nb(\clco(\subh),\clco(\subj))\big)
=\clispa(\subh\otimes\subj) \fin .
\end{equation}
\end{proposition}

\begin{corollary} \label{corlispa}
Let $\subh$, $\subj$ be linear subspaces of $\trc$ and $\trcjj$, respectively. Then,
we have:
\begin{equation}
\subh\otimes\subj\defi\clco\big(\nb(\subh,\subj)\big)=\clispa(\nb(\subh,\subj)) .
\end{equation}
\end{corollary}

\begin{corollary} \label{proten}
Every element of $\trc\altp\trcjj$ can be expressed as a  (finite) linear combination of states of the form
$\rho\otimes\sigma$; i.e.,
\begin{equation} \label{relispa}
\trc\altp\trcjj=\lispa\big(\nb(\sta,\stajj)\big)=
\lispa\{\rho\otimes\sigma\colon \rho\in\sta,\ \sigma\in\stajj\}.
\end{equation}
Therefore, we have:
\begin{equation} \label{reclispa}
\trc\otimes\trcjj=\clispa\big(\nb(\sta,\stajj)\big) \fin .
\end{equation}
\end{corollary}

\begin{corollary} \label{corsta}
Since $\sta=\clco(\pusta)$ and $\stajj=\clco(\pustajj)$, then
\begin{equation} \label{reldepu}
\sta\otimes\stajj\defi\clco\big(\nb(\sta,\stajj)\big)=
\clco\big(\nb(\pusta,\pustajj)\big)\ifed\pusta\otimes\pustajj
\end{equation}
and
\begin{equation} \label{charnat}
\trc\otimes\trcjj=\clispa\big(\nb(\sta,\stajj)\big)=
\clispa\big(\nb(\pusta,\pustajj)\big) \fin ,
\end{equation}
\begin{equation} \label{charnat-bis}
\trc\otimes\trcjj=\clispa(\sta\otimes\stajj)
=\clispa(\pusta\otimes\pustajj) \fin .
\end{equation}
\end{corollary}

We are now almost ready to derive the --- previously announced --- more precise characterization of
the natural tensor product $\trc\otimes\trcjj$. We only need to note a further technical fact.

\begin{lemma} \label{lefr}
Every operator of the form $\hvac\defi\vac$, for some vectors $\va,\vc\in\hhjj$, belongs to
the natural tensor product $\trc\otimes\trcjj$. Precisely, we have that
\begin{equation} \label{inclis}
\hvac\in\clispa\Big\{\hep\otimes\hcp\colon \eta,\phi\in\hh,\ \chi,\psi\in\jj\Big\}\subset
\trc\otimes\trcjj \fin .
\end{equation}
Moreover, for the linear space $\frhj=\lispa\big\{\hvac\colon\va,\vc\in\hhjj\big\}$ of all
\emph{finite rank operators} on $\hhjj$, we have:
\begin{equation} \label{inclfr}
\clttrnor(\frhj)=\clispa\Big\{\hep\otimes\hcp\colon \eta,\phi\in\hh,\ \chi,\psi\in\jj\Big\} \fin .
\end{equation}
In particular, putting $\hva\equiv\hvaa\defi{\vaa}$, $\va\in\hhjj$, we have that
\begin{equation} \label{inclis-bis}
\hva\in\clispar\Big\{\hph\otimes\hp\colon \phi\in\hh,\ \psi\in\jj\Big\}\subset
\trcsa\otimes\trcsajj \fin ,
\end{equation}
with $\hph\equiv\hpph$, $\hp\equiv\hpp$, and
\begin{equation} \label{inclfr-bis}
\clispar(\pustahj)=\clispar\Big\{\hph\otimes\hp\colon \phi\in\hh,\ \psi\in\jj\Big\} \fin ,
\end{equation}
where $\pustahj$ is the set of all \emph{rank-one projections} on $\hhjj$.
\end{lemma}

\begin{theorem} \label{pronatp}
For every pair of (separable, complex) Hilbert spaces $\hh$ and $\jj$, the natural tensor
product $\trc\otimes\trcjj$ coincides with the whole bipartite trace class $\trchj$. Indeed,
we have that
\begin{align}
\trc\otimes\trcjj
& =
\clispa\{\rho\otimes\sigma\colon \rho\in\sta,\ \sigma\in\stajj\}
\nonumber \\
& =
\clispa\{\pi\otimes\vp\colon \pi\in\pusta,\ \vp\in\pustajj\}
\nonumber \\ \label{idts}
& =
\clispa\Big\{\hep\otimes\hcp\colon \eta,\phi\in\hh,\ \chi,\psi\in\jj\Big\}
= \trchj \fin ,
\end{align}
where $\hep\equiv\ep$. Analogously, we have:
\begin{align}
\thsotjs
& =
\clispar\{\pi\otimes\vp\colon \pi\in\pusta,\ \vp\in\pustajj\}
\nonumber \\  \label{idts-bis}
& =
\clispar\{\rho\otimes\sigma\colon\rho\in\sta,\ \sigma\in\stajj\}
=\trchjsa \fin .
\end{align}
\end{theorem}

\begin{proof}
The first two equalities in~\eqref{idts} hold by relation~\eqref{charnat} in Corollary~\ref{corsta}.
Moreover, since --- as a consequence, say, of the \emph{singular value decomposition}
(Remark~\ref{sivadeb}) of a trace class operator ---
\begin{equation}
\trchj=\clttrnor(\frhj) \fin ,
\end{equation}
by relation~\eqref{inclfr} in Lemma~\ref{lefr} one must conclude that
\begin{equation}
\trchj=\clispa\big\{\hep\otimes\hcp\colon \eta,\phi\in\hh,\ \chi,\psi\in\jj\big\}
\subset\clttrnor(\trc\altp\trcjj)  \fin ,
\end{equation}
and hence $\clttrnor(\trc\altp\trcjj)=\trc\otimes\trcjj\subset\trchj\subset\clttrnor(\trc\altp\trcjj)$;
so that, actually,
\begin{equation}
\trchj=\clispa\big\{\hep\otimes\hcp\colon \eta,\phi\in\hh,\ \chi,\psi\in\jj\big\}=
\trc\otimes\trcjj \fin .
\end{equation}
Analogously, by the spectral decomposition of a selfadjoint trace class operator (that converges wrt
the trace norm) and by relation~\eqref{inclfr-bis}, we have:
\begin{align}
\trchjsa=\clispar(\pustahj)
& =
\clispar\{\pi\otimes\vp\colon \pi\in\pusta,\ \vp\in\pustajj\}
\nonumber \\ \label{idts-tris}
& \subset
\clttrnor(\trcsa\altp\trcsajj)=\thsotjs \fin .
\end{align}
Then, since $\thsotjs\subset\trchjsa$, the inclusion in relation~\eqref{idts-tris} is actually
an equality. It only remains to prove that the second equality in~\eqref{idts-bis} holds too.
In fact, this equality is an immediate consequence of the relation
\begin{align}
\lispar\{\rho\otimes\sigma\colon\rho\in\sta,\ \sigma\in\stajj\}
& \subset
\clispar\{\pi\otimes\vp\colon \pi\in\pusta,\ \vp\in\pustajj\}
\nonumber \\
& \subset
\clispar\{\rho\otimes\sigma\colon\rho\in\sta,\ \sigma\in\stajj\} \fin ,
\end{align}
where the first inclusion holds by the spectral decomposition of a density operator and by the
continuity of the natural bilinear map~\eqref{appnb}.
\end{proof}

\begin{remark} \label{rebop}
The fact that $\thotj=\trchj$, for any dimension of $\hh$ and $\jj$, is nontrivial. Indeed, it is
known (see~\cite{Kadison}, Example~{11.1.5}, Example~{11.1.6} and Exercise~{11.5.7}) that the \emph{natural,
or spatial, tensor product} $\bhobj$ --- i.e., the Banach space completion of the algebraic tensor product
$\bhabj$ wrt the norm $\tnormi$ of $\bophj$, also called the \emph{tensor product of the (concrete, or
represented) $\cast$-algebras} $\bop$ and $\bopjj$ (see, e.g., Section~{11.1} of~\cite{Kadison}) ---
coincides with the whole Banach space $\bophj$ iff $\min\{\dim(\hh),\dim(\jj)\}<\infty$, and in such a case
we actually have that $\bop\altp\bopjj=\bhobj=\bophj$; if $\dim(\hh)=\dim(\jj)=\infty$, instead,
$\bhobj$ is a \emph{proper} Banach subspace of $\bophj$. In the case where $\dim(\hh)=\dim(\jj)=\infty$, to
obtain the full Banach space $\bophj$, one needs to consider the so-called \emph{von~Neumann tensor product} of
the `local' von~Neumann algebras $\bop$ and $\bopjj$ (see Definition~{10.34} of~\cite{Moretti}, and the subsequent
discussion; also see Section~{11.1} of~\cite{Kadison}); i.e., the closure of the algebraic tensor product $\bhabj$
wrt the strong (equivalently, the weak) operator topology of $\bophj$. If $\min\{\dim(\hh),\dim(\jj)\}<\infty$,
the spatial tensor product and the von~Neumann tensor product coincide (Exercise~{11.2.2} of~\cite{Kadison}).
\end{remark}

Now, let $\bs$ be a complex Banach space, and let $\Bim\colon\trc\otimes\trcjj\rightarrow\bs$
be a bounded linear map. Then,
\begin{equation} \label{nalin}
\bim\defi\Bim\circ\nb\colon\trc\times\trcjj\rightarrow\bs
\end{equation}
is a bounded bilinear map, the bilinear map \emph{induced by} $\Bim$.

\begin{definition} \label{defnatlin}
A bounded bilinear map $\bim\colon\trc\times\trcjj\rightarrow\bs$ is said to be \emph{naturally linearizable}
if is of the form~\eqref{nalin}, for some \emph{bounded} linear map $\Bim\colon\trc\otimes\trcjj\rightarrow\bs$.
\end{definition}

\begin{proposition} \label{probim}
Let $\bs$ be a complex Banach space, and let $\bim\colon\trc\times\trcjj\rightarrow\bs$
be a bounded bilinear map. Then, there exists a unique linear map
$\Bima\colon\trc\altp\trcjj\rightarrow\bs$ such that
\begin{equation} \label{unipro}
\bim (S,T)=\Bima(S\otimes T) \fin , \quad \forall\cinque S\in\trc \fin ,\  \forall\cinque T\in\trcjj \fin .
\end{equation}
The bounded bilinear map $\bim$ is naturally linearizable iff the linear map $\Bima\colon\thatj\rightarrow\bs$
is bounded, $\thatj$ being endowed with (the restriction of) the norm $\ttrnor$. Moreover, if $\bim$ is naturally
linearizable, then the bounded linear map $\Bim\colon\trc\otimes\trcjj\rightarrow\bs$ --- the so-called
\emph{natural linearization} of (or \emph{natural linear map} inducing) $\bim$ --- is unique (i.e.,
the unique bounded extension of $\Bima$) and $\|\bim\norps\le\|\Bim\norps$.
\end{proposition}

\subsection{The projective tensor products}
\label{protenso}

We will denote by $\trc\prtp\trcjj$ the \emph{projective tensor product}~\cite{Grothendieck-book,Defant,Ryan,DFS,Kubrusly}
of the Banach spaces $\trc$ and $\trcjj$; i.e., the Banach space completion of the algebraic tensor product
$\trc\altp\trcjj$ wrt the norm
\begin{equation} \label{pronorm}
\|C\ptrn\defi\inf\big\{{\textstyle \sum_k \|S_k\trn\tre\|T_k\trn}\colon {\textstyle\sum_k
S_k\otimes T_k}=C \big\} , \quad C\in\trc\altp\trcjj \fin .
\end{equation}
Here, the infimum is taken over all (finite) decompositions of $C$ into elementary tensors. Note that,
by its definition, the norm~\eqref{pronorm} is the largest possible (sub-)cross norm on $\trc\altp\trcjj$,
and extends by continuity to a norm on $\trc\prtp\trcjj$ --- the so-called \emph{projective norm}
--- that will be still denoted by $\ptrnor$. Since $\|S\trn=\|S^\ast\trn$, for every $S\in\trc$,
it is clear that
\begin{equation} \label{reladj}
\|C\ptrn = \|C^\ast\ptrn \fin , \quad \forall\cinque C\in\trc\altp\trcjj \fin .
\end{equation}
This property actually extends to the full Banach space $\trc\prtp\trcjj$ (see Corollary~\ref{corcro} below).

\begin{remark} \label{tenhil}
As is well known, the Hilbert space tensor product $\hhjj$ is isomorphic to the Hilbert
space $\hsjth$ of all Hilbert-Schmidt operators from $\jj$ into $\hh$, and this isomorphism
can be implemented by extending --- by linearity, and, next, by continuity --- the mapping
\begin{equation} \label{canoma}
\phi\otimes\psi\mapsto\pjp \fin , \quad J\colon\jj\rightarrow\jj ,
\end{equation}
where $J$ is any \emph{complex conjugation} (an involutive antiunitary operator) on $\jj$.
Notice that the mapping~\eqref{canoma} admits a unique linear extension by the universal
property of the algebraic tensor product $\hh\altp\jj$, because the application
$\hh\times\jj\ni (\phi,\psi)\mapsto\pjp\in\hsjth$ is a bilinear map.
It is known~\cite{Grothendieck-book,Defant,Ryan,DFS,Kubrusly} --- also see~\cite{Schatten},
Theorem~{5.12} in Section~{10} of Chapter~{V} --- that, analogously, the Banach space
completion $\hh\prtp\jj$ of the algebraic tensor product $\hh\altp\jj$ wrt the projective norm
\begin{equation} \label{vpronorm}
\pnh{\va}\defi\inf\big\{{\textstyle \sum_k \|\phi_k\trn\tre \|\psi_k\trn}\colon
{\textstyle\sum_k \phi_k\otimes\psi_k}=\va \big\} , \quad \va\in\hh\altp\jj \fin ,
\end{equation}
is isomorphic to $\trjth$, the Banach space of trace class operators from $\jj$ into $\hh$.
Again, the isomorphism of Banach spaces $\hh\prtp\jj\cong\trjth$ can be implemented by extending
(by linearity, and, next, by continuity wrt the appropriate norms) the mapping~\eqref{canoma}.
\end{remark}

Analogously, we will denote by $\thsptjs$ the projective tensor product of the \emph{real}
Banach spaces $\trcsa$ and $\trcsajj$, i.e., the Banach space completion of $\thsatjs$
wrt the cross norm
\begin{equation} \label{defptrntr}
\ptrntr{C}\defi\inf\big\{{\textstyle \sum_k \|S_k\trn\tre\|T_k\trn}\colon
{\textstyle\sum_k S_k\otimes T_k}=C,\  S_k\otimes T_k\in\ethtjs \big\}  \fin ,
\end{equation}
for every $C\in\thsatjs$, which is the largest possible (sub-)cross norm on $\thsatjs$.
Here, recall that $\ethtjs\defi\big\{S\otimes T\colon S\in\trcsa,\ T\in\trcsajj\big\}$.
Clearly, since the infimum in~\eqref{defptrntr} is extended over all (finite) decompositions
of $C$ via \emph{selfadjoint} elementary tensors only, we have that
\begin{equation} \label{basine}
\|C\ptrn\le\ptrntr{C} \fin ,\quad \forall\cinque C\in\thsatjs \fin .
\end{equation}
We will call the norm $\ptrnortr$ on $\thsptjs$ --- obtained extending by continuity the norm~\eqref{defptrntr}
defined on the algebraic tensor product $\thsatjs$ --- the \emph{Hermitian projective norm}.

\begin{remark} \label{remajor}
By what previously observed, it follows easily that the cross norms $\ptrnor$ and $\ptrnortr$ majorize
any (sub-)cross norm on $\thptj$ and $\thsptjs$, respectively.
\end{remark}

It turns out that, actually, the norms $\ptrnor$ and $\ptrnortr$ are \emph{equivalent} on $\thsatjs$.

\begin{proposition} \label{pronorp}
For every trace class operator $C\in\trc\altp\trcjj$, $\|C\ttrn\le\|C\ptrn$. Moreover,
\begin{equation} \label{inenor}
\|C\ttrn\le\|C\ptrn\le\ptrntr{C}\le 2\tre\|C\ptrn \fin , \quad \forall\cinque C\in\thsatjs \fin .
\end{equation}
\end{proposition}

\subsection{Immersion maps, cross trace class operators and cross states}
\label{immermaps}

Since both the Banach spaces $\trc\prtp\trcjj$ and $\trc\otimes\trcjj$ are completions of the linear
space $\trc\altp\trcjj$ (wrt the norms $\ptrnor$ and $\ttrnor$, respectively) and, moreover, the
projective norm $\ptrnor$ majorizes the trace norm $\ttrnor$ on $\trc\altp\trcjj$, there is a natural
(linear, injective) \emph{immersion map}
\begin{equation}  \label{setimme}
\imme\colon\thptj\rightarrow\trc\otimes\trcjj=\trchj \fin .
\end{equation}
Specifically, the immersion map $\imme$ is (well) defined by $\imme(C)\defi C$, for all $C\in\trc\altp\trcjj$,
and
\begin{equation} \label{defimme}
\{C_n\}\subset\trc\altp\trcjj \fin , \ \ptrnorli_n C_n=\optp\in\trc\prtp\trcjj
\implies \imme(\optp)\defi\ttrnorli_n C_n \fin .
\end{equation}
We stress that the limit on the rhs of~\eqref{defimme} does exist, because, by the inequality
$\|C\ttrn\le\|C\ptrn$, for every $C\in\trc\altp\trcjj$ (Proposition~\ref{pronorp}), the sequence
$\{C_n\}$ is $\ttrnor$-Cauchy. Moreover, by the same inequality and by the continuity of the
norm(s), we have:
\begin{equation} \label{norimme}
\|\imme(\optp)\ttrn=\lim_n \|C_n\ttrn\le \lim_n \|C_n\ptrn=\|\optp\ptrn \fin .
\end{equation}

Analogously, taking into account that the Hermitian projective norm $\ptrnortr$ majorizes
the trace norm $\ttrnor$ on $\thsatjs$, there is a natural ($\erre$-linear, injective)
\emph{immersion map}
\begin{equation}
\immer\colon\thsptjs\rightarrow\thsotjs=\trchjsa \fin ,
\end{equation}
defined by $\immer(C)\defi C$, for all $C\in\thsatjs$, and
\begin{equation} \label{defimmer}
\{C_n\}\subset\trchjsa \fin , \ \ptrnortrli_n C_n=\opts\in\thsptjs
\implies \immer(\opts)\defi\ttrnorli_n C_n \fin .
\end{equation}
Here, since the antilinear map $\trchj\ni A\mapsto A^\ast\in\trchj$ is an isometry (in particular, it
is continuous), the limit $C\equiv\ttrnorli_n C_n$ is indeed contained in $\trchjsa$; i.e., $C=C^\ast$.
Then, by the first two inequalities in~\eqref{inenor},
\begin{equation} \label{norimmer}
\|\immer(\opts)\ttrn=\lim_n \|C_n\ttrn\le\lim_n \|C_n\ptrn
\le\lim_n \ptrntr{C_n}=\ptrntr{\opts} \fin .
\end{equation}

Finally, since the norm $\ptrnortr$ dominates the norm $\ptrnor$ on $\thsatjs$, there is
a natural ($\erre$-linear, injective) \emph{immersion map}
\begin{equation}
\immers\colon\thsptjs\rightarrow\thptj \fin ,
\end{equation} \label{mapimmers}
defined by $\immers(C)\defi C$, for all $C\in\thsatjs$, and
\begin{align}
\{C_n\}\subset\thsatjs \fin ,
& \
\ptrnortrli_n C_n = \opts\in\thsptjs
\nonumber \\  \label{defimmers}
& \implies
\immers(\opts)\defi\ptrnorli_n C_n \fin ;
\end{align}
hence, by the second inequality in~\eqref{inenor},
\begin{equation} \label{norimmers}
\|\immers(\opts)\ptrn=\lim_n\|{C_n}\ptrn\le \lim_n \ptrntr{C_n}=\ptrntr{\opts} \fin  .
\end{equation}
Moreover, by the first inequality in~\eqref{norimmer}, we also have that
\begin{equation}
\|\immer(\opts)\ttrn=\lim_n \|C_n\ttrn\le\lim_n \|C_n\ptrn=\|\immers(\opts)\ptrn \fin .
\end{equation}

\begin{remark}
With a slight abuse, we are considering a $\erre$-linear map --- i.e., $\immers$
(as well as the maps $\nimm\colon\trchjsa\rightarrow\trchj$ and~\eqref{recir} defined
in Proposition~\ref{prvarlin} below) --- from a \emph{real} vector space to a
\emph{complex} one. Full mathematical rigor is restored by simply regarding the
latter as a real vector space (by field restriction).
\end{remark}

The previously introduced immersion maps allow us to `compare' the norms $\ttrnor$, $\ptrnor$,
$\ptrnortr$.

\begin{proposition} \label{prvarlin}
The linear maps
\begin{equation}
\imme\colon\trc\prtp\trcjj\rightarrow\trchj, \quad \immer\colon\thsptjs\rightarrow\trchjsa
\end{equation}
and
\begin{equation}
\immers\colon\thsptjs\rightarrow\thptj
\end{equation}
are bounded, and
$\|\imme\norps=\|\immer\norps=\|\immers\norps=1$; therefore, in particular, we have:
\begin{equation} \label{majonor}
\|\imme(\optp)\ttrn\le\|\optp\ptrn \fin , \quad \forall\cinque\optp\in\thptj \fin .
\end{equation}
The norm $\ptrnortr$ \emph{dominates} the norm $\ptrnor$ and is \emph{dominated by} $2\tre\ptrnor$ on
$\thsptjs$. Precisely, denoting by
\begin{equation}
\nimm\colon\trchjsa\rightarrow\trchj
\end{equation}
the natural (isometric) immersion of $\trchjsa$ into $\trchj$, we have:
\begin{equation} \label{dominor}
\|\immer(\opts)\ttrn=\|\nimm\circ\immer(\opts)\ttrn
\le\|\immers(\opts)\ptrn\le\ptrntr{\opts}\le 2\tre\|\immers(\opts)\ptrn \fin ,
\quad \forall\cinque\opts\in\thsptjs \fin .
\end{equation}
Therefore, the norms $\ptrnor\circ\immers$ and $\ptrnortr$  on $\thsptjs$ are \emph{equivalent}.

Moreover, we have that
\begin{equation} \label{recir}
\imme\circ\immers=\nimm\circ\immer\colon\thsptjs\rightarrow\trchj \fin .
\end{equation}
\end{proposition}

The range of the immersion map $\imme$ is a \emph{dense} linear subspace (containing the algebraic
tensor product $\trc\altp\trcjj$) of the Banach space $\trchj$, and we put
\begin{equation} \label{dehtrchj}
\htrchj\defi\imme\big(\trc\prtp\trcjj\big)\subset\trchj \fin .
\end{equation}
Analogously, the range of the immersion map $\immer$ is a \emph{dense} linear subspace (containing
the $\erre$-linear span $\thsatjs$) of the real Banach space $\trchjsa$; we put
\begin{equation} \label{dehtrchjs}
\htrchjs\defi\immer\big(\thsptjs\big)\subset\trchjsa \fin .
\end{equation}
Clearly, $\htrchjs$ can be identified with a subset of $\htrchj$, because
\begin{align}
\htrchjs\defi\immer\big(\thsptjs\big)
& =
\nimm\circ\immer\big(\thsptjs\big)
\nonumber \\
& =
\imme\circ\immers\big(\thsptjs\big)
\nonumber \\ \label{idhtrchjs}
& \subset
\imme\big(\trc\prtp\trcjj\big)\ifed\htrchj \fin ,
\end{align}
where the second equality is just a natural identification obtained via the isometric immersion map
$\nimm\colon\trchjsa\rightarrow\trchj$, and next we have used relation~\eqref{recir}.

We also introduce the set
\begin{equation}  \label{dehstahj}
\hstahj\defi\stahj\cap\htrchj \fin ,
\end{equation}
which is a convex subset of the linear space $\htrchj$ containing, in particular, all
(finite) convex combinations of elementary tensors of the form $\rho\otimes\sigma$,
with $\rho\in\sta$ and $\sigma\in\stajj$.

\begin{definition}
We will call the trace class operators in $\htrchj$, the selfadjoint trace class operators in $\htrchjs$
and the density operators in $\hstahj$, respectively, the \emph{cross trace class operators}, the
\emph{selfadjoint cross trace class operators} and the \emph{cross states} (or \emph{cross density operators})
wrt the bipartition $\hhjj$.
\end{definition}

By its definition~\eqref{dehtrchjs} and by the inclusion relation $\htrchjs\subset\htrchj$ (see~\eqref{idhtrchjs}),
we know that the set of selfadjoint cross trace class operators must satisfy the relation
\begin{equation} \label{inchtrchjs}
\htrchjs\subset\trchjsa\cap\htrchj=\big\{C\in\htrchj\colon C=C^\ast\big\} \fin ,
\end{equation}
and, hence, for the set of cross states we have:
\begin{equation} \label{inchstahj}
\stahj\cap\htrchjs\subset\stahj\cap\htrchj\ifed\hstahj \fin .
\end{equation}
We will soon prove (see equation~\eqref{eqhtrchjs} in Corollary~\ref{corcro} below) that
in~\eqref{inchtrchjs} --- and, hence, in~\eqref{inchstahj} --- the inclusion relation is
actually an equality. To this end, and for our later purposes, we need to establish some
useful decompositions of cross trace class operators, which is our next task.

\subsection{The standard decomposition of a cross trace class operator}
\label{staco}

In the following, we will consider suitable decompositions of an element $\optp$ of the complex
Banach space $\trc\prtp\trcjj$ --- or of the real Banach space $\thsptjs$ --- of the general form
of a finite, or countably infinite, sum
\begin{equation} \label{gendec}
\optp=\sum_k c_k \tre (S_k\otimes T_k) \fin ,
\end{equation}
with $c_k\in\erre$, $0\neq S_k\otimes T_k\in\ethtj$ (respectively, $0\neq S_k\otimes T_k\in\ethtjs$).
The sum, if not finite, is supposed to be \emph{absolutely convergent} wrt the projective norm
$\ptrnor$ (for $\optp\in\thsptjs$, wrt the Hermitian projective norm $\ptrnortr$, or wrt the
equivalent norm $\ptrnor\circ\immers$); i.e.,
$\sum_k\tre |c_k|\tre\|S_k\otimes T_k\ptrn=\sum_k\tre |c_k|\tre\|S_k\trn\tre\|T_k\trn<\infty$,
so that
\begin{equation} \label{inegen}
\|\optp\ptrn\le\sum_k\tre |c_k|\tre\|S_k\otimes T_k\ptrn=\sum_k |c_k|
\tre\|S_k\trn\tre\|T_k\trn <\infty  \fin .
\end{equation}
Clearly, if $\optp\in\thsptjs$ and, for every value of the index $k$, $S_k\otimes T_k\in\ethtjs$
in~\eqref{gendec}, then an analogous inequality holds for $\ptrntr{\optp}$ too.

We will call an expansion of $\optp\in\thptj$ (or of $\optp\in\thsptjs$) of the form~\eqref{gendec} a
\emph{simple-tensor decomposition}.

\begin{remark} \label{remabs}
Since, for every $S\in\trc$ and $T\in\trcjj$, $\|S\otimes T\ptrn=\|S\trn\tre\|T\trn=\|S\otimes T\ttrn$,
it is clear that a series of the form $\sum_k c_k \tre (S_k\otimes T_k)$, with $c_k\in\erre$ and
$S_k\otimes T_k\in\ethtj$, is absolutely convergent wrt the projective norm $\ptrnor$ iff it is absolutely
convergent wrt the trace norm $\ttrnor$. Moreover, since the immersion map $\imme\colon\thptj\rightarrow\trchj$
is continuous (Proposition~\ref{prvarlin}), we have:
\begin{equation} \label{converel}
\optp=\ptrnsum_k c_k \tre (S_k\otimes T_k)\implies
\ttrnsum_k c_k \tre (S_k\otimes T_k)=\imme(\optp)\in\trchj \fin .
\end{equation}
Analogously, the series $\sum_k c_k\tre (S_k\otimes T_k)$, with
$c_k\in\erre$ and $S_k\otimes T_k\in\ethtjs$, is absolutely convergent wrt the Hermitian projective norm
$\ptrnortr$ iff it is absolutely convergent wrt the trace norm $\ttrnor$, and $\opts=\ptrntrsum_k c_k\tre
(S_k\otimes T_k)\implies\ttrnsum_k c_k \tre (S_k\otimes T_k)=\immer(\opts)\in\trchjsa$.
\end{remark}

\begin{remark} \label{remzer}
Some of the coefficients $\{c_k\}$ in~\eqref{gendec} may well be zero, even if $\optp\neq 0$, but
the corresponding zero terms in $\{c_k \tre (S_k\otimes T_k)\}$ are still to be envisioned as part
of the decomposition, because they may play some role in those arguments where this decomposition is
involved. A simple-tensor decomposition of $\optp\neq 0$, of the general form~\eqref{gendec}, will be
said to be \emph{sheer} if it does not contain any zero term.
\end{remark}

\begin{definition} \label{optimality}
A decomposition of $\optp\in\thptj$ of the form~\eqref{gendec} is called $\ptrnor$-\emph{optimal}
(or, for the sake of conciseness, simply \emph{optimal}) --- alternatively, in the case where,
for every value of the index $k$, $S_k\otimes T_k\in\ethtjs$ in~\eqref{gendec} (hence,
$\optp\in\thsptjs$), $\ptrnortr$-\emph{optimal} --- if the norm $\|\optp\ptrn$ (respectively,
$\ptrntr{\optp}$) is given precisely by
\begin{equation} \label{optcond}
\sum_k |c_k|\tre\|S_k\trn\tre\|T_k\trn\in\errep  \fin ;
\end{equation}
i.e., if the first inequality in~\eqref{inegen} (respectively, the inequality
$\ptrntr{\optp}\le\sum_k\tre |c_k|\tre\|S_k\trn\tre\|T_k\trn$) is saturated.
In this case, $\optp$ is said to be $\ptrnor$-\emph{optimally decomposable}, or simply
\emph{optimally decomposable} (respectively, $\optp\in\thsptjs$ is said to be
$\ptrnortr$-\emph{optimally decomposable}).
\end{definition}

\begin{definition} \label{norming}
A simple-tensor decomposition of the form~\eqref{gendec} --- with the scalar coefficients $\{c_k\}$
and the operators $\{S_k\otimes T_k\}$ \emph{subject to specific requirements} --- of a \emph{generic}
element $\optp$ of a certain subset $\suse$ of $\trc\prtp\trcjj$ (or of $\thsptjs$), will be said
to be a \emph{$\ptrnor$-norming decomposition} (respectively, a \emph{$\ptrnortr$-norming decomposition}),
if the norm $\|\optp\ptrn$ (or $\ptrntr{\optp}$) of $\optp\in\suse$ can be obtained as
\begin{equation}
\inf\big\{{\textstyle \sum_k |c_k|\tre\|S_k\trn\tre\|T_k\trn}\colon {\textstyle\sum_k
c_k \tre (S_k\otimes T_k)} = \optp \big\} \fin ,
\end{equation}
where the infimum is extended over all expansions of $\optp$ of the given specified form.
\end{definition}

\begin{theorem} \label{thdecop}
Every element $\optp$ of the complex Banach space $\big(\trc\prtp\trcjj,\ptrnor\big)$ admits a simple-tensor
decomposition of the form
\begin{equation} \label{decop}
\optp=\sum_k r_k \tre (X_k\otimes Y_k) \fin , \quad r_k\ge 0 \fin , \
X_k\otimes Y_k\in\ethtj \fin , \ \|X_k\trn\tre\|Y_k\trn=1 \fin ,
\end{equation}
where the sum is either finite or --- in the case where the Hilbert space $\hh\otimes\jj$ is
infinite-dimensional --- possibly countably infinite. In the latter case, the series is
supposed to converge absolutely wrt to the projective norm $\ptrnor$, and, in fact, the
non-negative real sequence $\{r_k\}$ is assumed to be summable, so that
\begin{equation}
\|\optp\ptrn\le\sum_k \|r_k \tre (X_k\otimes Y_k)\ptrn=
\sum_k r_k \tre\|X_k\trn\tre \|Y_k\trn=\sum_k r_k<\infty \fin .
\end{equation}
Moreover, for every $\optp\in\trc\prtp\trcjj$, we have that
\begin{equation} \label{pronorm-bis}
\|\optp\ptrn=\inf\big\{{\textstyle \sum_k r_k}\colon {\textstyle\sum_k
r_k \tre (X_k\otimes Y_k)}=\optp \big\} \fin ,
\end{equation}
where the infimum is taken over all possible expansions of $\optp$ of the general form~\eqref{decop};
namely, the decomposition~\eqref{decop} is $\ptrnor$-norming.

Analogously, every element $\opts$ of the real Banach space $(\thsptjs,\ptrnortr)$ admits a simple-tensor
decomposition of the form
\begin{equation} \label{decopsa}
\opts=\sum_k r_k \tre (X_k\otimes Y_k) \fin , \quad r_k\ge 0 \fin , \
X_k\otimes Y_k\in\ethtjs \fin , \ \|X_k\trn\tre\|Y_k\trn=1 \fin ,
\end{equation}
absolute convergence --- that holds because $\sum_k r_k<\infty$ --- and all other properties as
above, being understood wrt to the Hermitian projective norm $\ptrnortr$; in particular, the
decomposition~\eqref{decopsa} is $\ptrnortr$-norming.
\end{theorem}

\begin{proof}
In the case where $\dim(\hh\otimes\jj)<\infty$ and/or $\optp\in\trc\altp\trcjj$, the statement
is clear. Let us then assume that $\dim(\hh\otimes\jj)=\infty$ and
$\optp\in(\trc\prtp\trcjj)\setminus(\trc\altp\trcjj)$. By a classical result of
Grothendieck~\cite{Grothendieck-res} --- also see Proposition~{2.8} of~\cite{Ryan}, Proposition~{1.1.4}
of~\cite{DFS}, or Theorem~{7.9} of~\cite{Kubrusly} --- for every $\optp$ of this kind there is
a sequence of the form $\{S_k\otimes T_k\}$, with $S_k\in\trc$ and $T_k\in\trcjj$, such that
\begin{equation} \label{decop-bis}
\sum_k\tre \|S_k\trn\tre \|T_k\trn<\infty \quad \mbox{and} \quad
\optp=\sum_k S_k\otimes T_k \fin ,
\end{equation}
where, by the first of these relations, the series must converge wrt to the projective norm $\ptrnor$.
Moreover, it turns out that
\begin{equation}
\|\optp\ptrn\defi\inf\big\{{\textstyle \sum_k \|S_k\trn\tre\|T_k\trn}\colon {\textstyle\sum_k
S_k\otimes\tre T_k}=\optp \big\} \fin ,
\end{equation}
where the infimum is taken over all decompositions of $\optp$ of the form~\eqref{decop-bis}. Finally, by
suitably re-expressing the sequence $\{S_k\otimes T_k\}$ in the form $\{r_k \tre (X_k\otimes Y_k)\}$,
as specified in~\eqref{decop} (so that $r_k=\|S_k\trn\tre\|T_k\trn$), the statement follows.
An analogous result holds for the decomposition~\eqref{decopsa}, because the previously cited result
of Grothendieck applies to the projective tensor product of any pair of real or complex Banach spaces.
\end{proof}

\begin{remark} \label{obvrem}
In decomposition~\eqref{decop} --- or~\eqref{decopsa} --- we do not distinguish, in general,
between the various representations of the elementary tensor product $X_k\otimes Y_k$.
However, we can always assume, without loss of generality (wrt the condition that
$\|X_k\otimes Y_k\ttrn=\|X_k\trn\tre\|Y_k\trn=1$), that in~\eqref{decop} --- or in~\eqref{decopsa}
--- $\|X_k\trn=\|Y_k\trn=1$; moreover, in decomposition~\eqref{decopsa}, we can always suppose that
$X_k\in\trcsa$ and $Y_k\in\trcsajj$.
\end{remark}

\begin{remark} \label{remimme}
We stress that \emph{any} series $\sum_k r_k \tre (X_k\otimes Y_k)$ of the form specified
in~\eqref{decop} must converge in the Banach space $\thptj$ (assuming that $\sum_k r_k<\infty$),
because in such a case it is absolutely convergent wrt the projective norm $\ptrnor$, so that
the sequence of the partial sums is  $\ptrnor$-Cauchy. As noted in Remark~\ref{remabs}, it is
also absolutely convergent wrt the trace norm $\ttrnor$, and
$\imme(\optp)=\imme(\sum_k r_k \tre (X_k\otimes Y_k))=\sum_k r_k \tre (X_k\otimes Y_k)$, where
the latter sum converges in $\trchj$. Analogous facts hold true in the case of the real Banach
spaces $\thsptjs$ and $\trchjsa$.
\end{remark}

\begin{corollary} \label{corsep}
The Banach spaces $\thptj$ and $\thsptjs$ are separable.
\end{corollary}

\begin{proof}
By Theorem~\ref{thdecop}, a generic element $\optp$ of the Banach space $\trc\prtp\trcjj$ admits
a decomposition of the form~\eqref{decop}, which, for mere economy of this proof, in the following will
be assumed to be \emph{countably infinite}, possibly putting $r_k=0$ and choosing a whatever pair of
trace class operators $X_k\in\trc$, $Y_k\in\trcjj$ (with $\|X_k\trn=\|Y_k\trn=1$), for $k\ge\enne$,
for some suitable $\enne\in\nat$.

Now, let $\trco$, $\trcjjo$ be any pair of \emph{countable dense subsets} of the separable
Banach spaces $\trc$ and $\trcjj$, respectively. For every $n\in\nat$, we set
\begin{equation} \label{convser}
\optp_n=\sum_{k=1}^\infty \rkn \tre (\Xkn\otimes\Ykn) \fin , \quad 0\le\rkn\in\rati \fin , \
\Xkn\in\trco \fin ,\  \Ykn\in\trcjjo \fin ,
\end{equation}
where the non-negative sequence $\{\rkn\tre \|\Xkn\trn\tre \|\Ykn\trn\colon k\in\nat\}$ is supposed to be summable
--- so that the series will converge, wrt to the projective norm $\ptrnor$, to an element of $\trc\prtp\trcjj$
--- but this time we are \emph{not} assuming that $\|\Xkn\trn\tre\|\Ykn\trn=1$. By construction, the vector
$\optp_n$ belongs to the $\ptrnor$-closure of a \emph{countable} subset of $\trc\prtp\trcjj$; i.e., of the
(countable) $\rati$-linear span
\begin{equation}
\lispranb = \lispra\{S\otimes T\colon S\in\trco,\ T\in\trcjjo\} \fin .
\end{equation}

Let us show that the sequence $\{\optp_n\}$ --- subject to the specified conditions --- can be constructed
in such a way as to converge to a generic element $\optp$ (expressed by a countably infinite expansion of
the form~\eqref{decop}, with $\|X_k\trn=\|Y_k\trn=1$) of $\trc\prtp\trcjj$. To this end, we will assume that
\begin{equation} \label{esta}
|r_k-\rkn|<\frac{s}{n2^{k+1}} \quad \left(\Longrightarrow\
\rkn<r_k+\frac{s}{n2^{k+1}}\right) \fin ,
\end{equation}
with $s\equiv\sum_k r_k<\infty$, and
\begin{equation} \label{estb}
\|X_k-\Xkn\trn<\frac{3}{7n2^{k+1}} \ \left(\Longrightarrow\
\|\Xkn\trn<\|X_k\trn+\frac{3}{7n2^{k+1}}\right) , \quad
\|Y_k-\Ykn\trn<\frac{3}{7n2^{k+1}}  \fin ,
\end{equation}
where we recall once again that $\|X_k\trn=\|Y_k\trn=1$. Clearly, these assumptions are legitimate
by the density of the rationals $\rati$ in $\erre$, and of countable the sets $\trco$, $\trcjjo$
in $\trc$ and $\trcjj$, respectively. Note moreover that, since
\begin{equation}
0\le\rkn \tre \|\Xkn\trn\tre\|\Ykn\trn<\left(r_k+\frac{s}{n2^{k+1}}\right)\left(1+\frac{3}{7n2^{k+1}}\right)^2,
\end{equation}
where we have used the inequalities in brackets in~\eqref{esta} and~\eqref{estb}, then, by an elementary estimate,
\begin{equation}
\sum_{k=1}^\infty \|\rkn \tre (\Xkn\otimes\Ykn)\ptrn=\sum_{k=1}^\infty\rkn\tre\|\Xkn\trn\tre\|\Ykn\trn
< \frac{2s\tre(n+1)}{n} \fin ;
\end{equation}
hence, for every $n\in\nat$, the series in~\eqref{convser} is indeed absolutely convergent wrt to the
projective norm $\ptrnor$. At this point, observe that
\begin{align}
\|\optp-\optp_n\ptrn & \le\|{\textstyle\sum_k (r_k-\rkn) \tre (X_k\otimes Y_k)}\ptrn
+ \|{\textstyle\sum_k \rkn \tre ((X_k-\Xkn)\otimes Y_k)}\ptrn
\nonumber \\
& +
\|{\textstyle\sum_k \rkn \tre (\Xkn\otimes (Y_k-\Ykn))}\ptrn
\nonumber \\
& \le
{\textstyle\sum_k |r_k-\rkn|\tre\|X_k\trn\tre\|Y_k\trn} +
{\textstyle\sum_k \rkn\tre\|X_k-\Xkn\trn\tre\|Y_k\trn}
\nonumber \\
& +
{\textstyle\sum_k \rkn\tre\|X_k\trn\tre \|Y_k-\Ykn\trn}<\frac{s}{n} \fin.
\end{align}
Therefore, eventually we see that $\ptrnorli_n \optp_n=\optp$, and hence
\begin{equation}
\optp\in\clptrnor\big(\lispranb\big) ,
\end{equation}
which proves that the Banach space $\trc\prtp\trcjj$ is separable, as claimed.

The proof in case of the real Banach space $\thsptjs$ is analogous.
\end{proof}

Another relevant consequence of Theorem~\ref{thdecop} is that the cross trace class operators $\htrchj$
and the \emph{selfadjoint} cross trace class operators $\htrchjs$ admit an elementary characterization
in terms of simple tensors; moreover, they are, in  natural way, Banach spaces isomorphic to --- and,
therefore, will be identified with --- $\thptj$ and $\thsptjs$, respectively. For the sake of simplicity,
with a slight abuse, we will denote the norms of the two pairs of isomorphic Banach spaces by the same
symbols $\ptrnor$ and $\ptrnortr$.

\begin{corollary}[The standard decomposition] \label{corcro}
The set $\htrchj$ of all cross trace class operators wrt the bipartition $\hhjj$ --- i.e., the range
of the immersion map $\imme\colon\trc\prtp\trcjj\rightarrow\trchj$ --- is characterized as the complex
linear space  consisting of every trace class operator $C$ on $\hhjj$ that admits an absolutely
$\ttrnor$-convergent simple-tensor decomposition of the form
\begin{equation} \label{stardeco}
C=\sum_k r_k \tre (X_k\otimes Y_k) \fin ,
\end{equation}
with $r_k\ge 0$ $(\sum_k r_k<\infty)$, $X_k\otimes Y_k\in\ethtj$ and $\|X_k\trn\tre\|Y_k\trn=1$.
Moreover, here one can always assume that $X_k$ and $Y_k$ are rank-one partial isometries, i.e.,
\begin{equation}
X_k=\epk \fin , \ \ Y_k=\pck \fin , \quad
\mbox{for some normalized vectors $\eta_k,\phi_k\in\hh$ and $\psi_k,\chi_k\in\jj$.}
\end{equation}

The linear space $\htrchj$ is a selfadjoint subset of the bipartite trace class $\trchj$, i.e.,
\begin{equation}
C={\textstyle\sum_k r_k \tre (X_k\otimes Y_k)}\in\htrchj\implies
C^\ast={\textstyle\sum_k r_k \tre (X_k^\ast\otimes Y_k^\ast)}\in\htrchj \fin ,
\end{equation}
and becomes a separable Banach space --- isomorphic to $\trc\prtp\trcjj$ --- if endowed with the norm
\begin{equation} \label{pronorm-tris}
\|C\ptrn\defi\inf\big\{{\textstyle \sum_k r_k}\colon {\textstyle\sum_k r_k \tre
(X_k\otimes Y_k)}=C \big\} \fin , \quad C\in\htrchj \fin ,
\end{equation}
where the infimum is extended over all absolutely $\ttrnor$-convergent decompositions of $C$ of
the form~\eqref{stardeco}; to evaluate the norm $\|C\ptrn$, we can assume that $X_k$ and $Y_k$
therein are, in particular, rank-one partial isometries. Any such a simple-tensor decomposition
will converge absolutely to $C$ wrt the norm $\ptrnor$ of $\htrchj$, as well. Moreover, for every
$C\in\htrchj$, $\|C\ptrn=\|C^\ast\ptrn$.

The set $\htrchjs$ of all \emph{selfadjoint} cross trace class operators wrt the bipartition $\hhjj$
--- i.e., the range of the immersion map $\immer\colon\thsptjs\rightarrow\trchjsa$ --- coincides with
the real linear space of all selfadjoint trace class operators $C$ on $\hhjj$ that admit an absolutely
$\ttrnor$-convergent decomposition of the form
\begin{equation} \label{stardecosa}
C=\sum_k r_k \tre (X_k\otimes Y_k) \fin , \quad r_k\ge 0  \fin , \
X_k\otimes Y_k\in\ethtjs \fin , \ \|X_k\trn\tre\|Y_k\trn=1 \fin ,
\end{equation}
where the absolute convergence holds because $\sum_k r_k<\infty$. The linear space $\htrchjs$
becomes a separable (real) Banach space --- isomorphic to $\thsptjs$ --- if endowed with the norm
\begin{equation} \label{pronorm-sa}
\ptrntr{C}\defi\inf\big\{{\textstyle \sum_k r_k}\colon {\textstyle\sum_k r_k \tre
(X_k\otimes Y_k)}=C,\ X_k\otimes Y_k\in\ethtjs \big\} \fin ,
\end{equation}
where the infimum is extended over all absolutely $\ttrnor$-convergent decompositions of $C$ of
the specified form, and any such a simple-tensor decomposition will converge absolutely to $C$
wrt the norm $\ptrnortr$ of $\htrchjs$, as well.

Moreover, the following relation holds:
\begin{equation} \label{eqhtrchjs}
\htrchjs=\trchjsa\cap\htrchj=\big\{C\in\htrchj\colon C=C^\ast\big\} \fin .
\end{equation}

The complex Banach space $(\htrchj,\ptrnor)$ is closed wrt the standard product of operators
(i.e., composition), and --- once endowed with this product and with the adjoining map
$C\mapsto C^\ast$ --- becomes a Banach $\ast$-algebra.
\end{corollary}

\begin{proof}
The first assertion follows immediately from Theorem~\ref{thdecop}, taking into account
Remark~\ref{remimme}. Note that, by the singular value decomposition of a trace class
operator (Remark~\ref{sivadeb}), we can expand the operators $X_k\in\trc$, $Y_k\in\trcjj$ of
decomposition~\eqref{stardeco}, where we put $\|X_k\trn=\|Y_k\trn=1$, in terms of rank-one
partial isometries; precisely,
\begin{equation} \label{sivadxy}
X_k=\sum_j\sjk\sei\epjk \fin , \quad  Y_k=\sum_l\tlk\sei\pclk \fin ,
\end{equation}
where --- \emph{for each fixed value} of the index $k$ --- we have:
$\sum_j\sjk=\|X_k\trn=1=\|Y_k\trn=\sum_l\tlk$, with $\sjk,\tre\tlk>0$, and
$\big\{\ejk\big\},\big\{\pjk\big\}\subset\hh$ and $\big\{\plk\big\},\big\{\clk\big\}\subset\jj$
are orthonormal systems. Thus, the decomposition~\eqref{stardeco} of the cross trace class operator
$C\in\htrchj$ can be re-formulated as follows:
\begin{align}
C
& =
\sum_k r_k \tre (X_k\otimes Y_k)
\nonumber \\
& =
\sum_k\sum_j\sum_l r_k\sei\sjk\tre\tlk\tre\big(\epjk\otimes\pclk\big)
\nonumber \\
& =
\sum_{k}\sum_{m}r_k\sei\sjmk\sei\tlmk\tre\big(\epjmk\otimes\pclmk\big) \fin .
\nonumber \\ \label{decopis}
& =
\sum_n\rknn\sei\sjmkn\sei\tlmkn\tre\big(\epjmkn\otimes\pclmkn\big) \fin .
\end{align}
Here, a few comments are in order:
\begin{itemize}

\item The second equality in~\eqref{decopis} is obtained by considering the (possibly countably
infinite) sum in decomposition~\eqref{stardeco} as convergent wrt the trace norm $\ttrnor$, and
then exploiting the singular value decompositions~\eqref{sivadxy} of $X_k\in\trc$ and $Y_k\in\trcjj$
and the continuity of the natural bilinear form $\nb\colon\trc\times\trcjj\rightarrow\trchj$. Therefore,
the iterated sum $\sum_k\sum_j\sum_l\ldots$ converges wrt the norm $\ttrnor$.

\item Note that, for each fixed value of the index $k$, the double sequence (wrt the indices $j,l$)
\begin{equation}
\big\{\sjk\tre\tlk\tre\big(\epjk\otimes\pclk\big)\big\}
\end{equation}
is absolutely summable, because $\sjk\tre\tlk>0$, $\sum_j\sum_l\sjk\tre\tlk=1$ and
\begin{equation}
\left\|\epjk\otimes\pclk\right\ttrn=1=\left\|\epjk\otimes\pclk\right\ptrn .
\end{equation}
Therefore, we can convert the iterated sum $\sum_j\sum_l\ldots$  on the second line of~\eqref{decopis}
into a single sum $\sum_{m}\ldots$ by taking any bijection $m\mapsto\jlm$ (Fact~\ref{facdou}).

\item Next, since $r_k\tre\sjmk\sei\tlmk\ge 0$ and $\sum_k\sum_m r_k\tre\sjmk\sei\tlmk=
\sum_k r_k\sum_j\sum_l\sjk\tre\tlk=\sum_k r_k<\infty$, according to Fact~\ref{facdou}
we can convert the iterated sum $\sum_k\sum_m\ldots$ on the third line of~\eqref{decopis}
into a single sum $\sum_{n}\ldots$ by choosing any bijection $n\mapsto\kmn$.

\end{itemize}

Eventually, by suitably renaming the coefficients and the indices in the expansion on the last line
of~\eqref{decopis}, we get a decomposition of the form
\begin{equation} \label{decofin}
C=\sum_n c_n (\epn\otimes\pcn) \fin , \ \
c_n=\rknn\sei\sjmkn\sei\tlmkn\ge 0 \fin , \ \|\eta_n\|=\cdots=\|\chi_n\|=1 \fin ,
\end{equation}
where the sum converges absolutely wrt to the trace norm $\ttrnor$ and, in fact,
\begin{align}
\sum_n c_n
& =
\sum_n\rknn\sei\sjmkn\sei\tlmkn
\nonumber \\
& \label{eqcoefs}
=
\sum_{k}\sum_{m}r_k\sei\sjmk\sei\tlmk=\sum_k\sum_j\sum_l r_k\sei\sjk\tre\tlk
=\sum_k r_k<\infty \fin .
\end{align}

Let us prove the second assertion of the theorem. By~\eqref{pronorm-bis}, the norm on $\htrchj$
defined by~\eqref{pronorm-tris} coincides with the image, via the immersion map
$\imme\colon\trc\prtp\trcjj\rightarrow\trchj$, of the projective norm $\ptrnor$.
Hence, the linear space $\htrchj$, endowed with the norm~\eqref{pronorm-tris}, becomes a Banach space
isomorphic to $\trc\prtp\trcjj$. Moreover, since, for every $C\in\htrchj$, the decomposition~\eqref{stardeco}
can always be transformed into a decomposition of the special form~\eqref{decofin}, whose coefficients
$\{c_n\}$ are related to the coefficients $\{r_k\}$ of the more general decomposition~\eqref{stardeco}
in such a way that relation~\eqref{eqcoefs} holds --- i.e., $\sum_n c_n=\sum_k r_k$ --- we can assume
that the trace class operators $X_k$, $Y_k$ in~\eqref{pronorm-tris} are, in particular, rank-one partial
isometries; namely, we have:
\begin{equation}
\|C\ptrn=\inf\big\{{\textstyle \sum_k r_k}\colon {\textstyle\sum_k r_k \tre
(\epk\otimes\pck)}=C, \ r_k\ge 0,\ \|\eta_k\|=\cdots=\|\chi_k\|=1 \big\} \fin .
\end{equation}

As noted in Remark~\ref{remabs} (and next in Remark~\ref{remimme}) --- since, for every $S\in\trc$
and $T\in\trcjj$, $\|S\otimes T\ptrn=\|S\trn\tre\|T\trn=\|S\otimes T\ttrn$ --- a
series of the general form $\sum_k c_k \tre (S_k\otimes T_k)$, with $c_k\in\erre$ and $S_k\otimes T_k\in\ethtj$,
is absolutely convergent wrt the projective norm $\ptrnor$ iff it is absolutely convergent wrt the trace
norm $\ttrnor$, and, if $C=\ptrnsum_k c_k \tre (S_k\otimes T_k)$, then $\ttrnsum_k c_k\tre (S_k\otimes T_k)=C$.
Therefore, in particular, the $\ttrnor$-convergent simple-tensor decomposition~\eqref{stardeco}
must converge absolutely to $C$ wrt the norm $\ptrnor$ of $\htrchj$, as well.

Moreover, taking into account that the adjoining map $\trchj\ni C\mapsto C^\ast\in\trchj$ is an isometry (in
particular, it is continuous), for any $\ttrnor$-convergent expansion of $C\in\htrchj$ of the form~\eqref{stardeco},
$C^\ast=\sum_k r_k\tre (X_k\otimes Y_k)^\ast=\sum_k\tre r_k\tre (X_k^\ast\otimes Y_k^\ast)$, with
$\|X_k^\ast\trn\tre\|Y_k^\ast\trn=\|X_k\trn\tre\|Y_k\trn=1$. Thus, $C^\ast$ belongs to $\htrchj$, as well,
and, by definition~\eqref{pronorm-tris} of the norm $\ptrnor$, it is clear that $\|C\ptrn=\|C^\ast\ptrn$,
for every $C\in\htrchj$.

The assertion of the theorem regarding the characterization of the Banach space $\htrchjs$ follows, once again,
from Theorem~\ref{thdecop} and Remark~\ref{remimme}.

We now prove relation~\eqref{eqhtrchjs}. We already know that $\htrchjs\subset\trchjsa\cap\htrchj$
(recall relation~\eqref{inchtrchjs}). It is then sufficient to show that the reverse inclusion holds
too; i.e., that $C\in\htrchj, \ C=C^\ast\implies C\in\htrchjs$. Let $\sum_k r_k \tre (X_k\otimes Y_k)$
be an expansion of $C$ as specified in~\eqref{stardeco}. Note that each term of the form $X\otimes Y$
in this expansion (for notational convenience, here we are dropping the index $k$) can be expressed as
\begin{equation}
(\xu + \ima \sei\xd)\otimes(\yu + \ima \sei\yd) = \xu\otimes\yu - \xd\otimes\yd + \ima\tre
(\xu\otimes\yd+\xd\otimes\yu),
\end{equation}
where $\xu\defi\frac{1}{2}(X+X^\ast)$, $\xd\defi\frac{1}{2\tre\ima}(X-X^\ast)\in\trcsa$, and
$\yu,\yd\in\trcsajj$ are defined analogously. Thus, $C$ can be expressed as the sum of two
series, the first one containing all terms of the form $\xu\otimes\yu - \xd\otimes\yd$ ---
i.e., restoring the index $k$, the (absolutely convergent wrt each of the norms $\ttrnor$,
$\ptrnor$ and $\ptrnortr$) series
\begin{equation} \label{sapart}
\sum_k r_k \tre (\xuk\otimes\yuk - \xdk\otimes\ydk) \fin ,
\end{equation}
where $(\xuk\otimes\yuk),(\xdk\otimes\ydk)\in\ethtjs$ --- while the second one all terms of the form
$\ima\tre(\xu\otimes\yd+\xd\otimes\yu)$. By imposing the condition that $C=C^\ast$, we see that
the second sum must be zero, so that $C$ can be expressed by the series~\eqref{sapart} alone
(where, of course, the trace class operators $\xuk\otimes\yuk$, $\xdk\otimes\ydk$ can be suitably
renormalized by modifying the associated scalar coefficients), and hence, by the previously obtained
characterization of the Banach space $\htrchjs$, we conclude that indeed $C\in\htrchjs$, as we wished
to prove.

Let us eventually prove that $(\htrchj,\ptrnor)$ is, in natural way, a Banach algebra. We only need
to show that
\begin{equation}
C_1,C_2\in\htrchj \implies \mbox{$C_1\tre C_2\equiv C_1\circ C_2\in\htrchj$ and
$\|C_1\tre C_2\ptrn\le \|C_1\ptrn\tre\|C_2\ptrn$} \fin ;
\end{equation}
i.e., that the Banach space $(\htrchj,\ptrnor)$ is closed wrt to composition of operators
and its norm is submultiplicative. To this end, for any $\epsilon>0$, let us pick decompositions
(of the form~\eqref{stardeco})
\begin{equation} \label{deccjs}
C_j= \sum_k \rkj \tre (\xkj\otimes \ykj) \fin , \quad j=1,2 \fin ,
\end{equation}
of the cross trace class operators $C_1$ and $C_2$ (converging wrt both the norms $\ttrnor$ and $\ptrnor$),
such that
\begin{equation}
\sum_k \rkj \le\|C_j\ptrn + \epsilon  \fin , \quad j=1,2 \fin ,
\end{equation}
condition that can always be satisfied according to definition~\eqref{pronorm-tris}. Note that the
\emph{double series}
\begin{equation}
\sum_{kl} \rko\tre\rlt\big((\xko\xlt)\otimes(\yko\ylt)\big)
\end{equation}
is absolutely convergent --- wrt the both the norms $\ttrnor$ and $\ptrnor$ --- to some cross trace
class operator $C_\epsilon\in\htrchj$; indeed, $\sum_{kl}\rko\tre\rlt=\sum_k\rko\sum_l\rlt\le
(\|C_1\ptrn + \epsilon)(\|C_2\ptrn + \epsilon)$ and
\begin{align}
\|(\xko\xlt)\otimes(\yko\ylt)\ptrn
& =
\|(\xko\xlt)\otimes(\yko\ylt)\ttrn
\nonumber \\ \label{estima}
& =
\|\xko\xlt\trn\tre\|\yko\ylt\trn\le\|\xko\trn\tre\|\xlt\trn\tre\|\yko\trn\tre\|\ylt\trn= 1 \fin .
\end{align}
We claim that $C\equiv C_\epsilon=C_1\tre C_2$ (for any $\epsilon>0$).

Indeed, since (by Fact~\ref{trnsubm}) $\|A_1\tre A_2\ttrn\le\|A_1\ttrn\tre \|A_2\ttrn$, for all
$A_1,A_2\in\trchj$, then the linear maps (of right and left multiplication by a fixed trace class
operator)
\begin{equation} \label{limaps}
\trchj\ni A_1\mapsto A_1\tre A_2\in\trchj \ \ \mbox{and} \ \
\trchj\ni A_2\mapsto A_1\tre A_2\in\trchj
\end{equation}
are bounded. By the continuity of the \emph{first} linear map in~\eqref{limaps}, for any $C_1,C_2\in\htrchj$,
and exploiting the expansion~\eqref{deccjs} with $j=1$, we obtain that
\begin{equation} \label{relcca}
C_1\tre C_2=\bigg(\ttrnsum_k \rko\big(\xko\otimes\yko\big)\bigg) C_2
=\ttrnsum_k \rko\big((\xko\otimes\yko)\tre C_2\big),
\end{equation}
where, by the continuity of the \emph{second} linear map and by the expansion~\eqref{deccjs} with $j=2$,
we have:
\begin{align}
(\xko\otimes\yko)\sei C_2
& =
(\xko\otimes\yko)\dieci\ptrnsum_l\rlt\big(\xlt\otimes\ylt\big)
\nonumber \\ \label{relccb}
& =
\ttrnsum_l\rlt\big((\xko\xlt)\otimes(\yko\ylt)\big) .
\end{align}
Next, by relations~\eqref{relcca} and~\eqref{relccb}, we find that
\begin{align}
C_1\tre C_2
& =
\ttrnsum_k \rko\big((\xko\otimes\yko)\tre C_2\big)
\nonumber\\
& =
\ttrnsum_k \rko\dieci\ttrnsum_l\rlt\big((\xko\xlt)\otimes(\yko\ylt)\big)
\nonumber \\
& =
\ttrnsum_{kl} \rko\tre\rlt\big((\xko\xlt)\otimes(\yko\ylt)\big)
\nonumber \\ \label{douser}
& =
\ptrnsum_{kl} \rko\tre\rlt\big((\xko\xlt)\otimes(\yko\ylt)\big)\in\htrchj \fin .
\end{align}
Here, the iterated series on the second line is equal to the double series on the third line
by the absolute convergence of the latter (recall Fact~\ref{facdou}), actually wrt both the
norms $\ttrnor$ and $\ptrnor$. Therefore, we conclude that $C_1\tre C_2\in\htrchj$ and, for
every $\epsilon >0$,
$\|C_1\tre C_2\ptrn\le\sum_k \rko \sum_l\rlt\le(\|C_1\ptrn + \epsilon)(\|C_2\ptrn + \epsilon)$.
Hence, $\|C_1\tre C_2\ptrn\le\|C_1\ptrn\tre \|C_2\ptrn$, and the proof is now complete.
\end{proof}

\begin{remark} \label{remide}
In the light of Corollary~\ref{corcro}, the `abstract' complex Banach space $\trc\prtp\trcjj$
will be henceforth identified with the range $\htrchj$ of the immersion map $\imme$. Precisely,
from now on we will denote by $\trc\prtp\trcjj$ --- rather than by $\htrchj$ --- the Banach space
of all trace class operators of the form~\eqref{stardeco}, the \emph{cross trace class operators},
endowed with the norm $\ptrnor$ defined by~\eqref{pronorm-tris}, that will be called the
\emph{projective (cross) norm}. Once again, we stress  that, in the simple-tensor
decomposition~\eqref{stardeco}, absolute convergence is understood wrt the norm $\ptrnor$ or,
equivalently, wrt the trace norm $\ttrnor$, because the (absolute) summability and the sum itself
of the series do not depend on which if the two norms is considered (Remark~\ref{remimme}). With
an analogous identification, from now on we will denote by $\thsptjs$ the real vector space
$\htrchjs$ (the range of the immersion map $\immer$) --- i.e., the real vector space (recall
relation~\eqref{eqhtrchjs})
\begin{equation}
\thptjsa\defi\big\{C\in\thptj\equiv\htrchj\colon C=C^\ast\big\}
\end{equation}
of all selfadjoint cross trace class operators --- endowed with the \emph{Hermitian projective (cross) norm}
$\ptrnortr$, defined by~\eqref{pronorm-sa}. Therefore, skipping from now on all natural immersion maps, the
inequalities~\eqref{majonor} and~\eqref{dominor} relating the various norms on $\thptj$ and $\thsptjs$ can
be rewritten, respectively, as
\begin{equation} \label{inenorm}
\|C\ttrn\le\|C\ptrn \fin , \quad \forall\cinque C\in\thptj\equiv\htrchj ,
\end{equation}
and
\begin{equation} \label{inenorms}
\|C\ttrn\le\|C\ptrn\le\ptrntr{C}\le 2\tre\|C\ptrn \fin , \quad
\forall\cinque C\in\thsptjs\equiv\htrchjs \fin .
\end{equation}
\end{remark}

\begin{definition} \label{destade}
Any expansion of a cross trace class operator $C\in\trc\prtp\trcjj\equiv\htrchj$ of the
form~\eqref{stardeco} will be called a \emph{standard decomposition} of $C$.

Any expansion of a \emph{selfadjoint} cross trace class operator $C\in\thsptjs\equiv\htrchjs$
of the form~\eqref{stardeco} --- with $r_k\ge 0$ and $X_k\otimes Y_k\in\ethtjs$ ($\sum_k r_k<\infty$,
$\|X_k\trn\tre\|Y_k\trn=1$) --- will be called a \emph{Hermitian standard decomposition} of $C$.

A standard decomposition of $C\in\trc\prtp\trcjj$ (or a Hermitian standard decomposition of $C\in\thsptjs$)
of the form~\eqref{stardeco}, with $C\neq 0$, is said to be \emph{sheer} if it does not contain zero
summands; i.e., if $r_k>0$, for every $k$.
\end{definition}

\subsection{The cross trace class and optimal decompositions}
\label{thecross}

Taking into account the important Remark~\ref{remide}, we set the following:

\begin{definition} \label{defide}
The Banach spaces
\begin{equation}
\big(\thptj\equiv\htrchj,\ptrnor\big) \ \ \mbox{and} \ \ \big(\thsptjs\equiv\htrchjs,\ptrnortr\big)
\end{equation}
where the norms $\ptrnor$, $\ptrnortr$ are defined by~\eqref{pronorm-tris} and~\eqref{pronorm-sa},
will be called the \emph{cross trace class} and the \emph{selfadjoint cross trace class} of the
bipartite Hilbert space $\hhjj$, respectively.
\end{definition}

\begin{proposition}
The Banach spaces
\begin{equation}
\thptj\equiv\htrchj \ \ \mbox{and} \ \ \tjpth\equiv\htrcjh
\end{equation}
are isomorphic. This isomorphism is implemented by the (well defined) transposition map
\begin{equation} \label{detran}
\trsp\colon\thptj\ni C=\sum_k r_k \tre (X_k\otimes Y_k)\mapsto
\sum_k r_k \tre (Y_k\otimes X_k)\ifed\trsp(C)\in\tjpth \fin ,
\end{equation}
where $C=\sum_k r_k \tre (X_k\otimes Y_k)$ is any standard decomposition of the cross trace class operator
$C$ (and the mapping in~\eqref{detran} does not depend on the choice of this decomposition).

An analogous isomorphism holds for the Banach spaces $\thsptjs\equiv\htrchjs$ and $\tjspths\equiv\htrcjhs$.
\end{proposition}

Let us now derive some direct consequences of Corollary~\ref{corcro}.

\begin{proposition} \label{prolinc}
Every cross trace class operator --- wrt the bipartition $\hhjj$ --- can be expressed as a
linear combination of (at most) two selfadjoint cross trace class operators; i.e., the
following relation holds:
\begin{align}
\thptj
& =
\thsptjs + \ima\sei\big(\thsptjs\big)
\nonumber \\ \label{lincom}
& =
\big\{C_1 +\ima\sei C_2\colon C_1,C_2\in\thsptjs\big\} \fin .
\end{align}
\end{proposition}

\begin{proposition} \label{prcrost}
For the set $\hstahj$ of all cross states wrt the bipartition $\hhjj$, the following
relation holds:
\begin{equation} \label{eqhstahj}
\hstahj\defi\stahj\cap\thptj=\stahj\cap\thsptjs \fin .
\end{equation}
Moreover, $\hstahj$ is a $\ptrnor$-closed convex subset of $\thptj$ and a $\ptrnortr$-closed
convex subset of $\thsptjs$.
\end{proposition}

\begin{proof}
In fact, by relation~\eqref{eqhtrchjs}, we have that
\begin{align}
\stahj\cap\htrchjs
& =
\stahj\cap\trchjsa\cap\htrchj
\nonumber\\
& =
\stahj\cap\htrchj\ifed\hstahj \fin .
\end{align}
Thus, with the identifications $\thptj\equiv\htrchj$ and $\thsptjs\equiv\htrchjs$,
the first assertion follows. Now, let $\{D_n\}_{n\in\nat}$ be a $\ptrnor$-convergent
sequence in $\hstahj$. Clearly, by its definition, $\hstahj$ is a convex subset of
$\thptj$, and --- since the trace norm $\ttrnor$ is dominated by the projective norm
$\ptrnor$ on $\thptj$, and, moreover, $\stahj$ is a $\ttrnor$-closed subset of $\trchj$
--- we have that
\begin{equation}
\ttrnorli D_n=\ptrnorli D_n=D\in\stahj\cap\thptj\ifed\hstahj \fin .
\end{equation}
Therefore, $\hstahj$ is a $\ptrnor$-closed subset of $\thptj$. Recall that the projective norm
and the Hermitian projective norm are equivalent on $\thsptjs$, hence, a subset of $\thsptjs$ is
$\ptrnortr$-closed iff it is $\ptrnor$-closed.
\end{proof}

Another noteworthy fact about the convex set $\hstahj$ is that it is $\ttrnor$-dense in $\stahj$.
To clarify this point, let us consider the following subset of the Hilbert space $\hhjj$:
\begin{equation}
\hjo\defi\{\va\in\hhjj\setminus\{0\}\colon 1\le\srank(\va)<\infty\} \fin ,
\end{equation}
where $\srank(\va)$ is the Schmidt rank of $\va\neq 0$ (Fact~\ref{schmdec}). Clearly, by the
Schmidt decomposition of a vector in $\hhjj$, it is clear that $\hjo$ is norm-dense
in $\hhjj$. We also consider the associated subset $\pustahjo$ of $\pustahj$, defined by
\begin{equation}
\pustahjo\defi\{\vaa\in\pustahj\colon\va\in\hjo,\ \|\va\|=1\} \fin .
\end{equation}

\begin{proposition}
$\stahj=\clco(\pustahjo)=\clttrnor(\hstahj)$.
\end{proposition}

According to a well-known result~\cite{Reed}, the bipartite trace class $\trchj$ is a two-sided
$\ast$-ideal in $\bophj$: $\mbox{$K\in\trchj$, $L,M\in\bophj$}\implies\mbox{$K^\ast\in\trchj$
and $LKM\in\trchj$}$. For the cross trace class $\thptj\equiv\htrchj\subset\trchj\subset\bophj$,
a weaker --- yet important --- property holds; namely, $C\in\thptj\implies C^\ast\in\thptj$
(Corollary~\ref{corcro}) and, moreover, one can prove the following result (and the consequent
Corollary~\ref{corbou}):

\begin{proposition} \label{probou}
Given any cross trace class operator $C\in\thptj\equiv\htrchj$, and any bounded operators
$A,E\in\bop$ and $B,F\in\bopjj$, the trace class operator $(A\otimes B)\tre C(E\otimes F)$
is contained in the cross trace class $\thptj$ too, and the following relation holds:
\begin{equation} \label{relbou}
\|(A\otimes B)\tre C(E\otimes F)\ptrn\le \|A\nori \|B\nori\tre \|E\nori\tre \|F\nori\tre \|C\ptrn .
\end{equation}
Moreover, if, in particular, $C\in\thsptjs$, $E=A^\ast$ and $F=B^\ast$, then the selfadjoint
trace class operator $(A\otimes B)\tre C(E\otimes F)=(A\otimes B)\tre C(A\otimes B)^\ast$ is
contained in the selfadjoint cross trace class $\thsptjs$ too, and
\begin{equation} \label{relbou-bis}
\ptrntr{(A\otimes B)\tre C(A\otimes B)^\ast}\le (\|A\nori\tre \|B\nori)^2\sei \ptrntr{C} .
\end{equation}
\end{proposition}

\begin{corollary} \label{corbou}
For any $C\in\thptj$ and $L,M\in\bop\altp\bopjj$, $L\tre CM\in\thptj$.
\end{corollary}

Among all standard decompositions --- see Definition~\ref{destade} --- of a cross trace class operator,
there are some that deserve a special focus. Therefore, coherently with Definition~\ref{optimality},
and adopting the new notations established in Remark~\ref{remide} and Definition~\ref{defide} for
the Banach spaces of cross trace class operators, we set the following:

\begin{definition} \label{destanda}
A standard decomposition $\sum_k r_k \tre (X_k\otimes Y_k)$ of $C\in\trc\prtp\trcjj\equiv\htrchj$
--- with $r_k\ge 0$ and $X_k\otimes Y_k\in\ethtj$ ($\sum_k r_k<\infty$, $\|X_k\trn\tre\|Y_k\trn=1$) ---
will be said to be $\ptrnor$-\emph{optimal} (or, simply, \emph{optimal}) if $\|C\ptrn=\sum_k r_k$,
and a cross trace class operator admitting an optimal standard decomposition will be called
$\ptrnor$-\emph{optimally decomposable} (or simply \emph{optimally decomposable}). We will denote
by $\optrn(C)$ the --- possibly empty --- set of all $\ptrnor$-optimal standard decompositions of
$C$, and by
\begin{equation}
\opd\defi\big\{C\in\thptj\colon\optrn(C)\neq\varnothing\big\}
\end{equation}
the set of all $\ptrnor$-optimally decomposable operators in $\thptj$.

A \emph{Hermitian} standard decomposition $\sum_k r_k \tre (X_k\otimes Y_k)$ of $C\in\thsptjs\equiv\htrchjs$
--- with $r_k\ge 0$ and $X_k\otimes Y_k\in\ethtjs$ ($\sum_k r_k<\infty$, $\|X_k\trn\tre\|Y_k\trn=1$) ---
will be said to be $\ptrnortr$-\emph{optimal} if $\ptrntr{C}=\sum_k r_k$; moreover, a selfadjoint cross
trace class operator admitting a $\ptrnortr$-optimal (Hermitian) standard decomposition will be called
\emph{$\ptrnortr$-optimally decomposable}. We will denote by $\optrntr(C)$ the --- possibly empty ---
set of all $\ptrnortr$-optimal (Hermitian) standard decompositions of $C$, and by
\begin{equation}
\opds\defi\big\{C\in\thsptjs\colon\optrntr(C)\neq\varnothing\big\}
\end{equation}
the set of all $\ptrnortr$-optimally decomposable operators in $\thsptjs$.
\end{definition}

\begin{remark} \label{relineal}
It is clear that $\opd$ and $\opds$ are \emph{lineal} subsets of $\thptj$ and $\thsptjs$, respectively;
namely, every complex or real multiple of a vector in $\opd$ or $\opds$, respectively, belongs to the
same set.
\end{remark}

\begin{remark}
A selfadjoint cross trace class operator $C\in\thsptjs$ admits both Hermitian standard decompositions and
standard decompositions that are \emph{not} Hermitian. In general, a $\ptrnortr$-optimal Hermitian standard
decomposition (if it exists) will \emph{not} be optimal \emph{tout court}; in particular, in the case where
$\ptrntr{C}>\|C\ptrn$, a $\ptrnortr$-optimal Hermitian standard decomposition --- if it exists --- \emph{cannot}
be optimal. Therefore, if $C$ is $\ptrnortr$-optimally decomposable, then it may not be optimally decomposable,
and \emph{vice versa}; but the picture becomes somewhat simpler if we restrict to \emph{Hermitian} standard
decompositions only (see Proposition~\ref{proptimals} below).
\end{remark}

\begin{proposition} \label{proptimals}
Let $C\in\thsptjs$. The following facts hold true:
\begin{enumerate}[label=\tt{(O\arabic*)}]

\item \label{opta}
If $C$ admits an optimal --- i.e., a $\ptrnor$-optimal --- \emph{Hermitian} standard decomposition,
then this decomposition is also $\ptrnortr$-optimal and $\ptrntr{C}=\|C\ptrn$.

\item \label{optb}
Clearly, if, conversely, $C$ admits a $\ptrnortr$-optimal (Hermitian) standard decomposition and,
moreover, $\ptrntr{C}=\|C\ptrn$, then this decomposition is also optimal.

\item \label{optc}
By the previous two points, the set
\begin{align}
\hoptrn(C) \defi \big\{
&
{\textstyle \sum_k r_k \tre (X_k\otimes Y_k)}\colon r_k\ge 0,\
X_k\otimes Y_k\in\ethtjs,\ {\textstyle \sum_k r_k}<\infty,
\nonumber \\ \label{sethoptrn}
&
\|X_k\trn\tre\|Y_k\trn=1, \ {\textstyle \sum_k r_k \tre (X_k\otimes Y_k)}=C,
\ \|C\ptrn={\textstyle \sum_k r_k}\big\} \fin ,
\end{align}
consisting of all optimal \emph{Hermitian} standard decompositions of $C$, is empty if $\ptrntr{C}>\|C\ptrn$;
while, if $\hoptrn(C)$ is nonempty, then it must coincide with the set $\optrntr(C)$ of all $\ptrnortr$-optimal
(Hermitian) standard decompositions of $C$ and $\ptrntr{C}=\|C\ptrn$. Therefore, either $\hoptrn(C)$ is empty
or $\hoptrn(C)=\optrntr(C)\neq\varnothing$, and, in the latter case, $\ptrntr{C}=\|C\ptrn$.

\end{enumerate}
\end{proposition}

In the light of the previous result, we set the following:

\begin{definition}
A selfadjoint cross trace class operator $C\in\thsptjs$ is said to be \emph{Hermitian-optimally decomposable}
if the set $\hoptrn(C)$ defined by~\eqref{sethoptrn} is nonempty, i.e., if it admits an optimal Hermitian
standard decomposition (equivalently, by Proposition~\ref{proptimals}, if $C$ is $\ptrnortr$-optimally decomposable
and $\ptrntr{C}=\|C\ptrn$; also, if $\hoptrn(C)=\optrntr(C)\neq\varnothing$). We will denote by
\begin{align}
\hopd
& \defi
\big\{C\in\thsptjs\colon\hoptrn(C)\neq\varnothing\big\}
\nonumber \\
& \dodici =
\big\{C\in\opds\colon\ptrntr{C}=\|C\ptrn \big\}
\nonumber \\
& \dodici =
\big\{C\in\opds\colon\hoptrn(C)=\optrntr(C)\big\}
\nonumber \\ \label{multichar}
& \dodici \subset
\opds
\end{align}
the (lineal) set of all Hermitian-optimally decomposable selfadjoint cross trace class operators.
\end{definition}

The various equivalent characterizations of $\hopd$ in~\eqref{multichar} follow easily from
Proposition~\ref{proptimals}; precisely, the characterization on the second line follows from~\ref{opta}
and~\ref{optb}, whereas the characterization on the third line follows from~\ref{optc}.

\begin{proposition} \label{proexods}
If both $\hh$ and $\jj$ are finite-dimensional Hilbert spaces --- $\enne\equiv\dim(\hhjj)<\infty$ ---
then every linear operator in $\thptj$ is $\ptrnor$-optimally decomposable, and every selfadjoint linear
operator in $\thsptjs$ is $\ptrnortr$-optimally decomposable. In the first case, there is a $\ptrnor$-optimal
standard decomposition of the linear operator consisting of, at most, $2\enne^2+1$ simple tensors; in the
second case, a $\ptrnortr$-optimal Hermitian standard decomposition consisting of, at most, $\enne^2+1$
simple tensors.
\end{proposition}

\begin{proof}
By Remark~\ref{relineal}, it is sufficient to prove the result for any $C\in\thptj$ such that
$\|C\ptrn=1$ (respectively, for any $C\in\thsptjs$ such that $\ptrntr{C}=1$). Then, let $C$ belong
to the unit \emph{sphere} $\sph\big(\thptj\big)\defi\big\{C\in\thptj\colon\|C\ptrn=1\big\}$
in the $\enne^2$-dimensional \emph{complex} linear space $\bophj=\thptj$ (isomorphism of
linear spaces), which, regarded as a \emph{real} linear space, by field restriction, has
dimension $2\enne^2$. It is known --- see, e.g., Proposition~{2.2} of~\cite{Ryan} (also
see Proposition~{7.D} of~\cite{Kubrusly}) --- that the closed unit \emph{ball} in $\thptj$
coincides with the closed convex hull of its subset $\nb(\ball(\trc),\ball(\trcjj))$, where
the closure is taken wrt any norm topology, and $\ball(\trc)$, $\ball(\trcjj)$ are the
closed unit balls in $\trc$ and $\trcjj$, respectively; i.e.,
\begin{equation}
\ball\big(\thptj\big)=\clco\big(\nb(\ball(\trc),\ball(\trcjj))\big)=
\co\big(\nb(\ball(\trc),\ball(\trcjj))\big) .
\end{equation}
The second equality holds because $\enne\equiv\dim(\hhjj)<\infty$, hence the set
$\nb(\ball(\trc),\ball(\trcjj))$ and --- as a consequence of Carath\'eodory's theorem
(see, e.g., Theorem~{17.2} of~\cite{Rockafellar}) --- also its convex hull are compact
(wrt to any norm on $\bophj$). Now --- by Milman's partial converse of the Krein-Milman
Theorem (see Theorem~{2.10.15} of~\cite{Megginson}), and since the extreme points of the
unit ball $\ball(\bsb)$ of a Banach space $\bsb$ must be contained in the unit sphere
$\sph(\bsb)$ (Fact~\ref{facextp}) --- $\ext\big(\ball\big(\thptj\big)\big)\subset
\nb(\ball(\trc),\ball(\trcjj))\cap\sph(\thptj)=\nb(\sph(\trc),\sph(\trcjj))$
(one can actually prove that the extreme points of the unit ball of $\thptj$ are given
by $\ext\big(\ball\big(\thptj\big)\big)=\nb(\paih,\paij)\subset\sph\big(\thptj\big)$,
where $\paih$, $\paij$ are sets of rank-one partial isometries on $\hh$ and $\jj$, respectively).
Therefore, by the Minkowski-Carath\'eodory theorem (see, e.g., Theorem~{8.11} of~\cite{Simon_conv}),
and considering that the \emph{real} dimension of the linear space $\bophj=\thptj$ is $2\enne^2$,
$C$ can be expressed as a convex combination of --- at most, $2\enne^2+1$ --- simple tensors, i.e.,
\begin{equation}
C=\sum_{k=1}^{\ka} r_k \tre (X_k\otimes Y_k) \fin , \quad 1\le\ka\le 2\enne^2+1 \fin ,
\end{equation}
with $r_k> 0$, $\sum_k r_k=1=\|C\ptrn$ and $X_k\otimes Y_k\in\ext\big(\ball\big(\thptj\big)\big)
\subset\nb(\sph(\trc),\sph(\trcjj))$. Therefore, $C$ admits a $\ptrnor$-optimal standard decomposition
consisting of, at most, $2\enne^2+1$ simple tensors. The case of the $\enne^2$-dimensional \emph{real}
linear space $\bophjsa=\thsptjs$ is analogous.
\end{proof}

\begin{corollary} \label{corhop}
If both $\hh$ and $\jj$ are finite-dimensional Hilbert spaces, then
\begin{align}
\hopd
&=
\big\{C\in\thsptjs\colon\ptrntr{C}=\|C\ptrn\big\}
\nonumber \\
& =
\big\{C\in\thsptjs\colon\hoptrn(C)=\optrntr(C)\big\} .
\end{align}
\end{corollary}

\subsection{Natural isometric embeddings and some consequences}
\label{natisom}

Given closed subspaces $\vv$, $\ww$ of the Hilbert spaces $\hh$ and $\jj$, respectively, it is
quite natural to wonder how the Banach spaces $\trcvv\prtp\trcww$, $\trcvvs\prtp\trcwws$ are
related to $\thptj$ and $\thsptjs$. Before stating the main results, we start with some useful
preliminary facts.

First note that, for any closed subspace $\vv$ of the Hilbert space $\hh$, and every trace class operator
$S\in\trcvv$, there is a natural (isometric, linear) embedding of $S$ in $\trc$ --- let us still denote it by
the same symbol $S$ --- that is defined by putting $S\phi=S(\phi_1 +\phi_2)\defi S\phi_1$, where $\phi=\phi_1 +\phi_2$
is the decomposition of a generic vector $\phi\in\hh$ associated with the orthogonal sum decomposition
$\hh=\vv\oplus\vv^\perp$, where $\vv^\perp$ is the orthogonal complement of $\vv$ (which is a closed subspace of
$\hh$ too). Denoting by $\pi_1$, $\pi_2=I-\pi_1$ the orthogonal projections onto $\vv$ and $\vv^\perp$, respectively,
we see that $S\phi=S\tre\pi_1\phi=\pi_1 S\phi$ --- thus, $[S,\pi_1]=0=[S,I-\pi_1]=[S,\pi_2]$ and $\pi_2 S\tre\pi_2=S\tre\pi_2=0$
--- and hence $S=\pi_1 S\tre\pi_1+\pi_1 S\tre\pi_2+\pi_2 S\tre\pi_1+\pi_2 S\tre\pi_2=\pi_1 S\tre\pi_1+\pi_2 S\tre\pi_2
=\pi_1 S\tre\pi_1$.

Given any closed subspace $\ww$ of the Hilbert space $\jj$, we then also get a natural (injective,
linear) embedding of $\trcvv\altp\trcww$ into $\thatj$, that is isometric wrt the trace norms of the
Banach spaces $\trcvw$ and $\trchj$. Analogously, we have a natural (injective, linear) embedding of
$\trcvvs\altp\trcwws$ into $\thsatjs$, that is isometric wrt the trace norms of the Banach spaces
$\trcvwsa$ and $\trchjsa$.

\begin{proposition} \label{proimm}
Let $\vv$, $\ww$ be closed subspaces of the Hilbert spaces $\hh$ and $\jj$, respectively, and let
$\pi\in\bop$, $\varpi\in\bopjj$ be the orthogonal projections onto $\vv$ and $\ww$. The natural
--- injective, linear --- embedding map of $\trcvv\altp\trcww$ into $\trc\altp\trcjj$ is isometric
wrt the projective norms of these spaces. Hence, it extends to a linear isometry $\isot$ of the
Banach space $\trcvv\prtp\trcww$ into the Banach space $\thptj$; moreover, for ever $C\in\thptj$,
$(\pi\otimes\varpi)\tre C(\pi\otimes\varpi)$ is a cross trace class operator too, and
\begin{align}
\ran(\isot)
& =
\big\{C\in\thptj\colon C=(\pi\otimes\varpi)\tre C (\pi\otimes\varpi)\big\}
\nonumber \\
& =
\big\{(\pi\otimes\varpi)\tre C (\pi\otimes\varpi)\colon C\in\thptj\big\} \fin .
\end{align}
The linear map $\prot\colon\thptj\ni C\mapsto(\pi\otimes\varpi)\tre C (\pi\otimes\varpi)\in\thptj$
is a continuous projection; hence, $\ran(\prot)=\ran(\isot)$ is a complemented subspace of $\thptj$.

An analogous result holds for the Banach spaces $\trcvvs\prtp\trcwws$ and $\thsptjs$.
\end{proposition}

\begin{proof}
Let $\isoh\colon\trcvv\rightarrow\trc$, $\isoj\colon\trcww\rightarrow\trcjj$ be the natural linear
isometric embeddings of $\trcvv$, $\trcww$ into $\trc$ and $\trcjj$, respectively. Let us denote
by $\pi\in\bop$, $\varpi\in\bopjj$ the orthogonal projections onto the closed subspaces $\vv$, $\ww$
of $\hh$ and $\jj$, respectively. Then, we have:
\begin{align}
\ran(\isoh)
& =
\{S\in\trc\colon S=\pi\tre S\tre\pi\}=\{\pi\tre S\tre\pi\colon S\in\trc\} \fin ,
\nonumber \\ \label{relisohj}
\ran(\isoj)
& =
\{T\in\trcjj\colon T=\varpi\sei T\tre\varpi\}=\{\varpi\sei T\tre\varpi\colon T\in\trcjj\} \fin .
\end{align}

Let us define the linear map $\isoh\otimes\isoj\colon\trcvv\altp\trcww\rightarrow\thptj$ as
the standard tensor product of the linear maps $\isoh$ and $\isoj$ (see, e.g., Section~{3.3}
of~\cite{Kubrusly}), i.e.,
$\big(\isoh\otimes\isoj\big)(\ops\otimes\opt)\defi\isoh(\ops)\otimes\isoj(\opt)$, for all
$\ops\in\trcvv$ and $\opt\in\trcww$. Here, with a slight abuse, we have assumed the codomain of
the map $\isoh\otimes\isoj$ to be $\thptj$ rather than $\thatj$. Since the maps $\isoh$ and
$\isoj$ are injective, then $\isoh\otimes\isoj$ is injective too, and the linear subspace
$\ran(\isoh\otimes\isoj)$ of $\thatj$ can be identified with --- i.e., is linearly isomorphic
to --- $\ran(\isoh)\altp\ran(\isoj)=\trcvv\altp\trcww$. Thus, by natural embedding,
$\trcvv\altp\trcww$ can be identified with the linear subspace of $\trc\altp\trcjj$ consisting of
all operators of the latter space \emph{living} on the closed subspace $\vvww$ of $\hhjj$; i.e.,
$\trcvv\altp\trcww=\ran(\isoh)\altp\ran(\isoj)\equiv\ran(\isoh\otimes\isoj)=
\big\{C\in\trc\altp\trcjj\colon C=(\pi\otimes\varpi)\tre C(\pi\otimes\varpi)\big\}$.
Note that, by this remark, denoting by $\vwptrnor$ the projective norm of the algebraic tensor product
$\trcvv\altp\trcww$, it is clear that $\|C\ptrn\le \vevw C\ptrn$, for all $C\in\trcvv\altp\trcww$,
because the evaluation of the projective norm $\vevw C\ptrn$ involves decompositions of $C$ into
elementary tensors \emph{living on the subspace $\vvww$ only}. Let us prove that the reverse
inequality holds too. In fact, for every $C\in\trcvv\altp\trcww$, if $C=\sum_k S_k\otimes T_k$
is any finite decomposition, then we also have that $C=\sum_k(\pi\otimes\varpi)\tre(S_k\otimes T_k)
(\pi\otimes\varpi)$; hence:
\begin{align}
\|C\ptrn
& \defi
\inf\big\{{\textstyle \sum_k \|S_k\trn\tre \|T_k\trn}\colon {\textstyle\sum_k
S_k\otimes T_k}=C=(\pi\otimes\varpi)\tre C(\pi\otimes\varpi) \big\}
\nonumber \\
& \dodici \ge
\inf\big\{{\textstyle \sum_k \|\pi\tre S_k\tre \pi\trn\tre \|\varpi\sei T_k\tre\varpi\trn}
\colon {\textstyle\sum_k S_k\otimes T_k}=C=(\pi\otimes\varpi)\tre C(\pi\otimes\varpi) \big\}
\nonumber \\
& \dodici =
\inf\big\{{\textstyle \sum_k \|S_k\trn\tre \|T_k\trn} \colon {\textstyle\sum_k
\tre S_k\otimes T_k}=C ,\ S_k=\pi\tre S_k\tre\pi, \ T_k=\varpi\sei T_k\tre\varpi\big\}
= \vevw C\ptrn \fin ,
\end{align}
where the infimum on the first line is extended over \emph{all} (finite) decompositions of $C$
into elementary tensors, and, for obtaining the second line, we have used the fact that, for any
$S\in\trc$ and $T\in\trcjj$, $\|\pi\tre S\tre\pi\trn\le(\|\pi\nori)^2\tre\|S\trn=\|S\trn$ and
$\|\varpi\sei T\tre\varpi\trn\le\|T\trn$.

Therefore, actually, for every $C\in\ran(\isoh\otimes\isoj)\equiv\trcvv\altp\trcww$, we have that
$\vevw C\ptrn=\|C\ptrn$; i.e., the linear map $\isoh\otimes\isoj\colon\trcvv\altp\trcww\rightarrow\thptj$
is isometric wrt the norms $\vwptrnor$ and $\ptrnor$. Hence, by continuous extension we obtain an
isometric embedding --- a linear isometry --- of $\trcvv\prtp\trcww$ into $\trc\prtp\trcjj$. Summarizing,
we have proved that there is a linear isometry $\isot\colon\trcvv\prtp\trcww\rightarrow\thptj$ such that
$\isot(\ops\otimes\opt)=\isoh(\ops)\otimes\isoj(\opt)$, for all $\ops\in\trcvv$ and $\opt\in\trcww$.
Hence, for every $\opc=\sum_k\ops_k\otimes\opt_k\in\trcvv\prtp\trcww$ --- convergence wrt the norm
$\vwptrnor$ being understood --- we have that
\begin{align}
\isot(\opc)=\ptrnsum_k\isoh(\ops_k)\otimes\isoj(\opt_k)
& =
\ptrnsum_k\big(\pi\isoh(\ops_k)\tre\pi\big)\otimes\big(\varpi\isoj(\opt_k)\tre\varpi\big)
\nonumber \\
& =
(\pi\otimes\varpi)\tre\big({\textstyle \ptrnsum_k\isoh(\ops_k)\otimes\isoj(\opt_k)}\big)
\tre(\pi\otimes\varpi)
\nonumber \\
& =
(\pi\otimes\varpi)\sei\isot(\opc)\sei(\pi\otimes\varpi) \fin .
\end{align}
Here, the first equality holds by the continuity of $\isot$, and the second equality holds by
relations~\eqref{relisohj}, whereas, for the third equality, observe that, by Proposition~\ref{probou},
the linear map
\begin{equation} \label{defprot}
\prot\colon\thptj\ni C\mapsto(\pi\otimes\varpi)\tre C (\pi\otimes\varpi)\in\thptj
\end{equation}
is well defined --- i.e., $(\pi\otimes\varpi)\tre C (\pi\otimes\varpi)\in\thptj$ --- and bounded. Therefore,
$\ran(\isot)\subset\big\{(\pi\otimes\varpi)\tre C (\pi\otimes\varpi)\colon C\in\thptj\big\}
=\big\{C\in\thptj\colon C=(\pi\otimes\varpi)\tre C (\pi\otimes\varpi)\big\}$. Let us show that
the reverse inclusion relation holds too. Indeed, for every $C=\ptrnsum_k S_k\otimes T_k\in\thptj$
such that $C=(\pi\otimes\varpi)\tre C (\pi\otimes\varpi)$, we have:
\begin{align}
C=\ptrnsum_k S_k\otimes T_k
& =
(\pi\otimes\varpi)\tre\big({\textstyle \ptrnsum_k S_k\otimes T_k}\big)\tre(\pi\otimes\varpi)
\nonumber \\
& =
\ptrnsum_k\big(\pi\tre S_k\tre\pi\big)\otimes\big(\varpi\tre T_k\tre\varpi\big)
\nonumber \\
& =
\ptrnsum_k\isoh(\ops_k)\otimes\isoj(\opt_k)
=\isot\big({\textstyle \sum_k\ops_k\otimes\opt_k}\big) \fin ,
\end{align}
for some sets $\{\ops_k\}\subset\trcvv$, $\{\opt_k\}\subset\trcww$ --- note:
$\pi\tre S_k\tre\pi\in\ran(\isoh)$, $\varpi\tre T_k\tre\varpi\in\ran(\isoj)$ ---
satisfying $\|\ops_k\trn=\|\isoh(\ops_k)\trn=\|\pi\tre S\tre\pi\trn\le\|S\trn$ and
$\|\opt_k\trn=\|\isoj(\opt_k)\trn=\|\varpi\sei T\tre\varpi\trn\le\|T\trn$, so that
$\sum_k\ops_k\otimes\opt_k$ is absolutely convergent wrt the norm $\vwptrnor$ of
$\trcvv\prtp\trcww$.

Now, let us consider the bounded linear map $\prot\colon\thptj\rightarrow\thptj$ defined
by~\eqref{defprot}.
Clearly, $\prot^2=\prot$; hence, $\prot$ is a continuous projection. Therefore, by a well-known
result, the closed subspace $\ran(\prot)=\ran(\isot)$ of $\thptj$ is complemented (see Corollary~{3.2.15}
of~\cite{Megginson}, or Remark~{5.3} of~\cite{Kubrusly}).

The case of the Banach spaces $\trcvvs\prtp\trcwws$ and $\thsptjs$ is analogous.
\end{proof}

\begin{proposition} \label{profidec}
Let $\pi\in\bop$, $\varpi\in\bopjj$ be the orthogonal projections onto finite-dimensional
subspaces $\vv$, $\ww$ of $\hh$ and $\jj$, respectively. Then, every cross trace class
operator $C\in\thptj$ such that $C=(\pi\otimes\varpi)\tre C(\pi\otimes\varpi)$ is optimally
decomposable. Analogously, every selfadjoint cross trace class operator $C\in\thsptjs$ such
that $C=(\pi\otimes\varpi)\tre C(\pi\otimes\varpi)$ is $\ptrnortr$-optimally decomposable.
\end{proposition}

\begin{proof}
By Proposition~\ref{proimm} --- given any pair $\pi\in\bop$, $\varpi\in\bopjj$ of orthogonal
projections --- if $C\in\thptj$ is such that $C=(\pi\otimes\varpi)\tre C(\pi\otimes\varpi)$,
then --- putting $\vv=\ran(\pi)$ and $\ww=\ran(\varpi)$, and considering linear isometry
$\isot\colon\trcvv\prtp\trcww\rightarrow\thptj$ --- we can identify $C$ with the unique element
$\opc$ of $\trcvv\prtp\trcww$ such that $\isot(\opc)=C$. Therefore, we have:
\begin{equation}
\|C\ptrn=\|(\pi\otimes\varpi)\tre C(\pi\otimes\varpi)\ptrn=
\vevw (\pi\otimes\varpi)\tre C(\pi\otimes\varpi)\ptrn \fin .
\end{equation}
Here, $\vwptrnor$ denotes the norm of the projective tensor product $\trcvv\prtp\trcww$, and, with
a slight abuse of notation, $C=(\pi\otimes\varpi)\tre C(\pi\otimes\varpi)$ is regarded as an element
of both $\thptj$ and $\trcvv\prtp\trcww$. If the subspaces $\vv$, $\ww$ are finite-dimensional, then,
by Proposition~\ref{proexods}, there is an optimal standard decomposition $\sum_k r_k\tre (X_k\otimes Y_k)$
of the operator $(\pi\otimes\varpi)\tre C(\pi\otimes\varpi)\in\trcvv\prtp\trcww$; i.e.,
$\sum_k r_k=\vevw (\pi\otimes\varpi)\tre C(\pi\otimes\varpi)\ptrn=\|C\ptrn$. Then, by the isometric
embedding of $\trcvv\prtp\trcww$ in $\thptj$, we obtain an optimal decomposition of $C$.
\end{proof}

We have argued that, if $\dim(\hhjj)<\infty$, then both $\opd=\thptj$ and $\opds=\thsptjs$
(see Proposition~\ref{proexods}); moreover, by Corollary~\ref{corhop}, the set $\hopd$
coincides with $\big\{C\in\thsptjs\colon\ptrntr{C}=\|C\ptrn\big\}$. Now, as a consequence
of Proposition~\ref{profidec}, we obtain the following natural generalization of this fact:

\begin{corollary}
The sets $\opd$ and $\opds$ are norm-dense in $\thptj$ and $\thsptjs$, respectively.
$\hopd$ is a norm-dense subset of $\big\{C\in\thsptjs\colon\ptrntr{C}=\|C\ptrn\big\}$.
\end{corollary}

\begin{proposition} \label{proproj}
For every cross trace class operator $C\in\thptj$, and for every pair $\pi\in\bop$, $\varpi\in\bopjj$
of orthogonal projection operators, $(\pi\otimes\varpi)\tre C(\pi\otimes\varpi)$ is a cross trace class
operator too, and the following inequality holds:
\begin{equation} \label{reproj-bis}
\|(\pi\otimes\varpi)\tre C(\pi\otimes\varpi)\ptrn\le \|C\ptrn \fin .
\end{equation}

Let $\hh=\oplus_j\hj$, $\jj=\oplus_l\jl$ be orthogonal sum decompositions of the Hilbert spaces
$\hh$ and $\jj$, and let $\pij\in\bop$, $\pil\in\bopjj$ be the orthogonal projection operators onto
the closed subspaces $\hj\subset\hh$ and $\jl\subset\jj$, respectively. Then, for every cross trace
class operator $C\in\thptj$, $(\pijl)\tre C(\pimn)$ is a cross trace class operator too, and we have:
\begin{equation} \label{trine}
\|C\ptrn\le\sum_{jlmn}\left\|\cjlmn\right\ptrn , \quad
\mbox{where $\cjlmn\defi(\pijl)\tre C(\pimn)\in\thptj$.}
\end{equation}

If, in particular, the previous orthogonal sum decompositions of the Hilbert spaces $\hh$ and $\jj$
are finite --- i.e., if $\hh=\oplus_{j=1}^\nuno\hj$ and $\jj=\oplus_{l=1}^\ndue\jl$, with
$\nuno,\ndue\in\nat$ --- then we have that
\begin{equation} \label{trineb}
\frac{1}{\nn}\sum_{jlmn}\left\|\cjlmn\right\ptrn\le\|C\ptrn\le
\sum_{jlmn}\left\|\cjlmn\right\ptrn = \sum_{j,\tre m=1}^\nuno\otto\sum_{l,\tre n=1}^\ndue
\left\|\cjlmn\right\ptrn ,
\end{equation}
where
\begin{equation} \label{defnn}
\nn\defi\card\big\{\cjlmn\neq 0\colon j,m=1,\ldots,\nuno, \ l,n=1,\ldots,\ndue\big\}
\le(\nuno\ndue)^2 \fin ;
\end{equation}
hence, in the case where $[C,\pijl]=0$, for all $j\in\{1,\ldots,\nuno\}$ and $l\in\{1,\ldots,\ndue\}$
--- equivalently, for $C=\sum_{jl} \cjl$, with $\cjl\defi(\pijl)\tre C(\pijl)$ --- $\nn\le\nuno\ndue$
and we have:
\begin{equation} \label{trinec}
\frac{1}{\nuno\ndue}\sum_{jl}\left\|\cjl\right\ptrn\le\frac{1}{\nn}\sum_{jl}\left\|\cjl\right\ptrn\le
\|C\ptrn\le\sum_{jl}\left\|\cjl\right\ptrn .
\end{equation}
\end{proposition}

\begin{proof}
The first assertion follows immediately from Proposition~\ref{probou} --- see relation~\eqref{relbou}
--- taking into account the fact that $\|\pi\nori=\|\varpi\nori=1$.
Regarding the second assertion (relation~\eqref{trine}), note that --- since $\sum_{jl}\pijl=I$,
where the double sum, whenever infinite, converges (unconditionally, i.e., for every ordering
of the index set) wrt the strong operator topology of $\bophj$  --- for every $\va\in\hhjj$ we have:
\begin{align}
C\tre\va = C \sum_{mn}\big((\pimn)\tre\va\big)
& =
\sum_{mn}\big(C(\pimn)\tre\va\big)
\nonumber \\ \label{sameli}
& =
\sum_{jl}\sum_{mn}\big((\pijl)\tre C(\pimn)\tre\va\big)=\sum_{jl}\sum_{mn}\big(\cjlmn\va\big) .
\end{align}
Here, the second equality holds because $C$ is a bounded operator, and $\cjlmn\defi(\pijl)\tre C(\pimn)$
belongs to $\thptj$. Thus, by relation~\eqref{sameli},  the iterated sums $\sum_{jl}\sum_{mn}\cjlmn$ ---
whenever infinite --- converge (unconditionally) to $C$ wrt the strong operator topology.

If $\sum_{jlmn}\|\cjlmn\ptrn=\infty$, relation~\eqref{trine} is obvious. Let us suppose
that $\sum_{jlmn}\|\cjlmn\ptrn<\infty$, instead, so that the sum $\sum_{jlmn}\cjlmn$ --- whenever infinite ---
is absolutely (hence, unconditionally) convergent wrt the projective norm $\ptrnor$. Hence, it will converge
unconditionally, wrt the (subspace topology induced by the) strong operator topology --- that is weaker than
the projective norm topology --- to the same limit, which, by relation~\eqref{sameli}, must be precisely $C$;
i.e.,
\begin{equation}
\ptrnsumb_{jlmn}\cjlmn= \ptrnsumb_{jl}\ptrnsumb_{mn}\cjlmn=\sum_{jl}\sum_{mn}\cjlmn=C ,
\end{equation}
where the first series converges absolutely, whence we immediately get relation~\eqref{trine}.

Assuming, now, that both the orthogonal sum decompositions $\hh=\oplus_j\hj$ and $\jj=\oplus_l\jl$
are \emph{finite} --- specifically, that they involve $\nuno$ and $\ndue$ subspaces, respectively
--- the inequality $\|\cjlmn\ptrn=\|(\pijl)\tre C(\pimn)\ptrn\le\|C\ptrn$ also entails that
$\sum_{jlmn}\|\cjlmn\ptrn\le\nn\|C\ptrn$, where the number $\nn\le(\nuno\ndue)^2$ is defined
by~\eqref{defnn}. It follows that relation~\eqref{trineb} and --- in the case where $[C,\pijl]=0$,
with $j\in\{1,\ldots,\nuno\}$, $l\in\{1,\ldots,\ndue\}$ ---~\eqref{trinec} hold too.
\end{proof}

\begin{remark}
By relation~\eqref{reproj-bis} in Proposition~\ref{proproj}, and taking into account Proposition~\ref{proimm},
for every pair of orthogonal projections $\pi\in\bop$ and $\varpi\in\bopjj$, setting $\vv=\ran(\pi)$ and
$\ww=\ran(\varpi)$,
\begin{equation} \label{reproj-quater}
\vevw (\pi\otimes\varpi)\tre C(\pi\otimes\varpi)\ptrn=\|(\pi\otimes\varpi)\tre
C(\pi\otimes\varpi)\ptrn\le \|C\ptrn , \quad \forall\cinque C\in\thptj \fin ,
\end{equation}
where $\vwptrnor$ is the norm of the projective tensor product $\trcvv\prtp\trcww$.
\end{remark}

\subsection{The bipartite trace class vs the cross trace class}
\label{tcvsctc}

Exploiting the tools developed so far, we are now able to obtain further information about the Banach
space $\trc\prtp\trcjj$ of all cross trace class operators, with a special focus on its relation with
the \emph{bipartite trace class} $\trchj$, that coincides, as we have shown in Subsection~\ref{natens}
(see Theorem~\ref{pronatp}), with the natural (or spatial) tensor product $\trc\otimes\trcjj$ of the
trace classes $\trc$ and $\trcjj$, defined by~\eqref{natp}. Eventually, it will turn out that there is
a sharp distinction between the case where at least one of the Hilbert spaces of the bipartition $\hhjj$
is finite-dimensional and the \emph{genuinely infinite-dimensional} setting, i.e., the case where
$\dim(\hh)=\dim(\jj)=\infty$.
Let us focus, at first, on the case where \emph{at most one} of the Hilbert spaces
$\hh$ and $\jj$ is infinite-dimensional.

\begin{proposition} \label{profoc}
The normed vector space $\big(\thatj,\ptrnor\big)$ is complete --- otherwise stated,
$\thatj=\trc\prtp\trcjj$ --- iff $\hh$ and/or $\jj$ are finite-dimensional.
In the case where $\emme=\min\{\dim(\hh),\dim(\jj)\}<\infty$, every element $C$ of
$\thptj=\thatj$ admits a (finite) standard decomposition of the form
\begin{equation} \label{decopfin}
C=\sum_{k=1}^{\ka} r_k \tre (X_k\otimes Y_k) \fin ,
\end{equation}
with $\ka=\emme^2$, $r_k\ge 0$, $X_k\otimes Y_k\in\ethtj$ and $\|X_k\trn\tre\|Y_k\trn=1$;
moreover, $\{X_k\}_{k=1}^\ka$, if $\dim(\hh)=\min\{\dim(\hh),\dim(\jj)\}<\infty$
--- or $\{Y_k\}_{k=1}^\ka$, if $\dim(\jj)=\min\{\dim(\hh),\dim(\jj)\}<\infty$ --- can be
assumed to be any fixed (maximal) linearly independent set and, in particular, any fixed
orthonormal basis wrt the Hilbert-Schmidt scalar product.

Analogously, $\thsatjs=\thsptjs$ --- iff $\hh$ and/or $\jj$ are finite-dimensional, and, in the
case where $\emme=\min\{\dim(\hh),\dim(\jj)\}<\infty$, every element $C$ of the real Banach space
$\thsptjs=\thsatjs$ admits a (finite) standard decomposition of the form~\eqref{decopfin}, with
$\ka=\emme^2$, $r_k\ge 0$, $X_k\otimes Y_k\in\ethtjs$ and $\|X_k\trn\tre\|Y_k\trn=1$; the set
$\{X_k\}_{k=1}^\ka$, if $\dim(\hh)=\min\{\dim(\hh),\dim(\jj)\}<\infty$ --- or $\{Y_k\}_{k=1}^\ka$,
if $\dim(\jj)=\min\{\dim(\hh),\dim(\jj)\}<\infty$ --- can be assumed to be any fixed orthonormal
basis wrt the real Hilbert-Schmidt scalar product.
\end{proposition}

Our next step is to argue that, in the case where $\min\{\dim(\hh),\dim(\jj)\}<\infty$, the
cross trace class on $\hhjj$ is the trace class \emph{tout court} --- i.e., $\thptj=\trchj$
(where a set equality, or an isomorphism of linear spaces, is understood) --- and, moreover,
the projective norm and the trace norm are equivalent on this linear space; otherwise stated,
that the Banach space $\big(\thptj,\ptrnor\big)$ can be regarded as a \emph{renorming}~\cite{Guirao}
of the trace class $\big(\trchj,\ttrnor\big)$.

\begin{theorem} \label{theine}
Assume that $\emme=\min\{\dim(\hh),\dim(\jj)\}<\infty$.

Then, we have that
\begin{equation} \label{coinca}
\trchj=\trc\otimes\trcjj=\trc\prtp\trcjj=\trc\altp\trcjj
\end{equation}
and, hence,
\begin{equation}
\hstahj\defi\stahj\cap\thptj=\stahj \fin .
\end{equation}
The norms $\ttrnor$ and $\ptrnor$ are equivalent on the linear space $\trchj=\thptj$.
Indeed, for every $A\in\trchj=\thptj$, the following inequalities hold:
\begin{equation} \label{estgen}
\|A\ttrn\le\|A\ptrn\le\enne\sei\|A\ttrn , \quad
\mbox{where $\enne=\min\big\{4\emme,\emme^2\big\}$.}
\end{equation}
For the selfadjoint trace class operators and for the density operators on $\hhjj$, in particular,
stronger estimates hold, i.e.,
\begin{equation} \label{estsa}
\|A\ttrn\le\|A\ptrn\le 2\emme\sei\|A\ttrn , \quad \forall\cinque A\in\trchjsa \fin ,
\end{equation}
\begin{equation} \label{estdens}
1=\|D\ttrn\le\|D\ptrn\le\emme \fin , \quad \forall\cinque D\in\stahj=\hstahj \fin .
\end{equation}
Actually, in relation~\eqref{estdens}, both inequalities can be saturated; in fact,
$\|\stahj\ptrn=[1,\emme]$.

Analogously, we have that
\begin{equation} \label{coincb}
\trchjsa=\trcsa\otimes\trcsajj=\trcsa\prtp\trcsajj=\trcsa\altp\trcsajj \fin .
\end{equation}
The norms $\ttrnor$, $\ptrnor$ and $\ptrnortr$ are mutually equivalent on $\trchjsa=\thsptjs$,
because, for every $A\in\trchjsa=\thsptjs$, we have that
\begin{equation} \label{estsa-bis}
\|A\ttrn\le\|A\ptrn\le\ptrntr{A}\le 2\tre\|A\ptrn\le 4\emme\sei\|A\ttrn ;
\end{equation}
in particular, for every density operator $D\in\stahj=\hstahj$,
\begin{equation} \label{estdens-bis}
1=\|D\ttrn\le\|D\ptrn\le\ptrntr{D}\le 2\tre\|D\ptrn\le 2\emme \fin .
\end{equation}
\end{theorem}

\begin{proof}
Assume that $\emme\equiv\dim(\hh)<\infty$ and $\dim(\hh)\le\dim(\jj)\le\infty$ (obviously,
the case where $\dim(\jj)\le\dim(\hh)\le\infty$ and $\dim(\jj)=\emme<\infty$ is analogous).

Let us first prove relations~\eqref{coinca} and~\eqref{coincb}. To this end, it is sufficient
to show that --- with our initial assumption --- $\trchj\subset\thptj$. In fact, given any trace
class operator $A\in\trchj$, $A\neq 0$, we can write the singular value decomposition (recall
Remark~\ref{sivadeb})
\begin{equation}
A=\sum_{k} s_k\hvack \fin , \quad s_k>0, \ \hvack\equiv\vack
\end{equation}
--- that converges wrt the (bipartite) trace norm $\ttrnor$ --- where $\{\vak\}$, $\{\vck\}$ are
orthonormal systems in $\hhjj$. Now, let us consider \emph{extended} Schmidt decompositions of the
vectors $\vak,\vck\in\hhjj$ (recall Remark~\ref{exschmdec}), i.e.,
$\vak=\sum_{m=1}^\emme \tkm \big(\aekm\otimes\afkm\big)$ and
$\vck=\sum_{n=1}^\emme \ukn \big(\cekn\otimes\cfkn\big)$,
where $\{\aekm\}_{m=1}^\emme$, $\{\cekn\}_{n=1}^\emme$ are orthonormal bases in $\hh$,
$\{\afkm\}_{m=1}^\emme$, $\{\cfkn\}_{n=1}^\emme$ are orthonormal systems in $\jj$,
and $\tkm$, $\ukn$ are non-negative real numbers such that
$\sum_{m=1}^\emme \big(\tkm\big)^2=1=\sum_{n=1}^\emme \big(\ukn\big)^2\Longrightarrow
0\le\tkm,\ukn\le 1$.
At this point --- putting, as usual, $\hep\equiv\ep$ --- we can write
\begin{align}
A
& =
\ttrnsum_{k}s_k \sei \Big({\textstyle \sum_{m,\tre n=1}^\emme \tkm\tre\ukn}
\Big(\heemn\otimes\hffmn\Big)\Big)
\nonumber \\
& =
\sum_{m,\tre n=1}^\emme\ttrnsum_{k} s_k\sei \tkm\tre\ukn \Big(\heemn\otimes\hffmn\Big)
=\sum_{m,\tre n=1}^\emme A_{mn} \fin ,
\end{align}
where the trace class operator $A_{mn}\in\trchj$,
\begin{equation} \label{compa}
A_{mn}\defi\ttrnsum_{k} s_k\sei \tkm\tre\ukn \Big(\heemn\otimes\hffmn\Big) ,
\end{equation}
is contained in $\thptj$, as well, because
\begin{equation} \label{estco}
0<\sum_{k} s_k\sei \tkm\tre\ukn\le\sum_{k} s_k=\|A\ttrn ;
\end{equation}
hence, the series in~\eqref{compa} converges (absolutely) wrt the projective norm $\ptrnor$ too.

Therefore, a generic element of $\trchj$ can always be expressed as the sum of a finite
number of elements of $\thptj=\thatj$, so that $\trchj\subset\thptj$; hence, actually,
$\trchj=\thptj$. It follows that $\trchjsa\subset\thptj$, and hence --- recalling
relation~\eqref{eqhtrchjs} and Remark~\ref{remide} --- $\trchjsa=\trchjsa\cap\thptj=\thsptjs$,
as well.

Note that for every $A\in\trchj$ --- since, from~\eqref{compa} and~\eqref{estco}, one easily obtains
the inequality $\|A_{mn}\ptrn\le\sum_{k} s_k\sei\tkm\tre\ukn\le\|A\ttrn$ --- the  following simple
estimate holds:
\begin{equation} \label{estgena}
\|A\ptrn=\big\|{\textstyle\sum_{m,\tre n=1}^\emme A_{mn}}\bptrn\le
{\textstyle\sum_{m,\tre n=1}^\emme \|A_{mn}\ptrn}\le\emme^2\tre\|A\ttrn.
\end{equation}
However, in the case where $\emme\ge 5$ or $A=A^\ast$ --- and, in particular, for a density operator ---
stronger estimates hold. In fact, by Theorem~{14.1} of~\cite{Arveson}, relation~\eqref{estdens} holds
true and, actually, $\|\stahj\ptrn=[1,\emme]$. From this result, we easily conclude that relation~\eqref{estsa}
holds too. In fact, every $A\in\trchjsa$ can be expressed as a linear combination of two positive operators
(Fact~\ref{prolico}), i.e., $A=\apl\msei - \ami\msei = \aplu\tre\depl\msei - \amin\tre\demi$,
where $\apl\msei + \ami\msei=|A|$, $\apl\tre\ami=0=\ami\tre\apl$, $\aplu,\amin\in\errep$,
$\depl,\demi\in\stahj$, $\apl\msei=\aplu\tre\depl$ and $\ami\msei=\amin\tre\demi$; hence:
$\aplu,\amin\le\|A\ttrn\defi\tr(|A|)=\tr(\apl\msei + \ami\msei)=\aplu\msei +\amin$.
It follows that $\|A\ptrn\le\aplu\tre\|\depl\ptrn + \amin\tre\|\demi\ptrn\le\|A\ttrn\big(\|\depl\ptrn
+ \|\demi\ptrn\big)\le 2\emme\sei\|A\ttrn$.

Therefore, both relations~\eqref{estsa} and~\eqref{estdens} hold true, and --- since,
for every $A\in\trchjsa$, $\ptrntr{A}\le 2\|A\ptrn$ --- the estimates~\eqref{estsa-bis}
and~\eqref{estdens-bis} immediately follow.

Finally from the estimate~\eqref{estsa} --- since every $A\in\trchj$ can be written as
$A=A_1+\ima \sei A_2$, where $A_1$, $A_2$ are \emph{selfadjoint} trace class operators
in $\trchj$ such that $\|A_1\ttrn,\|A_2\ttrn\le\|A\ttrn$ (e.g.,
$\|A_1\ttrn\le\frac{1}{2}(\|A\ttrn+\|A^\ast\ttrn)=\|A\ttrn$) --- we see that
\begin{equation} \label{estgenb}
\|A\ptrn\le\|A_1\ptrn + \|A_2\ptrn\le 2\emme\sei(\|A_1\ttrn + \|A_2\ttrn)
\le 4\emme\sei\|A\ttrn .
\end{equation}
From~\eqref{estgena} and~\eqref{estgenb}, we now conclude that relation~\eqref{estgen} holds true.
\end{proof}

\begin{corollary} \label{corsemif}
Assume that $\emme=\min\{\dim(\hh),\dim(\jj)\}<\infty$. In this case, the Banach space dual $\big(\thptj\big)^\ast$
of $\thptj$ can be identified --- via the trace functional --- with the algebraic tensor product $\bop\altp\bopjj$,
endowed with a norm that is equivalent to the standard operator norm $\tnormi$. In fact, with our initial
assumption, $(\bop\altp\bopjj,\tnormi)$ is a Banach space, that coincides with $\bophj$, and can be identified
--- via the trace functional --- with a renorming of $\big(\thptj\big)^\ast$. Therefore, in particular,
every element of $\big(\thptj\big)^\ast$ can be expressed in the form
\begin{equation} \label{forbofu}
\sum_{j=1}^{\je} \tr\big(\argo (A_j\otimes B_j)\big)\colon\thptj\rightarrow\ccc \fin , \quad
A_j\in\bop \fin ,\ B_j\in\bopjj \fin ,
\end{equation}
with $\je\le\ka=\emme^2$.
\end{corollary}

\begin{proof}
Suppose that $\emme=\min\{\dim(\hh),\dim(\jj)\}<\infty$. Then, by Theorem~\ref{theine}, the
linear spaces $\thatj$, $\thptj$ and $\trchj$ coincide, and, moreover, the norms $\ptrnor$
and $\ttrnor$ are equivalent (hence, they induce the same topology). Therefore,
$\big(\thptj\big)^\ast$ can be identified (by means of the trace functional) --- as a linear space
--- with $\bophj$, and, as a Banach space, with a \emph{renorming} of the dual $\bophj=\trchjd$ of
$\trchj$ (this is a standard fact of the theory of renormings of Banach spaces; see, e.g.,
Chapter~{3} of~\cite{Guirao}).

Let us now focus on $\bophj$. It is a standard fact of the theory of operator algebras
(see~\cite{Kadison}, Section~{11.1.5} and Section~{11.1.6}) that --- in the case where at
least one of the Hilbert spaces $\hh$, $\jj$ is finite-dimensional --- denoting by
$\bop\otimes\bopjj$ the $\cast$-algebraic tensor product of $\bop$ and $\bopjj$, we have
that $\bop\altp\bopjj=\bop\otimes\bopjj=\bophj$; moreover, every element of $\bophj$ is of
the form $\sum_{j=1}^{\je} A_j\otimes B_j$, where $A_j\in\bop$, $B_j\in\bopjj$ and
$\je\le\ka=\emme^2$, with $\emme=\min\{\dim(\hh),\dim(\jj)\}<\infty$. (For further
details, see our previous Remark~\ref{rebop}.)

Eventually, combining all the preceding facts, we conclude that, with our previous assumptions,
every element of the Banach space $\big(\thptj\big)^\ast$ can be expressed in the form~\eqref{forbofu}.
\end{proof}

\begin{remark} \label{reminjec}
In the case where $\emme=\min\{\dim(\hh),\dim(\jj)\}<\infty$, we have an isomorphism of Banach spaces
\begin{equation} \label{isobasps}
\bop\injtp\bopjj\cong\big(\thptj\big)^\ast ,
\end{equation}
where $\bop\injtp\bopjj$ is the \emph{injective tensor product}~\cite{Defant,Ryan,DFS,Kubrusly}
of the Banach spaces $\bop$ and $\bopjj$; i.e., the Banach space completion of the algebraic
tensor product $\bop\altp\bopjj$ wrt the \emph{injective norm} $\innormi$. For every
$\sum_j A_j\otimes B_j\in\bop\altp\bopjj$, we have the following expression:
\begin{align}
\big\|{\textstyle\sum_j A_j\otimes B_j}\binnori
& \defi
\sup\big\{\big|{\textstyle\sum_j\funs (A_j)\tre\funt(B_j)}\big|\colon \funs\in\dbop,\ \funt\in\dbopjj,
\ \|\funs\norps=\|\funt\norps=1\big\}
\nonumber \\
& \dodici =
\sup\big\{\big|{\textstyle\sum_j\tr(SA_j)\tre\tr(TB_j)}\big|\colon S\in\trc,\ T\in\trcjj,
\nonumber \\ \label{norinje}
& \hspace{15.4mm}
\|S\trn=\|T\trn=1\big\} .
\end{align}
Here, we have used the fact that the unit balls $\ball(\trc)$ and $\ball(\trcjj)$ --- regarded
as subsets of the biduals $\trc^{\ast\ast}$ and $\trcjj^{\ast\ast}$ --- are \emph{norming sets}
for $\trc^{\ast}=\bop$ and $\trcjj^{\ast}=\bopjj$, respectively; see Subsection~{4.1}
in Chapter~{I} of~\cite{Defant}, or formula~{(3.3)} in Section~{3.1} of~\cite{Ryan}.
Note that the previous expression~\eqref{norinje} of the injective norm, as well as
all our further arguments below, \emph{do not} depend on the specific (finite) representation of an
element of the algebraic tensor product $\bop\altp\bopjj$. It turns out that, with our initial assumption
($\min\{\dim(\hh),\dim(\jj)\}<\infty$), $\bop\altp\bopjj$ --- which, as a linear space, coincides
with $\bophj$ (Corollary~\ref{corsemif}) --- is already complete wrt the injective norm $\innormi$.
In fact, by points~{(b1)} and~{(c1)} in Theorem~{7.17} of~\cite{Kubrusly} (also see the standard
isometric embeddings on the third line of~{(3.4)} in Section~{3.1} of~\cite{Ryan}), if $\bop\altp\bopjj$
is endowed with the norm $\innormi$, then the mapping
\begin{equation} \label{mapiso}
\bop\altp\bopjj\ni\sum_j A_j\otimes B_j\mapsto\sum_j\tr\big(\argo A_j\big)B_j\in\bo(\trc;\bopjj)
\end{equation}
is an isometric embedding, which extends to an isometric embedding of the injective tensor product
$\bop\injtp\bopjj$ into $\bo(\trc;\bopjj)$. Note that in~\eqref{mapiso}, and in~\eqref{mapiso-bis}
below, the sums are supposed to be \emph{finite}. Now --- assuming, e.g., that $\dim(\jj)<\infty$
(the case where $\dim(\hh)<\infty$ is analogous) --- it is easy to see that the map~\eqref{mapiso}
is surjective; hence, actually, $(\bop\altp\bopjj,\innormi)$ is a Banach space isomorphic to
$\bo(\trc;\bopjj)$. At this point, it is sufficient to recall that the map
\begin{equation} \label{mapiso-bis}
\bo(\trc;\bopjj)\ni\sum_j\tr\big(\argo A_j\big)B_j\mapsto
\sum_j \tr\big(\argo (A_j\otimes B_j)\big)\in\big(\thptj\big)^\ast
\end{equation}
is an isomorphism of Banach spaces (see Section~{2.2} of~\cite{Ryan}; also, point~{(e)}
in Remark~{7.13} of~\cite{Kubrusly})); hence, by composing the maps~\eqref{mapiso}
and~\eqref{mapiso-bis}, the isomorphism of Banach spaces~\eqref{isobasps} holds.
It is worth observing that, by this isomorphism and by relation~\eqref{norinje}, we have:
\begin{equation} \label{norinje-bis}
\big\|{\textstyle\sum_j A_j\otimes B_j}\binnori=\sup\big\{\big|{\textstyle\sum_j
\tr(C(A_j\otimes B_j))}\big|\colon C\in\thptj,\ \|C\ptrn=1\big\} .
\end{equation}
By this expression of the injective norm $\innormi$ and by the estimate~\eqref{estgen}, we see that
\begin{equation}
\enne^{-1}\big\|{\textstyle\sum_j A_j\otimes B_j}\btnori\le\big\|{\textstyle\sum_j
A_j\otimes B_j}\binnori\le\big\|{\textstyle\sum_j A_j\otimes B_j}\btnori , \quad
\mbox{where $\enne=\min\big\{4\emme,\emme^2\big\}$.}
\end{equation}
Thus, as stated in Corollary~\ref{corsemif}, the operator norm $\tnormi$ and the injective
norm $\innormi$ --- which can be identified with the norm of $\big(\thptj\big)^\ast$ --- are
equivalent on the linear space $\bop\altp\bopjj$; i.e., the injective tensor product
$\bop\injtp\bopjj\cong\big(\thptj\big)^\ast$ is a renorming of the Banach space
$(\bop\altp\bopjj,\tnormi)=\bophj$.
\end{remark}

By Theorem~\ref{theine}, we are now able to derive two further results characterizing the
trace class versus the \emph{cross} trace class of the Hilbert space $\hhjj$, in the
genuinely infinite-dimensional setting.

\begin{corollary} \label{unbonor}
In the case where $\dim(\hh)=\dim(\jj)=\infty$, the projective norm $\ptrnor$ --- and, hence, the
Hermitian projective norm $\ptrnortr$ --- is unbounded on the set $\hstahj$ of cross states.
\end{corollary}

\begin{proof}
Let $\{\psi_m\}_{m=1}^\infty$ be an orthonormal basis in the Hilbert space $\jj$, and, given any
$\emme\in\nat$, with $\emme\ge 2$, put $\jjm\defi\lispa\{\psi_m\}_{m=1}^\emme$. By Proposition~\ref{proimm},
the space $\trc\prtp\trcjjm=\trchjm$ can be (isometrically) identified with a linear subspace of $\thptj$.
Denoting by $\mptrnor$ the projective norm of $\trc\prtp\trcjjm$, by the assertion of Theorem~\ref{theine}
that follows relation~\eqref{estdens}, there is some density operator $D\in\trc\prtp\trcjjm\subset\thptj$
such that $\vem D\ptrn=\emme$. Since $\| D\ptrn=\vem D\ptrn=\emme$, then, by the arbitrariness of $\emme\ge 2$,
we conclude that the projective norm $\ptrnor$ --- hence, the Hermitian projective norm $\ptrnortr$ ---
is unbounded on $\hstahj$.
\end{proof}

To prove the second announced consequence of Theorem~\ref{theine} --- i.e., to show that, in the
\emph{genuinely infinite-dimensional setting}, there are states that are \emph{not} cross states
--- we first note a technical point that follows immediately from Proposition~\ref{proproj}.

\begin{lemma} \label{lemosd}
Let $\{\jl\}_{l=1}^\enne$, $\enne\in\nat$, be a set of mutually orthogonal closed subspaces of the the
Hilbert space $\jj$, and let $\pil\in\bopjj$ be the orthogonal projection operator onto the closed
subspace $\jl$. For every cross trace class operator $C\in\thptj$, $(\idpil)\tre C(\idpil)$ is a cross trace class
operator too, and, if $C=\sum_{l=1}^\enne(\idpil)\tre C(\idpil)$, then we have that
\begin{equation}
\enne^{-1}\sum_{l=1}^\enne\|C_l\ptrn\le\|C\ptrn\le\sum_{l=1}^\enne\|C_l\ptrn , \quad
\mbox{where $C_l\defi(\idpil)\tre C(\idpil)\in\thptj$.}
\end{equation}
\end{lemma}

\begin{corollary} \label{corsubne}
In the case where $\dim(\hh)=\dim(\jj)=\infty$, the following strict set inclusions hold:
$\thptj\subsetneq\trchj$ and, in particular, $\hstahj\subsetneq\stahj$.
\end{corollary}

\begin{proof}
Let us explicitly construct a suitable density operator $D\in\stahj\setminus\hstahj$. To this end,
given any orthonormal basis $\{\psi_n\}_{n=1}^\infty$ in $\jj$, we put
$\jl\defi\lispa\{\psi_{j(l)},\ldots,\psi_{k(l)}\}$, $l\in\nat$, with
\begin{equation}
j(l)\defi\sum_{n=0}^{l-1} 2^{2n}=\frac{2^{2l}-1}{3}\in\{1,5,\ldots\} \fin ,
\ k(l)\defi\sum_{n=1}^l 2^{2n}=4\tre j(l)=\frac{4\big(2^{2l}-1\big)}{3}\in\{4,20,\ldots\} \fin ,
\end{equation}
so that $\dim(\jl)=k(l)-j(l)+1=2^{2l}$ and $\jj=\oplus_{l=1}^\infty\jl$. By the assertion of
Theorem~\ref{theine} that follows relation~\eqref{estdens}, we can choose some density
operator $\Dl\in\stahjl\subset\trchjl=\trc\prtp\trcjl$ --- which, by Proposition~\ref{proimm},
can be isometrically embedded in $\thptj$ --- such that $\|\Dl\ptrn=\dim(\jl)=2^{2l}$.
Let us then select a density operator $D$ in the $\ttrnor$-closed convex hull
$\clco(\{\Dl\})\subset\stahj$; precisely, we set
\begin{equation} \label{defid}
D\defi\ttrnsum_{l=1}^\infty 2^{-l}\tre \Dl \in\stahj \fin .
\end{equation}
We will now argue that $D$ cannot belong to the set $\hstahj$ of all cross states wrt
the bipartition $\hhjj$. In fact, putting $\Den\defi\sum_{l=1}^\enne 2^{-l}\tre \Dl\in\thptj$,
for any $\enne\in\nat$, we have that
\begin{equation}
2^{-l}\tre \|\Dl\ptrn=2^l \ \ \mbox{and} \ \
\|\Den\ptrn\ge\enne^{-1}\tre{\textstyle\sum_{l=1}^\enne 2^{-l}\tre \|\Dl\ptrn}=\enne^{-1}\tre
{\textstyle\sum_{l=1}^\enne 2^l}=\big(2^{\enne+1}-2\big)/\enne \fin .
\end{equation}
Here, the inequality $\|\Den\ptrn\ge\enne^{-1}\sum_{l=1}^\enne 2^{-l}\tre \|\Dl\ptrn$ follows
directly from Lemma~\ref{lemosd} (with $C=\Den$), because, by construction,
$2^{-l}\tre \Dl=D_l=(\idpil)D(\idpil)=(\idpil)\Den(\idpil)$, where $\pil\in\bopjj$ is
the orthogonal projection operator onto the closed subspace $\jl$ of $\jj$. Hence, we have:
\begin{equation}
\|\Den\ptrn\ge\big(2^{\enne+1}-2\big)/\enne\ge 2^\enne/\enne \fin , \ \forall\cinque\enne\in\nat
\implies \lim_{\enne}\|\Den\ptrn=\infty \fin .
\end{equation}
Therefore, we must conclude that the series~\eqref{defid} cannot converge wrt the projective norm,
but this fact alone does not entail that $D\not\in\hstahj$. To prove that the density operator $D$
\emph{cannot be} a cross state, we will argue by contradiction. Indeed, first note that
$\Den=(I\otimes\pien)D(I\otimes\pien)$, where $\pien$ is the projection operator defined by
$\pien\defi\sum_{l=1}^\enne\pil$; i.e., the orthogonal projection onto $\oplus_{l=1}^\enne\jl$.
Next, suppose that $D\in\hstahj$. Then, by relation~\eqref{reproj-bis} --- with the obvious
identifications $\pi\equiv I$ and $\varpi\equiv\pien$ --- we should have that
\begin{equation}
\|D\ptrn\ge\|(I\otimes\pien)D(I\otimes\pien)\ptrn=\|\Den\ptrn\ge 2^\enne/\enne \fin , \quad
\forall\cinque\enne\in\nat \fin .
\end{equation}
Thus, we find a contradiction, and hence $D=\ttrnorli_{\enne}\Den\in\stahj\setminus\hstahj$.
\end{proof}

\begin{remark}
The state $D\in\stahj\setminus\hstahj$ constructed in the proof of Corollary~\ref{corsubne}
--- see~\eqref{defid} --- is \emph{not pure}. In Section~\ref{entanglement} --- see, in particular,
Example~\ref{xexchapu} --- we will show that, in the genuinely infinite-dimensional setting, there
also exist \emph{pure states} which are \emph{not} cross states.
\end{remark}

\subsection{The signed decomposition of a selfadjoint cross trace class operator}
\label{sigdec}

We will now introduce a suitable simple-tensor decomposition, of the general form~\eqref{gendec}, of an
element of the \emph{selfadjoint} cross trace class $\thsptjs$, where we relax the assumption that the
scalar coefficients are non-negative --- like in the standard decomposition~\eqref{stardecosa} ---
\emph{but}, instead, we assume the \emph{positivity} of the operators involved in the elementary tensor
product. Specifically, since these operators are also assumed to be normalized, we require them to be
\emph{density} operators, i.e., to belong to the set $\eshsj$.

\begin{lemma} \label{lemtrtr}
For every $S\in\trc$, $|\tr(S)|\le\|S\trn$, and $|\tr(S)|=\|S\trn$ iff, for some $c\in\toro$, $c\tre S\ge 0$;
i.e., $\tr(S)=\|S\trn$ iff $S\in\trcp$.
\end{lemma}

\begin{lemma} \label{lemposit}
If $X\in\trc$ and $Y\in\trcjj$ are such that $X\otimes Y\ge 0$ and $\|X\otimes Y\ttrn=1$, then
$X\otimes Y=\rho\otimes\sigma$, for some (uniquely determined) density operators $\rho\in\sta$
and $\sigma\in\stajj$.
\end{lemma}

\begin{theorem}[The signed decomposition] \label{thside}
The real Banach space $\thsptjs$ of selfadjoint cross trace class operators is characterized
as the set of all trace class operators on $\hhjj$ of the form
\begin{equation} \label{decosa}
C=\sum_k t_k \tre (\rho_k\otimes\sigma_k) \fin , \quad t_k\in\erre
\ ({\textstyle\sum_k |t_k|}<\infty)\fin , \ \rho_k\otimes\sigma_k\in\eshsj \fin ,
\end{equation}
where the (possibly countably infinite) sum is absolutely convergent wrt to the Hermitian projective
norm $\ptrnortr$ --- equivalently, wrt the projective norm $\ptrnor$ or the trace norm $\ttrnor$ --- and
\begin{equation}
\tr(C)=\sum_k t_k \fin , \quad \|C\ttrn\le\|C\ptrn\le\ptrntr{C}\le\sum_k |t_k| \fin .
\end{equation}
In fact, if $C=\sum_j\tre r_j \tre (X_j\otimes Y_j)$ is any Hermitian standard decomposition
of $C\in\thsptjs$, then $C$ admits a decomposition of the form~\eqref{decosa} such that
$\sum_k\tre |t_k|=\sum_j\tre r_j$ (and conversely).  Therefore, decomposition~\eqref{decosa} is
$\ptrnortr$-norming and, if $C$ is $\ptrnortr$-optimally decomposable (in particular, Hermitian-optimally
decomposable), then it admits a decomposition of the form~\eqref{decosa} satisfying $\ptrntr{C}=\sum_k |t_k|$
(in particular, $\ptrntr{C}=\|C\ptrn=\sum_k |t_k|$). Moreover, one can always assume --- still keeping the
previous properties --- that in decomposition~\eqref{decosa} the density operators $\{\rho_k\}$, $\{\sigma_k\}$ are,
specifically, pure states; i.e., $\rho_k=\pi_k\in\pusta$ and $\sigma_k=\varpi_k\in\pustajj$, for every $k$.

If $C\in\thsptjs$ is \emph{not} $\ptrnortr$-optimally
decomposable, then each expansion of $C$ of the form~\eqref{decosa} must contain both (strictly) positive
and (strictly) negative scalar coefficients.

In the case where $\emme=\min\{\dim(\hh),\dim(\jj)\}<\infty$, for every $C\in\thsptjs$, there is
a finite decomposition of the form~\eqref{decosa} consisting of, at most, $4\emme^2$ summands.

Let $\PP\in\thsptjs$. The following conditions are equivalent:

\begin{enumerate}[label=\tt{(P\arabic*)}]

\item \label{posa}
$\PP$ admits a decomposition of the form~\eqref{decosa}, that is, simultaneously, a
\emph{Hermitian standard} decomposition; i.e., if $\PP$ is
of the form
\begin{equation} \label{decpos}
\PP=\sum_k r_k \tre (\rho_k\otimes \sigma_k) \fin , \quad
r_k\ge 0 \ ({\textstyle\sum_k r_k}<\infty), \ \rho_k\otimes\sigma_k\in\eshsj \fin .
\end{equation}

\item \label{posb}
$\PP$ is positive, Hermitian-optimally decomposable and such that $\|\PP\ptrn=\|\PP\ttrn$.

\item \label{posb-bis}
$\PP$ is positive, $\ptrnortr$-optimally decomposable and such that $\ptrntr{\PP}=\|\PP\ttrn$.

\item \label{posc}
$\PP$ is positive, optimally decomposable and such that $\|\PP\ptrn=\|\PP\ttrn$.

\item \label{posd}
$\PP$ is positive, optimally decomposable and such that $\ptrntr{\PP}=\|\PP\ptrn=\|\PP\ttrn$.

\end{enumerate}
If any of the equivalent conditions~{\ref{posa}--\ref{posd}} holds, then any expansion of $\PP$
of the form~\eqref{decpos} is an optimal Hermitian standard decomposition; thus, it is, at the same time,
an optimal standard decomposition and a $\ptrnortr$-optimal (Hermitian) standard decomposition, and
\begin{equation} \label{norpode}
\tr(\PP)=\|\PP\ttrn=\|\PP\ptrn=\ptrntr{\PP}=\sum_k r_k \fin .
\end{equation}
\end{theorem}

\begin{proof}
We now prove the first part of the theorem, concerning the characterization of the selfadjoint
cross trace class $\thsptjs$ and the related facts. Let $C=\sum_j\tre r_j \tre (X_j\otimes Y_j)$
be a Hermitian standard decomposition of $C\in\thsptjs$. Let us express each of the selfadjoint
trace class operators $X_j\in\trcsa$, $Y_j\in\trcsajj$ as a (real) linear combination of two density
operators --- by Fact~\ref{prolico} --- i.e.,
\begin{equation} \label{sumdens}
X_j = \xjp\cinque\rjp\msei + \xjm\cinque\rjm \fin , \quad Y_j = \yjp\tre\sjp \msei + \yjm\tre\sjm \fin ,
\end{equation}
where $\xjp\msei=\frac{1}{2}\tr(X_j+|X_j|)\ge 0,\ldots,\yjm\msei=\frac{1}{2}\tr(Y_j-|Y_j|)\le 0$,
$\rjp,\rjm\in\sta$, $\sjp,\sjm\in\stajj$ and
\begin{align}
& \xjp\cinque\rjp=\frac{1}{2}\tre(X_j+|X_j|)\ge 0 \fin , \quad
\xjm\cinque\rjm=\frac{1}{2}\tre(X_j-|X_j|)\le 0 \fin ,
\nonumber \\
& \yjp\cinque\sjp=\frac{1}{2}\tre(Y_j+|Y_j|)\ge 0 \fin , \quad
\yjm\tre\sjm=\frac{1}{2}\tre(Y_j-|Y_j|)\le 0 \fin ;
\end{align}
hence: $\xjp=0 \ \mbox{for $X_j+|X_j|=0$}, \ldots, \yjm=0 \ \mbox{for $Y_j-|Y_j|=0$}$, and
$(\xjp\cinque\rjp)(\xjm\cinque\rjm)=0=(\yjp\tre\sjp)(\yjm\tre\sjm)$,
\begin{equation} \label{relabso}
|X_j|=\xjp\cinque\rjp\msei - \xjm\cinque\rjm \fin , \quad
|Y_j|=\yjp\tre\sjp \msei - \yjm\tre\sjm \fin .
\end{equation}
(Clearly, if, say, $\xjp=0$, then $\rjp\in\sta$ can be chosen arbitrarily, etc.)
Next, let us put
\begin{equation}
\tjpp\msei\defi\xjp\cinque\yjp\ge 0  \fin , \ \ \tjpm\msei\defi\xjp\cinque\yjm\le 0  \fin , \ \
\tjmp\msei\defi\xjm\cinque\yjp\le 0  \fin , \ \ \tjmm\msei\defi\xjm\cinque\yjm\ge 0  \fin ,
\end{equation}
so that
\begin{equation} \label{sumdens-bis}
X_j\otimes Y_j=\tjpp(\rjp\otimes\sjp)+\tjpm(\rjp\otimes\sjm)+
\tjmp(\rjm\otimes\sjm)+\tjmm(\rjm\otimes\sjm) \fin ,
\end{equation}
and --- since $\|X_j\trn=\tr(|X_j|)=\tr(\xjp\tre\rjp)-\tr(\xjm\tre\rjm)=\xjp\msei -\xjm$,
$\|Y_j\trn=\yjp\msei -\yjm$ (by~\eqref{relabso}) --- we have:
\begin{equation}
1=\|X_j\otimes Y_j\ttrn=\|X_j\trn\tre\|Y_j\trn=(\xjp\msei -\xjm)(\yjp\msei -\yjm)=
|\tjpp|+|\tjpm|+|\tjmp|+|\tjmm| \fin .
\end{equation}
By the last relation, we see that
\begin{equation}
0\le\sum_j\tre r_j=\sum_j\tre r_j \tre (|\tjpp|+|\tjpm|+|\tjmp|+|\tjmm|) \fin .
\end{equation}
Eventually, by applying the expressions~\eqref{sumdens-bis} to the Hermitian standard decomposition
of $C$, and by suitably renaming the coefficients and their indices --- i.e., introducing a new set
of coefficients $\{t_k\}\subset\erre$ (e.g., if $j=1,2,\ldots$, we can put $t_{4j-3}=r_1\tre\tjpp$,
$t_{4j-2}=r_1\tre\tjpm$, $t_{4j-1}=r_1\tre\tjmp$ and $t_{4j}=r_1\tre\tjmm$)
---  we obtain a relation of the form
\begin{equation}
C=\sum_j r_j \tre (X_j\otimes Y_j)=\sum_k t_k \tre (\rho_k\otimes \sigma_k) \fin ,
\quad t_k\in\erre \fin , \ \rho_k\in\sta \fin ,\ \sigma_k\in\stajj \fin ,
\end{equation}
where $0\le\sum_k\tre |t_k|=\sum_j\tre r_j \tre (|\tjpp|+|\tjpm|+|\tjmp|+|\tjmm|)=\sum_j\tre r_j<\infty$.
Therefore, it is proven that decomposition~\eqref{decosa}, with the specified properties, holds, and
converges absolutely wrt the Hermitian projective norm; it also converges (to the same limit) wrt the
projective norm, and wrt the trace norm, so that $\tr(C)=\sum_k t_k \tre\tr(\rho_k\otimes \sigma_k)=\sum_k t_k$.
(Conversely, from every decomposition of $C\in\thsptjs$ of the form~\eqref{decosa} --- setting,
by convention, $\sgn(0)\equiv 1$ --- we immediately get the Hermitian standard decomposition
$C=\sum_k\tre|t_k|\big((\sgn(t_k)\sei\rho_k)\otimes\sigma_k\big)$.)

Now, take any expansion of $C\in\thsptjs$ of the form~\eqref{decosa} and, exploiting the spectral
decomposition of a density operator, express each product state $\rho_k\otimes \sigma_k$ as a
convex combination of \emph{pure} states, i.e.,
\begin{equation} \label{expust}
\rho_k\otimes \sigma_k= \sum_{m,n} \pkm\tre\qkn\tre (\pikm\otimes\varpikn) \fin ,
\end{equation}
where $\pikm\in\pusta$, $\varpikn\in\pustajj$, $\pkm,\qkn>0$ and, for every $k$, $\sum_m \pkm=1=\sum_n\qkn$;
hence, the --- possibly countably infinite --- sum is absolutely convergent wrt the Hermitian projective norm
(equivalently, wrt the projective norm or the trace norm). Then, by suitably renaming the coefficients and their
indices, we get a decomposition of the form
\begin{equation} \label{decpusta}
C=\sum_l s_l \tre (\pi_l\otimes\varpi_l) \fin , \quad s_l\in\erre \fin , \ \pi_l\in\pusta \fin ,\
\varpi_l\in\pustajj \fin .
\end{equation}
Since decomposition~\eqref{decosa} is $\ptrnortr$-norming and
$\|\rho_k\otimes \sigma_k\ttrn=1=\sum_{m,n}\tre\pkm\tre\qkn\tre\|\pikm\otimes\varpikn)\ttrn$, then
decomposition~\eqref{decpusta} is $\ptrnortr$-norming too; moreover, if $D$ is $\ptrnortr$-optimally decomposable
(in particular, Hermitian-optimally decomposable), then it admits a decomposition of the form~\eqref{decpusta}
satisfying $\ptrntr{C}=\sum_l |s_l|$ (in particular, $\ptrntr{C}=\|C\ptrn=\sum_l |s_l|$).

Let $C\in\thsptjs$ be \emph{not} $\ptrnortr$-optimally decomposable (hence, in particular, $C\neq 0$).
For every decomposition of the form $C=\sum_k\tre t_k \tre (\rho_k\otimes \sigma_k)$, we have that
$0<\ptrntr{C}<\sum_k |t_k|$ (otherwise, by the first part of the proof, the second strict inequality
would be violated by some decomposition of $C$ of this form). By this fact, and by Lemma~\ref{lemtrtr},
we see that
\begin{equation}
|{\textstyle \sum_k t_k}|=|\tr(C)|\le\|C\ttrn\le\ptrntr{C} < {\textstyle \sum_k |t_k|}
\implies \{t_k\}\cap\erreps\neq\varnothing\neq\{t_k\}\cap\errems \fin .
\end{equation}

By Proposition~\ref{profoc}, if $\emme=\min\{\dim(\hh),\dim(\jj)\}<\infty$, then every $C\in\thsptjs$
admits a finite Hermitian standard decomposition consisting of, at most, $\emme^2$ summands. Hence,
arguing as above, one concludes that $C$ admits a finite decomposition of the form~\eqref{decosa}
consisting of, at most, $4\emme^2$ summands (because $X\otimes Y\in\ethtjs$ can be expressed as a
real linear combination of, at most, four density operators).

We next prove the second part of the theorem, i.e., the equivalence of properties~{\ref{posa}--\ref{posd}}.
If $\PP\in\thsptjs$ is of the form~\eqref{decpos}, then it is positive --- in fact, for every $D\in\stahj$
(in particular, for every pure state in $\pustahj$), since decomposition~\eqref{decpos} converges wrt
the trace norm, we have that $\tr(D\PP)=\sum_k\tre r_k \tre \tr(D(\rho_k\otimes \sigma_k))\ge 0$ ---
and, by relations~\eqref{inenorms},
\begin{equation}
\sum_k r_k=\tr(\PP)=\|\PP\ttrn\le\|\PP\ptrn\le\ptrntr{\PP}\le
\sum_k r_k\tre\|\rho_k\otimes \sigma_k\ttrn=\sum_k\tre r_k \fin ,
\end{equation}
so that, actually, $\ptrntr{\PP}=\|\PP\ptrn=\|\PP\ttrn=\tr(\PP)=\sum_k\tre r_k$; i.e., the
decomposition~\eqref{decpos} is both an optimal and a $\ptrnortr$-optimal Hermitian standard
decomposition of $\PP$ and $\ptrntr{\PP}=\|\PP\ptrn=\|\PP\ttrn$. Thus, condition~\ref{posa} implies
all properties~{\ref{posb}--\ref{posd}}. Moreover, condition~\ref{posb} implies~\ref{posb-bis},
because, if $\PP$ is Hermitian-optimally decomposable and $\|\PP\ptrn=\|\PP\ttrn$, then, by
point~\ref{opta} of Proposition~\ref{proptimals},  every optimal Hermitian standard decomposition
of $\PP$ is also $\ptrnortr$-optimal and $\ptrntr{\PP}=\|\PP\ptrn=\|\PP\ttrn$. Next, if
$\PP\in\thsptjs$ is $\ptrnortr$-optimally decomposable and $\ptrntr{\PP}=\|\PP\ttrn$, then ---
by~\eqref{inenorms} ($\ptrntr{\PP}\ge\|\PP\ptrn\ge\|\PP\ttrn$) --- $\ptrntr{\PP}=\|\PP\ptrn=\|\PP\ttrn$,
so that every $\ptrnortr$-optimal Hermitian standard decomposition of $\PP$ is also optimal
\emph{tout court}. Therefore, condition~\ref{posb-bis} implies~\ref{posc}. Let us finally prove
that condition~\ref{posc} implies~\ref{posa}, and hence conditions~{\ref{posa}--\ref{posd}} are
mutually equivalent (because
\ref{posa}$\implies$\ref{posb}$\implies$\ref{posb-bis}$\implies$\ref{posc}$\implies$\ref{posa}$\implies$\ref{posd}
and~\ref{posd} trivially implies~\ref{posc}).

In fact, if $\PP\in\thsptjs$ is positive and optimally decomposable, with $\PP\neq 0$
(otherwise, there is nothing to prove), then, picking any optimal standard decomposition
$\PP=\sum_k\tre r_k \tre (X_k\otimes Y_k)$, with $r_k> 0$ (for every $k$) --- i.e, a \emph{sheer}
decomposition --- we have:
\begin{align}
0<\|\PP\ttrn=\tr(\PP)
=
\sum_k r_k \sei\tr(X_k)\sei\tr(Y_k)
& =
\big|\textstyle{\sum_k r_k \sei\tr(X_k)\sei\tr(Y_k)}\big|
\nonumber\\
& \le
\sum_k r_k \sei|\tr(X_k)|\sei|\tr(Y_k)|
\nonumber\\ \label{estds}
& \le
\sum_k r_k \sei\|X_k\trn\tre\|Y_k\trn
=\sum_k r_k=\|\PP\ptrn \fin .
\end{align}
Here, the second inequality follows from the fact that $|\tr(X_k)|\le\|X_k\trn$ and
$|\tr(Y_k)|\le\|Y_k\trn$; see the first assertion of Lemma~\ref{lemtrtr}.

Now, if, moreover, we assume that $\|\PP\ptrn=\|\PP\ttrn$, then the inequalities on
the second and third lines of~\eqref{estds} must be saturated; i.e., we must have that
$\sum_k r_k \sei\tr(X_k)\sei\tr(Y_k)=\sum_k r_k$, or, equivalently,
\begin{equation}
\sum_k r_k \sei\tr(X_k\otimes Y_k)=\sum_k r_k \fin .
\end{equation}
Taking into account that $r_k> 0$ (for every $k$), this relation holds iff
$\sum_k\tre p_k \sei\tr(X_k\otimes Y_k)=1$, where
$|\tr(X_k\otimes Y_k)|=|\tr(X_k)\sei\tr(Y_k)|\le\|X_k\trn\tre\|Y_k\trn=1$
and we have put $p_k\defi r_k/(\sum_k\tre r_k)>0$ ($\sum_k\tre p_k=1$).
Since the extreme points
of the closed unit disc $\{z\in\ccc\colon |z|\le 1\}$ form the unit circle $\toro$, this last condition
is equivalent to requiring that, for every $k$, $\tr(X_k\otimes Y_k)=1=\|X_k\otimes Y_k\ttrn$; namely,
by the second assertion of Lemma~\ref{lemtrtr}, that, for every $k$, $X_k\otimes Y_k$ is a positive
selfadjoint trace class operator.

Summarizing, by the previous arguments, if $\PP\in\thsptjs$, with $\PP\neq 0$ (otherwise, there is
nothing to prove), is positive, optimally decomposable and, moreover, $\|\PP\ptrn=\|\PP\ttrn$, then,
given any sheer optimal standard decomposition $\PP=\sum_k\tre r_k \tre (X_k\otimes Y_k)$, the simple
tensor $X_k\otimes Y_k\in\ethtj$, for every $k$, must be a positive selfadjoint trace class operator;
i.e., since $\|X_k\otimes Y_k\ttrn=1$, by Lemma~\ref{lemposit} we must have that
\begin{equation}
X_k\otimes Y_k = \rho_k\otimes \sigma_k \fin ,\ \mbox{for some $\rho_k\in\sta$, $\sigma_k\in\stajj$}.
\end{equation}
Therefore, condition~\ref{posc} implies~\ref{posa}, as we wanted to prove, and hence
conditions~{\ref{posa}--\ref{posd}} are mutually equivalent.

Note that, if condition~\ref{posa}, or any of the other equivalent conditions~{\ref{posb}--\ref{posd}},
is satisfied, with $\PP\neq 0$, then every sheer optimal standard decomposition of $\PP$ is of the
form~\eqref{decosa}, provided that we do not distinguish between the various representations of the
same elementary tensor $\rho_k\otimes\sigma_k$ --- specifically, a decomposition of the form~\eqref{decpos},
with $r_k=t_k> 0$ (i.e., a decomposition of the form~\eqref{decosa}, with strictly positive coefficients)
--- and, moreover, by our previous arguments, every expansion of $\PP$ of the form~\eqref{decpos} is
an optimal Hermitian standard decomposition, so that relation~\eqref{norpode} holds. Also, using
essentially the same arguments as above, based on an estimate for the norm $\ptrnortr$ analogous
to~\eqref{estds}, one proves that, if one of the equivalent conditions~{\ref{posa}--\ref{posd}} is
verified, then every sheer $\ptrnortr$-optimal Hermitian standard decomposition of $\PP\neq 0$ is of
the form~\eqref{decpos} (with strictly positive coefficients $\{r_k\}$).
\end{proof}

\begin{definition} \label{deside}
An expansion of a selfadjoint cross trace class operator $C\in\thsptjs$ of the form~\eqref{decosa}
will be called a \emph{signed decomposition} of $C$. Specifically, if in decomposition~\eqref{decosa}
the density operators $\{\rho_k\}$, $\{\sigma_k\}$ are pure states --- i.e., for every $k$,
$\rho_k=\pi_k\in\pusta$, $\sigma_k=\varpi_k\in\pustajj$ --- then it  will be called a
\emph{pure-state signed decomposition}. If a (positive) selfadjoint operator $\PP\in\thsptjs$
admits a decomposition of the form~\eqref{decpos}, then it will be called \emph{positively decomposable},
and an expansion of the form~\eqref{decpos} will be called a \emph{positive decomposition} of $\PP$. In
the case where $C\neq 0$ (or $\PP\neq 0$), a signed decomposition (alternatively, a positive decomposition)
of $C$ (respectively, of $\PP$) is said to be \emph{sheer} if it does not contain zero summands.
The set of all positively decomposable operators in $\thsptjs$ we will denoted by $\pode$.
\end{definition}

\begin{remark} \label{reside}
By Theorem~\ref{thside}, a positive cross trace class operator $\PP\in\thsptjs$ is positively
decomposable iff it is optimally decomposable and such that $\|\PP\ptrn=\|\PP\ttrn$. From the
proof of the theorem, it is also clear that, if $\PP\in\thsptjs$, with $\PP\neq 0$, is positively
decomposable, then every sheer optimal standard decomposition --- as well as every sheer
$\ptrnortr$-optimal Hermitian standard decomposition  --- of $\PP$ is actually a signed decomposition
of the form~\eqref{decosa} (provided that we do not distinguish between different representations
of the same elementary tensor product operator), \emph{with strictly positive coefficients}, i.e.,
a (sheer) positive decomposition of the form~\eqref{decpos}; also, every signed decomposition
$\PP=\sum_k\tre t_k\tre(\rho_k\otimes \sigma_k)$ such that $\sum_k |t_k|=\|\PP\ptrn$, must be
a positive decomposition (indeed, if $\PP$ is positively decomposable, then we have that
$\|\PP\ptrn=\|\PP\ttrn\implies\|\PP\ttrn=\tr(\PP)=\sum_k t_k \le\sum_k |t_k|=
\|\PP\ptrn=\|\PP\ttrn\implies\{t_k\}\cap\errems=\varnothing$).
\end{remark}

\begin{proposition}[Characterization of the cross trace class] \label{prolinc-bis}
Every selfadjoint cross trace class operator --- wrt the bipartition $\hhjj$ --- can be expressed
as a real linear combination of (at most) two positively decomposable cross trace class operators;
hence, every cross trace class operator can be expressed as a complex linear combination of (at most)
four positively decomposable cross trace class operators. Precisely, the following relations hold:
\begin{align}
\thsptjs
& =
\pode - \pode
\nonumber \\ \label{lincom-bis}
& =
\big\{R_1 - R_2\colon R_1,R_2\in\pode\big\}
\end{align}
and
\begin{align}
\thptj
& =
\thsptjs + \ima\sei\big(\thsptjs\big)
\nonumber \\ \label{lincom-tris}
& =
\big\{(R_1 - R_2) +\ima\sei(R_3 - R_4)\colon R_1,R_2,R_3,R_4\in\pode\big\} \fin .
\end{align}
\end{proposition}

\begin{proof}
For every $C\in\thsptjs$, with $C\neq 0$ (the case where $C=0$ being trivial), we can re-express
any signed decomposition of $C$ as follows:
\begin{align}
C=\sum_k t_k \tre (\rho_k\otimes\sigma_k)
& =
\sum_{\tkg} t_k \tre (\rho_k\otimes\sigma_k) -
\sum_{\tkl} |t_k| \tre (\rho_k\otimes\sigma_k)
\nonumber \\
& =
R_1 - R_2 \fin , \quad  R_1,R_2\in\pode \fin ,
\end{align}
which proves the characterization~\eqref{lincom-bis} of the selfadjoint cross trace class.
Then, by Proposition~\ref{prolinc}, relation~\eqref{lincom-tris} follows immediately.
\end{proof}

\begin{proposition}[Characterization of cross states] \label{prdecdens}
The set $\hstahj$ of cross states wrt to the bipartition $\hhjj$ is a convex subset of $\trchj$,
that is characterized as the class of all positive trace class operators on $\hhjj$ that admit
a signed decomposition of the form
\begin{equation} \label{decodens}
D=\sum_k t_k \tre (\rho_k\otimes \sigma_k) \fin , \quad t_k\in\erres \
({\textstyle\sum_k |t_k|}<\infty) \fin , \ \rho_k\in\sta \fin ,\ \sigma_k\in\stajj \fin ,
\end{equation}
where the (possibly countably infinite) sum is absolutely convergent wrt to the Hermitian projective
norm $\ptrnortr$ --- equivalently, wrt the projective norm $\ptrnor$ or the trace norm $\ttrnor$ --- and
\begin{equation}
1=\tr(D)=\sum_k t_k \fin , \quad 1=\|D\ttrn\le\|D\ptrn\le\ptrntr{D}\le\sum_k |t_k| \fin .
\end{equation}
Decomposition~\eqref{decodens} is $\ptrnortr$-norming and, if $D$ is $\ptrnortr$-optimally
decomposable (in particular, Hermitian-optimally decomposable), then it admits a signed
decomposition of the form~\eqref{decodens} such that $\ptrntr{D}=\sum_k |t_k|$ (in particular,
$\ptrntr{D}=\|D\ptrn=\sum_k |t_k|$).

If $D\in\hstahj$ is \emph{not} $\ptrnortr$-optimally decomposable, then each expansion of $D$
of the form~\eqref{decodens} must contain both (strictly) positive and (strictly) negative
scalar coefficients.

In the case where $\emme=\min\{\dim(\hh),\dim(\jj)\}<\infty$, every $D\in\hstahj$ admits a
finite decomposition of the form~\eqref{decodens} consisting of, at most, $4\emme^2$ summands.

For every cross state $D\in\hstahj$, the following conditions are equivalent:

\begin{enumerate}[label=\tt{(D\arabic*)}]

\item \label{dosa}
$D$ admits a (sheer) positive decomposition of the form
\begin{equation} \label{dedepos}
D=\sum_k p_k \tre (\rho_k\otimes \sigma_k) \fin , \quad \mbox{$p_k>0$ $\big(\sum_k\tre p_k=1\big)$,
$\rho_k\in\sta$, $\sigma_k\in\stajj$} \fin .
\end{equation}

\item \label{dosb}
$D$ is Hermitian-optimally decomposable and $\|D\ptrn=1$.

\item \label{dosb-bis}
$D$ is $\ptrnortr$-optimally decomposable and $\ptrntr{D}=1$.

\item \label{dosc}
$D$ is optimally decomposable and $\|D\ptrn=1$.

\item \label{dosd}
$D$ is optimally decomposable and $\ptrntr{D}=\|D\ptrn=1$.

\end{enumerate}
If any of the equivalent conditions~{\ref{dosa}--\ref{dosd}} is satisfied, then the expansion~\eqref{dedepos}
is a (sheer) optimal Hermitian standard decomposition and, thus, it is both an optimal standard decomposition
and a $\ptrnortr$-optimal Hermitian standard decomposition; moreover, every sheer optimal standard decomposition,
and every sheer $\ptrnortr$-optimal Hermitian standard decomposition, of the cross state $D$ --- as well as every
sheer signed decomposition $D=\sum_k\tre t_k\tre(\rho_k\otimes \sigma_k)$ such that $\sum_k |t_k|=1$ --- is
actually a positive decomposition of the form~\eqref{dedepos}, provided that one does not distinguish between
different representations of the same elementary tensor product operator.
\end{proposition}

\begin{proof}
All claims follow immediately from Theorem~\ref{thside} and Remark~\ref{reside}.
\end{proof}

\subsection{The universal bilinear map on trace classes}
\label{ubil}

The application
\begin{equation} \label{defub}
\ub\colon\trc\times\trcjj\ni (S,T) \mapsto S\otimes T\in\trc\prtp\trcjj
\end{equation}
is a  bounded bilinear map, that we will call henceforth the \emph{universal bilinear map} on
$\trc\times\trcjj$. Note that this map shares both the same \emph{domain} (the Cartesian product
$\trc\times\trcjj$) and \emph{range} (the elementary tensor products of trace class operators)
with the natural bilinear map $\nb\colon\trc\times\trcjj\rightarrow\trchj$ introduced in
Subsection~\ref{natens} --- and, moreover, $\nb(S,T)=\ub(S,T)$ --- \emph{but} it has a different
\emph{codomain}, i.e., the Banach space $(\thptj,\ptrnor)$ versus the trace class $(\trchj,\ttrnor)$.
This fact has remarkable consequences; see the next subsection.

Clearly, $\ub\in\bothptj$ --- like the natural bilinear map $\nb$ --- is a bilinear isometry since
\begin{equation}
\|\ub(S,T)\ptrn=\|S\otimes T\ttrn=\|S\trn\tre\|T\trn \fin , \quad
\forall\cinque (S,T)\in\trc\times\trcjj \fin ,
\end{equation}
so that $\|\ub\quattro\norps=1$.

\begin{notation} \label{notensb}
In analogy with our previously introduced Notation~\ref{notensa}, given (nonempty) subsets $\subh$,
$\subj$ of $\trc$ and $\trcjj$, respectively, we put
\begin{equation}
\subh\prtp\subj\defi\clcop\big(\ub(\subh,\subj)\big)=
\clptrnor\big(\subh\altp\subj\big)\subset\trc\prtp\trcjj\subset\trchj \fin ;
\end{equation}
i.e., $\subh\prtp\subj$ is the $\ptrnor$-closed convex hull $\clcop(\ub(\subh,\subj))$ of $\ub(\subh,\subj)$.
Since $\ub(\subh,\subj)=\nb(\subh,\subj)$ is a subset of both the cross trace class $\thptj$ and the bipartite
trace class $\thotj=\trchj$, for the sake of simplicity in the following we will use the notation $\nb(\subh,\subj)$
for this set even in the case where it should be regarded as a subset of $\thptj$. E.g., we will write
$\subh\prtp\subj\defi\clcop(\nb(\subh,\subj))\equiv\clcop\big(\ub(\subh,\subj)\big)$.
\end{notation}

In the case where $\subh=\trc$ and $\subj=\trcjj$, Notation~\ref{notensb} is consistent with our definition
of the Banach space $\trc\prtp\trcjj$; see relation~\eqref{reprtpa} below. Once again, privileging the convex
structures wrt to the linear ones in Notation~\ref{notensb},  as previously done for Notation~\ref{notensa},
is motivated by the applications we have in mind.

\begin{proposition}
Let $\subh$, $\subj$ be linear subspaces of $\trc$ and $\trcjj$, respectively. Then,
we have:
\begin{equation} \label{reprtpa}
\subh\prtp\subj\defi\clcop(\nb(\subh,\subj))=\clispap(\nb(\subh,\subj))\equiv
\clispaptrn(\nb(\subh,\subj)) \defi\clptrnor\big(\lispa(\nb(\subh,\subj))\big) .
\end{equation}
Moreover, the following relations hold:
\begin{equation}
\sta\prtp\stajj\defi\clcop\big(\nb(\sta,\stajj)\big)=
\clcop\big(\nb(\pusta,\pustajj)\big)\ifed\pusta\prtp\pustajj
\end{equation}
and
\begin{equation}
\trc\prtp\trcjj=\clispap\big(\nb(\sta,\stajj)\big)=
\clispap\big(\nb(\pusta,\pustajj)\big) \fin ,
\end{equation}
\begin{equation}
\trc\prtp\trcjj=\clispap\big(\sta\prtp\stajj\big)
=\clispap\big(\pusta\prtp\pustajj\big) \fin .
\end{equation}
\end{proposition}

\subsection{The canonical linearization of bilinear maps on trace classes}
\label{canlin}

The term ``universal bilinear map'' introduced for the application~\eqref{defub} is justified by
the following:

\begin{theorem}[Universal mapping property of the cross trace class]  \label{isothm}
Let $\bs$ be a complex Banach space, and let $\Bimp\colon\trc\prtp\trcjj\rightarrow\bs$
be a bounded linear map. Then,
\begin{equation}
\bim\defi\Bimp\circ\ub\colon\trc\times\trcjj\rightarrow\bs
\end{equation}
is a bounded bilinear map and, moreover, $\|\bim\norps=\|\Bimp\norps$; i.e.,
\begin{align}
\|\bim\norps\defi\sup_{\substack{0\neq S\in\trc \\ 0\neq T\in\trcjj}}
\frac{\|\bim(S,T)\|}{\|S\trn\tre\|T\trn}
& =
\sup_{\substack{0\neq S\in\trc \\ 0\neq T\in\trcjj}}\frac{\|\Bimp(S\otimes T)\|}{\|S\trn\tre\|T\trn}
\nonumber \\
& =
\sup_{0\neq C\in\thptj}\frac{\|\Bimp(C)\|}{\|C\ptrn}\ifed\|\Bimp\norps \fin .
\end{align}

Conversely, let $\bs$ be a complex Banach space and $\bim\colon\trc\times\trcjj\rightarrow\bs$
a bounded bilinear map. Then, there exists a unique bounded linear map
$\Bimp\colon\trc\prtp\trcjj\rightarrow\bs$ such that
\begin{equation} \label{unipro-bis}
\bim (S,T)=\Bimp(S\otimes T) \fin , \quad \forall\cinque S\in\trc \fin ,\  \forall\cinque T\in\trcjj \fin .
\end{equation}
Moreover, $\|\bim\norps=\|\Bimp\norps$.

Therefore, the following isomorphism of
Banach spaces holds:
\begin{equation}
\bo\big(\trc\prtp\trcjj;\bs\big)\cong\bo(\trc,\trcjj;\bs) \fin .
\end{equation}
This isomorphism is implemented by the mapping
\begin{equation}
\bo\big(\trc\prtp\trcjj;\bs\big)\ni\Lambda\mapsto\Lambda\circ\ub\in\bo(\trc,\trcjj;\bs) \fin ,
\end{equation}
where $\ub\colon\trc\times\trcjj\rightarrow\trc\prtp\trcjj$ is the universal bilinear map on
$\trc\times\trcjj$.
\end{theorem}

\begin{proof}
Apply Theorem~{2.9} of~\cite{Ryan}, or Theorem~{7.12} of~\cite{Kubrusly}, to the projective
tensor product $\thptj$.
\end{proof}

\begin{definition} \label{defbibi}
We will call the application $\Bimp\colon\trc\prtp\trcjj\rightarrow\bs$ the \emph{canonical linearization}
of (or the \emph{canonical linear map} induced by) $\bim\colon\trc\times\trcjj\rightarrow\bs$; conversely,
we call $\bim$ the \emph{canonical bilinearization} of $\Bimp$.
\end{definition}

\begin{remark} \label{natcan}
In the case where $\bim\colon\trc\times\trcjj\rightarrow\bs$ is naturally linearizable, by
relations~\eqref{unipro} and~\eqref{unipro-bis} it is clear that the associated induced linear
maps $\Bim$ and $\Bimp$ coincide once restricted to the algebraic tensor product $\trc\altp\trcjj$,
and this restriction is precisely the unique linear map $\Bima\colon\trc\altp\trcjj\rightarrow\bs$
of Proposition~\ref{probim}.
\end{remark}

\begin{corollary} \label{coroptp}
Let $\ptph\colon\trc\rightarrow\trc$, $\ptpj\colon\trcjj\rightarrow\trcjj$ be bounded linear maps.
Then, there is a unique bounded linear map $\ptphj\colon\thptj\rightarrow\thptj$ such that
\begin{equation}
\big(\ptphj\big)(S\otimes T)=\ptph(S)\otimes\ptpj(T) \fin , \quad
\forall\cinque S\in\trc \fin , \forall\cinque T\in\trcjj \fin ,
\end{equation}
and, moreover, $\big\|\ptphj\bnorps=\|\ptph\noro\sei\|\ptpj\noro$.
\end{corollary}

\begin{corollary}
Suppose that $\hh$ and/or $\jj$ are finite-dimensional Hilbert spaces; precisely,
assume that $\emme=\min\{\dim(\hh),\dim(\jj)\}<\infty$. Then, in this case,
\begin{equation} \label{coinc}
\trc\otimes\trcjj=\trc\prtp\trcjj=\trc\altp\trcjj \fin ,
\end{equation}
a set equality (or an isomorphism of linear spaces) being understood,
and any bounded bilinear map $\bim\colon\trc\times\trcjj\rightarrow\bs$ is naturally linearizable;
moreover --- denoting by $\Bim\colon\trc\otimes\trcjj\rightarrow\bs$ the natural linearization of $\bim$
--- we have that $\Bim(A)=\Bimp(A)$, for all $A\in\trc\altp\trcjj$, and
\begin{equation} \label{compnor}
\|\bim\norps=\|\Bimp\norps\le\|\Bim\norps\le\enne\sei\|\Bimp\norps \fin , \quad
\mbox{where $\enne=\min\big\{4\emme,\emme^2\big\}$.}
\end{equation}
\end{corollary}

\begin{proof}
By Theorem~\ref{theine}, relation~\eqref{coinc} holds and, moreover, the norms $\ttrnor$ and $\ptrnor$
are equivalent on $\trc\otimes\trcjj=\thptj=\trc\altp\trcjj$. Therefore, since, according to Theorem~\ref{isothm},
there is a unique bounded (wrt the norms $\ptrnor$ and $\norbs$) linear map $\Bimp\colon\trc\prtp\trcjj\rightarrow\bs$
satisfying
\begin{equation}
\bim (S,T)=\Bimp(S\otimes T) \fin , \quad \forall\cinque S\in\trc \fin ,\  \forall\cinque T\in\trcjj
\end{equation}
--- that must coincide with the (unique) linear map $\Bima\colon\trc\altp\trcjj\rightarrow\bs$ of
Proposition~\ref{probim} --- then there is a unique linear map $\Bm\colon\trc\otimes\trcjj\rightarrow\bs$
such that
\begin{equation}
\Bm(S\otimes T)=\Bima(S\otimes T)=\Bimp(S\otimes T)=\bim(S,T) \fin , \quad \forall\cinque S\in\trc \fin ,
\  \forall\cinque T\in\trcjj \fin ,
\end{equation}
which is bounded (wrt the norms $\ttrnor$ and $\norbs$). Thus, the bilinear map
$\bim\colon\trc\times\trcjj\rightarrow\bs$ is naturally linearizable (Definition~\ref{defnatlin})
and, by Proposition~\ref{probim}, and its (unique) natural linearization $\Bim$ must be precisely
the map $\Bm$. Hence, $\Bim(A)=\Bima(A)=\Bimp(A)$, for all $A\in\trc\altp\trcjj$, and we have that
$\|\Bimp\norps=\|\bim\norps\le\|\Bim\norps$. Moreover, by relation~\eqref{estgen}, for every
bipartite trace class operator $A\in\trc\altp\trcjj=\trchj$, $A\neq 0$, we get the estimate
\begin{equation} \label{nacaine}
\frac{\|\Bimp(A)\|}{\|A\ptrn}\le\frac{\|\Bimp(A)\|}{\|A\ttrn}=\frac{\|\Bim(A)\|}{\|A\ttrn}
\le\enne\sei\frac{\|\Bimp(A)\|}{\|A\ptrn} \fin ,
\end{equation}
where $\enne=\min\big\{4\emme,\emme^2\big\}$, with $\emme=\min\{\dim(\hh),\dim(\jj)\}<\infty$.
Now, the lhs inequality in~\eqref{nacaine} is coherent with the previously obtained relation
$\|\Bimp\norps\le\|\Bim\norps$, while, by the rhs inequality, we find out that
\begin{equation}
\|\Bim\norps=\sup_{A\neq 0}\frac{\|\Bim(A)\|}{\|A\ttrn} \le\enne\sei\sup_{A\neq 0}
\frac{\|\Bimp(A)\|}{\|A\ptrn}=\enne\sei\|\Bimp\norps
\end{equation}
as well, so that relation~\eqref{compnor} holds true.
\end{proof}

In this case where $\min\{\dim(\hh),\dim(\jj)\}<\infty$, the Banach space dual $\big(\thptj\big)^\ast$
of $\thptj$ is characterized by Corollary~\ref{corsemif}; in the general case, we have the following
result:

\begin{corollary} \label{cornorc}
Denoting by $\big(\thptj\big)^\ast=\bo\big(\trc\prtp\trcjj;\ccc\big)$ the Banach space dual of
the cross trace class $\thptj$, the following isomorphism of Banach spaces holds:
\begin{equation}
\big(\thptj\big)^\ast\cong\ftrcjj\equiv\bo(\trc,\trcjj;\ccc) \fin .
\end{equation}
This isomorphism is implemented by the mapping
\begin{equation}
\big(\thptj\big)^\ast\ni\Gamma\mapsto\Gamma\circ\ub\in\ftrcjj \fin ,
\end{equation}
where $\ub\colon\trc\times\trcjj\rightarrow\trc\prtp\trcjj$ is the universal bilinear map on
$\trc\times\trcjj$.
As a consequence, since, for every cross trace class operator $C\in\thptj$, we have that
\begin{equation} \label{forptrn}
\|C\ptrn=\max\big\{|\Gamma(C)|\colon\Gamma\in\big(\thptj\big)^\ast,\ \|\Gamma\norps=1\big\},
\end{equation}
then, given any simple-tensor decomposition of $C$ --- namely, $C=\sum_k S_k\otimes T_k$,
where the sum, whenever not finite, is supposed to converge absolutely wrt the projective
norm $\ptrnor$, i.e., to be such that
$\sum_k\tre\|S_k\otimes T_k\ptrn=\sum_k\tre\|S_k\trn\tre \|T_k\trn<\infty$ ---
the following relation holds:
\begin{equation}
\|C\ptrn=\max\big\{|{\textstyle \sum_k \gamma(S_k,T_k)}|\colon\gamma\in\ftrcjj,
\ \|\gamma\norb=1\big\} \fin .
\end{equation}
\end{corollary}

\begin{proof}
The isomorphism between the Banach spaces $\big(\thptj\big)^\ast$ and $\ftrcjj$ is a direct
consequence of Theorem~\ref{isothm} (by taking $\bs=\ccc$). Moreover, for every $C\in\thptj$,
the existence of a \emph{norming functional} $\Gamma_0\in\big(\thptj\big)^\ast$ such that
$\|\Gamma_0\norps=1$ and $\|C\ptrn=\Gamma_0(C)$ is a well-known consequence of the Hahn-Banach
extension theorem (see, e.g., Corollary~{2.3} of~\cite{Fabian}). It follows that we have:
\begin{align}
\|C\ptrn
& =
\max\big\{|\Gamma(C)|\colon\Gamma\in\big(\thptj\big)^\ast,\ \|\Gamma\norps=1\big\}
\nonumber \\
& =
\max\big\{|\Gamma({\textstyle\sum_k S_k\otimes T_k})|\colon\Gamma\in\big(\thptj\big)^\ast,
\ \|\Gamma\norps=1\big\}
\nonumber \\
& =
\max\big\{|{\textstyle\sum_k \Gamma(S_k\otimes T_k)}|\colon\Gamma\in\big(\thptj\big)^\ast,
\ \|\Gamma\norps=1\big\}
\nonumber \\
& =
\max\big\{|{\textstyle \sum_k \gamma(S_k,T_k)}|\colon\gamma\in\ftrcjj,
\ \|\gamma\norb=1\big\} \fin ,
\end{align}
where $C=\sum_k S_k\otimes T_k$ is any (absolutely $\ptrnor$-convergent) simple-tensor decomposition
of the cross trace class operator $C\in\thptj$.
\end{proof}

\section{Separability, pseudo-mixtures and barycentric decompositions}
\label{sepsta}

After recalling some basic facts about the separable states in $\trchj$ --- including the
general definition of this class of bipartite states (see Definition~\ref{defsest} below)
--- we will next prove that they are cross states, and we will highlight some remarkable
connections of separability with the the projective norm $\ptrnor$ on $\trc\prtp\trcjj$
and with the Hermitian projective norm $\ptrnortr$ on $\thsptjs$. In this section, we will
always assume that $\dim(\hh),\dim(\jj)\ge 2$, because otherwise ``there is no entanglement''
and, accordingly, separability is an empty concept.

\subsection{Pseudo-mixtures of product states and discretely separable states}
\label{pseudomi}

The set $\hstahj$ of cross states wrt to the bipartition $\hhjj$ is a convex subset
(of the cross trace class $\thptj$ and) of $\trchj$, characterized, according to
Proposition~\ref{prdecdens}, as the set of all \emph{positive} trace class operators
on the Hilbert space $\hhjj$ that admit a decomposition of the form
\begin{equation} \label{decode}
D=\sum_k t_k \tre (\rho_k\otimes \sigma_k) \fin , \quad
t_k\in\erres \fin , \ \rho_k\in\sta \fin ,\ \sigma_k\in\stajj \fin ,
\end{equation}
where $1=\tr(D)=\sum_k t_k\le\sum_k |t_k|<\infty$, so that the --- possibly, countably infinite
--- sum is absolutely convergent wrt the Hermitian projective norm (equivalently, wrt the projective
norm or the trace norm). Moreover, for every $D\in\hstahj$, we have that
\begin{equation} \label{pronorm-quater}
1=\tr(D)=\|D\ttrn\le\|D\ptrn\le\ptrntr{D}=\inf\big\{{\textstyle \sum_k |t_k|}\colon
{\textstyle\sum_k t_k \tre (\rho_k\otimes\sigma_k)}=D \big\}\le 2\tre\|D\ptrn \fin ,
\end{equation}
where the infimum is taken over all possible expansions of $D$ of the form~\eqref{decode}.

Coherently with well-established literature on this topic~\cite{Vidal},
we call a density operator $D$ on $\hhjj$ of the form~\eqref{decode} a \emph{pseudo-mixture} of
the product states $\{\rho_k\otimes \sigma_k\}\subset\eshsj$; in particular, $D$ is said to be
a (statistical) \emph{mixture} of this collection of product states if it can be expressed in
the form~\eqref{decode} with $\{t_k\}\subset\erreps$ (i.e., if $\{t_k\}$ is a probability distribution).

Therefore, \emph{every cross state on $\hhjj$ can be expressed as a pseudo-mixture of a suitable collection
of product states}; in particular, as a pseudo-mixture of product states of the form $\pi\otimes\varpi$,
with $\pi\in\pusta$ and $\varpi\in\pustajj$. In fact, the following result holds:

\begin{proposition}[The pure-state signed decomposition of a cross state] \label{prdecdens-bis}
Every cross density operator $D\in\hstahj$ admits a pure-state signed decomposition of the form
\begin{equation} \label{pudeco}
D=\sum_l s_l \tre (\pi_l\otimes\varpi_l) \fin , \quad s_l\in\erres \fin , \ \pi_l\in\pusta \fin ,\
\varpi_l\in\pustajj \fin ,
\end{equation}
where the --- possibly countably infinite --- sum is absolutely convergent wrt the Hermitian projective
norm (equivalently, wrt the projective norm or the trace norm), and
\begin{equation}
1=\tr(D)=\sum_l s_l=\|D\ttrn\le\|D\ptrn\le\ptrntr{D}\le\sum_l |s_l| \fin .
\end{equation}
The decomposition~\eqref{pudeco} is $\ptrnortr$-norming and, if $D$ is $\ptrnortr$-optimally decomposable
(in particular, Hermitian-optimally decomposable), then it admits a decomposition of the form~\eqref{pudeco}
satisfying $\ptrntr{D}=\sum_l |s_l|$ (in particular, $\ptrntr{D}=\|D\ptrn=\sum_l |s_l|$).

If the cross state
$D\in\hstahj$ is \emph{not} $\ptrnortr$-optimally decomposable, then each expansion of $D$ of the
form~\eqref{pudeco} must contain both (strictly) positive and (strictly) negative scalar coefficients.

Moreover, for a cross state $D\in\hstahj$, each of conditions~{\ref{dosa}--\ref{dosd}} in
Proposition~\ref{prdecdens} is equivalent to the following:
\begin{enumerate}[label=\tt{(D\arabic*)}]
\setcounter{enumi}{5}

\item \label{dosa-bis}
$D$ is a statistical mixture of tensor products of pure states; i.e., it admits a (sheer) positive
decomposition of the form
\begin{equation} \label{dedepos-bis}
D=\sum_l p_l \tre (\pi_l\otimes\varpi_l) \fin , \quad \mbox{$p_l>0$ $\big(\sum_l\tre p_l=1\big)$,
$\pi_l\in\pusta$, $\varpi_l\in\pustajj$} \fin .
\end{equation}

\end{enumerate}
\end{proposition}

In particular, if $D\in\hstahj$ is positively decomposable (recall Definition~\ref{deside}), then the
scalar coefficients $\{t_k\}$ in the expansion~\eqref{decode} can be assumed to be \emph{positive};
i.e., $D$ can be expressed as a genuine \emph{statistical mixture} of the form
\begin{equation} \label{podecde}
D=\sum_k p_k\tre (\rho_k\otimes\sigma_k) \fin , \quad \mbox{$p_k>0$ $\big(\sum_k p_k=\tr(D)=1\big)$,
$\rho_k\in\sta$, $\sigma_k\in\stajj$,}
\end{equation}
which provides an optimal Hermitian standard decomposition of $D$. It is clear that \emph{every
trace class operator $D$ of the form}~\eqref{podecde} is a (positively decomposable) cross state
characterized, equivalently, by any of conditions~{\ref{dosb}--\ref{dosd}} in Proposition~\ref{prdecdens}
or, also, by condition~\ref{dosa-bis} in Proposition~\ref{prdecdens-bis}; moreover, every sheer optimal
Hermitian standard decomposition of $D$ --- as well as every sheer signed decomposition
$D=\sum_k\tre t_k\tre(\rho_k\otimes \sigma_k)$ such that $\sum_k |t_k|=1$ --- must be of the
form~\eqref{podecde} (with the convention that we do not distinguish between different representations
of the same elementary tensor product operator in the decomposition).

\begin{definition} \label{definse}
A state $\varsigma\in\stahj$ is said to be \emph{discretely separable}, or \emph{strongly separable},
wrt the bipartition $\hh\otimes\jj$ if it is of the form
\begin{equation} \label{discdec}
\varsigma=\sum_k p_k\tre (\rho_k\otimes\sigma_k) \fin , \quad \mbox{with $p_k>0$, $\sum_k p_k=1$,
$\rho_k\in\sta$, $\sigma_k\in\stajj$,}
\end{equation}
where the sum --- whenever not finite --- converges absolutely wrt each of the norms $\ttrnor$,
$\ptrnor$ and $\ptrnortr$; otherwise stated, if it is a positively decomposable cross state (equivalently,
a cross state satisfying any of conditions~{\ref{dosb}--\ref{dosd}} in Proposition~\ref{prdecdens}, or
condition~\ref{dosa-bis} in Proposition~\ref{prdecdens-bis}) so that it admits an optimal
decomposition in the form of a countable convex combination of product states and
$\ptrntr{\varsigma}=\|\varsigma\ptrn=\|\varsigma\ttrn=\tr(\varsigma)=1$. We will denote
by $\dsta$ the subset of $\stahj$ consisting of all such states.

A discretely separable state $\varsigma\in\dsta$ is said to be \emph{finitely separable} if it is of
the form~\eqref{discdec}, where the sum is \emph{finite}, and we will denote by $\fsta$ the subset of
$\dsta$ consisting of all such states. Otherwise stated, $\fsta\defi\co(\nb(\sta,\stajj))=\sta\altp\stajj$,
where we have used Notation~\ref{notensa}.
\end{definition}

\begin{proposition}
$\dsta$ is a convex subset of the cross trace class $\thptj$. Moreover, we have that
\begin{align}
\thsptjs
& =
\errep\sei\dsta - \errep\sei\dsta
\nonumber \\ \label{lincom-quat}
& =
\big\{r_1\tre\varsigma_1 - r_2\tre\varsigma_2\colon r_1,r_2\in\errep,\ \varsigma_1,\varsigma_2\in\dsta\big\}
\end{align}
and
\begin{align}
\thptj
& =
\thsptjs + \ima\sei\big(\thsptjs\big)
\nonumber \\ \label{lincom-quin}
& =
\errep\sei\dsta - \errep\sei\dsta +\ima\sei\big(\errep\sei\dsta - \errep\sei\dsta\big) .
\end{align}
\end{proposition}

Let us put
\begin{equation}
\ndsta\defi\hstahj\setminus\dsta \fin .
\end{equation}
This set is characterized as the class of all cross states that \emph{cannot} be expressed as a
discrete statistical mixture of product states. Otherwise stated, $\ndsta$ is the collection of
all states whose signed decompositions \emph{must} contain \emph{both} (strictly) positive and
(strictly) negative scalar coefficients. Note that, by Proposition~\ref{prdecdens-bis}, $\ndsta$
contains the set of all cross states that are \emph{not} $\ptrnortr$-optimally decomposable.

\subsection{Characterizing cross states in terms of discretely separable states}
\label{charasep}

Given a cross state $D\in\hstahj$, we can re-elaborate a signed decomposition of $D$ as follows;
we put
\begin{align}
D=\sum_k t_k \tre (\rho_k\otimes\sigma_k)
& =
\sum_{\tkg} t_k \tre (\rho_k\otimes\sigma_k) +
\sum_{\tkl} t_k \tre (\rho_k\otimes\sigma_k)
\nonumber \\ \label{signbis}
& =
\alpha\tre \dpl + (1-\alpha)\tre \dmi ,
\end{align}
--- with $\{t_k\}\subset\erre$, $\{\rho_k\}\subset\sta$, $\{\sigma_k\}\subset\stajj$ --- where
the sums, whenever not finite, converge absolutely wrt each of the norms $\ttrnor$,
$\ptrnor$ and $\ptrnortr$ (and with the obvious convention that an `empty sum' is identically zero),
and we have set
\begin{equation} \label{setalpha}
\alpha\equiv\alpha(\{t_k\})\defi\sum_{\tkg} t_k\ge\sum_{k} t_k=\tr(D)=1 \fin , \quad
1-\alpha=\sum_{\tkl} t_k \le 0 \fin ,
\end{equation}
\begin{equation} \label{setdpl}
\dpl\equiv\dpl(\{t_k\tre(\rho_k\otimes \sigma_k)\})\defi\alpha^{-1}\sum_{\tkg} t_k
\tre (\rho_k\otimes \sigma_k)\in\dsta \fin ,
\end{equation}
\begin{align}
\dsta\sqcup\{0\}\ni\dmi
& \equiv
\dmi(\{t_k\tre(\rho_k\otimes \sigma_k)\})
\nonumber \\ \label{setdmi}
& \defi
\begin{cases}
(1-\alpha)^{-1}\sum_{t_k<0}\tre t_k \tre (\rho_k\otimes \sigma_k) & \mbox{if $\alpha>1$}
\\
0 & \mbox{if $\alpha=1$}
\end{cases} \quad .
\end{align}
Note that $\dpl$ is a discretely separable state, whereas $\dmi\in\pode$ is either a
discretely separable state or, if $\alpha=1$ --- i.e., if $D=\dpl$ --- identically zero.
Moreover, the coefficient $\alpha\equiv\alpha(\{t_k\})\ge 1$, as well as the positively
decomposable cross trace class operators $\dpl\equiv\dpl(\{t_k\tre(\rho_k\otimes \sigma_k)\})$
and $\dmi\equiv\dmi(\{t_k\tre(\rho_k\otimes \sigma_k)\})$, depend on the specific signed
decomposition~\eqref{signbis} chosen for the cross state $D\in\hstahj$. Also note that
by~\eqref{setalpha} we have:
\begin{equation} \label{relalpha}
2\tre\alpha(\{t_k\})-1=\sum_{\tkg} t_k + \sum_{\tkl} (-t_k) = \sum_k |t_k|\ge 1 \fin ,
\end{equation}
where
\begin{equation} \label{alphauno}
\alpha(\{t_k\})=1 \iff D=\dpl(\{t_k\tre(\rho_k\otimes \sigma_k)\})\in\dsta \fin .
\end{equation}

\begin{remark}
If $D\in\ndsta\defi\hstahj\setminus\dsta$, then this cross state must be of the form
$D=\alpha\tre\don + (1-\alpha)\tre\dtw$, for some $\don,\dtw\in\dsta$ and
$\alpha\in(-\infty,0)\cup(1,+\infty)$; namely, $D$ must be a \emph{non-convex affine
combination} of two discretely separable states. Clearly, without loss
of generality, we can assume that $\alpha\in(1,+\infty)$. Moreover, since the (Hermitian)
projective norm of a discretely separable state is equal to $1$, we have that
\begin{equation} \label{relalpha-bis}
\ptrntr{D}\le\alpha\tre\ptrntr{D_1} + (\alpha-1)\tre\ptrntr{D_2}=2\tre\alpha-1 \fin , \quad
\alpha>1 \fin .
\end{equation}
\end{remark}

The previous arguments lead us to formulate the following result:

\begin{proposition} \label{proaffco}
A cross state $D\in\hstahj$ is either a discretely separable state --- i.e., a state the
form~\eqref{discdec} --- or else a non-convex affine combination of two discretely separable
states, i.e.,
\begin{equation}
D=\alpha\tre\don + (1-\alpha)\tre\dtw\in\ndsta , \quad \alpha>1 \fin , \
\don,\dtw\in\dsta \fin .
\end{equation}
Moreover, if $D\in\dsta$, then $\|D\ttrn=\|D\ptrn=\ptrntr{D}=1$. If, instead,
\begin{equation}
D\in\ndsta\defi\hstahj\setminus\dsta
\end{equation}
--- namely, if $D$ is a cross state on $\hhjj$ that is not discretely separable --- then
\begin{align}
1 =\|D\ttrn \le \|D\ptrn & \le \ptrntr{D}
\nonumber \\ \label{relptrntr}
& =
\inf \big\{2\tre\alpha-1>1\colon \alpha\tre\don + (1-\alpha)\tre\dtw=D, \
\don,\dtw\in\dsta\big\} \fin .
\end{align}
\end{proposition}

\begin{proof}
By our previous discussion, we only need to prove the second and the third assertions.
In fact --- as already noted in Definition~\ref{definse} --- if $D\in\dsta$, then,
by the equivalence of conditions~{\ref{dosa}--\ref{dosd}} in Proposition~\ref{prdecdens},
we have, in particular, that $\|D\ttrn=\|D\ptrn=\ptrntr{D}=1$. If, instead, $D\in\ndsta$,
then
\begin{align}
1 =\|D\ttrn \le \|D\ptrn \le & \ptrntr{D}
\nonumber\\
= &
\inf\big\{{\textstyle \sum_k |t_k|}\colon {\textstyle\sum_k
t_k \tre (\rho_k\otimes\sigma_k)}=D \big\}
\nonumber \\
= &
\inf \big\{2\alpha(\{t_k\})-1> 1 \colon
\ {\textstyle\sum_k t_k \tre (\rho_k\otimes\sigma_k)}=D\big\}
\nonumber \\
= & \inf \big\{2\alpha(\{t_k\})-1>1  \colon
\nonumber \\
& \
\alpha(\{t_k\})\tre\dpl(\{t_k\tre(\rho_k\otimes \sigma_k)\}) + (1-\alpha(\{t_k\}))
\tre\dmi(\{t_k\tre(\rho_k\otimes \sigma_k)\})=D\big\}
\nonumber\\ \label{nondis}
= &
\inf\big\{2\tre\alpha-1>1\colon \alpha\tre\don + (1-\alpha)\tre\dtw=D, \
\don,\dtw\in\dsta\big\} \fin .
\end{align}
Note that the second line of~\eqref{nondis} --- i.e., the first equality therein --- follows from
relation~\eqref{pronorm-quater} (the infimum being extended over all possible expansions of $D$
of the form~\eqref{decode}) and the second equality from relation~\eqref{relalpha} (together
with~\eqref{alphauno}); next, the third equality follows from decomposition~\eqref{signbis}
(together with definitions~\eqref{setalpha}--\eqref{setdmi}) and the last one from~\eqref{relalpha-bis}
(where $D=\alpha\tre\don + (1-\alpha)\tre\dtw$, with $\don,\dtw\in\dsta$).
\end{proof}

\begin{remark}
Clearly, if $D=\alpha\tre\don + (1-\alpha)\tre\dtw$ --- with $D,\don,\dtw\in\dsta$ and $\alpha\ge 1$
--- then $1=\ptrntr{D}\le\alpha\tre\ptrntr{D_1}+(1-\alpha)\tre\ptrntr{D_2}=2\tre\alpha-1$, where
the inequality is saturated when $\alpha=1$. By this fact and by Proposition~\ref{proaffco} --- in
particular, observing that if $D\in\ndsta$ and $D=\alpha\tre\don + (1-\alpha)\tre\dtw$,
with $\don,\dtw\in\dsta$, then $\alpha\not\in[0,1]$ (actually, we can assume, without loss of generality,
that $\alpha\in(1,+\infty)$) --- for \emph{every} cross state $D\in\hstahj$ we have:
\begin{equation} \label{nondis-bis}
\ptrntr{D}=\inf\big\{2\tre\alpha-1\ge 1\colon \alpha\tre\don + (1-\alpha)\tre\dtw=D, \
\don,\dtw\in\dsta\big\} \fin .
\end{equation}
We stress that, for $D\in\dsta$, the infimum on the rhs of~\eqref{nondis-bis} is attained as a \emph{minimum}
(i.e., $2\tre\alpha-1=1$), whereas, for $D\in\ndsta\defi\hstahj\setminus\dsta$ such that $\ptrntr{D}=1$, it cannot be
(since $2\tre\alpha-1\in(1,+\infty)$). We will show that the latter class of cross states consists precisely
of those states on $\hhjj$ that are separable --- according to the general definition of separability that will
be recalled below --- \emph{but not} discretely separable.
\end{remark}

\subsection{Separable bipartite states: the general case}
\label{sepgen}

The definition of \emph{discrete separability} introduced in the previous subsection coincides with the
notion of separability \emph{tout court} considered by some authors (see, e.g., Sect.~{9.5} of~\cite{Busch}),
but is a particular case of a more general notion of separability we will now analyze.

We know that every \emph{discretely separable state} --- i.e., every trace class operator $D$ on $\hhjj$ of the
form~\eqref{podecde} --- is a cross state. It is then natural to wonder whether the \emph{most general} separable
state is a cross state, as well.

We start by providing a formal definition of \emph{separability} of bipartite states, a notion that seems to
have been first established in its full generality (and on the physical ground) by Werner~\cite{Werner}.

\begin{definition} \label{defsest}
The set $\sesta\subset\stahj$ of all \emph{separable states} wrt the bipartition $\hh\otimes\jj$ is
defined as
\begin{equation}
\sesta\defi\clttrnor\big\{\varsigma\in\trc\altp\trcjj\colon \mbox{$\varsigma=
\sum_j p_j\tre(\rho_j\otimes\sigma_j)$, $\rho_j\in\sta$, $\sigma_j\in\stajj$}\big\}.
\end{equation}
Here, the decomposition of the density operator $\varsigma\in\trc\altp\trcjj$ ranges over all (finite)
convex combinations of \emph{product states} of the form $\rho\otimes\sigma$, with $\rho\in\sta$
and $\sigma\in\stajj$. Otherwise stated, $\sesta$ coincides with the $\ttrnor$-closed convex hull
of the set of all product states wrt the bipartition $\hh\otimes\jj$; i.e., recalling Notation~\ref{notensa},
\begin{equation}
\sesta=\clco(\nb(\sta,\stajj))\ifed\sta\otimes\stajj \fin.
\end{equation}

A state $\tau\in\stahj\setminus\sesta$ is said to be \emph{entangled}.
\end{definition}

It is quite natural to consider, in this context, an analogous notion of separability where the role
of the trace norm $\trnor$ is played, instead, by the projective norm $\ptrnor$.

\begin{definition} \label{defsest-bis}
The set $\csta\subset\thptj$ of \emph{cross separable states} wrt the bipartition
$\hh\otimes\jj$ is defined as
\begin{equation} \label{defcsta}
\csta\defi\clptrnor\big\{\varsigma\in\trc\altp\trcjj\colon \mbox{$\varsigma=
\sum_j p_j\tre(\rho_j\otimes\sigma_j)$, $\rho_j\in\sta$, $\sigma_j\in\stajj$}\big\}.
\end{equation}
Otherwise stated, $\csta$ is the $\ptrnor$-closed convex hull $\clcop(\nb(\sta,\stajj))$ of the
set of all product states wrt the bipartition $\hh\otimes\jj$; i.e., recalling Notation~\ref{notensb},
\begin{equation}
\csta=\sta\prtp\stajj \fin .
\end{equation}
\end{definition}

\begin{proposition} \label{procsta}
$\csta$ is a closed convex subset of the $\ptrnor$-closed convex set $\hstahj$ of all cross states
wrt the bipartition $\hhjj$, and
\begin{equation} \label{relcdsta}
\dsta\subset\clptrnor(\dsta)=\csta\subset\sesta \fin .
\end{equation}
In particular, in the case where at least one of the Hilbert spaces of the bipartition $\hhjj$ is
finite-dimensional, $\csta=\sesta$.
\end{proposition}

\begin{remark} \label{receq}
Recall that the norms $\ptrnor$ and $\ptrnortr$ are equivalent on $\thsptjs$. Therefore, in
the definition of the set $\csta\subset\hstahj\subset\thsptjs$ --- see~\eqref{defcsta} --- the
\emph{closure} can actually be taken wrt either of these norms.
\end{remark}

\begin{proposition} \label{pronose}
For every cross separable state $\varsigma\in\csta$, $\|\varsigma\ptrn=\ptrntr{\varsigma}=1$.
\end{proposition}

Recalling relation~\eqref{reldepu}, we have:
\begin{equation} \label{propust}
\sesta=\clco(\nb(\sta,\stajj))=\clco(\nb(\pusta,\pustajj))\ifed\pusta\otimes\pustajj \fin .
\end{equation}
We will denote by $\sestap$ the set of all \emph{separable pure states} on $\hh\otimes\jj$, i.e.,
\begin{equation}
\sestap\defi\pustahj\cap\sesta \fin .
\end{equation}

\subsection{Separable states and barycentric decompositions}
\label{barycentric}

The precise relation between the convex sets $\sesta$ and $\csta$ is not immediately clear
in the genuinely infinite-dimensional setting (i.e., for $\dim(\hh)=\dim(\jj)=\infty$); we
only know that $\csta\subset\sesta$. To further investigate this relation --- specifically,
to check whether $\csta=\sesta$ in the most general case --- we first need to establish some
preliminary results. To this end, we will use, in particular, basic facts regarding the integration
of vector-valued functions; see, e.g., the standard references~\cite{Hytonen,Diestel}. In fact,
we will consider \emph{Bochner integrals} of functions that take values in a \emph{separable}
Banach space (i.e., $\trchj$ or $\trc\prtp\trcjj$). By a well-known result --- the so-called
\emph{Pettis Measurability Theorem}~\cite{Hytonen} --- such a function is \emph{strongly measurable}
iff it is \emph{weakly measurable}; to indicate this type of vector-valued functions, we will therefore
synthetically use the term \emph{measurable}. It is, moreover, a standard technical fact that, whenever
integration wrt a \emph{finite} measure (in particular, a probability measure) is involved, every
\emph{norm-bounded} measurable function is \emph{Bochner-integrable}; see, e.g., Proposition~{1.2.2}
of~\cite{Hytonen}.

\begin{lemma} \label{lemtop}
Let us endow Cartesian product $\sta\times\stajj$ with the product topology --- $\sta$ and $\stajj$
being endowed with their standard topology, see Fact~\ref{statop} --- and the sets $\nb(\sta,\stajj)$,
$\ub(\sta,\stajj)$ with their subspace topology (wrt the natural norm topology of the Banach spaces
$\trc\otimes\trcjj=\trchj$ and $\trc\prtp\trcjj$, respectively). Then, these sets are mutually homeomorphic,
because the bijective maps
\begin{equation} \label{mapa}
\sta\times\stajj\ni (\rho,\sigma)\mapsto\rho\otimes\sigma\in\nb(\sta,\stajj)
\end{equation}
and
\begin{equation} \label{mapb}
\nb(\sta,\stajj)\ni \rho\otimes\sigma\mapsto\rho\otimes\sigma\in\ub(\sta,\stajj)
\end{equation}
are homeomorphisms. Moreover, $\nb(\sta,\stajj)$ and $\nb(\pusta,\pustajj)$ are a norm-closed
subsets of $\trchj$.
\end{lemma}

\begin{proof}
By Fact~\ref{factdenso}, the maps~\eqref{mapa} and~\eqref{mapb} are bijective. Let us then prove
that, in particular, they are homeomorphisms. First note that the Cartesian product $\sta\times\stajj$,
endowed with the product topology, is a second countable space, because $\sta$, $\stajj$ are endowed with
their (second countable) standard topology, that coincides with  the subspace topology wrt the separable
Banach spaces $\trc$ and $\trcjj$, respectively. Similarly, $\nb(\sta,\stajj)$ and $\ub(\sta,\stajj)$ ---
once endowed with the subspace topology wrt the separable Banach spaces $\trchj$ and $\thptj$, respectively
--- are second countable spaces too. Therefore, \emph{a fortiori} these spaces satisfy the first axiom of
countability, so that in the following we can argue in terms of \emph{sequences} (a map between topological
spaces, with first countable domain, is continuous iff it is sequentially continuous). Now, a sequence
$\{(\rho_n,\sigma_n)\}\subset\sta\times\stajj$ converges to some $(\rho,\sigma)\in\sta\times\stajj$
iff $\lim_n\|\rho-\rho_n\trn=0$ and $\lim_n \|\sigma-\sigma_n\trn=0$. Moreover, we have that
\begin{align}
\|\rho\otimes\sigma-\rho_n\otimes\sigma_n\ttrn
& \le
\|(\rho-\rho_n)\otimes\sigma\ttrn+\|\rho_n\otimes(\sigma-\sigma_n)\ttrn
\nonumber \\
& =
\|\rho-\rho_n\trn\tre\|\sigma\trn+\|\rho_n\trn\tre\|\sigma-\sigma_n\trn
=\|\rho-\rho_n\trn+\|\sigma-\sigma_n\trn \fin .
\end{align}
Therefore, the bijective map~\eqref{mapa} is continuous and, by the boundedness of the partial traces
$\trjj\colon\trc\otimes\trcjj\rightarrow\trc$ and $\trhh\colon\trc\otimes\trcjj\rightarrow\trcjj$,
its inverse --- i.e., the mapping
\begin{equation} \label{mapc}
\nb(\sta,\stajj)\ni\rho\otimes\sigma\mapsto\big(\rho=\trjj(\rho\otimes\sigma),
\sigma=\trhh(\rho\otimes\sigma)\big)\in\sta\times\stajj
\end{equation}
--- is continuous too (the two coordinate maps being continuous). Let us next to show that the
bijection~\eqref{mapb} is a homeomorphism, as well. To this end, note that --- since the trace
norm $\ttrnor$ is majorized by the projective norm $\ptrnor$ --- for every sequence
$\{(\rho_n,\sigma_n)\}\subset\sta\times\stajj$, we have:
\begin{align}
\|\rho\otimes\sigma-\rho_n\otimes\sigma_n\ttrn
& \le
\|\rho\otimes\sigma-\rho_n\otimes\sigma_n\ptrn
\nonumber \\
& \le
\|(\rho-\rho_n)\otimes\sigma\ptrn+\|\rho_n\otimes(\sigma-\sigma_n)\ptrn
\nonumber \\ \label{estrs}
& =
\|\rho-\rho_n\trn\tre\|\sigma\trn+\|\rho_n\trn\tre\|\sigma-\sigma_n\trn
=\|\rho-\rho_n\trn+\|\sigma-\sigma_n\trn \fin .
\end{align}
Hence, $\lim_n\|\rho\otimes\sigma-\rho_n\otimes\sigma_n\ttrn=0\iff\lim_n
\|\rho\otimes\sigma-\rho_n\otimes\sigma_n\ptrn=0$, where for the direct implication
``$\Longrightarrow$'' we have also exploited the continuity of the map~\eqref{mapc}
(otherwise stated, by the second inequality in~\eqref{estrs} the map
$\sta\times\stajj\ni (\rho,\sigma)\mapsto\rho\otimes\sigma\in\ub(\sta,\stajj)$
is continuous, and, composed with the continuous map~\eqref{mapc}, gives the map~\eqref{mapb},
which is then continuous too).

Let us now finally prove that both the sets $\nb(\sta,\stajj)$ and $\nb(\pusta,\pustajj)$
are closed in $\trchj$. In fact, given any convergent sequence of \emph{product states}
$\{\rho_n\otimes\sigma_n\}\subset\trchj$ --- $\ttrnorli_n\rho_n\otimes\sigma_n=D\in\trchj$;
but note that, actually, $D\in\stahj$, because $\stahj$ is a closed subset of $\trchj$ ---
by the continuity of the partial trace we have that
\begin{equation}
\rho\equiv\trjj(D)=\trjj\Big(\ttrnorli_n\rho_n\otimes\sigma_n\Big)=
\trnorli_n\trjj(\rho_n\otimes\sigma_n)=\trnorli_n\rho_n\in\sta
\end{equation}
and, analogously, $\sigma\equiv\trhh(D)=\trnorli_n\sigma_n\in\stajj$. It follows that the sequence
$\{(\rho_n,\sigma_n)\}\subset\sta\times\stajj$ converges to $(\rho,\sigma)\in\sta\times\stajj$;
hence, by the continuity of the map~\eqref{mapa}, we must have that
$\ttrnorli_n\rho_n\otimes\sigma_n=\rho\otimes\sigma\in\nb(\sta,\stajj)$.
Thus, $\nb(\sta,\stajj)$ is a closed subset of $\trchj$.
Moreover, by the continuity
of the product of operators in $\trchj$ (the trace class is a Banach algebra), we see that
\begin{equation}
\ttrnorli_n\rho_n\otimes\sigma_n=\rho\otimes\sigma
\implies \ttrnorli_n (\rho_n\otimes\sigma_n)^2=(\rho\otimes\sigma)^2=
\rho^2\otimes\sigma^2 \fin .
\end{equation}
Therefore, if, in particular, $\{\rho_n\otimes\sigma_n\}\subset\nb(\pusta,\pustajj)$, then
\begin{equation}
\rho\otimes\sigma=\ttrnorli_n\rho_n\otimes\sigma_n=\ttrnorli_n \sei(\rho_n\otimes\sigma_n)^2
=(\rho\otimes\sigma)^2 \fin ,
\end{equation}
so that $\rho\otimes\sigma\in\nb(\pusta,\pustajj)$. Hence, $\nb(\pusta,\pustajj)$ is a closed
subset of $\trchj$ too.
\end{proof}

\begin{remark} \label{rempol}
Endowed with with their relative topology wrt $\trchj$, both $\nb(\sta,\stajj)$ and $\nb(\pusta,\pustajj)$
are Polish spaces, being closed subsets of the Polish space $\trchj$. Moreover, by Lemma~\ref{lemtop},
their Polish space topology coincides with the topology induced by the projective norm $\ptrnor$.
Otherwise stated, we can identify, also topologically, the spaces $\nb(\sta,\stajj)$ and $\ub(\sta,\stajj)$,
as well as the spaces $\nb(\pusta,\pustajj)$ and $\ub(\pusta,\pustajj)$. Therefore, the spaces $\nb(\sta,\stajj)$
and $\ub(\sta,\stajj)$ possess the same \emph{natural Borel structure} (the smallest Borel structure containing
all the open sets of their common topology), wrt which they are (isomorphic) \emph{standard Borel spaces}; see
Subsection~\ref{measures}, and references therein. Analogous facts hold for the spaces $\nb(\pusta,\pustajj)$
and $\ub(\pusta,\pustajj)$ too.
\end{remark}

\begin{lemma} \label{lembar}
Let $\mathcal{C}$ be a closed subset of $\stahj$. Then, denoting by $\prm(\mathcal{C})$ the set of
Borel probability measures on $\mathcal{C}$, we have (Bochner integrals being understood):
\begin{equation}
\clco(\mathcal{C})=\big\{{\textstyle \int_\mathcal{C} \xi\otto\de\mu(\xi)}\colon\mu\in
\prm(\mathcal{C})\big\};
\end{equation}
i.e., the $\ttrnor$-closed convex hull of the set $\mathcal{C}$ coincides with the set of
\emph{barycenters} of all Borel probability measures on $\mathcal{C}$.
\end{lemma}

\begin{proof}
This is a well-known result; see Lemma~{1} of~\cite{HoShiWe} (also see~\cite{HoShiWe-bis,Holevo-EBC}).
\end{proof}

\begin{lemma} \label{lemprem}
Let $(\mesp,\mu)$ be a Borel probability measure space, and let $f\colon\mesp\rightarrow\trchj$
be a measurable function such that $f(\mesp)\subset\nb(\sta,\stajj)$. Then, $f$ is Bochner-integrable
and
\begin{equation} \label{intego}
\int_\mesp f(x)\dodici\de\mu(x)\subset\clco(f(\mesp))\subset\clco(\nb(\sta,\stajj))=\sesta \fin ,
\end{equation}
where, on the lhs, a Bochner integral of $\trchj$-valued functions is understood. Analogously, let
$g\colon\mesp\rightarrow\thptj$ be a measurable function such that $g(\mesp)\subset\nb(\sta,\stajj)$;
then, $g$ is Bochner-integrable and
\begin{equation}  \label{integp}
\int_\mesp g(x)\dodici\de\mu(x)\subset\clcop(g(\mesp))\subset\clcop(\nb(\sta,\stajj))=\csta \fin ,
\end{equation}
where, on the lhs, a Bochner integral of $\thptj$-valued functions is understood. Moreover, if
$f(x)=g(x)$, for $\mu$-almost all $x\in\mesp$, then the integrals in~\eqref{intego} and~\eqref{integp}
are equal; i.e., they converge to the same density operator on $\hhjj$.
\end{lemma}

\begin{proof}
Let us prove the first assertion. Since here integration wrt a finite measure $\mu$ is involved, and the
vector-valued function $f$ is norm-bounded, then $f$ is Bochner-integrable. Moreover, by a well-known
result --- see~\cite{Diestel}, Corollary~{8} in Section~{2} of Chapter~{2}, or Proposition~{1.2.12}
of~\cite{Hytonen} --- the Bochner integral $\int_\mesp f(x)\dodici\de\mu(x)$ is contained in the
$\ttrnor$-closed convex hull $\clco(f(\mesp))$. The proof of the second assertion is analogous.
Let us prove the final assertion. To this end, just recall that the linear immersion map
$\imme\colon\thptj\rightarrow\trchj$ is bounded (Proposition~\ref{prvarlin}), and therefore,
since $f(x)=\imme\tre\big(g(x)\big)=\big(\imme\circ g\big)(x)$, for $\mu$-almost all $x\in\mesp$,
by a classical result of Hille (see Theorem~{1.2.4} of~\cite{Hytonen}), we have that
$\int_\mesp f(x)\dodici\de\mu(x)=\int_\mesp\big(\imme\circ g\big)(x)\dodici\de\mu(x)=
\imme\tre\big(\int_\mesp g(x)\dodici\de\mu(x)\big)$; i.e.\ --- omitting, as usual, the immersion map
$\imme$ --- the integrals in~\eqref{intego} and~\eqref{integp} are equal.
\end{proof}

\begin{theorem} \label{thechase}
A state $\varsigma\in\stahj$ is separable --- i.e., $\varsigma\in\sesta=\sta\otimes\stajj$ --- iff there
is a Borel probability measure $\mu$ on $\nb(\sta,\stajj)=
\big\{\rho\otimes\sigma\colon\rho\in\sta,\ \sigma\in\stajj\big\}$ such that
\begin{equation} \label{bardeca}
\varsigma=\int_{\nb(\sta,\stajj)}\mdodici\rho\otimes\sigma\dodici\de\mu(\rho\otimes\sigma) \fin ;
\end{equation}
also, iff there is a Borel probability measure $\nu$ on $\nb(\pusta,\pustajj)=
\big\{\pi\otimes\varpi\colon\pi\in\pusta,\ \varpi\in\pustajj\big\}$ such that
\begin{equation} \label{bardecb}
\varsigma=\int_{\nb(\pusta,\pustajj)}\mdodici\pi\otimes\varpi\dodici\de\nu(\pi\otimes\varpi) \fin .
\end{equation}
In~\eqref{bardeca} and~\eqref{bardecb}, a Bochner integral of $\trchj$-valued functions --- or,
equivalently (i.e., with the integral converging to the same state $\varsigma$), a Bochner integral
of $\trc\prtp\trcjj$-valued functions --- is understood.
\end{theorem}

\begin{proof}
By the final assertion of Lemma~\ref{lemtop}, $\nb(\sta,\stajj)$ is a closed subset of $\trchj$.
Therefore, we can apply Lemma~\ref{lembar}, with $\mathcal{C}=\nb(\sta,\stajj)$, and conclude that
if a state $\varsigma\in\stahj$ is separable --- i.e., if $\varsigma\in\sesta=\clco(\nb(\sta,\stajj))$
--- then it is of the form~\eqref{bardeca}, for some Borel probability measure $\mu$ on $\nb(\sta,\stajj)$,
where the integral is a Bochner integral of $\trchj$-valued functions. Conversely, if $\varsigma\in\stahj$
is of the form~\eqref{bardeca} --- where a Bochner integral of $\trchj$-valued functions is understood ---
then, by the first assertion of Lemma~\ref{lemprem}, it is contained in $\ttrnor$-closed convex hull
$\clco(\nb(\sta,\stajj))=\sesta$. Let us prove that the integral on the rhs of~\eqref{bardeca} can be
regarded as a Bochner integral of $\trc\prtp\trcjj$-valued functions, as well. Indeed, by the first
assertion of Lemma~\ref{lemtop}, the map
\begin{equation} \label{mapd}
\nb(\sta,\stajj)\ni \rho\otimes\sigma\mapsto\rho\otimes\sigma\in\trc\prtp\trcjj
\end{equation}
is continuous --- hence, \emph{a fortiori}, (strongly) $\mu$-measurable --- and
\begin{equation}
\int_{\nb(\sta,\stajj)}\mdodici\|\rho\otimes\sigma\ptrn\dodici\de\mu(\rho\otimes\sigma)
=\int_{\nb(\sta,\stajj)}\mdodici\|\rho\otimes\sigma\ttrn\dodici\de\mu(\rho\otimes\sigma)=1<\infty \fin .
\end{equation}
Hence, it is Bochner-integrable. Moreover, by the final assertion Lemma~\ref{lemprem}, the Bochner integral
in~\eqref{bardeca} converges --- wrt both the norms $\ttrnor$ and $\ptrnor$ --- to a \emph{unique} state
$\varsigma$.

Analogously, by Lemma~\ref{lemtop}, $\sestap=\nb(\pusta,\pustajj)$ is a closed subset of $\trchj$,
and, restricting the domain of the map~\eqref{mapd}, we see that the mapping
\begin{equation}
\nb(\pusta,\pustajj)\ni\pi\otimes\varpi\mapsto\pi\otimes\varpi\in\trc\prtp\trcjj
\end{equation}
is continuous too. Then arguing as above, we conclude that
\begin{equation}
\varsigma\in\sesta=\clco(\nb(\sta,\stajj))=\clco(\nb(\pusta,\pustajj))
\end{equation}
--- where we have used relation~\eqref{propust} --- iff it is of the form~\eqref{bardecb}, where,
once again, the integral can be regarded as a Bochner integral of $\trchj$-valued functions or,
equivalently, as a Bochner integral of $\trc\prtp\trcjj$-valued functions. Once again, by
Lemma~\ref{lemprem}, the Bochner integral in~\eqref{bardecb} converges --- wrt both the norms
$\ttrnor$ and $\ptrnor$ --- to a \emph{unique} state $\varsigma$.
\end{proof}

We will call the general expression~\eqref{bardeca} --- or~\eqref{bardecb} --- of a separable state
$\varsigma\in\sesta$ a \emph{barycentric decomposition} of $\varsigma$. It is, however, worth observing
that, by Theorem~\ref{thechase}, a separable state also admits a decomposition that may be regarded as
a direct generalization of the (positive) decomposition of a discretely separable state.

\begin{corollary} \label{coronint}
There exists a Borel isomorphism
\begin{equation}
[0,1]\ni x\mapsto\big(\rho\otimes\sigma\big)(x)\in\nb(\sta,\stajj) \fin .
\end{equation}
Given any such a (measurable) map, a state $\varsigma\in\stahj$ is separable iff
it is of the form
\begin{equation} \label{bardeca-bis}
\varsigma=\int_0^1 \big(\rho\otimes\sigma\big)(x)\dodici\de\zeta(x) \fin ,
\end{equation}
for some Borel probability  measure $\zeta$ on the interval $[0,1]$. Here, a Bochner integral of
$\trchj$-valued functions --- or, equivalently, a Bochner integral of $\trc\prtp\trcjj$-valued
functions --- is understood.

Analogously, and with the same twofold possible interpretation of the Bochner integral,
there is a Borel isomorphism
\begin{equation}
[0,1]\ni x\mapsto\big(\pi\otimes\varpi\big)(x)\in\nb(\pusta,\pustajj) \fin ,
\end{equation}
and, given any such a map, a state $\varsigma\in\stahj$ is separable iff
it is of the form
\begin{equation} \label{bardecb-bis}
\varsigma=\int_0^1\big(\pi\otimes\varpi\big)(x)\dodici\de\xi(x) \fin .
\end{equation}
for some a Borel probability measure $\xi$ on $[0,1]$.
\end{corollary}

\begin{proof}
Since $\nb(\sta,\stajj)$ and $\nb(\pusta,\pustajj)$ are Polish spaces (Remark~\ref{rempol}),
and they both have the power of the continuum (for $\dim(\hhjj)\ge 2$), they are
Borel isomorphic to the unit interval $[0,1]$ (Fact~\ref{faposp}). Now, let
$[0,1]\ni x\mapsto\varphi(x)\equiv\big(\rho\otimes\sigma\big)(x)\in\nb(\sta,\stajj)$
be any Borel isomorphism. For every Borel probability measure $\zeta$ on $[0,1]$, by the
change-of-variables formula (see, e.g., Proposition~{1.2.6} of~\cite{Hytonen}) we have that
\begin{equation}
\int_0^1 \big(\rho\otimes\sigma\big)(x)\dodici\de\zeta(x)
=\int_{\nb(\sta,\stajj)}\mdodici\rho\otimes\sigma\dodici\de\mu(\rho\otimes\sigma)
\end{equation}
where $\mu=\varphi_{\ast}\tre\zeta$; i.e., $\mu$ is the push-forward of the measure $\zeta$:
$\mu(\bose)\defi\zeta(\varphi^{-1}(\bose))$, for every Borel subset $\bose$
of $\nb(\sta,\stajj)$. Note that, since $\varphi$ is Borel isomorphism, the push-forward
$\varphi_{\ast}$ establishes a bijection between the set of all Borel probability measures on
the interval $[0,1]$ and the set of all Borel probability measures on $\nb(\sta,\stajj)$.
Then, by Theorem~\ref{thechase}, the first assertion follows, and the proof of the second
assertion is analogous.
\end{proof}

\begin{example}  \label{exnonds}
Let us suppose that $\hh=\jj=\elleds\equiv\elledsc$, where $\lem$ is the Lebesgue measure on the
interval $[0,1]$ (i.e., $\de\lem(x)=\de x$), and let $\big\{\phi_k=\eik\colon k\in\itg\big\}$ be the
standard trigonometric basis in $\elleds$, so that the Fourier coefficients $\{c_k(\psi)\}_{k\in\itg}$
of any $\psi\in\elleds$ are given by
$c_k(\psi)\defi\langle\phi_k,\psi\rangle=\int_0^1 \eikx\sei \psi(x) \dodici \de x$, $k\in\itg$.
Now, denoting by $\toro$ the circle group, the mapping
\begin{equation}
\toro\ni z\equiv\eit\mapsto V(z)\defi\sum_{k\in\itg} z^{-k}\sei\hphk
\end{equation}
--- where $\hphk\equiv\phk$, and $V(z)$ is a unitary operator on $\elleds$ (here defined via
its spectral decomposition) --- is a strongly continuous unitary representation of the compact
abelian group $\toro$. Precisely, $\toro\ni z\mapsto V(z)$ is the \emph{regular representation}~\cite{Folland-AA}
of $\toro$ --- the orthogonal sum of exactly one copy of each unitary character of $\toro$ ---
and, in fact, we have that $\big(V\big(\eit\big)\psi\big)(x)=\psi([x-\vartheta])$, where
$\vartheta\in[0,1)$, $[x-\vartheta]\in[0,1]$ and $[x-\vartheta]\equiv x-\vartheta\pmod{1}$. Consider the
associated (diagonal or) \emph{inner tensor product representation}~\cite{Folland-AA} $V\otimes V$ of $\toro$ in
$\hhjj=\elleds\otimes\elleds$. Given any pair of pure states $\pi=\hpu\equiv\hhpu\in\pusta$
and $\varpi=\hpd\in\pustajj$, we can define a separable state $\varsigma\in\sesta$ by setting
\begin{equation} \label{nonds}
\varsigma\defi\int_0^1 \Big(\act\big(\eit\big)(\pi\otimes\varpi)\Big) \dieci \de\vartheta \fin ,
\end{equation}
where, for every trace class operator $A\in\trchj$,
\begin{equation}
\act\big(\eit\big)\sei A \defi
\big(V\otimes V\big)\big(\eit\big)^\ast A \otto\big(V\otimes V\big)\big(\eit\big)\in\trchj \fin ;
\end{equation}
thus, in particular, $\act\big(\eit\big)(\pi\otimes\varpi)\in\nb(\pusta,\pustajj)$.
In formula~\eqref{nonds}, a Bochner integral of $\trchj$-valued functions (equivalently, by the final
assertion of Lemma~\ref{lemprem}, a Bochner integral of $\trc\prtp\trcjj$-valued functions) is understood.
Note that --- by Proposition~{4.1} of~\cite{Aniello_BM} --- the map $\toro\ni z\mapsto\act(z)$ is a
\emph{strongly continuous} isometric representation of the circle group $\toro$ in the Banach space $\trchj$.
As a consequence, the mapping (which depends on both the representation $V$ and the pure state $\pi\otimes\varpi$)
\begin{equation}
[0,1]\ni x\mapsto\varphi(x)\equiv\Big(\act\big(\eix\big)(\pi\otimes\varpi)\Big)\in\nb(\pusta,\pustajj)
\end{equation}
--- where $\nb(\pusta,\pustajj)\subset\trchj$ is endowed with the subspace topology --- is continuous;
in particular, it is measurable. Therefore, one can define the \emph{push-forward measure} $\nu\defi\varphi_{\ast}\tre\lem$
($\de\lem(x)=\de x$), that is a Borel probability measure on $\nb(\pusta,\pustajj)$ (depending on $V$ and
on the state $\pi\otimes\varpi$), and by the change-of-variables formula, we obtain the barycentric decomposition
\begin{equation} \label{nonds-bis}
\varsigma=\int_{\nb(\pusta,\pustajj)}\mdodici\pi\otimes\varpi\dodici\de\nu(\pi\otimes\varpi) \fin ,
\quad \nu=\varphi_{\ast}\tre\lem \fin , \ \varphi\colon [0,1]\rightarrow\nb(\pusta,\pustajj) \fin .
\end{equation}
It can be proved that if, for every $k\in\itg$, $c_k(\psi_1)\defi\langle\phi_k,\psi_1\rangle\neq 0\neq c_k(\psi_2)$
(i.e., if the state vectors $\psi_1\in\hh$ and $\psi_2\in\jj$ have non-vanishing Fourier coefficients), then
the separable state $\varsigma$ --- defined by~\eqref{nonds} with $\pi=\hpu$, $\varpi=\hpd$, or, equivalently,
by~\eqref{nonds-bis} --- is \emph{not} discretely separable (see Theorem~{3} of~\cite{HoShiWe}). Therefore,
if $\dim(\hh)=\dim(\jj)=\infty$, then $\dsta\subsetneq\sesta$.
\end{example}

Let us derive some further consequences of Theorem~\ref{thechase}.

\begin{corollary} \label{corsestap}
A separable state $\varsigma\in\sesta$ is pure --- i.e., a rank-one projection --- iff it is of
the form~\eqref{bardecb}, where the probability measure $\nu$ is a Dirac measure or, equivalently, such that $\supp(\nu)$ is a
singleton set. Therefore, the set $\sestap$ of all separable pure states on $\hh\otimes\jj$ is
given explicitly by
\begin{equation}
\sestap=\nb(\pusta,\pustajj)=
\big\{\pi\otimes\varpi\colon\pi\in\pusta,\ \varpi\in\pustajj\big\} ,
\end{equation}
and $\sesta=\clco(\sestap)$. Moreover, denoting by $\ext(\sesta)$ the set of all extreme points of
the convex set $\sesta$, we have that
\begin{equation} \label{rextre}
\ext(\sesta)=\nb(\pusta,\pustajj)=\sestap \fin .
\end{equation}
\end{corollary}

\begin{proof}
First note that a probability measure $\nu$ on $\nb(\pusta,\pustajj)$ is a Dirac measure iff
$\supp(\nu)$ is a singleton set (by Remark~\ref{remdir} and Fact~\ref{famesp}, because $\supp(\nu)$
is a second countable Hausdorff space). Therefore, given any $\varsigma\in\sesta$, let us prove that
$\varsigma$ is pure iff it is of the form~\eqref{bardecb}, where $\nu=\delta_D$, with
$\{D\}=\supp(\nu)\subset\nb(\pusta,\pustajj)$. Clearly, if the latter condition holds, then
$\varsigma=D=\pi_0\otimes\varpi_0$, for some $\pi_0\in\pusta$ and $\varpi_0\in\pustajj$. To prove
the reverse implication, suppose that $\varsigma\in\sestap$ and $\supp(\nu)$ is \emph{not} a singleton set.
Then, since $\supp(\nu)\neq\varnothing$ (again by Fact~\ref{famesp}, $\nb(\pusta,\pustajj)$ being second countable),
there are at least \emph{two different pure states} of the form $\pi_1\otimes\varpi_1$ and $\pi_2\otimes\varpi_2$
in $\supp(\nu)$. It follows that there exist \emph{disjoint open neighborhoods} $\neigh$, $\neighp$ of
$\pi_1\otimes\varpi_1$ and $\pi_2\otimes\varpi_2$, respectively, so that $\nu(\neigh),\nu(\neighp)>0$
and $\neighp\subset\neighc$ (thus, for the closed set $\neighc$, $\nu(\neighc)>0$ too). Hence, we have:
\begin{align}
\varsigma
& =
\int_{\nb(\pusta,\pustajj)}\mdodici\big(\chne(\pi\otimes\varpi) + \chnec(\pi\otimes\varpi)\big)
\sei\pi\otimes\varpi\dodici\de\nu(\pi\otimes\varpi)
\nonumber\\
& =
\int_{\neigh}\msei \pi\otimes\varpi \dodici\de\nu(\pi\otimes\varpi) +
\int_{\neighc}\msei \pi\otimes\varpi \dodici\de\nu(\pi\otimes\varpi)
\nonumber\\ \label{twoints}
& =
\epsilon\tre D_1 + (1-\epsilon)\tre D_2  , \quad \epsilon\in(0,1) \fin , \
D_1, D_2\in\sesta \fin .
\end{align}
In the last line of~\eqref{twoints}, $0<\epsilon=\nu(\neigh)<1$, $0<1-\epsilon=\nu(\neighc)<1$ and
\begin{equation}
D_1=\nu(\neigh)^{-1}\int_{\neigh}\msei \pi\otimes\varpi \dodici\de\nu(\pi\otimes\varpi)
=\int_{\nb(\pusta,\pustajj)}\mdodici\pi\otimes\varpi\dodici\de\mu_1(\pi\otimes\varpi) \fin ,
\end{equation}
\begin{equation}
D_2=\int_{\nb(\pusta,\pustajj)}\mdodici \pi\otimes\varpi \dodici\de\mu_2(\pi\otimes\varpi) \fin ,
\end{equation}
where $\mu_1$, $\mu_2$ are the Borel probability measures on $\nb(\sta,\stajj)$ determined by
\begin{equation}
\de\mu_1(\pi\otimes\varpi)=\nu(\neigh)^{-1}\tre\chne(\pi\otimes\varpi)\sei\de\nu(\pi\otimes\varpi) \fin , \
\de\mu_2(\pi\otimes\varpi)=\nu(\neighc)^{-1}\tre\chnec(\pi\otimes\varpi)\sei\de\nu(\pi\otimes\varpi) \fin .
\end{equation}
Therefore, we find a contradiction, because $\varsigma$ was assumed to be an extreme point of
$\stahj$, and we must conclude that, if $\varsigma$ is a pure state, then the nonempty set $\supp(\nu)$
must be a singleton.

Let us now prove the second assertion. First note that, since $\sesta\subset\stahj$, then
$\ext(\stahj)\cap\sesta\ifed\sestap\subset\ext(\sesta)$ (by Fact~\ref{factext}). Next,
suppose that $\varsigma\in\sesta\setminus\sestap$, and consider a barycentric decomposition
of $\varsigma$ of the form~\eqref{bardecb}. Then, $\nu$ cannot be a Dirac measure and there
must be at least \emph{two different pure states} of the form $\pi_1\otimes\varpi_1$ and
$\pi_2\otimes\varpi_2$ in $\supp(\nu)\neq\varnothing$. Arguing as above, we find that
$\varsigma= \epsilon\tre D_1 + (1-\epsilon)\tre D_2$, $\epsilon\in(0,1)$,
$D_1, D_2\in\sesta$; i.e., $\varsigma$ is not contained in $\ext(\sesta)$ --- hence,
$\ext(\sesta)\subset\sestap$, as well --- so that, actually, $\ext(\sesta)=\sestap$.
\end{proof}

We are now able to generalize the final assertion of Proposition~\ref{procsta}: The cross separable
states wrt to the bipartition $\hhjj$ are precisely the separable states \emph{tout court} in the
genuinely infinite-dimensional setting ($\dim(\hh)=\dim(\jj)=\infty$), as well; in fact, the
following result holds.

\begin{corollary} \label{corconta}
The set $\sesta=\sta\otimes\stajj$ of all separable states is contained in the set $\hstahj$
of cross states and coincides with the set $\csta=\sta\prtp\stajj$ of cross separable states;
in fact, we have:
\begin{equation}
\pusta\otimes\pustajj=\sta\otimes\stajj=\sta\prtp\stajj
=\pusta\prtp\pustajj\subset\hstahj \fin .
\end{equation}
\end{corollary}

\begin{proof}
We have already shown that $\csta\subset\sesta$ (Proposition~\ref{procsta}), and the reverse
inclusion holds too. In fact, by Theorem~\ref{thechase}, if a state $\varsigma\in\stahj$ is
separable, then it is of the form~\eqref{bardeca}, where the integral converges wrt the
projective norm $\ptrnor$ as well (to the same state), and hence, by the second assertion of
Lemma~\ref{lemprem}, if $\varsigma\in\sesta$, then  it is contained in the $\ptrnor$-closed
convex hull $\clcop(\nb(\sta,\stajj))=\csta$ too. Therefore, recalling relation~\eqref{propust},
\begin{align}
\pusta\otimes\pustajj=\sta\otimes\stajj
& =
\sesta
\nonumber \\
& =
\csta=\sta\prtp\stajj\subset\hstahj \fin .
\end{align}
It remains to prove that $\sta\prtp\stajj=\pusta\prtp\pustajj$. Clearly,
$\sta\prtp\stajj\supset\pusta\prtp\pustajj$. Conversely, by Theorem~\ref{thechase}, every
separable state $\varsigma\in\sesta=\sta\otimes\stajj=\sta\prtp\stajj$ is of the
form~\eqref{bardecb}, where a Bochner integral of $\thptj$-valued functions is
understood; hence, by the first inclusion relation in~\eqref{integp} (Lemma~\ref{lemprem}),
$\sta\prtp\stajj\subset\clcop(\nb(\pusta,\pustajj))\ifed\pusta\prtp\pustajj$ too.
\end{proof}

\begin{corollary} \label{corequi}
For every cross state $D\in\hstahj$, we have that
\begin{equation} \label{nondis-tris}
\ptrntr{D}=\inf \big\{2\tre\alpha-1\ge 1\colon \alpha\tre\don + (1-\alpha)\tre\dtw=D, \
\don,\dtw\in\sesta\big\} \fin .
\end{equation}
\end{corollary}

\begin{proof}
Since $\dsta\subset\sesta$, by relation~\eqref{nondis-bis}, for every $D\in\hstahj$ we have that
\begin{equation} \label{nondis-quater}
\ptrntr{D}\ge\inf \big\{2\tre\alpha-1\ge 1\colon \alpha\tre\don + (1-\alpha)\tre\dtw=D, \
\don,\dtw\in\sesta\big\} \fin .
\end{equation}
On the other hand, by Corollary~\ref{corconta}, $\sesta=\sta\otimes\stajj=\sta\prtp\stajj=\csta$.
Hence, by Proposition~\ref{procsta}, $\sesta$ coincides with the $\ptrnortr$-closure of the convex
set $\dsta$ (the projective norm and the Hermitian projective norm being equivalent), so that, if
\begin{equation}
D=\alpha\tre\don(\alpha) + (1-\alpha)\tre\dtw(\alpha) \fin , \quad \alpha\ge 1 \fin , \
\don,\dtw\in\sesta \fin ,
\end{equation}
then, for every $\epsilon>0$, there are some $\don(\alpha,\epsilon),\dtw(\alpha,\epsilon)\in\dsta$
such that
\begin{equation} \label{dsapprox}
\ptrntr{\don(\alpha)-\don(\alpha,\epsilon)}<\frac{\epsilon}{2\alpha} \ \ \mbox{and} \ \
\ptrntr{\dtw(\alpha)-\dtw(\alpha,\epsilon)}<\frac{\epsilon}{2(\alpha-1)} \fin .
\end{equation}
It follows that $\ptrntr{D}=\ptrntr{\alpha\tre\don(\alpha) + (1-\alpha)\tre\dtw(\alpha)}
\le\alpha\tre\ptrntr{\don(\alpha)-\don(\alpha,\epsilon)} + \alpha\tre\ptrntr{\don(\alpha,\epsilon)}
+ (\alpha-1)\tre\ptrntr{\dtw(\alpha)-\dtw(\alpha,\epsilon)}+ (\alpha-1)\tre\ptrntr{\dtw(\alpha,\epsilon)}
<2\tre\alpha -1 +\epsilon$, where, for obtaining the second inequality, we have used relations~\eqref{dsapprox}
and the fact that $\ptrntr{\don(\alpha,\epsilon)}=1=\ptrntr{\dtw(\alpha,\epsilon)}$ (by Proposition~\ref{proaffco},
since $\don(\alpha,\epsilon),\dtw(\alpha,\epsilon)\in\dsta$). By the previous estimate, we conclude that
\begin{equation} \label{nondis-quinques}
\ptrntr{D}\le\inf \big\{2\tre\alpha-1\ge 1\colon \alpha\tre\don + (1-\alpha)\tre\dtw=D, \
\don,\dtw\in\sesta\big\} \fin ,
\end{equation}
as well. Hence, by inequalities~\eqref{nondis-quater} and~\eqref{nondis-quinques}, we see that,
actually, relation~\eqref{nondis-tris} holds true.
\end{proof}

\begin{corollary} \label{corimplise}
For every cross state $D\in\hstahj$ we have that $\ptrntr{D}\ge\|D\ptrn\ge 1$, and
\begin{equation} \label{implise}
\varsigma\in\sesta \implies \ptrntr{\varsigma}=\|\varsigma\ptrn = 1 \fin .
\end{equation}
\end{corollary}

\begin{proof}
Indeed, for every $D\in\hstahj$ we have: $\ptrntr{D}\ge\|D\ptrn\ge\|D\ttrn=1$. In particular,
for every state $\varsigma\in\sesta=\csta$ (Corollary~\ref{corconta}), by Proposition~\ref{pronose},
$\|\varsigma\ptrn=\ptrntr{\varsigma}=1$. Alternatively, using the characterization~\eqref{bardeca} of
a separable state, we obtain the estimate
\begin{equation}
1\le\|\varsigma\ptrn\le\ptrntr{\varsigma}\le
\int_{\nb(\sta,\stajj)}\mdodici\ptrntr{\rho\otimes\sigma}\dodici\de\mu(\rho\otimes\sigma)=1 \fin ,
\end{equation}
i.e., $\ptrntr{\varsigma}=\|\varsigma\ptrn=1$. Here, we have used the fact that the integral in~\eqref{bardeca}
converges wrt the projective norm $\ptrnor$ too, and a well-known property of the Bochner integral.
\end{proof}

\begin{corollary} \label{fidiset}
Assume that both $\hh$ and $\jj$ are finite-dimensional: $\enne\equiv\dim(\hhjj)<\infty$.
Then, for every density operator $D\in\stahj$, the following facts are equivalent:

\begin{enumerate}[label=\tt{(S\arabic*)}]

\item \label{fisepa}
$D\in\fsta$, i.e., $D$ is finitely separable;

\item \label{fisepb}
$D\in\dsta$, i.e., $D$ is discretely separable;

\item \label{fisepc}
$D\in\sesta$, i.e., $D$ is separable;

\item \label{fisepd}
$\ptrntr{D}= 1$;

\item \label{fisepe}
$\|D\ptrn = 1$.

\end{enumerate}

Therefore, if $\dim(\hhjj)<\infty$, we have:
\begin{align}
\fsta=\dsta=\sesta
& =
\big\{D\in\stahj\colon \ptrntr{D}= 1\big\}
\nonumber \\
&
=
\big\{D\in\stahj\colon \|D\ptrn= 1\big\} .
\end{align}

Moreover, in this case, every separable density operator $\varsigma\in\sesta$ admits a finite convex
decomposition of the form
\begin{equation} \label{finconvde}
\varsigma=\sum_{k=1}^\mathsf{K}\tre p_k\tre (\pi_k\otimes\varpi_k) \fin , \
\mbox{where $\mathsf{K}\le\enne^2$, $p_k>0$ $\big(\sum_{k=1}^\mathsf{K}\tre p_k=1\big)$,
$\pi_k\in\pusta$, $\varpi_k\in\pustajj$} \fin .
\end{equation}
\end{corollary}

\begin{proof}
Clearly, \ref{fisepa}$\implies$\ref{fisepb}$\implies$\ref{fisepc}. Moreover, by Corollary~\ref{corimplise},
\ref{fisepc} implies~\ref{fisepd}, and, since, for every $D\in\hstahj$, $\ptrntr{D}\ge\|D\ptrn\ge\|D\ttrn= 1$,
\ref{fisepd} implies~\ref{fisepe}. Let us prove that \ref{fisepe}$\implies$\ref{fisepc}$\implies$\ref{fisepa},
and hence the facts \ref{fisepa}--\ref{fisepe} are mutually equivalent.

To prove that \ref{fisepe} implies~\ref{fisepc}, recall that, by Proposition~\ref{proexods}, if both $\hh$
and $\jj$ are finite-dimensional, then every linear operator in $\thptj$ is optimally decomposable. Hence,
by the equivalence of conditions~\ref{dosa} and~\ref{dosd} in Proposition~\ref{prdecdens}, we have:
\ref{fisepe}$\implies$\ref{fisepb}$\implies$\ref{fisepc}.

Let us finally prove that \ref{fisepc} implies~\ref{fisepa} and decomposition~\eqref{finconvde}, simultaneously.
In fact, in the case were $\enne\equiv\dim(\hhjj)<\infty$, $\sesta$ is a (norm-bounded and norm-closed, hence)
$\ttrnor$-compact convex subset of the $\enne^2$-dimensional \emph{real} linear space $\bophjsa=\trchjsa$. Then,
by the classical Minkowski-Carath\'eodory theorem --- see, e.g., Theorem~{8.11} of~\cite{Simon_conv}, or
Theorem~{I.6.13} of~\cite{Alfsen} 
--- a separable density operator $\varsigma$ on $\hhjj$ can always be expressed as a convex combination of, at most,
$\enne^2$ points of $\ext(\sesta)$; in this regard, note that the \emph{affine dimension} of the convex set
$\sesta=\clco(\nb(\sta,\stajj))=\clco(\nb(\pusta,\pustajj))=\co(\nb(\sta,\stajj))=\co(\nb(\pusta,\pustajj))$
is not larger than $\enne^2-1$ --- see Theorem~{8.8} of~\cite{Simon_conv}, and the subsequent definition
of dimension of a convex set --- since $\stahj$ is a convex body in the affine hyperplane
$\afhy\defi\{A\in\bophjsa=\trchjsa\colon\tr(A)=1\}$. Now, 
the aforementioned at most $\enne^2$ points must be contained in the set $\sestap=\nb(\pusta,\pustajj)=\ext(\sesta)$
of all separable pure states. Thus, relation~\eqref{finconvde} holds, and~\ref{fisepc} implies~\ref{fisepa}.
\end{proof}

\section{Further characterizations of separable states}
\label{further}

In Subsection~\ref{sepgen} (recall Corollary~\ref{corimplise}), we have shown that if a bipartite
state $\varsigma\in\stahj$ is separable --- i.e., if $\varsigma\in\sesta\subset\thptj$ --- then
$\|\varsigma\ptrn=1$. We will now prove that the reverse implication holds true, as well, so that
the set of all separable states (wrt the bipartition $\hhjj$) is a $\ptrnor$-closed convex subset
of $\thptj$, completely characterized as the intersection of $\stahj$ with a level set of the
projective norm (this result is already proven in the finite-dimensional setting only; recall
Corollary~\ref{fidiset}).

Let us denote by $\stahjb$ the intersection of the set $\stahj$ with the \emph{unit sphere}
$\sph\big(\thptj\big)$ --- equivalently, since $\|D\ptrn\ge\|D\ttrn= 1$, for every $D\in\hstahj$,
with the \emph{unit ball} $\ball\big(\thptj\big)\defi\big\{C\in\thptj\colon\|C\ptrn\le 1\big\}$
--- in $\thptj$:
\begin{align}
\stahjb
& \defi
\stahj\cap\sph\big(\thptj\big)
\nonumber \\
& \dodici =
\big\{D\in\hstahj\colon\|D\ptrn=1\big\}
\nonumber \\
& \dodici =
\big\{D\in\hstahj\colon\|D\ptrn\le 1\big\}
\nonumber \\ \label{basph}
& \dodici =
\stahj\cap\ball\big(\thptj\big)\ifed\shjb .
\end{align}
The third equality in~\eqref{basph} implies that $\stahjb=\shjb$ is a \emph{convex subset} (of
$\hstahj$ and) of $\stahj$. In the following, we will use both notations --- $\stahjb$ or $\shjb$
--- for indicating the same set; the former notation being preferred to emphasize the sharp value
of the projective norm, the latter being chosen to stress that we are dealing with a convex set.
With these notations, our claim is expressed by the relation $\sesta=\stahjb=\shjb$.
We will also prove that $\sesta$ coincides with the intersection of $\stahj$ with the unit sphere
$\sph\big(\thsptjs\big)=\big\{C\in\thsptjs\colon\ptrntr{C}=1\big\}$ in the Hermitian trace
class $\thsptjs$. For technical reasons, to prove these claims it will be convenient to regard
the elements of the cross trace class $\thptj$ as \emph{functionals} on a suitable unital
$\cast$-algebra $\alghj$ of bounded operators on the Hilbert space $\hhjj$. This construction
involves various intermediate steps that will be outlined in Subsection~\ref{strategy}, after
defining the $\cast$-algebra $\alghj$.

\subsection{The relevant operator algebra}
\label{algebra}

In the case where $\dim(\hhjj)=\infty$, we will denote by $\alghj$ the \emph{unitization}~\cite{Murphy} of
the (non-unital) $\cast$-algebra $\cphj$ of all \emph{compact operators} on the Hilbert space $\hhjj$;
namely, we consider the unital $\cast$-algebra
\begin{equation}
\alghj\defi\cphj+\ccc\sei I \fin ,
\end{equation}
endowed with the norm $\|A + z \sei I\|\defi\sup\{\|AB + z\sei B\tnori\colon B\in\cphj,\ \|B\tnori=1\}$,
that coincides with the restriction to $\alghj$ of the standard operator norm $\tnormi$ of $\bophj$.
Therefore, $\alghj$ is, in a natural way, a (unital) $\cast$-subalgebra of the ambient $\cast$-algebra
$\bophj$ (see Example~{10.4.8} of~\cite{Kadison}). In the case where $\elle=\dim(\hhjj)<\infty$, we
simply put $\alghj\equiv\cphj=\bophj$ (say, the unital $\cast$-algebra of $\elle\times\elle$ complex
matrices endowed with the spectral norm).

We can endow the dual $\alghjd$ of the $\cast$-algebra $\alghj$ with its $\wst$-topology
(weak$^\ast$ topology), i.e., with the initial topology induced by the family of maps
\begin{equation}
\big\{\pax\colon K\in\alghj\big\} \fin , \quad \pax\colon\alghjd\rightarrow\ccc ,
\end{equation}
where $\pail K,\xi\pair\equiv \xi(K)\in\ccc$ is the pairing between $K\in\alghj$ and $\xi\in\alghjd$;
equivalently, with the initial topology induced by the family of semi-norms
$\big\{|\pax|\colon K\in\alghj\big\}$, where $|\pax|\colon\alghjd\rightarrow\errep$. Endowed with this
topology --- that will be called the \emph{$\awst$-topology} --- $\alghjd$ becomes a locally convex
topological vector space. $\alghjd$ contains $\trchj$ as a distinguished linear subspace via the
(injective, linear) immersion map
\begin{equation} \label{defident}
\identi\colon\trchj\ni A\mapsto\xia\in\alghjd
\end{equation}
where the pairing between the functional $\xia\in\alghjd$ ($A\in\trchj$) and a vector $K$ of
$\alghj$ is provided by the trace, i.e., $\xia(K)\equiv\pail K,\xia\pair\defi\tr(KA)$. The map
$\identi$ is a linear isometry, because --- by the fact that, for every $A\in\trchj$ and $K\in\alghj$,
$|\xia(K)|\le\|A\ttrn\tre\|K\tnori$, and by Fact~\ref{norsub} (i.e., $\cphj$ is a norming subspace of
$\bophj=\trchjd$) --- we have:
\begin{equation} \label{renorsub-bis}
\|\xia\norps\le\|A\ttrn=\sup\{|\xia(K)|=|\tr(AK)|\colon K\in\cphj, \ \|K\tnori=1\}\le\|\xia\norps \fin .
\end{equation}
Clearly, the rhs inequality is due to the fact that the supremum is taken over the unit ball of
$\cphj\subset\alghj$. Hence, actually, $\|\xia\norps=\|A\ttrn$.

\begin{remark} \label{rehaba}
For every $A\in\trchj$ --- taking into account the equality in relation~\eqref{renorsub-bis}
and the fact that, by the same relation, $\|\xia\norps=\|A\ttrn$ --- we see that the functional
$\identi(A)=\xia\in\alghjd$ is a Hahn-Banach extension of the functional
$\big(\cphj\ni K\mapsto\tr(KA)\in\ccc\big)\in\cphj^\ast$.
\end{remark}

The relative (or subspace) topology on $\trchj\equiv\identi(\trchj)$ wrt the $\awst$-topology
--- i.e., the initial topology on $\trchj$, via the map $\identi$, associated with the
$\awst$-topology on $\alghjd$ --- will be called the \emph{extended $\wst$-topology} or,
in short, the \emph{$\ewst$-topology}. Every $\ewst$-open subset of $\trchj$ is of the form
$\identi^{-1}(\neigh)$, for some $\awst$-open subset $\neigh$ of $\alghjd$. Therefore, endowed
with the $\ewst$-topology, $\trchj$ is a topological space homeomorphic to $\identi(\trchj)$
(equipped with its subspace topology wrt the $\awst$-topology of $\alghj$).  A net $\{A_i\}$ in
$\trchj$ converges to some $A\in\trchj$ wrt the $\ewst$-topology iff $\lim_i\tr(KA_i)=\tr(KA)$,
for all $K\in\alghj$.

\begin{fact} \label{relclo}
By a well-known property of the closure of a set wrt the subspace topology, for every subset
$\subh$ of $\trchj$, we have that
\begin{equation} \label{relaclos}
\clwsub=\clawsub\cap\tre\trchj \fin ,
\end{equation}
where $\clwsub\equiv\ewstcl(\subh)$, $\clawsub\equiv\awstcl(\subh)$ denote the $\ewst$-closure and
the $\awst$-closure of $\subh$, respectively, and all the obvious identifications via the immersion
map $\identi\colon\trchj\rightarrow\alghjd$ are understood; i.e., relation~\eqref{relaclos} is a
shorthand form of $\identi\Big(\clwsub\otto\Big)=\clawsubid\cap\tre\identi(\trchj)$.
\end{fact}

\begin{proposition} \label{proclo}
Let $\subh$, $\subj$ be nonempty subsets of $\trchj$ and $\alghjd$, respectively. Then, for the
$\ewst$-closed convex hull of $\subh$ and the $\awst$-closed convex hull of $\subj$, we have that
\begin{equation} \label{clcotworels}
\clcow(\subh)=\ewstcl(\co(\subh)) \fin , \quad \clcoaw(\subj)=\awstcl(\co(\subj)) \fin .
\end{equation}
\end{proposition}

\begin{example} \label{exaclo}
Assume that $\dim(\hhjj)=\infty$. In this case, the zero vector in $\trchj$ belongs to the closure
$\wstcl(\stahj)$ --- wrt the $\wst$-topology of $\trchj$ regarded as the Banach space dual of $\cphj$
(via the trace functional) --- of the convex set $\stahj$. Indeed, given any orthonormal basis
$\{\vek\}_{k=1}^\infty$ in the Hilbert space $\hhjj$, and putting $D_l\defi l^{-1}\sum_{k=1}^l\veek\in\stahj$,
$l\in\nat$, for every $A\in\trchj$ we have that
\begin{equation} \label{limidl}
\lim_l\sei\tr(D_l A)=
\lim_l\Big(l^{-1}\sei\tr\big({\textstyle (\sum_{k=1}^l\veek) A}\big)\Big)=0  \fin ,
\end{equation}
because $\lim_l \tr\big({\textstyle (\sum_{k=1}^l\veek) A}\big)=
\sum_{k=1}^\infty\langle\vek,A\sei\vek\rangle=\tr(A)$. Since the trace class $\trchj$ is
$\tnormi$-dense in $\cphj$, given any $K\in\cphj$ and $A\in\trchj$, by the estimate
\begin{align}
|\tr(D_l K)|\le |\tr(D_l (K-A))|+ |\tr(D_l A)|
& \le
\|D_l (K-A)\ttrn + \|D_l A \ttrn
\nonumber \\
& \le
\|D_l\ttrn\sei\|K-A\tnori + \|D_l\tnori\sei\|A\ttrn
\nonumber \\
& =
\|K-A\tnori+l^{-1}\sei\|A\ttrn  \fin ,
\end{align}
it follows that relation~\eqref{limidl} actually holds for every $A\in\cphj$. Therefore,
$\wstli_l D_l=0\in\wstcl(\stahj)$. By this fact, it also follows that, for any $D\in\stahj$
and $s\in [0,1]$, putting $D_l(s)\defi s\sei D + (1-s)D_l\in\stahj$, we have that
$\wstli_l D_l(s)=s\sei D\in\wstcl(\stahj)$; hence:
$\co(\stahj\cup\{0\})=\{s\sei D\colon s\in [0,1],\ D\in\stahj\}\subset\wstcl(\stahj)$.

Moreover, by relation~\eqref{limidl}, which holds true for every $A\in\cphj$, one can easily
check that $\awstli_l\identi(D_l)=\Xi\in\alghjd\setminus\identi(\trchj)$, where $\Xi$ is the
positive linear functional on $\alghj$ determined by setting $\Xi(K)=0$, for all $K\in\cphj$,
and $\Xi(I)=1$. As a consequence, the sequence $\{D_l\}_{l=1}^\infty\subset\stahj$ does not
admit a limit in $\trchj$ wrt the $\ewst$-topology.
\end{example}

The following result justifies the term ``extended $\wst$-topology'' introduced above.

\begin{proposition} \label{protwotops}
Suppose that $\dim(\hhjj)=\infty$. Then, the $\ewst$-topology is strictly stronger --- i.e.,
strictly finer --- than the $\wst$-topology of $\trchj$, regarded as the Banach space dual of
$\cphj$ via the trace functional; i.e., the $\wst$-topology is a strict subtopology of the
$\ewst$-topology. Moreover, $\stahj$ is $\ewst$-closed, but not $\wst$-closed; $\identi(\stahj)$
is not $\awst$-closed.
\end{proposition}

In spite of the fact that, when $\dim(\hhjj)=\infty$, the $\ewst$-topology is strictly stronger than
the $\wst$-topology of $\trchj$, the associated subspace topologies on $\stahj$ \emph{do coincide}.

\begin{proposition} \label{protop}
The relative topologies on $\stahj$ wrt the $\wst$-topology and the $\ewst$-topology of $\trchj$
both coincide with the standard topology --- {\rm recall Fact~\ref{statop}} --- of $\stahj$.
\end{proposition}

\begin{proof}
By Proposition~\ref{protwotops}, the relative $\ewst$-topology on $\stahj$ is stronger than
the relative $\wst$-topology on $\stahj$ inherited from $\trchj$. The latter, in turn, is
stronger than the relative topology induced on $\stahj$ by the weak operator topology of $\bophj$;
i.e., the standard topology. Conversely, the relative topology on
$\stahj$ wrt the trace norm (the Schatten $1$-norm $\ttrnor$) topology --- i.e., again,
the standard topology of $\stahj$ --- is stronger than the relative $\ewst$-topology,
which is stronger than the relative $\wst$-topology. Hence, all these topologies must
coincide with the standard topology of $\stahj$.
\end{proof}

By Proposition~\ref{protop}, we get to an important point, i.e., Theorem~\ref{thclcow} below.

\begin{lemma} \label{lemthclcow}
With the obvious identifications $\clcow(\nb(\pusta,\pustajj))\equiv\identi\big(\clcow(\nb(\pusta,\pustajj))\big)$,
$\nb(\pusta,\pustajj)\equiv\identi(\nb(\pusta,\pustajj))$ and $\stahj\equiv\identi(\stahj)$,
the following relation holds:
\begin{equation} \label{relemthcl}
\clcow(\nb(\pusta,\pustajj))=\clcoaw(\nb(\pusta,\pustajj))\cap\stahj \fin .
\end{equation}
\end{lemma}
\begin{proof}
By Fact~\ref{relclo} and relations~\eqref{clcotworels}, and with suitable identifications via
the immersion map $\identi$, we have:
\begin{align}
\clcow(\nb(\pusta,\pustajj))
& =
\ewstcl\big(\co(\nb(\pusta,\pustajj))\big)
\nonumber \\
& =
\awstcl\big(\co(\nb(\pusta,\pustajj))\big)\cap\trchj
\nonumber \\
& =\clcoaw(\nb(\pusta,\pustajj))\cap\trchj \fin .
\end{align}
We need to prove that $\clcoaw(\nb(\pusta,\pustajj))\cap\trchj=\clcoaw(\nb(\pusta,\pustajj))\cap\stahj$.
In fact, let $\{D_i\}\subset\co(\nb(\pusta,\pustajj))\subset\sesta$ be a net of (finitely) separable states
converging --- wrt the $\awst$-topology --- to $A\equiv\paxa\in\trchj\equiv\identi(\trchj)\subset\alghjd$
(with obvious identifications via the immersion map $\identi$); i.e., $\lim_i\tr(KD_i)=\lim_i\pail K,\xidi
\pair=\pail K,\xia\pair=\tr(KA)$, for all $K\cinque\in\alghj$. Now, by imposing this condition for every
$K\cinque\in\alghj$, with $K\ge 0$ (in particular, for every rank-one projection on $\hhjj$), we conclude
that $A\ge 0$, and, moreover, with $K=I$, we see that $\tr(A)=1$. Thus, $A=\awstli_i D_i\in\stahj$, and
the proof is complete.
\end{proof}

\begin{theorem} \label{thclcow}
The convex set of all separable states wrt the bipartition $\hhjj$ can be described as a closed convex hull
wrt the $\ewst$-topology of $\trchj$, i.e.,
\begin{equation} \label{inclrela}
\sesta=\clcow(\nb(\pusta,\pustajj)) \fin .
\end{equation}
\end{theorem}

\begin{proof}
By relation~\eqref{propust}, we have: $\sesta=\clco(\nb(\sta,\stajj))=\clco(\nb(\pusta,\pustajj))$.
Let us prove that $\clco(\nb(\pusta,\pustajj))=\clcow(\nb(\pusta,\pustajj))$. In fact, by
Proposition~\ref{protop}, the relative topology on $\stahj$ wrt the $\ewst$-topology of $\trchj$
coincides with the standard topology of $\stahj$, and hence with the relative topology on $\stahj$
wrt the trace norm topology of $\trchj$. Therefore, since
$\clcow(\nb(\pusta,\pustajj))=\ewstcl\big(\co(\nb(\pusta,\pustajj))\big)$ (recall the first of
relations~\eqref{clcotworels}), by a well-known property of the closure of a set wrt the subspace
topology, we have that
$\clcow(\nb(\pusta,\pustajj))\cap\stahj=\clco(\nb(\pusta,\pustajj))\cap\stahj=\clco(\nb(\pusta,\pustajj))$.
It just remains to notice that, by Lemma~\ref{lemthclcow}, $\clcow(\nb(\pusta,\pustajj))\subset\stahj$.
\end{proof}

\subsection{Our strategy}
\label{strategy}

For the reader's convenience, we now briefly outline our plan for deriving, in the rest of
this section, a complete characterization of the set $\sesta$ in terms of the projective norms.
As a byproduct, we will also obtain various interesting characterizations of the cross trace class
$\thptj$.

\begin{enumerate}

\item As a starting point, we know that $\sesta\subset\stahjb=\shjb$ (see Corollary~\ref{corimplise}
and relation~\eqref{basph}), and we want to prove, among other facts, that the reverse inclusion
relation holds too. To this end, we need to establish some further technical results.

\item We then introduce two mutually related norms $\norf$ and $\norg$ --- respectively, on the spaces
$\thptj$ and $\bophj$ --- and we show that these norms are dominated, respectively, by the projective
norm $\ptrnor$ (Proposition~\ref{threno} below) and by the standard operator norm $\tnormi$
(Proposition~\ref{pronorg}). We next obtain more precise characterizations of these norms.

\item It is immediate to realize that, in the finite-dimensional setting, $\norf$ is
\emph{precisely} (another way of defining) the projective norm $\ptrnor$ of $\thptj$;
see Remark~\ref{remfind} below. With no assumption on the dimension of the Hilbert spaces
$\hh$ and $\jj$, the norm $\norg$ can be regarded as a natural extension to $\bophj$ of the
\emph{injective norm}~\cite{Defant,Ryan,DFS,Kubrusly} of the algebraic tensor product
$\cph\altp\cpj$, or, equivalently, of the injective norm of $\bhabj$; see Remark~\ref{retwonors}.
Exploiting standard duality relations between the injective and the projective norms, we then
prove that the projective tensor product $\thptj$ is isomorphic to the \emph{dual} of the
Banach space completion of the algebraic tensor product $\cph\altp\cpj$ wrt to the injective
norm $\norg$ (i.e., the injective tensor product $\cph\injtp\cpj$ of $\cph$ and $\cpj$).
This fact entails that $\norf=\ptrnor$ in the infinite-dimensional setting too; see
Theorem~\ref{neothe} below. Therefore, the norm $\norf$ may be regarded as a
`useful reformulation' of the projective norm.

\item As previously observed, the restriction of the norm $\norg$ to the algebraic tensor
product $\bop\altp\bopjj$ coincides with the injective norm $\innormi$ of $\bop\altp\bopjj$
(Remark~\ref{retwonors}). Then, further exploring the duality relations involving the cross
trace class, it turns out that the injective tensor product $\bop\injtp\bopjj$ --- namely,
the Banach space completion of the normed space $(\bop\altp\bopjj,\norg=\innormi)$ --- is
(isomorphic to) a closed subspace of the Banach space dual $\big(\thptj\big)^\ast$ of $\thptj$;
see Theorem~\ref{proreno} below. In particular, in the case where at least one of the Hilbert
spaces $\hh$ and $\jj$ is finite-dimensional, the norms $\ttrnor$ and $\norf=\ptrnor$ are
equivalent on $\trchj=\thptj$ (a set equality being understood); moreover, in this case,
the linear space $\bophj$, endowed with the norm $\norg$, is a Banach space --- which is
isomorphic to the injective tensor product $\bop\injtp\bopjj$, and a \emph{renorming}~\cite{Guirao}
of the standard Banach space $(\bophj,\tnormi)$ of bounded operators --- and can be identified
with the Banach space dual of the cross trace class $\thptj$.

\item  Let us now restrict our focus on the cross trace class regarded as a set of bounded
functionals. The technical reason why this picture is of interest to us is related to the
peculiar properties of a weak$^\ast$  algebra~\cite{Megginson}. As previously noted, the
elements of $\thptj$ can be identified with the bounded functionals on the injective tensor
product $\cph\injtp\cpj$. This fact is crucial, e.g., for proving the coincidence of the norms
$\norf$ and $\ptrnor$. But, according to Remark~\ref{remnocl} and Proposition~\ref{procwst} below,
if $\dim(\hhjj)=\infty$, the mentioned identification is \emph{not} suitable to reach our final target
of proving that $\sesta=\stahjb$.

\item Recall that, in Subsection~\ref{algebra}, we have regarded the elements of the
trace class $\trchj$ --- hence, in particular, of the cross trace class $\thptj$ ---
as bounded functionals on the $\cast$-algebra $\alghj$, via the linear isometry
$\identi\colon\trchj\ni A\mapsto\xia\in\alghjd$. In addition, by means of the norms
$\norf$ and $\norg$ (and of the identification of the norms $\norf$ and $\ptrnor$),
the elements of the cross trace class $\thptj$ can also be regarded as \emph{special}
bounded functionals on $\alghj$. In fact, the Banach space $\thptj$ can be identified
with (i.e., is isomorphic to) a linear subspace $\halghj$ of the dual $\alghjd$ of $\alghj$,
provided that this subspace is equipped with a suitable norm $\nordg$; i.e., there is an
isomorphism of Banach spaces $\iden\colon\thptj\rightarrow\halghj\subset\alghjd$, that is
a (domain and codomain) restriction of the immersion map $\identi\colon\trchj\rightarrow\alghjd$.
Moreover, it turns out that the Banach space $\big(\halghj,\nordg\big)\cong\thptj$ is a
closed subspace of the dual $\alghjgd\equiv\big(\alghj,\norg\big)^\ast$ of the normed space
$\alghjg\equiv\big(\alghj,\norg\big)$; in particular, the norm $\nordg$ is the restriction
to the subspace $\halghj$ of the dual norm $\dunorg$ of $\norg$, the latter being regarded
as a norm on the linear space $\alghj$ (see Proposition~\ref{pronor} below). Summarizing,
we have \emph{two} further natural characterizations of the elements of the cross trace
class $\thptj$ as bounded functionals.

\item At this stage, one may ask whether we should consider the \emph{cross states} $\hstahj$
--- i.e., our main objects of interest --- as (positive) functionals on the $\cast$-algebra
$\alghj$ or, rather, on the closely related normed space $\alghjg$. In this regard, a remarkable
fact is the following (see Remark~\ref{refuns} below): $\alghjgd$ is, in a natural way, a linear
subspace of $\alghjd$, and, moreover, the $\sigma\big(\alghjgd,\alghjg\big)$-topology coincides
with the subspace topology on $\alghjgd\subset\alghjd$ wrt the $\sigma(\alghjd,\alghj)$-topology
(the $\awst$-topology). Therefore, for our main purposes, we may think of the elements of $\hstahj$
as functionals in $\halghj$ (i.e., as elements of the set $\iden(\hstahj)\subset\halghj\subset\alghjgd$)
or, \emph{equivalently}, as functionals in $\alghjd$ (i.e., as elements of $\identi(\stahj)\subset\alghjd$),
because, eventually, the $\awst$-topology, and the associated subspace topology on
$\trchj\equiv\identi(\trchj)$ (the $\ewst$-topology), will turn out to be the relevant mathematical tool.

\item By the previous reasoning, the convex set $\sesta\subset\thptj$ of all separable states
wrt the bipartition $\hhjj$ can eventually be identified with a convex subset of $\alghjd$. We can
then use the topological properties of the space $\alghjd$ as the dual of the $\cast$-algebra
$\alghj$ to prove that, denoting by $\clcow(\nb(\pusta,\pustajj))$ the $\ewst$-closed convex
hull of $\nb(\pusta,\pustajj)$ (equivalently, by the first of relations~\eqref{clcotworels},
the closure of $\co(\nb(\pusta,\pustajj))$ wrt the $\ewst$-topology), we have:
$\stahjb=\shjb\subset\clcow(\nb(\pusta,\pustajj))$, and hence --- recalling also that
$\sesta=\clcow(\nb(\pusta,\pustajj))$ (Theorem~\ref{thclcow}) --- actually, $\sesta=\stahjb$
(see Theorem~\ref{chasest} below). As a byproduct, it turns out that $\sesta$ can also be described as
the intersection of the set $\hstahj$ of all cross states with the unit sphere in the \emph{Hermitian}
trace class $\thsptjs$; see Corollary~\ref{chasest-bis} below.

\end{enumerate}

It is worth stressing, however, that regarding the elements of the cross trace class $\thptj$
as bounded functionals on the normed space $\alghjg\equiv\big(\alghj,\norg\big)$ --- or, also,
on the normed space $\bophjg\equiv\big(\bophj,\norg\big)$ (see Remark~\ref{redigub} below) ---
will be fundamental, in Section~\ref{entanglement}, for introducing the so-called
\emph{entanglement function}.

\subsection{Two useful norms and duality relations for the cross trace class}
\label{norms}

With every bounded operator $L\in\bophj$, we can associate a linear functional $\fu$ on the Banach
space $\thptj$, defined by
\begin{equation}
\fu\colon\thptj\ni C\mapsto\tr(CL)\in\ccc \fin .
\end{equation}
The functional $\fu$ is bounded, because $|\tr(CL)|\le\|C\ttrn\sei\|L\tnori\le\|C\ptrn\sei\|L\tnori$,
and we have:
\begin{equation} \label{denogamma}
\|\fu\norps\defi\big\{|\tr(CL)|\colon C\in\thptj,\ \|C\ptrn=1\big\}\le\|L\tnori \fin ;
\end{equation}
in particular,
\begin{equation}
\fu(S\otimes T)=\tr((S\otimes T)\tre L)=\bif(S,T) \fin , \quad S\in\trc \fin , \ T \in\trcjj \fin ,
\end{equation}
where $\bif\colon\trc\times\trcjj\rightarrow\ccc$, $\bif\in\ftrcjj$, is the canonical bilinearization
of the bounded linear functional $\fu\in\big(\thptj\big)^\ast$ (Definition~\ref{defbibi}, with
$\bs=\ccc$). Note that, by Theorem~\ref{isothm}, the norm $\|\fu\norps$ of the linear functional
$\fu\colon\thptj\rightarrow\ccc$ satisfies the following important relation:
\begin{equation} \label{relfu}
\|\fu\norps=\|\bif\norb=\sup\{|\tr((S\otimes T)\tre L)|\colon S\otimes T\in\nb(\trc,\trcjj),
\ \|S\trn=\|T\trn=1\}
\fin .
\end{equation}

For every cross trace class operator $C\in\thptj$, we can now introduce the quantity
\begin{equation} \label{defnorf}
\|C\nof\defi\sup\{|\fk(C)|=|\tr(CK)|\colon K\in\alghj, \ \|\fk\norps=1\} \fin .
\end{equation}

\begin{proposition} \label{threno}
Formula~\eqref{defnorf} defines a norm $\norf$ on the cross trace class $\thptj$, which dominates (the
restriction of) the trace norm and is dominated by the projective norm; i.e., for every $C\in\thptj$,
\begin{align}
\|C\ttrn\le\|C\nof
& \le
\|C\ptrn
\nonumber \\ \label{inenorf}
& =
\sup\big\{|\Gamma(C)|\colon\Gamma\in\big(\thptj\big)^\ast,\ \|\Gamma\norps=1\big\} .
\end{align}
\end{proposition}

\begin{proof}
Let us prove that formula~\eqref{defnorf} defines a norm. Clearly, $\|C\nof\ge 0$ and,
for every $c\in\ccc$, $\|c\sei C\nof=|c|\sei\|C\nof$. Moreover, $\|C\nof=0\iff C=0$,
because, if $\|C\nof=0$, then, in particular, $0=|\tr(C\tre\vac)|=|\langle\vc, C\tre\va\rangle|$,
for all $\va,\vc\in\hhjj$ ($\vac\in\cphj$); hence, $C=0$. It remains to prove the
triangle inequality; in fact, for all $C_1,C_2\in\thptj$, we have:
\begin{align}
\|C_1 + C_2\nof
& =
\sup\{|\fk(C_1 +C_2)|=|\tr((C_1 + C_2)\tre K)|\colon K\in\alghj, \ \|\fk\norps=1\}
\nonumber \\
& \le
\sup\{|\fk(C_1)| + |\fk(C_2)|\colon K\in\alghj, \ \|\fk\norps=1\}
\nonumber \\
& \le
\sup\{|\fk(C_1)|\colon K\in\alghj, \ \|\fk\norps=1\}
\nonumber \\
& +
\sup\{|\fk(C_2)|\colon K\in\alghj, \ \|\fk\norps=1\} = \|C_1\nof + \|C_2\nof \fin .
\end{align}
Let us now show that $\|C\ttrn\le\|C\nof$. Indeed, by Fact~\ref{norsub} and, next, by the
inequality $\|\fk\norps\le\|K\tnori$ (for any $K\in\cphj$), we have that
\begin{align}
\|C\ttrn
& =
\sup\{|\tr(CK)|\colon K\in\cphj, \ \|K\tnori=1\}
\nonumber \\
& =
\sup\{|\fk(C)|=|\tr(C\tre K)|\colon K\in\cphj, \ \|K\tnori\le 1\}
\nonumber \\
& \le
\sup\{|\fk(C)|\colon K\in\cphj, \ \|\fk\norps\le 1\}
\nonumber \\
& \le
\sup\{|\fk(C)|\colon K\in\alghj, \ \|\fk\norps\le 1\}
= \|C\nof \fin .
\end{align}
Finally, recalling relation~\eqref{forptrn}, it is clear that $\|C\nof\le\|C\ptrn$, as well, and
relation~\eqref{inenorf} holds.
\end{proof}

Note that, for every bounded operator $L\in\bophj$, we can put
\begin{equation} \label{defnog}
\|L\nog\defi\|\fu\norps\le\|L\tnori \fin .
\end{equation}

\begin{proposition} \label{pronorg}
The map $\norg\colon\bophj\ni L\mapsto\|\fu\norps\in\errep$ is a norm --- dominated by the operator
norm $\tnormi$ --- and, for every $L\in\bophj$, we have that
\begin{align}
\|L\nog
& \defi
\sup\big\{|\tr(CL)|\colon C\in\thptj,\ \|C\ptrn=1\big\}
\nonumber \\
& \dodici =
\sup\{|\tr((S\otimes T)\tre L)|\colon S\otimes T\in\nb(\trc,\trcjj), \ \|S\trn=\|T\trn=1\}
\nonumber \\
& \dodici =
\sup\{|\langle\phi\otimes\psi,L(\eta\otimes\chi)\rhj|\colon
\phi,\eta\in\hh,\ \psi,\chi\in\jj, \ \|\phi\|=\|\eta\|=\|\psi\|=\|\chi\|=1\}
\nonumber \\ \label{fornog}
& \dodici =
\|L^\ast\nog  \fin .
\end{align}
Moreover, for every $A\in\bop$ and $B\in\bopjj$, we have that $\|A\otimes B\nog=\|A\nori\tre\|B\nori=
\|A\otimes B\tnori$, and, for every pair $U$, $V$ of unitary operators on $\hh$ and $\jj$, respectively,
\begin{equation} \label{invlocuni}
\|L\nog=\|(U\otimes V)\tre L\tre (U\otimes V)^\ast\nog \fin .
\end{equation}

Finally, the algebraic tensor product $\frh\altp\frj$ is $\norg$-dense in $\cph\altp\cpj$.
\end{proposition}

\begin{proof}
From definition~\eqref{defnog}, it is clear that $\norg$ is a semi-norm. Moreover, since the linear
mapping
\begin{equation}
\bophj\ni L\mapsto\fu\in\big(\thptj\big)^\ast
\end{equation}
is injective, this semi-norm is also point-separating; i.e., actually, a norm.  For the characterization
of $\|L\nog$ on the second line of~\eqref{fornog}, recall relation~\eqref{relfu}. Let us next derive the
characterization of $\|L\nog$ on the third line of~\eqref{fornog}. For any $S\in\sph(\trc)$ and
$T\in\sph(\trcjj)$, by the singular value decomposition of a trace class operator (Remark~\ref{sivadeb}),
we have that $S=\trnsum_k s_k\tre\epk$ and $T=\trnsum_l t_l\tre\cpl$, where $s_k,t_l>0$
($\sum_k s_k=\|S\trn=1=\|T\trn=\sum_l t_l$), and $\{\eta_k\},\{\phi_k\}\subset\hh$,
$\{\chi_l\},\{\psi_l\}\subset\jj$ are orthonormal systems. Hence, for every $L\in\bophj$, we obtain
the estimate
\begin{align}
|\tr((S\otimes T)L)|
& =
\big|{\textstyle \sum_{kl}}\sei s_k\tre t_l\sei\langle\phi_k\otimes\psi_l,L\tre(\eta_k\otimes\chi_l)\rhj\big|
\nonumber \\
& \le
{\textstyle \sum_{kl}}\sei s_k\tre t_l\sei|\langle\phi_k\otimes\psi_l,L\tre(\eta_k\otimes\chi_l)\rhj|
\nonumber \\
& \le
{\textstyle \sup_{kl}}\tre|\langle\phi_k\otimes\psi_l,L\tre(\eta_k\otimes\chi_l)\rhj|
\nonumber \\
& \le
\sup\{|\langle\phi\otimes\psi,L(\eta\otimes\chi)\rhj|\colon
\phi,\eta\in\hh,\ \psi,\chi\in\jj,
\nonumber \\ \label{estten}
& \phantom{\le\sup\{}
\|\phi\|=\|\eta\|=\|\psi\|=\|\chi\|=1\} \fin , \quad
\forall\cinque S\in\sph(\trc) \fin , \ \forall\cinque T\in\sph(\trcjj) \fin .
\end{align}
Moreover, recalling relation~\eqref{relfu}, we have that
\begin{align}
\|L\nog
& =
\|\fu\norps
\nonumber \\
& =
\sup\{|\tr((S\otimes T)L)|\colon S\otimes T\in\nb(\trc,\trcjj), \|S\trn=\|T\trn=1\}
\nonumber \\
& \ge
\sup\{|\langle\phi\otimes\psi,L(\eta\otimes\chi)\rhj|\colon
\phi,\eta\in\hh,\ \psi,\chi\in\jj, \ \|\phi\|=\|\eta\|=\|\psi\|=\|\chi\|=1\}
\nonumber \\
& \ge \label{samese}
\sup\{|\tr((S\otimes T)L)|\colon S\otimes T\in\nb(\sph(\trc),\sph(\trcjj))\}
\fin ,
\end{align}
where the first inequality holds by set restriction to the rank-one partial isometries of the form
$\ep\otimes\chp=\ecpp$, whereas the second one follows immediately
from the estimate~\eqref{estten}. But the supremum on the second line of~\eqref{samese} and the
supremum on the fourth line are actually taken over the same set, so that relation~\eqref{fornog}
holds true.

By any of the various equivalent expressions in~\eqref{fornog} of $\|L\nog$, it is also clear that
$\|L\nog=\|L^\ast\nog$, and by the final expression relation~\eqref{invlocuni} follows immediately.

Therefore, for every $A\in\bop$ and $B\in\bopjj$, we have:
\begin{align}
\|A\otimes B\nog
& =
\sup\{|\langle\phi\otimes\psi, (A\otimes B)(\eta\otimes\chi)\rhj|\colon
\|\phi\|=\|\eta\|=\|\psi\|=\|\chi\|=1\}
\nonumber \\
& =
\sup\{|\langle\phi,A\eta\rhh|\sei |\langle\psi, B\chi\rjj|\colon
\|\phi\|=\|\eta\|=\|\psi\|=\|\chi\|=1\}
\nonumber \\
& =
\sup\{|\langle\phi,A\eta\rhh|\colon \|\phi\|=\|\eta\|=1\} \sei
\sup\{|\langle\psi, B\chi\rjj|\colon\|\psi\|=\|\chi\|=1\}
\nonumber \\ \label{aotb}
& =
\|A\nori\tre\|B\nori \fin ,
\end{align}
where the last equality follows by Fact~\ref{fabop}.

Finally, let $\{A_k\}_{k=1}^\ka,\{E_k\}_{k=1}^\ka\subset\bop$, $\{B_k\}_{k=1}^\ka,\{F_k\}_{k=1}^\ka
\subset\bopjj$ ($\ka\in\nat$) be $\ka$-tuples of bounded operators. Note that $\|\sum_k A_k\otimes B_k-
\sum_k E_k\otimes F_k\nog\le\sum_k\|A_k\otimes B_k- E_k\otimes F_k\nog\le\sum_k(\|(A_k-E_k)\otimes B_k\nog
+\|E_k\otimes(B_k- F_k)\nog)=\sum_k(\|A_k-E_k\nori\sei\|B_k\nori+\|E_k\nori\sei\|B_k- F_k\nori)$, where
for the last equality we have used relation~\eqref{aotb}. By this estimate, it is clear that --- since
the linear subspaces of finite rank operators $\frh$, $\frj$ are $\normi$-dense in $\cph$ and $\cpj$,
respectively --- the algebraic tensor product $\frh\altp\frj$ is $\norg$-dense in $\cph\altp\cpj$.
\end{proof}

\begin{remark} \label{remfind}
In the case where $\dim(\hhjj)<\infty$, we have that the finite-dimensional linear space
$\alghj=\cphj=\bophj$ --- once endowed  with the appropriate norm, that is,
$\bophj\ni L\mapsto\|L\nog\equiv\|\fu\norps$ --- can be identified with the Banach space
dual $\thptjd$ of the projective trace class (recall~\eqref{denogamma}), and hence, by
definition~\eqref{defnorf} and by the final equation in~\eqref{inenorf}, we see that
$\norf=\ptrnor$. We will soon prove that the equality of the two norms actually holds in
the infinite-dimensional setting too (see the first assertion of Theorem~\ref{neothe} below).
\end{remark}

\begin{remark}
Identifying $\trc$, $\trcjj$ with the Banach space duals of $\cph$ and $\cpj$, respectively,
from the expression of the norm $\norg$ on the second line of~\eqref{fornog} it is clear that
this norm can be regarded as a natural extension to $\bophj$ of the \emph{injective norm} --- see,
e.g., Chapter~{I} of~\cite{Defant}, Chapter~{3} of~\cite{Ryan}, or Chapter~{7} of~\cite{Kubrusly} ---
defined on the algebraic tensor product $\cph\altp\cpj$.
\end{remark}

\begin{remark} \label{retwonors}
The restriction of the norm $\norg\colon\bophj\rightarrow\errep$ to the algebraic tensor
product $\bop\altp\bopjj$ of the Banach spaces $\bop$ and $\bopjj$ coincides with the
injective norm $\innormi$ of $\bop\altp\bopjj$. In fact, as noted in Remark~\ref{reminjec},
for every $\sum_j A_j\otimes B_j\in\bop\altp\bopjj$,
we have that
\begin{align}
\big\|{\textstyle\sum_j A_j\otimes B_j}\binnori
& \defi
\sup\big\{\big|{\textstyle\sum_j\funs (A_j)\tre\funt(B_j)}\big|\colon \funs\in\dbop,\ \funt\in\dbopjj,
\ \|\funs\norps=\|\funt\norps=1\big\}
\nonumber \\
& \dodici =
\sup\big\{\big|{\textstyle\sum_j \tr(SA_j)\tre\tr(TB_j)}\big|\colon S\in\trc,\ T\in\trcjj,
\nonumber \\ \label{norinje-tris}
& \hspace{15.4mm}
\|S\trn=\|T\trn=1\big\}=\big\|{\textstyle \sum_j A_j\otimes B_j}\big\nog \fin ,
\end{align}
where we have used the fact that the unit balls $\ball(\trc)$ and $\ball(\trcjj)$ ---
regarded as subsets of $\trc^{\ast\ast}$ and $\trcjj^{\ast\ast}$ --- are \emph{norming sets}
for $\trc^{\ast}=\bop$ and $\trcjj^{\ast}=\bopjj$, respectively (see Subsection~{4.1}
in Chapter~{I} of~\cite{Defant}, or formula~{(3.3)} in Section~{3.1} of~\cite{Ryan}),
and next the characterization of $\big\|\sum_j A_j\otimes B_j\big\nog$ as on the second
line of~\eqref{fornog}. Therefore, the Banach space completion of $\bop\altp\bopjj$ wrt the norm
$\norg$ is precisely the \emph{injective tensor product}~\cite{Defant,Ryan,DFS,Kubrusly}
$\bop\injtp\bopjj$. Since the norm $\norg$ is majorized by the standard operator norm
$\tnormi$ (recall~\eqref{defnog}), it turns out that $\bhobj\subset\bop\injtp\bopjj$,
where $\bhobj$ is the natural (or spatial) tensor product of $\bop$ and $\bopjj$, i.e.,
the Banach space completion of the algebraic tensor product $\bhabj$ wrt the norm $\tnormi$
of $\bophj$, and $\bhabj=\bhobj=\bophj$ iff $\min\{\dim(\hh),\dim(\jj)\}<\infty$ (recall
Remark~\ref{rebop}, and references therein). Also note that the injective norm of $\cph\altp\cpj$
coincides with the restriction of the injective norm $\innormi$ of $\bop\altp\bopjj$ and (thus)
of the norm $\norg$. This fact can be verified by comparing relation~\eqref{norinje-tris} with the
definition of the injective norm of $\cph\altp\cpj$ (see Chapter~{3} of~\cite{Ryan}, or
Chapter~{7} of~\cite{Kubrusly}), where  the duals of $\cph$, $\cpj$ are identified via the
trace with the Banach spaces $\trc$ and $\trcjj$, respectively. Alternatively, one can use
the fact that the injective tensor product $\cph\injtp\cpj$ is isomorphic to a closed subspace
of $\bop\injtp\bopjj$ --- by virtue of Proposition~{3.2} of~\cite{Ryan} (as observed in the
discussion following this result therein), or by Theorem~{7.19} of~\cite{Kubrusly} --- because
$\cph$, $\cpj$ are (closed) subspaces of $\bop$ and $\bopjj$, respectively, where the isomorphism
is obtained by continuously extending the natural immersion of $\cph\altp\cpj$, endowed with its
natural injective norm, into $\bop\injtp\bopjj$.
\end{remark}

In the light of the preceding two remarks, with a slight abuse, we will call the norm $\norg$ the
\emph{injective norm} of $\bophj$.

By combining the previous results with fundamental duality relations between the \emph{projective}
and the \emph{injective} norms~\cite{Defant,Ryan}, we obtain the following important fact:

\begin{theorem} \label{neothe}
The norm $\norf$ on $\thptj$, defined by~\eqref{defnorf}, coincides with the projective norm $\ptrnor$,
and, in fact, for every $C\in\thptj$ we have:
\begin{align}
\|C\nof
& \defi
\sup\{|\tr(CK)|\colon K\in\alghj, \ \|K\nog=1\}
\nonumber \\
& \dodici =
\sup\{|\tr(CK)|\colon K\in\cphj, \ \|K\nog=1\}
\nonumber \\
& \dodici =
\sup\{|\tr(CK)|\colon K\in\cph\altp\cpj, \ \|K\nog=1\}
\nonumber \\
& \dodici =
\sup\{|\tr(CF)|\colon F\in\frh\altp\frj, \ \|F\nog=1\}
\nonumber \\
& \dodici =
\sup\{|\tr(CL)|\colon L\in\bophj, \ \|L\nog=1\}
\nonumber \\ \label{varexps}
& \dodici =
\sup\big\{|\Gamma(C)|\colon\Gamma\in\big(\thptj\big)^\ast,\ \|\Gamma\norps=1\big\}
= \|C\ptrn \fin .
\end{align}
Moreover, the injective norm of $\cph\altp\cpj$ coincides with the (restriction of the) norm $\norg$
defined by~\eqref{defnog}, and the isomorphism of Banach spaces
\begin{equation} \label{isoba}
\thptj\cong\Big(\cph\intp\cpj\Big)^\ast
\end{equation}
holds, where $\cph\intp\cpj\equiv\cph\injtp\cpj$ is the injective tensor product of the Banach spaces
$\cph$ and $\cpj$, i.e., the Banach space completion of the algebraic tensor product $\cph\altp\cpj$
wrt (the restriction of) the norm $\norg$. This isomorphism is implemented by the mapping
\begin{equation} \label{isdua}
\thptj\ni C\mapsto\kac\in\Big(\cph\intp\cpj\Big)^\ast,
\end{equation}
where $\kac\colon\cph\intp\cpj\rightarrow\ccc$ is a bounded linear functional determined by
$\kac(K)=\tr(CK)$, for every $K\in\cph\altp\cpj$.
\end{theorem}

\begin{proof}
The trace classes $\trc$, $\trcjj$ can be identified --- via the trace functional ---
with the Banach space duals of $\cph$ and $\cpj$, respectively; moreover, the Banach
spaces $\trc$ and $\trcjj$ have both the approximation property and the Radon-Nikod\'ym
property (recall Fact~\ref{appran}). Therefore, by a classical result --- see, e.g., the
theorem in Subsection~{16.6}, Chapter~{II} of~\cite{Defant}, or Theorem~{5.33} of~\cite{Ryan}
--- the cross trace class $\thptj$, regarded as a projective tensor product, is isomorphic to
the Banach space dual of the injective tensor product $\cph\injtp\cpj$ (the Banach space
completion of the algebraic tensor product $\cph\altp\cpj$ wrt the injective norm). Note that,
by Remark~\ref{retwonors}, the injective norm of $\cph\altp\cpj$ coincides with the restriction
of both the injective norm $\innormi$ of $\bop\altp\bopjj$ and the norm $\norg$, i.e., the
so-called injective norm of $\bophj$. The mentioned isomorphism of Banach spaces is induced by
the natural pairing --- see Subsection~{16.6}, Chapter~{II} of~\cite{Defant} ---
$\tr(SA)\tr(TB)=\tr((S\otimes T)(A\otimes B))$ between the elementary tensors $S\otimes T$ and
$A\otimes B$, with $S\in\trc=\cph^\ast$, $S\in\trcjj=\cpj^\ast$, $A\in\cph$ and $B\in\cpj$; i.e.,
by the mapping
\begin{equation}
\big(\thatj,\ptrnor\big)\ni{\textstyle \sum_k}\tre S_k\otimes T_k \mapsto
{\textstyle \sum_k}\tre \tr((S_k\otimes T_k)\argo)\in\big(\cph\altp\cpj,\norg\big)^\ast,
\end{equation}
where, with a slight abuse of notation, $\norg$ is the restriction to the algebraic tensor product
$\cph\altp\cpj$ of the injective norm of $\bophj$ (i.e., the injective norm of $\cph\altp\cpj$) and
\begin{equation}
(\cph\altp\cpj,\norg)^\ast=\Big(\cph\intp\cpj\Big)^\ast\equiv\big(\cph\injtp\cpj\big)^\ast .
\end{equation}
Therefore, denoting by $\kac$ the bounded linear functional on $\cph\intp\cpj$, corresponding
to the cross trace class operator $C$ via the isomorphism~\eqref{isdua}, we have that
\begin{equation} \label{realt}
\kac(K)=\tr(CK) \fin , \quad \forall\cinque C\in\thatj, \ \forall\cinque K\in\cph\altp\cpj \fin .
\end{equation}
Let us prove that this relation extends to \emph{every} cross trace class operator $C\in\thptj$.

Indeed, if $\{C_n\}_{n\in\nat}\subset\thatj$ is any $\ptrnor$-convergent sequence --- $\ptrnorli_n C_n=C$,
for some $C\in\thptj$ --- then
\begin{equation} \label{repro}
\kac(K)=\lim_n\kacn(K)=\lim_n\tr(C_n K)=\tr(CK) \fin , \quad
\forall\cinque K\in\cph\altp\cpj \fin .
\end{equation}
Here, the first equality holds because of the isomorphism~\eqref{isdua}, the second one by
relation~\eqref{realt}, while for the last equality we have:
$\lim_n\tr(C_n K)=\tr\big(\ttrnorli_n C_n K\big)=\tr\big(\ptrnorli_n C_n K\big)=\tr(CK)$
(the trace norm $\ttrnor$ being majorized by the projective norm $\ptrnor$).

By relation~\eqref{repro}, for every $C\in\thptj$, the following relation holds:
\begin{equation}
\|C\ptrn=\sup\{|\kac(K)|=|\tr(CK)|\colon K\in\cph\altp\cpj, \ \|K\nog=1\}\le\|C\nof \fin .
\end{equation}
Here, the inequality follows from definition~\eqref{defnorf}, taking into account that $\|K\nog\defi\|\fk\norps$.
Now, by Proposition~\ref{threno}, $\|C\ptrn\ge\|C\nof$, as well, so that, actually,
$\ptrnor=\norf$; in particular, for every $C\in\thptj$ we have:
\begin{equation} \label{thexp}
\|C\nof=\sup\{|\kac(K)|=|\tr(CK)|\colon K\in\cph\altp\cpj, \ \|K\nog=1\}=\|C\ptrn \fin .
\end{equation}

Finally, by relation~\eqref{thexp} --- and taking into account that $\alghj\supset\cphj\supset\cph\altp\cpj$,
that $\frh\altp\frj$ is $\norg$-dense in $\cph\altp\cpj$ (by the final assertion of Proposition~\ref{pronorg})
and that $\cph\altp\cpj\subset\bophj\subset(\thptj\big)^\ast$ --- it is also clear that the various equivalent
expressions~\eqref{varexps} of the norm $\norf=\ptrnor$ hold true.
\end{proof}

\begin{remark} \label{reneothe}
In the case where $C\in\thsptjs$, in the various equivalent expressions~\eqref{varexps}
of $\|C\nof=\|C\ptrn$ we can assume that the operators $K$, $F$ and $L$ therein (belonging
to their respective spaces) are \emph{selfadjoint}. Indeed, e.g., putting
$\kre\defi\frac{1}{2}(K+K^\ast)$, we have that
\begin{align}
\|C\ptrn=\|C\nof
& \defi
\sup\{|\tr(CK)|\colon K\in\alghj, \ \|K\nog=1\}
\nonumber \\
& \dodici =
\sup\{|\repa(\tr(CK))|\colon K\in\alghj, \ \|K\nog=1\}
\nonumber \\
& \dodici =
\sup\{|\tr(C\kre)|\colon K\in\alghj, \ \|K\nog=1\}
\nonumber \\
& \dodici =
\sup\{|\tr(CK)|\colon K\in\alghj, \ K=K^\ast, \ \|K\nog=1\}
\nonumber \\
& \dodici =
\sup\{\tr(CK)\colon K\in\alghj, \ K=K^\ast, \ \|K\nog=1\} \fin .
\end{align}
Here, the second equality is due to the fact that $\|K\nog=\|zK\nog$, for every $z\in\toro$,
while, for the third equality, notice that $\repa(\tr(CK))=\frac{1}{2}(\tr(CK)+\tr(CK)^\ast)=
\frac{1}{2}(\tr(CK)+\tr(CK^\ast))=\tr(C\kre)$, where we have used the fact that $C=C^\ast$ and
the cyclic property of the trace. Next, for the fourth equality, note that
$\|\kre\nog\le\frac{1}{2}(\|K\nog+\|K^\ast\nog)=\|K\nog$, which would imply that the supremum
on the third line is smaller that the supremum on the fourth line; but, comparing the latter with
the supremum on the first line, we see that the reverse inequality holds too, and we actually have
an equality. Finally, the fifth line obtains noting that $|\tr(CK)|=\max\{\tr(CK),\tr(C(-K))\}$ therein.
\end{remark}

In the light of Theorem~\ref{neothe}, we can now endow the Banach space $\thptj$ with the
$\sigma\big(\thptj,\cph\injtp\cpj\big)$-topology, i.e., the weak$^\ast$ topology pertaining
to the Banach space dual $\thptj$ of the injective tensor product $\cph\injtp\cpj$. In the
following, we will call this topology the \emph{cross $\wst$-topology} of $\thptj$ or, in short,
the \emph{$\cwst$-topology}. Since $\cph\altp\cpj$ is norm-dense in $\cph\injtp\cpj$, by
Fact~\ref{wstatop}, a \emph{$\ptrnor$-bounded} net $\{C_i\}\subset\thptj$ will converge to some
$C\in\thptj$ wrt the $\cwst$-topology iff $\lim_i\tr(C_iK)=\tr(CK)$, for all $K\in\cph\altp\cpj$;
or, equivalently, iff
\begin{equation}
\lim_i\tr(C_i(K_1\otimes K_2))=\tr(C(K_1\otimes K_2)) \fin , \quad \forall\cinque K_1\in\cph, \
\forall\cinque K_2\in\cpj \fin .
\end{equation}

\begin{remark} \label{remnocl}
Assume that $\dim(\hhjj)=\infty$. In this case, the zero vector belongs to the $\cwst$-closure
\begin{equation}
\cwstclb\equiv\cwstcl(\shjb)
\end{equation}
of the convex set $\shjb=\stahjb$. In fact, given any pair of orthonormal bases $\{\phi_m\}_{m\in\indm}$,
$\{\phi_n\}_{n\in\indn}$ in $\hh$ and $\jj$, respectively ---
where $\indm$, $\indn$ are countable index sets, at least one of which can be chosen to be $\nat$ ---
let $\{\vek\}_{k=1}^\infty$ denote the orthonormal basis in the Hilbert space $\hhjj$
obtained by any ordering of the set $\{\phi_m\otimes\psi_n\colon (m,n)\in\indm\times\indn\}$.
Then, arguing as in Example~\ref{exaclo}, we put
\begin{align}
\varsigma_l\defi l^{-1}\sum_{k=1}^l\veek
& \in
\co(\nb(\pusta,\pustajj))
\nonumber \\
& \subset
\fsta\defi\co(\nb(\sta,\stajj))\subset\shjb \fin , \quad l\in\nat \fin ,
\end{align}
and for every $A\in\trchj$ we have that
\begin{equation} \label{limitre}
\lim_l\sei\tr(\varsigma_l A)=
\lim_l\Big(l^{-1}\sei\tr\big({\textstyle (\sum_{k=1}^l\veek) A}\big)\Big)=0  \fin ,
\end{equation}
because $\lim_l \tr\big({\textstyle (\sum_{k=1}^l\veek) A}\big)=
\sum_{k=1}^\infty\langle\vek,A\sei\vek\rangle=\tr(A)$. Since the trace class $\trchj$ is
$\tnormi$-dense in $\cphj$, given any $K\in\cphj$ and $A\in\trchj$, by the estimate
\begin{align}
|\tr(\varsigma_l K)|\le |\tr(\varsigma_l (K-A))|+ |\tr(\varsigma_l A)|
& \le
\|\varsigma_l (K-A)\ttrn + \|\varsigma_l A \ttrn
\nonumber \\
& \le
\|\varsigma_l\ttrn\sei\|K-A\tnori + \|\varsigma_l\tnori\sei\|A\ttrn
\nonumber \\
& =
\|K-A\tnori+l^{-1}\sei\|A\ttrn  \fin ,
\end{align}
it follows that relation~\eqref{limitre} actually holds for every $A\in\cphj$; hence, in particular,
for every $A\in\cph\altp\cpj$. Therefore,
$\cwstli_l \varsigma_l=0\in\clcocws(\nb(\pusta,\pustajj))\subset\cwstcl(\shjb)$. By this fact, it also
follows that, for any $D\in\shjb$ and $s\in [0,1]$, putting $C_l\defi s\sei D + (1-s)\varsigma_l$,
we have that $\cwstli_l C_l=s\sei D$; hence: $\co(\shjb\cup\{0\})=\{s\sei D\colon s\in [0,1],
\ D\in\shjb\}\subset\cwstcl(\shjb)$.
\end{remark}

\begin{proposition} \label{procwst}
The convex set $\shjb$ is a
$\cwst$-precompact subset of $\thptj$, that is closed --- hence, compact --- iff $\dim(\hhjj)<\infty$.
In the case where $\dim(\hhjj)=\infty$, the $\cwst$-compact set $\cwstcl(\shjb)$ satisfies the relation
\begin{align}
\sesta\subsetneq\clcocws(\nb(\pusta,\pustajj))
& \subset
\clcocws(\shjb\cup\{0\})
\nonumber\\
& =
\cwstcl(\shjb)\subset\balp \fin ,
\end{align}
where $\balp\defi\big\{C\in\thptj\colon C\ge 0, \ \|C\ptrn\le 1\big\}$ is the convex set of all
positive cross trace class operators belonging to the unit ball of $\thptj$.
\end{proposition}

\begin{proof}
To prove the first claim, observe that $\shjb$ is a subset of the unit ball in $\thptj$,
which, by the Banach-Alaoglu theorem (Theorem~{2.6.18} of~\cite{Megginson}), is compact wrt
its weak$^\ast$ topology, i.e., the cross $\wst$-topology. Therefore, $\shjb$ is precompact wrt
the $\cwst$-topology. In the case where $\dim(\hhjj)<\infty$, the $\cwst$-topology coincides with
any norm topology on the finite-dimensional vector space $\thptj$, and hence $\shjb$ is a $\cwst$-closed
set, because it is the intersection of the norm-closed sets $\stahj$ and $\ball\big(\thptj\big)$.

Assume now that $\dim(\hhjj)=\infty$. By Remark~\ref{remnocl}, $\shjb$ is not $\cwst$-closed,
and its $\cwst$-closure contains $\clcocws(\nb(\pusta,\pustajj))$, which, since the
$\cwst$-topology is a subtopology of the $\ptrnor$-topology, contains
$\clcop(\nb(\sta,\stajj))=\sesta$ (Corollary~\ref{corconta}); moreover, $0\in\clcocws(\nb(\pusta,\pustajj))$.
Thus, $\sesta\subsetneq\clcocws(\nb(\pusta,\pustajj))\subset\clcocws(\shjb\cup\{0\})$, where the last
inclusion relation is due to the fact that $\nb(\pusta,\pustajj)\subset\shjb$. In Remark~\ref{remnocl},
we also argue that $\co(\shjb\cup\{0\})\subset\cwstcl(\shjb)$; hence, we have:
$\cwstcl(\shjb)\subset\cwstcl(\co(\shjb\cup\{0\}))=\clcocws(\shjb\cup\{0\})\subset\cwstcl(\shjb)$, where
the first inclusion relation is obvious and the equality holds by Remark~\ref{recloco}.
Let us finally prove that the compact subset $\cwstcl(\shjb)$ of the unit ball in $\thptj$ consists of
\emph{positive} operators.
To this end, it is sufficient to to observe that, if $\{D_i\}$ is a net in $\shjb$ that converges, wrt the
$\cwst$-topology, to some $C\in\thptj$ --- i.e., $\lim_i\tr(D_i K)=\tr(C K)$, for all $K\in\cph\altp\cpj$ ---
then, in particular, we have:
\begin{equation}
0\le\lim_i\sei\langle\va,D_i\tre\va\rhj=\lim_i\sei\tr(D_i\tre\vaa)=\tr(C\tre\vaa)=\langle\va,C\tre\va\rhj ,
\quad \forall\cinque\va\in\hh\altp\jj \fin ;
\end{equation}
hence, $\langle\va,C\tre\va\rhj\ge 0$, for all $\va\in\hhjj$. Thus, $\cwstcl(\shjb)\subset\balp$.
\end{proof}

According to Remark~\ref{remnocl} and Proposition~\ref{procwst}, in the case where $\dim(\hhjj)=\infty$,
the cross $\wst$-topology is too weak to be useful for our purpose of suitably characterizing the set
$\sesta$; e.g., $0\in\clcocws(\nb(\pusta,\pustajj))\supsetneq\sesta$. It will turn out that, instead,
the extended $\wst$-topology of $\trchj$ --- or, equivalently, the standard topology of $\stahj$ ---
does the job. In this regard, as it will be clear soon, a crucial point is that, for our main purposes,
we may consider the elements of the cross trace class $\thptj$ as bounded functionals on the
$\cast$-algebra $\big(\alghj,\tnormi\big)$ (that allows us to define the $\ewst$-topology of $\trchj$)
--- or, equivalently, on the normed space $(\alghj,\norg)$ --- rather than on the Banach space
$\cph\injtp\cpj$ (as it would be natural in the light of Theorem~\ref{neothe}); see Proposition~\ref{pronor}
and Remark~\ref{refuns} below.

Before moving forward in this direction, for the sake of completeness, it is worth observing that
the injective tensor product $\bop\injtp\bopjj$ can be identified
with a closed subspace of the Banach space dual of the cross trace class $\thptj$. Moreover --- in the
case where at least one of the Hilbert spaces of the bipartition $\hhjj$ is finite-dimensional --- the
norms $\ttrnor$ and $\norf=\ptrnor$ on $\trchj=\thptj$ are mutually equivalent (recall Theorem~\ref{theine}),
and, analogously, the norms $\norg$ and $\tnormi$ on $\bophj$ are equivalent; in addition, in this case we
can identify the dual space of $\thptj$ with $\bop\injtp\bopjj$, and the latter space can be realized precisely
as the renorming $(\bophj,\norg)$ of the Banach space $(\bophj,\tnormi)$. In fact, the following result holds:

\begin{theorem} \label{proreno}
The injective tensor product $\bop\injtp\bopjj$ is (isomorphic to) a closed subspace of the Banach
space dual $\big(\thptj\big)^\ast$ of $\thptj$.

Assume now that $\emme=\min\{\dim(\hh),\dim(\jj)\}<\infty$.

Then, the norms $\ttrnor$, and $\norf=\ptrnor$ are equivalent on $\trchj=\thptj$ (a set equality
being understood). Indeed, for every $C\in\trchj=\thptj$, the following inequalities hold:
\begin{equation} \label{estgen-b}
\|C\ttrn\le\|C\nof=\|C\ptrn\le\enne\sei\|C\ttrn , \quad
\mbox{where $\enne=\min\big\{4\emme,\emme^2\big\}$;}
\end{equation}
in particular, for the selfadjoint trace class operators and for the density operators on $\hhjj$,
the following stronger estimates hold:
\begin{equation} \label{estsa-b}
\|C\ttrn\le\|C\nof=\|C\ptrn\le 2\emme\sei\|C\ttrn , \quad \forall\cinque C\in\trchjsa=\thsptjs \fin ,
\end{equation}
\begin{equation} \label{estdens-b}
1=\|D\ttrn\le\|D\nof=\|D\ptrn\le\emme \fin , \quad \forall\cinque D\in\stahj=\hstahj \fin .
\end{equation}

Moreover, the algebraic, the natural (or spatial) and the injective tensor products of $\bop$ and
$\bopjj$ all coincide, as linear spaces, with $\bophj$, i.e.,
\begin{equation} \label{coinclin}
\bhabj=\bhobj=\bop\injtp\bopjj=\bophj \fin ,
\end{equation}
and the injective norm $\norg$ of $\bophj$ is equivalent to the the operator norm $\tnormi$; precisely,
for every bounded operator $L\in\bophj$,
\begin{equation} \label{estgen-bis}
\|L\nog\le\|L\tnori\le\enne\sei\|L\nog , \quad
\mbox{where $\enne=\min\big\{4\emme,\emme^2\big\}$.}
\end{equation}
The dual space $\big(\thptj\big)^\ast$ of $\thptj$ can be identified --- as a linear space and
via the trace functional --- with $\bophj=\bhabj$, endowed with the injective norm $\norg$; in
particular, every functional $\Gamma\in\big(\thptj\big)^\ast$ is of the form
\begin{equation} \label{forbofu-bis}
\Gamma(C)=\sum_{j=1}^{\je} \tr\big(C\tre (A_j\otimes B_j)\big) \fin , \quad
\forall\cinque C\in\trchj=\thptj \fin ,
\end{equation}
for some $\je$-tuples of bounded operators $\{A_j\}_{j=1}^{\je}\subset\bop$ and
$\{B_j\}_{j=1}^{\je}\subset\bopjj$, with $\je\le\ka=\emme^2$.

Finally, for every $C\in\trchj=\thptj$, we have:
\begin{equation} \label{idedua}
\|C\ptrn=\max\{|\tr(CL)|\colon L\in\bophj, \ \|L\nog=1\} \fin .
\end{equation}
{\rm (In comparison with the fourth line of~\eqref{varexps}, here we have a \emph{maximum}
rather than a \emph{supremum}.)}
\end{theorem}

\begin{proof}
Let us prove the first assertion. By Remark~\ref{retwonors}, the restriction of the norm $\norg$
to the algebraic tensor product $\bop\altp\bopjj$ coincides with the injective norm $\innormi$ of
$\bop\altp\bopjj$. Therefore, for every $L\in\bop\altp\bopjj$,
\begin{equation}
\|L\innori=\|L\nog=\|\fu\norps\defi\big\{|\tr(CL)|\colon C\in\thptj,\ \|C\ptrn=1\big\}.
\end{equation}
Hence, the map $(\bop\altp\bopjj,\innormi)\ni L\mapsto\tr(\argo L)\in\big(\thptj\big)^\ast$
is a linear isometry that extends uniquely to an isometry of $\bop\injtp\bopjj$ into
$(\thptj\big)^\ast$. Otherwise stated, $\bop\injtp\bopjj$ is isomorphic to a closed subspace of
$\big(\thptj\big)^\ast$.

Suppose that, in particular, $\emme=\min\{\dim(\hh),\dim(\jj)\}<\infty$.

Let us prove the second assertion. In fact, in the case where $\emme=\min\{\dim(\hh),\dim(\jj)\}<\infty$,
by Theorem~\ref{theine}, the linear spaces $\thptj$ and $\trchj$ coincide, and the associated norms
$\ptrnor$ and $\ttrnor$ are equivalent; specifically --- depending on which case applies --- they
satisfy relations~\eqref{estgen}, \eqref{estsa} or~\eqref{estdens}. Moreover, by Theorem~\ref{neothe},
for every trace class operator $C\in\trchj=\thptj$, we have that $\|C\ttrn\le\|C\nof=\|C\ptrn$. Hence,
relations~\eqref{estgen-b}, \eqref{estsa-b} and~\eqref{estdens-b} hold true.

To prove the third assertion, note that, by Corollary~\ref{corsemif} (recall also Remark~\ref{rebop}),
in this case the dual space $\big(\thptj\big)^\ast$ can be identified --- as a linear space --- with
$\bhabj=\bhobj=\bophj$ (and every bounded linear functional $\Gamma$ on $\thptj$ can be expressed in
the form~\eqref{forbofu-bis}), endowed with a norm that is equivalent to the operator norm $\tnormi$,
and, by definition ($\|L\nog\defi\|\fu\norps$), is precisely the norm $\norg$, which is the injective
norm of $\bhabj=\bophj$ (Remark~\ref{retwonors}). Now, since, for every $C\in\trchj=\thptj$, we have
that $\|C\ttrn\le\|C\ptrn\le\enne\sei\|C\ttrn$ --- by standard Banach space polarity arguments (see,
e.g., Sect.~{3.1} of~\cite{Guirao}) --- the corresponding (equivalent) dual norms $\tnormi$ and $\norg$
must verify relation~\eqref{estgen-bis}. Moreover, since the two norms $\tnormi$ and $\norg$ are equivalent,
the algebraic tensor product $\bhabj$ must be complete wrt the injective norm $\norg$ too; hence, we have
that $\bhabj=\bop\injtp\bopjj$ (set equality), as well, and relation~\eqref{coinclin} holds true.

Finally, by our previous identification of the dual space $\big(\thptj\big)^\ast$ with a renorming
of the Banach space $\bophj$, relation~\eqref{idedua} holds too.
\end{proof}

\subsection{The cross trace class operators as suitable functionals}
\label{functionals}

The norm $\norg$ allows us to relate the projective norm $\norf=\ptrnor$ of $\thptj$ to a norm
on a suitable class of bounded functionals on $\alghj$. Precisely, we select those functionals
of $\alghjd$ that are induced by cross trace class operators; namely, we introduce the linear
subspace $\halghj$ of $\alghjd$ defined by
\begin{equation}
\halghj\defi\big\{\digu\equiv\padigu\in\alghjd\colon C\in\thptj\big\} .
\end{equation}
with $\pail K,\digu\pair\equiv\digu(K)=\tr(CK)\in\ccc$ denoting the canonical pairing between
$K\in\alghj$ and $\digu\in\alghjd$. Recalling the isometry $\identi\colon\trchj\rightarrow\alghjd$
defined by~\eqref{defident}, we see that $\halghj=\identi\big(\thptj\big)$.

For every $C\in\thptj$, by definitions~\eqref{defnorf} and~\eqref{defnog}, and by Theorem~\ref{neothe},
we have that
\begin{align}
\|C\ptrn=\|C\nof
& =
\sup\{|\fk(C)|=|\tr(CK)|\colon K\in\alghj, \ \|\fk\norps=1\}
\nonumber \\
& = \label{defnordg}
\sup\{|\digu(K)|=|\tr(CK)|\colon K\in\alghj, \ \|K\nog=1\}\ifed\|\digu\nodg \fin .
\end{align}
We then obtain a norm $\nordg$ on the linear space $\halghj=\identi\big(\thptj\big)\subset\alghjd$
of all bounded functionals on $\alghj$ of the form $\digu\equiv\padigu$, with $C\in\thptj$; see
Proposition~\ref{pronor} below. Note that, since the norm $\norg$ is dominated by the operator norm
$\tnormi$ (and $\identi$ is an isometry),
\begin{align}
\|\digu\nodg\ge\|\digu\norps
& \defi
\sup\big\{|\digu(K)|=|\tr(CK)|=\big|\big(\identi(C)\big)(K)\big|\colon K\in\alghj, \ \|K\tnori=1\big\}
\nonumber \\
& \dodici =
\|C\ttrn \fin , \quad \forall\cinque \digu\in\halghj\subset\alghjd .
\end{align}

\begin{proposition} \label{pronor}
By setting $\|\digu\nodg\defi\sup\{|\digu(K)|=|\tr(CK)|\colon K\in\alghj, \ \|K\nog=1\}$,
for every $C\in\thptj$, we obtain a norm $\nordg$ on the linear space $\halghj\subset\alghjd$.
This norm satisfies the following relation:
\begin{equation} \label{redig}
\|C\ttrn\le\|\digu\nodg=\|C\nof=\|C\ptrn \fin , \quad \forall\cinque C\in\thptj \fin .
\end{equation}
Then, by the rhs equality above, the bijective linear map
\begin{equation} \label{liboco}
\iden\colon\big(\thptj,\ptrnor\big)\ni C\mapsto\digu\in\big(\halghj,\nordg\big)
\end{equation}
is an isometry, and, as a consequence, $\big(\halghj,\nordg\big)$ is a Banach space, isomorphic to
the cross trace class $\thptj\cong\big(\cph\injtp\cpj\big)^\ast$; precisely, it is a closed subspace
of the Banach space
\begin{equation}
\alghjgd\equiv\big(\alghj,\norg\big)^\ast,
\end{equation}
i.e., the dual of the normed space $\alghjg\equiv\big(\alghj,\norg\big)$. In particular, the norm $\nordg$
is the restriction to the subspace $\halghj$ of the dual norm $\dunorg$ of $\norg$.

In the case where $\emme=\min\{\dim(\hh),\dim(\jj)\}<\infty$, $\alghjg\equiv\big(\alghj,\norg\big)$
is --- up to normalization of the norm, i.e., replacing the norm $\norg$ with $\enne^2\norg$, where
$\enne=\min\big\{4\emme,\emme^2\big\}$ --- a (unital) Banach algebra, which is a renorming of the
$\cast$-algebra $\big(\alghj,\tnormi\big)$. Moreover, in this case, $\halghj=\identi(\trchj)$
(set equality), and the Banach subspace $\big(\halghj,\nordg\big)$ of
$\alghjgd\equiv\big(\alghj,\norg\big)^\ast$ is a renorming of the closed subspace
\begin{equation}
\identi(\trchj)\cong\trchj\cong\cphj^\ast
\end{equation}
of $\alghjd$; specifically, for every $C\in\trchj=\thptj$ (set equality), we have:
\begin{equation} \label{relreno}
\|\digu\norps=\|C\ttrn\le\|\digu\nodg=\|C\nof=\|C\ptrn\le\enne\sei\|C\ttrn=\enne\sei\|\digu\norps
\fin ,
\end{equation}
where $\|\digu\norps$ is the norm of $\digu$ as a functional on $\big(\alghj,\tnormi\big)$ and
$\enne=\min\big\{4\emme,\emme^2\big\}$.
\end{proposition}

\begin{proof}
Since, by~\eqref{defnordg}, $\|\digu\nodg=\|C\nof=\|C\ptrn$, $C\in\thptj$, and, moreover, the mapping
$\thptj\ni C\mapsto\digu\equiv\padigu=\identi(C)\in\halghj=\identi\big(\thptj\big)\subset\alghjd$ is a
linear bijection, we conclude that $\nordg\colon\halghj\rightarrow\errep$ is indeed a norm, and the
bijective linear map~\eqref{liboco} is an isometry. Note that, by~\eqref{inenorm}, relation~\eqref{redig}
holds. By~\eqref{defnordg}, it is also clear that $\big(\halghj,\nordg\big)$ is a (closed) subspace of the
Banach space $\big(\alghj,\norg\big)^\ast$, and the norm $\nordg$ is precisely the restriction to
the subspace $\halghj$ of the dual norm $\dunorg$ of $\norg$.

Assume now that at least one of the Hilbert spaces $\hh$ and $\jj$ is finite-dimensional,
i.e., that $\emme=\min\{\dim(\hh),\dim(\jj)\}<\infty$. Then, by the third assertion of
Theorem~\ref{proreno}, the injective norm $\norg$ of $\bophj$ is equivalent to the operator
norm $\tnormi$. Hence, the restrictions of these norms to $\alghj$ are equivalent norms too, so
that $\alghjg\equiv\big(\alghj,\norg\big)$ is a Banach space, because $\big(\alghj,\tnormi\big)$
is a Banach space. Specifically, $\alghjg$ is --- up to normalization of its norm --- a (unital)
Banach algebra since $\|K_1\tre K_2\nog\le\|K_1\tnori\tre\|K_2\tnori\le\enne^2\sei\|K_1\nog\tre\|K_2\nog$,
with $\enne=\min\big\{4\emme,\emme^2\big\}$, and hence $\big(\alghj,\enne^2\norg\big)$ is a
unital Banach algebra and a renorming of the $\cast$-algebra $\big(\alghj,\tnormi\big)$. Moreover,
by the second assertion of Theorem~\ref{proreno}, we have that $\thptj=\trchj$ (set equality),
and the norms $\ttrnor$ and $\norf=\ptrnor$ on $\trchj$ are equivalent (precisely, they satisfy
relation~\eqref{estsa-b}). The Banach space $\trchj\equiv\big(\trchj,\ttrnor\big)$ can be
identified, via the trace, with the dual of $\cphj$, and the Banach space $\trchj\cong\cphj^\ast$
can be identified with a closed subspace of $\alghjd$; in fact, recalling Remark~\ref{rehaba},
we have: $\trchj\cong\cphj^\ast\cong\identi(\trchj)$. In conclusion, if
$\min\{\dim(\hh),\dim(\jj)\}<\infty$, then $\halghj=\identi\big(\thptj\big)=\identi(\trchj)$, and
$\big(\halghj=\identi(\trchj),\nordg\big)\cong\big(\trchj,\norf=\ptrnor\big)$ is a renorming of
the closed subspace $\identi(\trchj)\cong\trchj\cong\cphj^\ast$ of $\alghjd$; precisely,
by~\eqref{estsa-b} and by the fact that $\identi\colon\trchj\rightarrow\alghjd$ is an isometry,
for every $C\in\trchj=\thptj$, relation~\eqref{relreno} holds.
\end{proof}

We stress that the isomorphism of Banach spaces~\eqref{liboco} is a map restriction (wrt both the domain
and the codomain) of the linear isometry $\identi\colon\trchj\rightarrow\alghjd$, see~\eqref{defident}.

\begin{remark} \label{redigub}
For every $C\in\thptj$ the linear functional $\digu\in\halghj\subset\alghjgd$ admits a natural extension
to a bounded linear functional $\digub$ on the normed space $\bophjg\equiv(\bophj,\norg)$, defined by
$\digub(L)\defi\tr(CL)$, for all $L\in\bophj$. In fact, the functional $\digub$ is bounded, because
$|\tr(CL)|\le\|C\ttrn\sei\|L\tnori\le\|C\ptrn\sei\|L\tnori$, and, moreover, by relation~\eqref{varexps},
$\|\digub\norps\defi\sup\{|\tr(CL)|\colon L\in\bophj, \ \|L\nog=1\}=\|C\ptrn=\|\digu\nodg=\|\digu\dunog$.
\end{remark}

\begin{remark} \label{refuns}
We are now regarding the elements of the cross trace class $\thptj$ as bounded functionals
on the normed space $\alghjg\equiv(\alghj,\norg)$, rather than on the Banach space $\cph\injtp\cpj$
(as it would be natural according to Theorem~\ref{neothe}), or on the $\cast$-algebra
$\alghj\equiv\big(\alghj,\tnormi\big)$ (that allows us to introduce the $\awst$-topology). However,
by Fact~\ref{fanonotri}, with the obvious identifications $\bsb\equiv\alghj$, $\no\equiv\norg$ and
$\notri\equiv\tnormi$, we have that $\alghjgd$ is, in a natural way, a linear subspace of $\alghjd$,
and the $\sigma\big(\alghjgd,\alghj\big)$-topology --- or, with a more precise notation, the
$\sigma\big(\alghjgd,\alghjg\big)$-topology --- coincides with the subspace topology on
$\alghjgd\subset\alghjd$ wrt the $\sigma(\alghjd,\alghj)$-topology, i.e., the $\awst$-topology.
Therefore, in the following, we may think of the cross states $\hstahj$ as functionals in $\halghj$
(i.e., as elements of the set $\iden(\hstahj)\subset\halghj$) or, equivalently, as functionals in
$\alghjd$ (i.e., as elements of the set $\identi(\stahj)\subset\alghjd$), because the $\awst$-topology
will turn out to be the relevant one.
\end{remark}

\subsection{The main result: characterization of separable states}
\label{main}

Let us now further consider the intersection $\shjb$ of the convex set $\hstahj$ of all cross states
with the unit ball in $\thptj$; see~\eqref{basph}. As previously noted, $\shjb$ is a convex set and,
by the the fact that $\norf=\ptrnor$, by third equality in~\eqref{basph} and by Corollary~\ref{corimplise},
we have:
\begin{equation} \label{relshj}
\shjb\defi\big\{D\in\hstahj\colon\|D\nof=\|D\ptrn\le 1\big\}=\stahjb\supset\sesta \fin .
\end{equation}
By the definition of the norms $\norf$ and $\nordg$ (see~\eqref{defnorf} and~\eqref{defnordg},
respectively), and by the fact that the norm $\nordg$ is the restriction of the dual norm
$\dunorg$ of $\norg$ to the Banach subspace $\halghj$ of $\alghjgd$ (Proposition~\ref{pronor}),
we have:
\begin{align}
\shj
& \defi
\iden(\shjb)
\nonumber \\
& \dodici =
\big\{\xid\in\iden(\hstahj)\colon\|D\ptrn\le 1\big\}
\nonumber \\
& \dodici =
\big\{\xid\in\iden(\hstahj)\colon\|\xid\nodg=\|\xid\dunog\le 1\big\}
\nonumber \\
& \dodici =
\big\{\xid\in\iden(\hstahj)\colon\mbox{$|\xid(K)|\defi|\tr(DK)|\le\|K\nog$,
$\forall K\cinque\in\alghj$}\big\}
\nonumber \\  \label{defishj}
& \dodici \subset
\halghj\subset\alghjgd\subset\alghjd ,
\end{align}
where, for the last inclusion relation, recall Remark~\ref{refuns}.

We will also consider the following convex subset of $\alghjd$:
\begin{align}
\shjt
& \defi
\big\{\xi\in\alghjd\colon\mbox{$\xi\ge 0$, $\xi(I)=1$, and $|\xi(K)|\le\|K\nog$,
$\forall\cinque K\in\alghj$}\big\}
\nonumber\\  \label{defshjt}
& \hspace{1.2mm} =
\big\{\xi\in\alghjgd\colon\mbox{$\xi\ge 0$, $\xi(I)=1$ and $\|\xi\dunog\le 1$}\big\}
\supset\shj \fin .
\end{align}
Here, $\xi\ge 0$ (positivity) means that $\xi(K)\ge 0$, for all positive $K\in\alghj$.
Clearly, $\shjt$ is a convex subset of the convex set of all normalized, positive
functionals on the $\cast$-algebra $\alghj$. Regarding definition~\eqref{defshjt}, we
stress the following fact. For every $\xi\in\alghjd$, the conditions $\xi\ge 0$ and
$\xi(I)=1$ imply that $\|\xi\norps\defi\sup\{|\xi(K)|\colon K\in\alghj,\ \|K\tnori=1\}=1$
--- see, e.g., Corollary~{3.3.4} of~\cite{Murphy} --- but not, in general, that
$|\xi(K)|\le\|K\nog$, for all $K\in\alghj$ (or even that $\xi\in\alghjgd$, in the case
where $\dim(\hh)=\dim(\jj)=\infty$). Indeed, the latter condition is satisfied precisely
by those (positive) functionals on $\alghj$ that are contained in the unit ball of the
linear subspace $\alghjgd$ of $\alghjd$.

In the following, for the sake of notational conciseness, it will be convenient to identify
the set $\shjb\subset\hstahj$ with the convex subset $\shj$ of $\alghjd$; i.e., with the image,
via the isomorphism of Banach spaces $\iden\colon\thptj\rightarrow\halghj$ (see~\eqref{liboco}),
or via the isometry $\identi$ (see~\eqref{defident}), of the convex set $\shjb$. More generally,
(any subset of) $\trchj$ will be identified with (a corresponding subset of) the linear subspace
$\identi(\trchj)$ of $\alghjd$.

At this point, it is also worth recalling that the $\ewst$-topology on $\trchj$ is defined as
the initial topology induced by the map $\identi\colon\trchj\rightarrow\alghjd$, $\alghjd$
being endowed with the $\awst$-topology (see Subsection~\ref{algebra}).

\begin{remark}
By Fact~\ref{wstatop-bis} --- with the identifications: $\bs\equiv\alghjg$, regarded as a
normed space endowed with the norm $\norg$, which is obviously a $\norg$-dense linear subspace
of its Banach space completion $\bsb\equiv\alghjgc$, and  $\mathfrak{B}\equiv\shjt$, which is
a $\dunorg$-bounded subset of $\bs^\ast\equiv\alghjgd=\alghjgcd\equiv\bsb^\ast$ --- we conclude
that the relative topology on $\shjt$ wrt the $\sigma\big(\alghjgd,\alghjgc\sei\big)$-topology
coincides with the relative topology on $\shjt$ wrt the $\sigma(\alghjgd,\alghjg)$-topology;
i.e., by Remark~\ref{refuns}, with the relative topology on $\shjt$ wrt the $\awst$-topology
of $\alghjd$. All of which is to say that, as far as topological arguments are concerned in
the following, the functionals in $\shjt$ can be regarded, interchangeably, as bounded functionals
on the spaces $\alghj$ (endowed with the norm $\tnormi$), $\alghjg$ or $\alghjgc$.
\end{remark}

To prove the main theorem, it remains to establish a further technical fact.

\begin{lemma} \label{lemcoco}
The convex set $\shj$ is a precompact subset of $\alghjd$ wrt the $\awst$-topology, and, in fact,
with the obvious identifications via the immersion map $\identi$, we have that
\begin{align}
\shj\equiv\shjb
& \subset
\clcow(\nb(\pusta,\pustajj))
\nonumber \\
& \equiv
\identi\big(\clcow(\nb(\pusta,\pustajj))\big)
\nonumber \\ \label{relaw}
& \subset
\clcoaw(\nb(\pusta,\pustajj))=\shjt \fin ,
\end{align}
where the $\awst$-closed convex hull $\clcoaw(\nb(\pusta,\pustajj))=\shjt$ of the set
\begin{equation}
\nb(\pusta,\pustajj)\equiv\identi(\nb(\pusta,\pustajj))
\end{equation}
is a $\awst$-compact convex subset of $\alghjd$. Moreover, identifying $\trchj$, $\stahj$ with
$\identi(\trchj)$ and $\identi(\stahj)$, respectively, we have:
\begin{align}
\clcow(\nb(\pusta,\pustajj))
& \equiv
\identi\big(\clcow(\nb(\pusta,\pustajj))\big)
\nonumber \\
& =
\clcoaw(\nb(\pusta,\pustajj))\cap\trchj
\nonumber \\  \label{alsorec}
& =
\clcoaw(\nb(\pusta,\pustajj))\cap\stahj
=\shjt\cap\stahj \fin .
\end{align}
\end{lemma}

\begin{proof}
As previously noted (see~\eqref{defishj} and~\eqref{defshjt}), $\shj\defi\iden(\shjb)=\identi(\shjb)$
is contained in the convex set of all \emph{states} $\xi$ of the $\cast$-algebra $\alghj$ such that
$|\xi(K)|\le\|K\nog$, for all $K\in\alghj$; i.e., we have:
\begin{align}
\shj
& =
\big\{\xid\in\iden(\hstahj)\colon\mbox{$|\xid(K)|\defi|\tr(DK)|\le\|K\nog$,
$\forall K\cinque\in\alghj$}\big\}
\nonumber \\
& =
\{\xid\in\alghjd\colon\mbox{$D\in\hstahj$, $|\xid(K)|\le\|K\nog$,
$\forall K\cinque\in\alghj$}\}
\nonumber \\
& \subset
\{\xi\in\alghjd\colon\mbox{$\xi\ge 0$, $\xi(I)=1$, and $|\xi(K)|\le\|K\nog$,
$\forall K\cinque\in\alghj$}\}
\nonumber \\
& \ifed
\shjt \fin .
\end{align}

By Theorem~{6.3} of~\cite{Arveson} --- also see the proof of this theorem {\it ibidem},
and suitably define the subset $V$ of $\hhjj$ therein as $V\defi\{\phi\otimes\psi\colon
\phi\in\hh, \ \psi\in\jj, \ \|\phi\|=\|\psi\|=1\}$ --- the convex subset $\shjt$ of
$\alghjd$ is $\awst$-compact and coincides with the $\awst$-closed convex hull of the set
$\nb(\pusta,\pustajj)\equiv\identi(\nb(\pusta,\pustajj))$, so that
\begin{equation} \label{crucincl}
\shj\subset\shjt=\clcoaw\big(\identi(\nb(\pusta,\pustajj))\big) \fin .
\end{equation}
Thus, the $\awst$-closure of the set $\shj$ is contained in the $\awst$-compact set $\shjt$;
hence, it is $\awst$-compact too, and $\shj$ is a precompact subset of $\alghjd$ wrt the
$\awst$-topology.

Now, by Fact~\ref{relclo}, $\clcow(\nb(\pusta,\pustajj))=\clcoaw(\nb(\pusta,\pustajj))\cap\trchj$, and
actually, by Lemma~\ref{lemthclcow}, $\clcow(\nb(\pusta,\pustajj))=\clcoaw(\nb(\pusta,\pustajj))\cap\stahj$,
so that relation~\eqref{alsorec} holds true. To complete the proof, just note that
$\shjb\subset\clcow(\nb(\pusta,\pustajj))$.

In fact, by relations~\eqref{alsorec} and~\eqref{crucincl}, we argue that
\begin{align}
\shj\equiv\shjb
&=
\shjb\cap\stahj
\nonumber \\
& \subset
\shjt\cap\stahj
\nonumber\\
& =
\clcoaw(\nb(\pusta,\pustajj))\cap\stahj= \clcow(\nb(\pusta,\pustajj)) \fin ,
\end{align}
all obvious identifications via the immersion map $\identi$ being understood. Hence, we finally obtain
relation~\eqref{relaw}.
\end{proof}

At this point, eventually, the convex set $\sesta$ can be characterized as the intersection of the
convex set $\hstahj$ of all cross states with the unit ball wrt the norm $\ptrnor=\norf$.

\begin{theorem} \label{chasest}
The set $\sesta$ admits the following characterization:
\begin{align}
\sesta=\stahjb
& =
\shjb
\nonumber \\  \label{relsesta}
& =
\iden^{-1}(\shj)=\big\{D\in\hstahj\colon\|\xid\dunog=1\big\} \fin .
\end{align}
Moreover, $\sesta\equiv\iden(\sesta)=\identi(\sesta)=\shj$ is a $\awst$-precompact subset of $\alghjd$, because
\begin{align}
\sesta=\clcow(\nb(\pusta,\pustajj))
& =
\clcoaw(\nb(\pusta,\pustajj))\cap\trchj
\nonumber \\
& =
\clcoaw(\nb(\pusta,\pustajj))\cap\stahj
\nonumber \\ \label{furels}
& =
\shjt\cap\stahj \fin
\end{align}
--- the obvious identifications $\nb(\pusta,\pustajj)\equiv\identi(\nb(\pusta,\pustajj))$,
$\trchj\equiv\identi(\trchj)$ and $\stahj\equiv\identi(\stahj)$ being understood --- where
\begin{equation}
\clcoaw(\nb(\pusta,\pustajj))=\shjt
\end{equation}
is a $\awst$-compact convex subset of $\alghjd$.
\end{theorem}

\begin{proof}
Taking into account relation~\eqref{relshj}, to prove~\eqref{relsesta} it is sufficient to show
that the reverse inclusion $\stahjb=\shjb=\iden^{-1}(\shj)\subset\sesta$ holds too. In fact, by
Lemma~\ref{lemcoco} (see, in particular, relation~\eqref{relaw}), the convex set $\shj$ is contained
in the $\ewst$-closed convex hull of the set $\nb(\pusta,\pustajj)$,
so that
\begin{equation}
\shj\equiv\shjb\subset\clcow(\nb(\pusta,\pustajj))=\sesta \fin .
\end{equation}
Here, the equality after the inclusion relation holds by Theorem~\ref{thclcow}, and
natural identifications via the linear isometry $\identi\colon\trchj\rightarrow\alghjd$
--- or, equivalently, via the isomorphism of Banach spaces $\iden$ --- are understood.
At this point, regarding the last equality in relation~\eqref{relsesta}, just note that
$\shj\defi\iden(\shjb)=\iden(\sesta)\equiv\sesta$, and hence
\begin{align}
\sesta
& =
\iden^{-1}(\shj)
\nonumber \\
& =
\big\{D\in\hstahj\colon\|\xid\nodg\le 1\big\}
\nonumber \\
& =
\big\{D\in\hstahj\colon\|\xid\nodg=1\big\}=\big\{D\in\hstahj\colon\|\xid\dunog=1\big\},
\end{align}
where,
for the last two equalities, recall that $1=\|D\ttrn\le\|D\ptrn=\|\xid\nodg=\|\xid\dunog$
(Proposition~\ref{pronor}). By relation~\eqref{alsorec}, the characterization~\eqref{furels}
of $\sesta=\clcow(\nb(\pusta,\pustajj))$ follows immediately. Finally, recall that,
by Lemma~\ref{lemcoco}, $\clcoaw(\nb(\pusta,\pustajj))=\shjt$ --- with $\nb(\pusta,\pustajj)
\equiv\identi(\nb(\pusta,\pustajj))$ --- is a \emph{$\awst$-compact} convex subset of $\alghjd$.
\end{proof}

\begin{corollary} \label{chasest-bis}
The set $\sesta$ admits the following further characterization, in terms of the Hermitian
projective norm:
\begin{equation} \label{relsesta-bis}
\sesta=\big\{D\in\hstahj\colon\ptrntr{D}=1\big\} \fin .
\end{equation}
\end{corollary}

\begin{proof}
Observe, indeed, that
\begin{align}
\sesta
& \subset
\stahj\cap\sph\big(\thsptjs\big)
\nonumber \\
& =
\big\{D\in\hstahj\colon\ptrntr{D}=1\big\}
\nonumber \\
& =
\stahj\cap\ball\big(\thsptjs\big)
\nonumber \\
& \subset
\stahj\cap\ball\big(\thptj\big)
=\stahjb=\sesta \fin .
\end{align}
In the previous argument, the first inclusion relation, follows from Corollary~\ref{corimplise},
and the second equality is a direct consequence of the fact that $\ptrntr{D}\ge 1$, for all
$D\in\hstahj$; whereas the second inclusion relation follows from the fact that
$\ball\big(\thsptjs\big)\subset\ball\big(\thptj\big)$, because the projective norm $\ptrnor$
is dominated on $\thsptjs$ by the Hermitian projective norm $\ptrnortr$. Finally, we have used
relation~\eqref{relsesta} in Theorem~\ref{chasest}.
\end{proof}

\subsection{The cross norm criterion in any Hilbert space dimension}
\label{ecnc}

It is now worth collecting some of our main findings concerning the characterization
of separable states in a unique statement. As previously noted, in the case where
$\min\{\dim(\hh),\dim(\jj)\}<\infty$ our results can be compared with the results
obtained by Rudolph~\cite{Rudolph,Rudolph-further} (who considered the finite-dimensional
setting only) and by Arveson~\cite{Arveson} (who considered multipartite systems, as well).
Here, we have extended these results to include the genuinely infinite-dimensional setting
--- i.e., the case where $\dim(\hh)=\dim(\jj)=\infty$ --- and therefore we call the associated
separability criterion the \emph{Extended Cross Norm Criterion} (ECNC).

\begin{theorem}[The Extended Cross Norm Criterion]  \label{xextcnc}
Let us suppose, at first, that at least one of the Hilbert spaces $\hh$, $\jj$
is finite-dimensional: $\emme\equiv\min\{\dim(\hh),\dim(\jj)\}<\infty$. Then, we have that
\begin{equation} \label{xcoispa}
\trchj=\thptj=\thatj \fin ,
\end{equation}
\begin{equation} \label{xcoispa-bis}
\trchjsa=\thsptjs=\thsatjs  \quad \mbox{and} \quad
\stahj=\hstahj
\end{equation}
--- where all equalities have to be understood as set equalities --- and the norms $\ttrnor$
(trace norm) and $\ptrnor$ (projective norm) are equivalent on $\trchj$, whereas their restrictions
to the real Banach space $\trchjsa$ are equivalent to the Hermitian projective norm $\ptrnortr$.
These norms, evaluated on the density operators $\stahj$ --- i.e., the $\ttrnor$-normalized positive
trace class operators --- satisfy the relations
\begin{equation} \label{xrenors}
1=\|D\ttrn\le\|D\ptrn\le\emme \quad \mbox{and} \quad
1\le\|D\ptrn\le\ptrntr{D}\le2\tre\|D\ptrn\le 2\emme \fin , \quad \forall\cinque D\in\stahj \fin .
\end{equation}

Moreover, given a density operator $D\in\stahj$, the following facts are equivalent:
\begin{itemize}

\item $D\in\sesta$, i.e., $D$ is a separable state.

\item $\|D\ptrn=1$, i.e., $D$ is normalized wrt the projective norm $\ptrnor$.

\item $\ptrntr{D}=1$, i.e., $D$ is normalized wrt the Hermitian projective norm $\ptrnortr$.

\end{itemize}

In particular, if $\enne\equiv\dim(\hhjj)<\infty$, then every separable density operator
$\varsigma\in\sesta$ admits a finite convex decomposition of the form
\begin{equation} \label{xappcara}
\varsigma=\sum_{k=1}^\mathsf{K}\tre p_k\tre (\pi_k\otimes\varpi_k) \fin , \
\mbox{where $\mathsf{K}\le\enne^2$, $p_k>0$ $\big(\sum_{k=1}^\mathsf{K}\tre p_k=1\big)$,
$\pi_k\in\pusta$, $\varpi_k\in\pustajj$}.
\end{equation}

Suppose now that $\dim(\hh)=\dim(\jj)=\infty$. Then, we have that
\begin{equation} \label{xnoncoispa}
\trchj\supsetneq\thptj\supsetneq\thatj \fin ,
\end{equation}
\begin{equation} \label{xnoncoispa-bis}
\trchjsa\supsetneq\thsptjs\supsetneq\thsatjs \quad \mbox{and} \quad
\stahj\supsetneq\hstahj \fin ,
\end{equation}
--- where all relations have to be understood merely as strict set containments --- and the norms
$\ptrnor$ and $\ptrnortr$ are unbounded on the convex set $\hstahj$ of all cross states on $\hhjj$.

Moreover, given a density operator $D\in\stahj$, the following facts are equivalent:
\begin{itemize}

\item $D\in\sesta$, i.e., $D$ is a separable state.

\item $D\in\hstahj$ and $\|D\ptrn=1$, i.e., $D$ is a cross state, and normalized wrt the
projective norm $\ptrnor$.

\item $D\in\hstahj$ and $\ptrntr{D}=1$, i.e., $D$ is a cross state, and normalized wrt the
Hermitian projective norm $\ptrnortr$.

\end{itemize}

Finally, for any dimension of the Hilbert spaces $\hh$ and $\jj$, and for every entangled cross state
$D\in\hstahj$ --- i.e., for every $D\in\hstahj\setminus\sesta$ --- $\ptrntr{D}\ge\|D\ptrn>1$.

\end{theorem}

\begin{proof}
For the first assertion of the theorem, see Theorem~\ref{theine}, and, for the
inequalities~\eqref{xrenors}, see, in particular, relations~\eqref{estdens}
and~\eqref{estdens-bis} therein.
For the equivalence of all the various characterizations of the set $\sesta$ of
separable states, see Theorem~\ref{chasest} and Corollary~\ref{chasest-bis},
that hold for any dimension of the Hilbert spaces $\hh$, $\jj$ of the bipartition.
In particular, for $\enne\equiv\dim(\hhjj)<\infty$, see Corollary~\ref{fidiset}.
Focusing now on the case where $\dim(\hh)=\dim(\jj)=\infty$, relations~\eqref{xnoncoispa}
and~\eqref{xnoncoispa-bis} follow from Corollary~\ref{corsubne}, and, for the assertion
concerning the unboundedness of the norms $\ptrnor$ and $\ptrnortr$ on $\hstahj$,
see Corollary~\ref{unbonor}.
Finally, with no assumption on the dimension of the Hilbert spaces $\hh$ and $\jj$, for every
cross state $D\in\hstahj$, $\ptrntr{D}\ge\|D\ptrn\ge 1$ (Corollary~\ref{corimplise}), and,
for every \emph{entangled} cross state $D\in\hstahj$, $\|D\ptrn\neq 1$; hence, for every
$D\in\hstahj\setminus\sesta$, $\ptrntr{D}\ge\|D\ptrn>1$.
\end{proof}

\section{The projective norm and the entanglement function}
\label{entanglement}

For every density operator $D\in\stahj$, let us define the quantity
\begin{equation}  \label{dentfun}
\ent(D)\defi\sup\{|\tr(DL)|\colon L\in\bophj, \ \|L\nog=1\} \fin ,
\end{equation}
where $\norg$ is the injective norm of $\bophj$; see~\eqref{fornog}. Considering the case where $L=I$
on the rhs of~\eqref{dentfun}, we see that $1\le\ent(D)(\le +\infty)$. We then obtain an extended
real-valued function
\begin{equation} \label{entfun}
\ent\colon\stahj\ni D\mapsto\ent(D)\in[1,+\infty] \fin ,
\end{equation}
that will be called the \emph{entanglement function}. We will prove that the minimum value $\ent(D)=1$
is actually attained by the function $\ent$ on certain density operators $D$, and --- in the genuinely
infinite-dimensional setting where $\dim(\hh)=\dim(\jj)=\infty$ --- the value $+\infty$ is attained too.
Therefore, in the following, where appropriate, standard arithmetics and ordering of the extended real
number system should be understood. In particular, as usual in the context of convex analysis, we will
adopt the following conventions: $(+\infty)+(+\infty)=+\infty$, $0\sei(+\infty)=0$ and, for any $r>0$,
$r\sei(+\infty)=+\infty$.
The reader will also notice that, for every cross state  $D\in\hstahj$, $\ent(D)$ is precisely the norm
$\|\digud\norps=\|D\ptrn$ of the bounded linear functional $\digud\colon\bophjg\ni L\mapsto\tr(DL)\in\ccc$,
where $\bophjg$ denotes the normed space $\big(\bophj,\norg\big)$ (recall Remark~\ref{redigub}).
Specifically --- in the case where $\min\{\dim(\hh),\dim(\jj)\}<\infty$, and hence, by Theorem~\ref{theine},
$\stahj=\hstahj$ --- the entanglement function $\ent$ is just the restriction to $\stahj$ of the
projective norm $\ptrnor$. If $\dim(\hh)=\dim(\jj)=\infty$, instead, $\hstahj\subsetneq\stahj$
(Corollary~\ref{corsubne}), and $\ent$ coincides with $\ptrnor$ on the cross states $\hstahj$,
whereas it takes the value $+\infty$ on all those states that are \emph{not} cross states.
In fact, the following result holds:

\begin{theorem} \label{theoentfu}
The extended real-valued entanglement function $\ent\colon\stahj\rightarrow[1,+\infty]$
satisfies the following properties:
\begin{enumerate}[label=\tt{(E\arabic*)}]

\item \label{enta}
it is a convex function, i.e.,
$\ent(t D_1+ (1-t) D_2)\le t\sei\ent(D_1) + (1-t)\sei \ent(D_2)$, for all
$t\in[0,1]$ and $D_1, D_2\in\stahj$;

\item \label{entb}
$\ent(D)=\|D\ptrn$, for all $D\in\hstahj$;

\item \label{entc}
$\ent(D)=+\infty$, for all $D\in\stahj\setminus\hstahj$, i.e., for all density operators $D$
on $\hhjj$ that are not cross states (which form a nonempty set iff $\dim(\hh)=\dim(\jj)=\infty$);

\item \label{entd}
$\ent$ is lower semicontinuous wrt the standard topology of $\stahj$;

\item \label{ente}
$\ent(D)=1$ if $D\in\sesta$ {\rm ($D$ separable)}, whereas $\ent(D)>1$ if $D\in\stahj\setminus\sesta$
{\rm ($D$ entangled)};

\item \label{entf}
$\ent(D)=\sup\{\tr(DF)\colon F\in\frh\altp\frj, \ F=F^\ast, \ \|F\nog=1\}$;

\item \label{entg}
for every $D\in\stahj$, and every pair $U$, $V$ of unitary operators on $\hh$ and $\jj$, respectively,
\begin{equation} \label{invlocuni-bis}
\ent(D)=\ent((U\otimes V)\tre D\tre (U\otimes V)^\ast) \fin ;
\end{equation}

\item \label{enth}
for every $r\in[1,+\infty)$, the sublevel set $\stahjr\defi\{D\in\stahj\colon\ent(D)\le r\}$ of $\ent$ is
closed in $\stahj$, wrt the standard topology of $\stahj$.

\end{enumerate}
\end{theorem}

\begin{proof}
Property~\ref{enta} holds by the fact that, for every $t\in[0,1]$, every $L\in\bophj$ and any
$D_1,D_2\in\stahj$, we have: $|\tr((t D_1+ (1-t) D_2)L)| \le t\sei|\tr(D_1L)| + (1-t)\sei|\tr(D_2L)|$;
hence, if $D=t D_1+ (1-t) D_2$, then
\begin{align}
\ent(D)
& \defi
\sup\{|\tr(DL)|\colon L\in\alghj, \ \|L\nog=1\}
\nonumber\\
& \dodici \le
\sup\{t\sei|\tr(D_1L)| + (1-t)\sei|\tr(D_2L)|\colon L\in\bophj, \ \|L\nog=1\}
\nonumber\\
& \dodici \le
t\sei\ent(D_1) + (1-t)\sei \ent(D_2)
\fin .
\end{align}

Next, by definition~\eqref{dentfun}, and by Theorem~\ref{neothe} (see relation~\eqref{varexps}),
property~\ref{entb} holds too.

Let us now prove~\ref{entc}. We reason by contradiction. Supposing that $\dim(\hh)=\dim(\jj)=\infty$
--- otherwise, by Theorem~\ref{theine}, $\stahj=\hstahj$ --- let $D\in\stahj\setminus\hstahj$, and
assume that $\ent(D)<+\infty$. Then, since $\cph\altp\cpj\subset\bophj$, we have:
\begin{align}
+\infty>\ent(D)
& \defi
\sup\{|\tr(DL)|\colon L\in\bophj, \ \|L\nog=1\}
\nonumber\\
& \dodici \ge
\sup\{|\tr(DK)|\colon K\in\cph\altp\cpj, \ \|K\nog=1\} \fin .
\end{align}
Therefore, the mapping $\cph\altp\cpj\ni K\mapsto\tr(DK)\in\ccc$ is a bounded functional on the normed
space $(\cph\altp\cpj,\norg)$, that (uniquely) extends to a bounded functional on the the Banach space
completion of this normed space, namely, to the injective tensor product $\cph\injtp\cpj$. By
Theorem~\ref{neothe} --- see~\eqref{isoba} --- the isomorphism of Banach spaces
$\thptj\cong\big(\cph\injtp\cpj\big)^\ast$ holds, this isomorphism being implemented by the map $\kac$
(definition~\eqref{isdua}). Thus, there must be a cross trace class operator $C\in\thptj$ such that
$\tr(DK)=\kac(K)=\tr(CK)$, for all $K\in\cph\altp\cpj$. But this conclusion would imply that, actually,
$D=C$  --- e.g., because $\tr(D(\phikl\otimes\psimn))=\tr(C(\phikl\otimes\psimn))=
\langle\phi_l\otimes\psi_n, C(\phi_k\otimes\psi_m)\rangle$, where $\{\phi_k\}_{k\in\nat}$,
$\{\phi_m\}_{m\in\nat}$ are orthonormal bases in $\hh$ and $\jj$, respectively --- which leads us to
a contradiction wrt our initial assumption that $D\not\in\thptj$.

To prove~\ref{entd}, it is sufficient to note that the entanglement function
$\ent\colon\stahj\rightarrow[1,+\infty]$ is defined as the pointwise supremum over a collection of
functions of the form $\stahj\ni D\mapsto|\tr(DL)|\in\erre$, $L\in\bophj$. Since $\bophj=\trchj^\ast$,
all these functions are continuous --- hence, lower semicontinuous --- wrt the relative topology on
$\stahj$ induced by the (trace) norm topology of $\trchj$, i.e., wrt the standard topology of $\stahj$
(Fact~\ref{statop}). Therefore, by a well-known result (see, e.g., Proposition~{7.11} of~\cite{Folland-RA}),
$\ent$ is lower semicontinuous wrt the standard topology of $\stahj$.

Let us prove~\ref{ente}. By Theorem~\ref{chasest}, $\sesta=\stahjb$, i.e., for $D\in\hstahj$,
$\ent(D)=\|D\ptrn=1$ iff $D\in\sesta$. Therefore, $\ent(D)>1$ iff $D\in\stahj\setminus\sesta$;
in particular,  $\ent(D)=\|D\ptrn>1$, if $D\in\hstahj\setminus\sesta$ (see Theorem~\ref{xextcnc}), and
$\ent(D)=+\infty$, if $D\in\stahj\setminus\hstahj$ (by property~\ref{entc} of $\ent$).

We next prove property~\ref{entf}. Arguing as in Remark~\ref{reneothe}, but with $\alghj$ replaced
with $\frh\altp\frj$, one shows that~\ref{entf} holds true for every $D\in\hstahj$, because, in this case,
$\ent(D)=\|D\ptrn$. Let us now suppose that, instead, $D\in\stahj\setminus\hstahj$. By the proof
of~\ref{entc}, we know that the linear functional $\cph\altp\cpj\ni K\mapsto\tr(DK)\in\ccc$ is
\emph{not} bounded in this case, and therefore --- recalling final assertion of Proposition~\ref{pronorg}
--- it cannot be bounded on the $\norg$-dense linear subspace $\frh\altp\frj$ of $\cph\altp\cpj$
(and hence of $\cph\altp\cpj$), as well. Therefore, we have:
\begin{equation}
\sup\{|\tr(DF)|\colon F\in\frh\altp\frj, \ \|F\nog=1\}=+\infty=\ent(D) \fin .
\end{equation}
At this point, we can reason as in Remark~\ref{reneothe}; i.e., considering, for every
$F\in\frh\altp\frj$, the selfadjoint component $\fre\defi\frac{1}{2}(F+F^\ast)\in\frh\altp\frj$
of $F$, we have that
\begin{align}
+\infty=\ent(D)
& =
\sup\{|\tr(DF)|\colon F\in\frh\altp\frj, \ \|F\nog=1\}
\nonumber \\
& \dodici =
\sup\{|\repa(\tr(DF))|\colon F\in\frh\altp\frj, \ \|F\nog=1\}
\nonumber \\
& \dodici =
\sup\{|\tr(D\fre)|\colon F\in\frh\altp\frj, \ \|F\nog=1\}
\nonumber \\
& \dodici =
\sup\{|\tr(DF)|\colon F\in\frh\altp\frj, \ F=F^\ast, \ \|F\nog=1\}
\nonumber \\
& \dodici =
\sup\{\tr(DF)\colon F\in\frh\altp\frj, \ F=F^\ast, \ \|F\nog=1\} \fin .
\end{align}
Here, the second equality follows from the fact that $\|F\nog=\|zF\nog$, for every $z\in\toro$,
while, for the third equality, notice that $\repa(\tr(DF))=\frac{1}{2}(\tr(DF)+\tr(DF)^\ast)=
\frac{1}{2}(\tr(DF)+\tr(DF^\ast))=\tr(D\fre)$, where we have used the fact that $D=D^\ast$ and
the cyclic property of the trace. For the fourth equality, observe that
$\|\fre\nog\le\frac{1}{2}(\|F\nog+\|F^\ast\nog)=\|F\nog$, which would imply that the supremum
on the third line is smaller that the supremum on the fourth line; but, comparing the latter with
the supremum on the first line, we see that the reverse inequality holds too, and then we have
an equality. The fifth line obtains just noting that $|\tr(DF)|=\max\{\tr(DF),\tr(D(-F))\}$ therein.

Let us prove property~\ref{entg}. Since $\tr((U\otimes V)\tre D\tre (U\otimes V)^\ast F)=
\tr(D\tre (U\otimes V)^\ast F\tre(U\otimes V))$, for every $F\in\frh\altp\frj$, and, moreover, the
set $\{F\in\frh\altp\frj\colon F=F^\ast, \ \|F\nog=1\}$ is invariant wrt the local unitary transformation
$F\mapsto(U\otimes V)^\ast F\tre(U\otimes V)$ (recall that, by relation~\eqref{invlocuni},
$\|F\nog=\|(U\otimes V)^\ast F\tre(U\otimes V)\nog$), this is a straightforward consequence of
property~\ref{entf}.

Let us finally prove property~\ref{enth}. In fact, by~\ref{entd}, the entanglement function
$\ent\colon\stahj\rightarrow[1,+\infty]$ is lower semicontinuous wrt the standard topology
of $\stahj$; i.e., for every $r\in[1,+\infty)$, the superlevel set $\{D\in\stahj\colon\ent(D)>r\}$
of $\ent$ is open, or, equivalently, its complement --- the sublevel set $\stahjr$ --- is closed in
$\stahj$, wrt the standard topology of $\stahj$.
\end{proof}

\begin{remark}
Note that the entanglement function $\ent\colon\stahj\rightarrow[1,+\infty]$ can be extended,
in a natural way to the whole bipartite trace class $\trchj$. To the best of our knowledge,
the first proposal of an entanglement measure based on the projective norm is due to
Rudolph~\cite{Rudolph-measures}, who considered a finite-dimensional setting. A general
entanglement function of the kind we are considering here was defined, later on, by Arveson
--- see Part~{2} of~\cite{Arveson} --- who also proved that, in the case where at most one
of the Hilbert spaces of the bipartition (or, more generally, of a multipartition) is
infinite-dimensional, the entanglement function coincides with the projective norm. Here,
we have proved that, in the genuinely infinite-dimensional setting, the entanglement function
coincides with the projective norm \emph{when restricted to cross states}, whereas it assumes
the value $+\infty$ \emph{precisely} on all other states.
\end{remark}

\begin{remark} \label{renewenta}
Note that $\{F\in\frh\altp\frj\colon F=F^\ast\}=\frsah\altp\frsaj$, where $\frsah$, $\frsaj$
are the linear spaces of finite-rank selfadjoint operators on $\hh$ and $\jj$. In fact, the
set on the lhs of the equality contains the set on the rhs. Conversely, if
$F=\sum_k G_k\otimes H_k\in\frh\altp\frj$ and $F=F^\ast$, then
$F=\frac{1}{2}(F+F^\ast)=\frac{1}{4}\sum_k\big(G_k+G_k^\ast\big)\otimes\big(H_k+H_k^\ast\big)+\frac{1}{4}
\sum_k\big(\ima\big(G_k-G_k^\ast\big)\big)\otimes\big(-\ima\big(H_k-H_k^\ast\big)\big)\in\frsah\altp\frsaj$.
Therefore, by property~\ref{entf} in Theorem~\ref{theoentfu}, we have:
\begin{equation} \label{newenta}
\ent(D)=\sup\{\tr(DF)\colon F\in\frsah\altp\frsaj, \ \|F\nog=1\} \fin , \quad
\forall\cinque D\in\stahj \fin .
\end{equation}
\end{remark}

\begin{remark} \label{remalg}
Observe that, for every density operator $D\in\stahj$, we have:
\begin{align}
\ent(D)
& =
\sup\{|\tr(DK)|\colon K\in\alghj, \ \|K\nog=1\}
\nonumber \\
& =
\sup\{|\tr(DK)|\colon K\in\alghj, \ K=K^\ast, \ \|K\nog=1\}
\nonumber \\
& =
\sup\{\tr(DK)\colon K\in\alghj, \ K=K^\ast, \ \|K\nog=1\} \fin .
\end{align}
These relations follow immediately from definition~\eqref{dentfun} and property~\ref{entf}
in Theorem~\ref{neothe}, simply noting that $\{F\in\frh\altp\frj\colon F=F^\ast\}\subset
\{K\in\alghj\colon K=K^\ast\}\subset\alghj\subset\bophj$, and, for obtaining the third equality,
that, if $K=K^\ast$, then $|\tr(DK)|=\max\{\tr(DK),\tr(D(-K))\}$. By the same reasoning, the supremum
in~\eqref{dentfun} can be taken over the selfadjoint part $\bophjsa$ of $\bophj$ only, and,
in this case, one may replace $|\tr(DL)|$ with $\tr(DL)$ therein.
\end{remark}

\begin{remark}
Let $\st(\bophj)$ be the convex set of \emph{all states} of the $\cast$-algebra $\bophj$,
i.e., the set of all normalized, positive (hence, bounded) functionals on
$\bophj$~\cite{Holevo,Emch,Moretti,Murphy}:
\begin{equation}
\stbhj)\defi\{\stal\in\dbophj\colon\stal\ge 0,\ \stal(I)=1\} \fin .
\end{equation}
The entanglement function $\ent\colon\stahj\rightarrow[1,+\infty]$ admits a natural extension, i.e.,
\begin{equation}
\entex\colon\stbhj\ni \stal\mapsto\entex(\stal)\in[1,+\infty] \fin ,
\end{equation}
where, for every $\stal\in\stbhj$, we set
\begin{align}
\entex(\stal)
& \defi
\sup\{|\stal(B)|\colon B\in\bophj, \ \|B\nog=1\}
\nonumber \\
& \dodici =
\sup\{|\stal(B)|\colon B\in\bophjsa, \ \|B\nog=1\} \fin .
\end{align}
The second line above obtains arguing as in Remark~\ref{reneothe}, and taking into account
the well-known fact that $\stal(B^\ast)=\stal(B)^\ast$, for every positive functional on
the $\cast$-algebra $\bophj$ (see, e.g., Theorem~{3.3.2} of~\cite{Murphy}). It is worth
observing that, if $\dim(\hhjj)=\infty$, in the definition of the extended entanglement
function $\entex$ one may \emph{not} replace the $\cast$-algebra $\bophj$ with, say, the
smaller algebra $\cast$-algebra $\alghj$, as it can be done for the entanglement function
$\ent$ (recall Remark~\ref{remalg}). Indeed, there are states of the $\cast$-algebra $\bophj$
--- which form the convex set $\stbhjs$ of \emph{singular states} --- that vanish identically
on the compact operators $\cphj$. Precisely, every state $\stal\in\stbhj$ can be uniquely
decomposed into a convex combination $\stal=p\sei\stald + (1-p)\tre\stals$, $p\in[0,1]$,
where $\stald=\tr(D\argo)$, for some $D\in\stahj$, is a normal state of the $\cast$-algebra
$\bophj$, whereas $\stals\in\stbhjs$ is a singular state; see Section~{3}, in Chapter~{8}
of~\cite{Emch}, and Proposition~{10.4.3} of~\cite{Kadison}.
\end{remark}

\begin{remark}
Along lines similar to those followed for introducing the entanglement function $\ent$, one can
define an \emph{Hermitian entanglement function} $\enthe\colon\stahj\rightarrow[1,+\infty]$, which
is an extended real-valued entanglement function such that $\enthe(D)=\ptrntr{D}$, for all cross
states $D\in\hstahj$, and $\enthe(D)=+\infty$, for all other states $D\in\stahj\setminus\hstahj$.
\end{remark}

\begin{corollary}
The convex subset $\hstahj$ of $\stahj$ admits the following characterization:
\begin{equation}
\hstahj=\big\{D\in\stahj\colon\xid=\tr(D\argo)\in\alghjgd\big\} \fin .
\end{equation}
Moreover, for every $D\in\hstahj$, $\|\xid\dunog=\ent(D)=\|D\ptrn$, and hence
\begin{equation} \label{sestafun}
\sesta=\big\{D\in\stahj\colon\mbox{\rm $\xid\in\alghjgd$ and $\|\xid\dunog=1$}\big\} \fin .
\end{equation}
\end{corollary}

\begin{proof}
By Remark~\ref{remalg}, and by points~\ref{entb} and~\ref{entc} in
Theorem~\ref{theoentfu}, given a density operator $D\in\stahj$, the associated
functional $\xid\colon\alghjg\ni K\mapsto\tr(DK)\in\ccc$ is bounded iff $D\in\hstahj$;
moreover, for every $D\in\hstahj$, $\|\xid\dunog=\ent(D)=\|D\ptrn$, and hence, by
point~\ref{ente} in Theorem~\ref{theoentfu}, the characterization~\eqref{sestafun}
of $\sesta$ holds.
\end{proof}

In the case where at least one of the Hilbert spaces $\hh$ and
$\jj$ is finite-dimensional, the characterizations~\eqref{relsesta} and~\eqref{sestafun} of
$\sesta$ are equivalent, because, in this case, $\hstahj=\stahj$ and, moreover,  the Banach
space $\alghjgd\equiv\big(\alghj,\norg\big)^\ast$ is a renorming of the closed subspace
$\identi(\trchj)\cong\trchj$ of $\alghjd$; whereas, in the genuinely infinite-dimensional
case --- i.e., if $\dim(\hh)=\dim(\jj)=\infty$ --- $\hstahj\subsetneq\stahj$, and, for every
$D\in\stahj\setminus\hstahj$, the linear functional $\xid\colon\alghjg\rightarrow\ccc$ is
\emph{not} bounded. Therefore, in the latter case, relation~\eqref{sestafun} provides a stronger
characterization of separable states.

The entanglement function has a regular behavior wrt to restriction to (or extension from)
subspaces of the Hilbert spaces of the bipartition:

\begin{proposition}
Let $\vv$, $\ww$ be closed subspaces of the Hilbert spaces $\hh$ and $\jj$, respectively, and
let $\entvw\colon\stavw\rightarrow[1,+\infty]$ be the entanglement function on $\stavw$. Then,
for every $D\in\stavw$,
\begin{equation}
\entvw(D)=\ent(D) \fin ,
\end{equation}
where, with a slight abuse, we have identified the density operator $D$ on $\vv\otimes\ww$ with its
image via the natural (isometric) embedding of $\trcvw$ into $\trchj$.
\end{proposition}

\begin{proof}
Let $\pi\in\bop$, $\varpi\in\bopjj$ be the orthogonal projections onto $\vv$ and $\ww$.
The Banach space $\bopvw$ can be identified with the closed subspace
$\big\{L\in\bophj\colon L=(\pi\otimes\varpi)\tre L\tre (\pi\otimes\varpi)\big\}$ of $\bophj$,
i.e., the range of the continuous projection
$\bophj\ni L\mapsto(\pi\otimes\varpi)\tre L\tre (\pi\otimes\varpi)\in\bophj$
(which is a norm-closed by a well-known result; see, e.g., Corollary~{3.2.10} of~\cite{Megginson}).
Let us denote by $\norg$ the injective norm of $\bopvw$ (defined as in~\eqref{defnog}); namely, for
every $L\in\bopvw$, $\|L\nog\defi\sup\big\{|\tr(CL)|\colon C\in\trcvv\prtp\trcww,\ \vevw C\ptrn=1\big\}$,
where $\vwptrnor$ denotes the norm of the projective tensor product $\trcvv\prtp\trcww$ (recall that,
by Proposition~\ref{proimm}, $\vevw C\ptrn=\|C\ptrn$). Observe that $\norg$ coincides with the restriction
to $\bopvw$ of the injective norm of $\bophj$ (which justifies our use of a unique symbol for denoting
both norms). This fact can be easily checked, e.g., recalling relation~\eqref{fornog}, and observing
that, for every $L\in\bopvw$ --- i.e., for every $L\in\bophj$ such that
$L=(\pi\otimes\varpi)\tre L\tre (\pi\otimes\varpi)$ --- we have:
\begin{align}
\|L\nog
& =
\sup\{|\langle\phi\otimes\psi,L(\eta\otimes\chi)\rangle|\colon
\phi,\eta\in\vv,\ \psi,\chi\in\ww, \ \|\phi\|=\|\eta\|=\|\psi\|=\|\chi\|=1\}
\nonumber \\
& =
\sup\{|\langle\phi\otimes\psi,L(\eta\otimes\chi)\rangle|\colon
\phi,\eta\in\hh,\ \psi,\chi\in\jj, \ \|\phi\|=\|\eta\|=\|\psi\|=\|\chi\|=1\} \fin .
\end{align}
By the previous facts, it follows that, for every $D\in\stavw$,
\begin{align}
\entvw(D)
& =
\sup\{|\tr(DL)|\colon L\in\bophj, \ L=(\pi\otimes\varpi)\tre L\tre (\pi\otimes\varpi),
\ \|L\nog=1\}
\nonumber \\
& \le
\ent(D)
\nonumber \\
& =
\sup\{|\tr((\pi\otimes\varpi)\tre D\tre (\pi\otimes\varpi)\tre L)|\colon L\in\bophj, \ \|L\nog=1\}
\nonumber \\
& =
\sup\{|\tr(D\tre (\pi\otimes\varpi)\tre L\tre(\pi\otimes\varpi))|\colon L\in\bophj, \ \|L\nog=1\}
\nonumber \\
& \le
\sup\{|\tr(D\tre (\pi\otimes\varpi)\tre L\tre(\pi\otimes\varpi))|\colon L\in\bophj, \
\|(\pi\otimes\varpi)\tre L\tre(\pi\otimes\varpi)\nog=1\}
\nonumber \\
& =
\sup\{|\tr(DL)|\colon L\in\bophj, \ L=(\pi\otimes\varpi)\tre L\tre (\pi\otimes\varpi), \ \|L\nog=1\}
= \entvw(D) \fin .
\end{align}
Here, the first (and the last) equality, and the first inequality, hold by the fact that the injective
norm of $\bopvw$ coincides with the restriction to $\bopvw$ of the injective norm of $\bophj$, while,
for the second inequality, we have used the fact that
\begin{align}
\|(\pi\otimes\varpi)\tre L\tre(\pi\otimes\varpi)\nog
& =
\sup\{|\langle\phi\otimes\psi,L(\eta\otimes\chi)\rangle|\colon
\phi,\eta\in\vv,\ \psi,\chi\in\ww, \ \|\phi\|=\cdots=\|\chi\|=1\}
\nonumber \\
& \le\|L\nog \fin , \quad \forall\cinque\in\bophj \fin .
\end{align}
Therefore, in conclusion, for every density operator $D\in\stavw$, $\entvw(D)=\ent(D)$.
\end{proof}

The entanglement function $\ent$ is well-behaved wrt the action of suitable `local quantum maps'; i.e.,
the action of such a map \emph{does not increase} the entanglement of a cross state. To illustrate this fact,
let us recall that a any \emph{positive, trace-preserving linear map} $\ptph\colon\trc\rightarrow\trc\in\ltrc$
--- in short, a PTP map --- is bounded and \emph{mildly contractive}, i.e., $\|\ptph\noro=1$ (see Proposition~{1}
of~\cite{Aniello_CSP}). The following result that further characterizes the entanglement function $\ent$ holds:

\begin{proposition} \label{proptp}
Let $\ptph\colon\trc\rightarrow\trc$, $\ptpj\colon\trcjj\rightarrow\trcjj$ be positive, trace-preserving
linear maps. Then, there is a unique bounded linear map $\ptphj\colon\thptj\rightarrow\thptj$ such that
$\big(\ptphj\big)(S\otimes T)=\ptph(S)\otimes\ptpj(T)$, for all $S\in\trc$ and $T\in\trcjj$. The map
$\ptphj$ is trace preserving, i.e., $\tr\big(\big(\ptphj\big)(C)\big)=\tr(C)$, for all $C\in\thptj$.
If it is also positive --- namely, if $C\ge 0\implies\big(\ptphj\big)(C)\ge 0$ --- then we have:
\begin{itemize}

\item $\big(\ptphj\big)(\hstahj)\subset\hstahj$.

\item For every $D\in\hstahj$, $\ent\big(\big(\ptphj\big)(D)\big)\le\ent(D)$.

\end{itemize}
\end{proposition}

\begin{proof}
The linear maps $\ptph\colon\trc\rightarrow\trc$ and $\ptpj\colon\trcjj\rightarrow\trcjj$,
being positive and trace-preserving, are bounded, and $\|\ptph\noro=\|\ptpj\noro=1$.
By Corollary~\ref{coroptp}, there is a unique bounded linear map
$\ptphj\colon\thptj\rightarrow\thptj$ such that $\big(\ptphj\big)(S\otimes T)=\ptph(S)\otimes\ptpj(T)$,
for all $S\in\trc$ and $T\in\trcjj$, and, moreover, $\big\|\ptphj\bnorps=\|\ptph\noro\sei\|\ptpj\noro=1$.
Let us prove that this map is trace-preserving. In fact, for every $C\in\thptj$, there is a decomposition
of the form $C=\ptrnsum_k S_k\otimes T_k$, with $T_k\otimes S_k\in\ethtj$ (Corollary~\ref{corcro}).
Then, we have that
\begin{equation}
\tr\big(\big(\ptphj\big)(C)\big)=\tr\big({\textstyle \ptrnsum_k\big(\ptphj\big)(S_k\otimes T_k)}\big)
={\textstyle \sum_k \tr(\ptph(S_k))\sei\tr(\ptpj(T_k))}=\tr(C) \fin ,
\end{equation}
where, for the last equality, we have used the fact that $\ptph$ and $\ptpj$ are trace-preserving.
Therefore, if $\ptphj$ is also positive, then $\big(\ptphj\big)(\hstahj)\subset\hstahj$. Moreover,
since $\big\|\ptphj\bnorps=1$, then for every $D\in\hstahj$,
$\ent\big(\big(\ptphj\big)(D)\big)=\big\|\big(\ptphj\big)(D)\bptrn\le\|D\ptrn=\ent(D)$.
\end{proof}

We stress that the hypothesis that the map $\ptphj$ in Proposition~\ref{proptp} be positive is nontrivial,
and is realized, e.g., in the case where both the Hilbert spaces $\hh$, $\jj$ are finite-dimensional, and,
moreover, the PTP maps $\ptph$ and $\ptpj$ therein are, in particular, \emph{completely positive}~\cite{Paulsen,Stormer}.

By Theorem~\ref{theoentfu}, we will now prove that the entanglement of any given state
$D\in\stahj$ can be detected by a special type of quantum observables, a so-called
\emph{entanglement witness} (for the entangled state $D$)~\cite{Horodecki-bis,Terhal}.
Precisely, for every \emph{entangled} state $D\in\stahj$, there is a quantum observable
$E$ on $\hhjj$, such that the expectation value of $E$ is strictly larger than $1$ ---
and, in particular, by a suitable choice of $E$, arbitrarily close to $\ent(D)>1$, if
$\ent(D)<+\infty$ (i.e., if $D$ is a cross state) --- when measuring the observable on $D$,
whereas it is \emph{not} larger than $1$, not even in modulus, when measuring the observable
on \emph{any separable state} $\varsigma\in\sesta$. To make a long story short, a state is
entangled iff it admits a witness observable.

\begin{corollary}[Existence of suitable entanglement witnesses] \label{corentawi}
Let $D$ be an entangled state on the bipartite Hilbert space $\hhjj$; i.e., let
$D\in\stahj\setminus\sesta$. Then, there exists a bounded selfadjoint operator $E$ on
$\hhjj$ (an entanglement witness for $D$) satisfying the conditions
\begin{equation} \label{condsew}
E\in\frsah\altp\frsaj \fin , \ \|E\nog=1
\end{equation}
--- where $\frsah$ is the linear space of finite-rank selfadjoint operators on $\hh$ ---
and such that
\begin{equation}
\tr(DE)>1 \fin , \ \mbox{whereas} \ \; |\tr(\varsigma E)|\le 1 \fin , \
\forall\cinque\varsigma\in\sesta \fin .
\end{equation}
Moreover, if $D\in\hstahj\setminus\sesta$ --- i.e., if $D$ is an entangled cross state ---
then, for every $\epsilon>0$, the finite-rank selfadjoint operator $E\in\frsah\altp\frsaj$
can be chosen in such a way that
\begin{equation} \label{xcondent}
\|D\ptrn-\epsilon<\tr(DE)\le\|D\ptrn=\ent(D) \fin .
\end{equation}
If, instead, $D\in\stahj\setminus\hstahj$ --- i.e., if $D$ is not a cross state, and hence
$\ent(D)=+\infty$ --- then, for every $n\in\nat$, the finite-rank selfadjoint operator
$E\in\frsah\altp\frsaj$ can be chosen in such a way that
\begin{equation} \label{xcondent-bis}
\tr(DE)>n  \fin .
\end{equation}
\end{corollary}

\begin{proof}
Let $D\in\stahj\setminus\sesta$. By properties~\ref{ente} and~\ref{entf}
in Theorem~\ref{theoentfu}, and taking into account Remark~\ref{renewenta}, we have:
\begin{equation} \label{relent}
1<\ent(D)=\sup\{\tr(DF)\colon F\in\frsah\altp\frsaj, \ \|F\nog=1\} \fin .
\end{equation}

Thus, if $D\in\hstahj\setminus\sesta$, then by property~\ref{entb} in Theorem~\ref{theoentfu},
$\ent(D)=\|D\ptrn$, and, for every $\epsilon>0$, there is some $E\in\frsah\altp\frsaj$ ---
with $\|E\nog=1$ --- such that $\|D\ptrn-\epsilon<\tr(DE)\le\|D\ptrn=\ent(D)$. Assume that,
in particular, $0<\epsilon<\|D\ptrn-1$, and hence
\begin{equation}
1<\|D\ptrn-\epsilon<\tr(DE)\le\|D\ptrn=\ent(D) \fin .
\end{equation}
Therefore, we can and do choose $E$ in such a way that $\tr(DE)>1$ and condition~\eqref{xcondent}
is satisfied.

If, instead, $D\in\stahj\setminus\hstahj$ --- i.e., if $D$ is not a cross state --- then
$\ent(D)=+\infty$, and hence, by~\eqref{relent}, for every $n\in\nat$, there is some
$E\in\frsah\altp\frsaj$ --- with $\|E\nog=1$ --- such that $\tr(DE)>n\ge 1$ (thus, also
condition~\eqref{xcondent-bis} is satisfied).

In both cases, $\tr(DE)>1$, and, by Theorem~\ref{chasest} (or by property~\ref{ente}
in Theorem~\ref{theoentfu}), for every separable state $\varsigma\in\sesta$, we have
that $\ent(\varsigma)=\|\varsigma\ptrn=1$; hence, taking into account that $\|E\nog=1$,
we have: $|\tr(\varsigma E)|\le\ent(\varsigma)=\|\varsigma\ptrn=1$, for all $\varsigma\in\sesta$.
Thus, the bounded operator $E$ satisfies all the required conditions, and the proof is complete.
\end{proof}

\begin{remark} \label{rementawi}
To conform to the convention adopted in the literature~\cite{Horodecki-bis,Terhal}, we should
redefine our entanglement witness by considering the bounded selfadjoint operator $W=I-E\in\alghj$,
where $E\in\frsah\altp\frsaj$ is as in Corollary~\ref{corentawi}, so that the expectation value of
$W$ is strictly smaller than zero when measuring the witness on a given entangled state $D$ (whose
entanglement we wish to detect), whereas it is \emph{not} smaller than zero when measuring the
observable on any separable state. The standard proof relies on a suitable geometric version of the
Hahn-Banach theorem~\cite{Horodecki}. The advantage of our present result (and definition of the
entanglement witness) is twofold. On the one hand, we prove that a witness can always be chosen,
for every dimension of the local Hilbert spaces $\hh$ and $\jj$, in the linear space $\frsah\altp\frsaj$
of all finite linear combinations of finite-rank \emph{local observables}. On the other hand, we
show that this witness can be chosen in such a way as to estimate (by measuring it), the value
$\ent(D)>1$ of the entanglement function on the entangled state $D$ we are interested to detect,
provided that $D$ is a cross state, i.e., $\ent(D)<+\infty$.
\end{remark}

\begin{definition}[Entanglement witnesses]
We call a bounded selfadjoint operator $E\in\bophjsa$ an \emph{entanglement witness} if
$|\tr(DE)|>1$, for some entangled state $D\in\stahj\setminus\sesta$, whereas, for every
separable state $\varsigma\in\sesta$, $|\tr(\varsigma E)|\le 1$; in such a case, we say
that $E$ is \emph{an entanglement witness for $D$}. Moreover, we say that an entangled
state $D\in\stahj\setminus\sesta$ is \emph{detected} by a certain entanglement witness
$E\in\bophjsa$ if $|\tr(DE)|>1$.
\end{definition}

\begin{proposition} \label{proentwi}
If $E\in\bophjsa$ is an entanglement witness, then $\|E\tnori>1$. A bounded selfadjoint operator
$E\in\bophjsa$ is an entanglement witness if one of the following two equivalent conditions is
satisfied:
\begin{enumerate}[label=\tt{(W\arabic*)}]

\item \label{enwia}
$\|E\nog=1$ and $|\tr(DE)|>1$, for some state $D\in\stahj$ (which in then entangled).

\item  \label{enwib}
$\|E\tnori>\|E\nog=1$.

\end{enumerate}
Every entangled state $D\in\stahj\setminus\sesta$ is detected by some entanglement witness;
in particular, it can be detected by a finite-rank entanglement witness $E$ such that
$\|E\nog=1$, and, whenever $D$ is a cross state, one can further require that the expectation
value $\tr(DE)$ of the observable $E$ when measuring $D$ is arbitrarily close to $\ent(D)$,
with $1<\tr(DE)\le\ent(D)<+\infty$.
\end{proposition}

\begin{proof}
If $E\in\bophjsa$ is an entanglement witness, then $|\tr(DE)|>1$, for some entangled state
$D\in\stahj\setminus\sesta$, so that, by Fact~\ref{fabop-bis}, $\|E\tnori>1$. Let us now
prove the second assertion. In fact, if $\|E\nog=1$, then $|\tr(\varsigma E)|\le 1$, for
every separable state $\varsigma\in\sesta$, because $\sup\{|\tr(\varsigma E)|\colon
\varsigma\in\sesta\}\le\sup\big\{|\tr(CE)|\colon C\in\thptj, \ \|C\ptrn=1\big\}\ifed\|E\nog=1$,
where the inequality holds by the fact that $\varsigma\in\sesta\implies\|\varsigma\ptrn=1$.
Thus, if, in addition, $|\tr(DE)|>1$, for some state $D\in\stahj$, then this state must be
entangled, and hence $E$ is an entanglement witness; moreover, by Fact~\ref{fabop-bis},
$\|E\tnori>1=\|E\nog$, and hence condition~\ref{enwia} implies~\ref{enwib}. Now, if
$E\in\bophjsa$ is such that $\|E\tnori>1$, then, again by Fact~\ref{fabop-bis}, $|\tr(DE)|>1$,
for some state $D\in\stahj$; hence, condition~\ref{enwib} implies~\ref{enwia}, and both conditions
imply that $E$ is an entanglement witness. This proves the second assertion.

The third assertion is essentially a rephrasing of Corollary~\ref{corentawi}.
\end{proof}

Let us now study the projective norm $\ptrnor$ and the entanglement function on the \emph{pure states}
$\pustahj$ (the rank-one projections on $\hhjj$). We will denote by
\begin{equation}
\hpustahj\defi\pustahj\cap\hstahj=\pustahj\cap\thptj
\end{equation}
the set of all pure states on $\hhjj$ that are \emph{cross} states.

\begin{example} \label{xexchapu}
Let $\va\in\hhjj$ be a normalized nonzero vector --- $\|\va\|=1$ --- and let
\begin{equation} \label{xschdeca-bis}
\va=\sum_{l=1}^\sra a_l\tre(\neel\otimes\nffl) \fin , \ \
1\le\sra=\srank(\va)\le\min\{\dim(\hh),\dim(\jj)\} \fin , \
\{a_l\}_{l=1}^\sra=\scfs(\va)\subset\erreps ,
\end{equation}
be the Schmidt decomposition of $\va$ (recall Fact~\ref{schmdec}). As usual, we set
$\hva\equiv\hvaa\defi\vaa\in\pustahj$. Then, we have the expansion
\begin{equation} \label{xexphva}
\hva=\sum_{k,\tre l=1}^\sra a_k\tre a_l \Big(\nhekl\otimes\nhfkl\Big), \ \
\nhekl\otimes\nhfkl\equiv\phikl\otimes\psikl=\phipsikl \fin ,
\end{equation}
which converges absolutely wrt the projective norm $\ptrnor$ if
$\sum_{k,\tre l=1}^\sra a_k\tre a_l=\sum_{k=1}^\sra a_k\sum_{l=1}^\sra a_l<\infty$ ---
i.e., if $\sum_{l=1}^\sra a_l<\infty$ --- and, otherwise, convergence wrt the trace norm
$\ttrnor$ should be understood, namely, $\hva=\ttrnsum_{k,\tre l=1}^\sra a_k\tre a_l\sei\phipsikl$.
Precisely, $\ttrnor$-convergence \textit{\`a la} Pringsheim holds; i.e., for every $\epsilon>0$,
there is some $\nep\in\nat$ such that
$\|\hva-\sum_{k=1}^m \sum_{l=1}^n a_k\tre a_l\tre(\phipsikl)\ttrn<\epsilon$, for all $m,n>\nep$.

Let us also consider the related vector
$\vc_\enne\equiv\vc_\enne(\va)=\sum_{l=1}^\enne \neel\otimes\nffl\in\hhjj$,
--- where $1\le\enne\le\sra$, if $\sra=\srank(\va)<\infty$, and $\enne\in\nat$, otherwise
--- and the associated rank-one positive operator
\begin{equation} \label{xropop}
E_\enne\equiv E_\enne(\hva)\defi\vccn=\sum_{k,\tre l=1}^\enne \nhekl\otimes\nhfkl\in
\big(\frh\altp\frj\big)\cap\big(\erreps\sei\pustahj\big) .
\end{equation}

Recalling relation~\eqref{fornog}, we easily conclude that $\|E_\enne\nog=1$, because
$\langle\phi_1\otimes\psi_1, E_\enne(\phi_1\otimes\psi_1)\rangle=1$ and, moreover, for every
$\phi,\eta\in\hh$, $\psi,\chi\in\jj$, with $\|\phi\|=\|\eta\|=\|\psi\|=\|\chi\|=1$, we have
\begin{align}
1\le|\langle\phi\otimes\psi, E_\enne(\eta\otimes\chi)\rangle|
=|\langle\phi\otimes\psi,\vc_\enne\rangle\langle\vc_\enne,\eta\otimes\chi\rangle|
& =
\big|{\textstyle \sum_{k,\tre l=1}^\enne \langle\phi\otimes\psi,\phi_k\otimes\psi_k\rangle
\langle\phi_l\otimes\psi_l,\eta\otimes\chi\rangle}\big|
\nonumber \\
& =
\big|{\textstyle \sum_{k,\tre l=1}^\enne \langle\phi,\phi_k\rangle \langle\psi,\psi_k\rangle
\langle\phi_l,\eta\rangle \langle\psi_l,\chi\rangle}\big|
\nonumber \\ \label{estinorg}
& =
{\textstyle \big|\sum_{k=1}^\enne\overline{\alpha_k}\tre\beta_k\big|
\big|\sum_{l=1}^\enne\overline{\gamma_l}\tre\delta_l\big|}\le 1 \fin .
\end{align}
Here, $\alpha_k\equiv\langle\phi_k,\phi\rangle$, $\beta_k\equiv\langle\psi,\psi_k\rangle$,
$\gamma_l\equiv\langle\eta,\phi_l\rangle$ and $\delta_l\equiv\langle\psi_l,\chi\rangle$
--- with $\sum_{k=1}^\enne|\alpha_k|^2\le\|\phi\|^2=1$, $\sum_{k=1}^\enne|\beta_k|^2\le\|\psi\|^2=1$,
$\sum_{l=1}^\enne|\gamma_l|^2\le\|\eta\|^2=1$ and $\sum_{l=1}^\enne|\delta_l|^2\le\|\chi\|^2=1$ ---
so that the final inequality follows directly from the Cauchy-Schwarz inequality.

Notice that $\|E_\enne\tnori=\|\vc_\enne\|^2=\enne$; hence, for every $\enne>1$,
$\|E_\enne\tnori>\|E_\enne\nog=1$, so that $E_\enne$ is an entanglement witness
(recall the sufficient condition~\ref{enwib} in Proposition~\ref{proentwi}).

Let us assume, at first, that $\va\in\hhjj$ is such that $\sra=\srank(\va)<\infty$; hence,
$\hva\in\hpustahj$, because, by~\eqref{xexphva}, $\hva\in\frh\altp\frj$. Exploiting the fact that,
by Theorem~\ref{neothe} (see relation~\eqref{varexps}),
\begin{equation}
\|\hva\ptrn=\sup\big\{|\tr(\hva\tre F)|=|\langle\va,F\va\rangle|\colon F\in\frh\altp\frj, \ \|F\nog=1\big\}
\end{equation}
and the triangle inequality, one concludes that, for every $\enne\le\sra$,
\begin{align}
{\textstyle \big(\sum_{l=1}^\enne a_l\big)^2}=\langle\va,\vc_\enne\rangle^2
=|\langle\va,\vc_\enne\rangle|^2=\tr(\hva\sei E_\enne)
& \le
\|\hva\ptrn
\nonumber \\
& \le
{\textstyle \sum_{k,\tre l=1}^\sra a_k\tre a_l \big\|\nhekl\otimes\nhfkl\bptrn}
\nonumber \\ \label{xesthva}
& =
{\textstyle \big(\sum_{l=1}^\sra a_l\big)^2} ,
\end{align}
where $E_\enne\in\big(\frh\altp\frj\big)\cap\big(\erreps\sei\pustahj\big)$, with $\|E_\enne\nog=1$,
is the rank-one positive operator defined by~\eqref{xropop}. Then, the estimate~\eqref{xesthva},
for $\enne=\sra=\srank(\va)$, provides a simple expression of $\|\hva\ptrn$ in terms of the Schmidt
coefficients of $\va$, i.e.,
\begin{equation} \label{pnhva}
\|\hva\ptrn=\big({\textstyle \sum_{l=1}^\sra a_l\big)^2 = 1+ \sum_{k\neq l} a_k\tre a_l}  \fin ,
\end{equation}
where, for obtaining the second equality, we have used the fact that $\sum_{k=1}^\sra a_k^2=1$.
It is clear that, if $\min\{\dim(\hh),\dim(\jj)\}<\infty$ (hence, $\sra=\srank(\va)<\infty$), this
is actually the general expression of the projective norm of a pure state in $\pustahj=\hpustahj$.

Let us now focus on the genuinely infinite-dimensional case where $\dim(\hh)=\dim(\jj)=\infty$.
We claim that, in this case, $\pustahj\supsetneq\hpustahj$; specifically, given a normalized
nonzero vector $\va\in\hhjj$ --- with $\sra=\srank(\va)\le\infty$ --- we have that
\begin{equation} \label{xclaima}
\hva\in\hpustahj \iff \sum_{l=1}^\sra a_l<\infty
\end{equation}
and, for every $\hva\in\hpustahj$, $\|\hva\ptrn=\big(\sum_{l=1}^\sra a_l\big)^2$.
Otherwise stated, $\hva\in\hpustahj$ iff either $\srank(\va)<\infty$, or $\srank(\va)=\infty$
and the sequence its Schmidt coefficients $\{a_l\}_{l\in\nat}$ is contained in $\elluno\subsetneq\elldue$
(i.e., $\{a_l\}_{l\in\nat}$ is, up to normalization, a probability distribution). As a consequence,
we have:
\begin{equation}
\pustahj\setminus\hpustahj=\big\{\hva\in\pustahj\colon\mbox{$\srank(\va)=\infty$ and
$\{a_l\}_{l\in\nat}\in\elldue\setminus\elluno$}\big\}\neq\varnothing \fin .
\end{equation}
Moreover, relation~\eqref{pnhva} actually provides the most general expression of the projective
norm of a pure state $\hva$ in $\hpustahj$.

Let us prove relation~\eqref{xclaima}. We first consider the ``if part'' of the claim.
Then, let $\va\in\hhjj$ be a normalized nonzero vector, with a Schmidt decomposition of the
form~\eqref{xschdeca-bis}, and such that $\sum_{l=1}^\sra a_l<\infty$. In this case, as previously
noted, the expansion~\eqref{xexphva} is absolutely convergent wrt the projective norm $\ptrnor$
if $\sum_{k,\tre l=1}^\sra a_k\tre a_l=\sum_{k=1}^\sra a_k\sum_{l=1}^\sra a_l<\infty$,
i.e., if $\sum_{l=1}^\sra a_l<\infty$; hence, $\hva\in\hpustahj$. Consider, now, the
``only if part'' of the claim. It is sufficient to suppose that $\hva\in\hpustahj$,
$\sra=\srank(\va)=\infty$ and $\sum_{l=1}^\infty a_l=\infty$, and show that this assumption
leads to a contradiction. To this end, let us put, for every $n\in\nat$,
\begin{equation} \label{xdefhvan}
\hva_n\defi\vaan=\sum_{k,\tre l=1}^n a_k\tre a_l \Big(\nhekl\otimes\nhfkl\Big)\in\thatj \fin , \ \
\mbox{where}\ \va_n\defi\sum_{l=1}^n a_l\tre(\neel\otimes\nffl) \fin .
\end{equation}
As previously noted, $\hva=\ttrnorli_{m,n} \vamn=\ttrnorli_n\hva_n$, whereas --- since,
by relation~\eqref{pnhva} (that, of course, admits a direct generalization to any rank-one
positive operator),
\begin{equation} \label{xliminfi}
\|\hva_n\ptrn=\big({\textstyle \sum_{l=1}^n a_l\big)^2}
\implies\lim_n\|\hva_n\ptrn=\infty
\end{equation}
--- the sequence $\big\{\hva_n\big\}_{n\in\nat}$ cannot converge to $\hva$ wrt the
projective norm. This fact, however, does not \emph{automatically} imply that
$\hva\not\in\hpustahj$. Let us prove that this is indeed the case; i.e., that we find
a contradiction wrt our initial assumption. Consider the orthogonal projection operator
\begin{equation}
P_n\defi {\textstyle \Big(\sum_{k=1}^n \nhekk\Big)\otimes\Big(\sum_{l=1}^n \nhfll\Big)}
= \sum_{k,\tre l=1}^n \nhekk\otimes\nhfll
\end{equation}
onto the finite-dimensional subspace $\lispa\{\neek\otimes\nffl\colon k,l=1,\ldots, n\}$
of $\hhjj$. Note that
\begin{equation}
\va_n=P_n\va \quad \mbox{and} \quad \hva_n=P_n\hva\sei P_n \fin ,
\end{equation}
and, moreover, $P_n=\pi_n\otimes\varpi_n$, where $\pi_n=\sum_{k=1}^n\nhekk$,
$\varpi_n=\sum_{l=1}^n\nhfll$ are orthogonal projections onto
$\lispa\{\neek\colon k=1,\ldots, n\}\subset\hh$ and $\lispa\{\nffl\colon l=1,\ldots, n\}\subset\jj$,
respectively. Hence, by Proposition~\ref{proproj}
(see relation~\eqref{reproj-bis}),
\begin{equation} \label{xlenorm}
\|\hva_n\ptrn\le\|\hva\ptrn \fin , \quad \forall\cinque n\in\nat \fin .
\end{equation}
Therefore, by~\eqref{xliminfi} and~\eqref{xlenorm}, we must conclude that if
$\srank(\va)=\infty$ and $\sum_{l=1}^\infty a_l=\infty$, then
$\hva=\ttrnorli_n\hva_n\not\in\hpustahj$.

Having proved our first claim~\eqref{xclaima}, let us next prove that, for
\emph{every} pure state $\hva\in\hpustahj$, relation~\eqref{pnhva} holds true.
By what previously proved, we know that $\sum_{l=1}^\infty a_l<\infty$. At this point,
defining the rank-one positive operator $\hva_n$ as in~\eqref{xdefhvan}, in this case we
have that $\hva=\ptrnorli_n\hva_n$ and
\begin{equation} \label{xliminfi-bis}
\|\hva_n\ptrn=\big({\textstyle \sum_{l=1}^n a_l\big)^2}\implies
\|\hva\ptrn=\lim_n\|\hva_n\ptrn=\big({\textstyle \sum_{l=1}^\infty a_l\big)^2} .
\end{equation}

Summarizing, we have now shown --- with no assumption on the dimension of the Hilbert
spaces $\hh$ and $\jj$ ---  that a pure state $\hva\in\pustahj$ belongs to
$\hpustahj\defi\pustahj\cap\hstahj$ iff $\sum_{l=1}^\sra a_l<\infty$, where
$\{a_l\}_{l=1}^\sra$ is the set of the Schmidt coefficients of $\va$, and, moreover,
for every $\hva\in\hpustahj$,
$\|\hva\ptrn=\big({\textstyle \sum_{l=1}^\sra} a_l\big)^2= 1+ \sum_{k\neq l} a_k\tre a_l$.

It is worth observing that the expansion~\eqref{xexphva} of $\hva\in\hpustahj$, which converges
absolutely wrt the projective norm $\ptrnor$, is an \emph{optimal standard decomposition} ---
recall Definition~\ref{destanda} --- because
\begin{equation}
\big\|\nhekl\otimes\nhfkl\bptrn=\big\|\nhekl\trnb\sei\big\|\nhfkl\trnb=1
\end{equation}
and $\|\hva\ptrn=\sum_{k,\tre l=1}^\sra a_k\tre a_l$.

By the previous results, and by Theorem~\ref{theoentfu}, the value of the entanglement
function~\eqref{entfun} on a pure state is now completely determined, i.e.,
\begin{equation}
\ent(\hva)=
\begin{cases}
\|\hva\ptrn=\big({\textstyle \sum_{l=1}^\sra a_l\big)^2}
&
\mbox{if $\hva\in\hpustahj\iff\sum_{l=1}^\sra a_l<\infty$}
\\
+\infty
&
\mbox{if $\hva\in\pustahj\setminus\hpustahj\iff\sum_{l=1}^\sra a_l=\infty$}
\end{cases} \quad .
\end{equation}

Note that, in the case where $\emme\min\{\dim(\hh),\dim(\jj)\}<\infty$ (hence,
$\sra=\srank(\va)\le\emme<\infty$), $\hva$ is \emph{maximally entangled} iff
$\sra=\emme$ and, moreover, $a_1=\cdots=a_\sra=1/\sqrt{s}$, and, in such a case,
$\ent(\hva)=\|\hva\ptrn=s=\emme$.

Even if the fact that the entanglement function is not finite on those pure states that
are not cross states is a consequence of Theorem~\ref{theoentfu}, it is interesting to
see a direct proof of this fact.

For every $\hva\in\pustahj$, we have:
\begin{equation} \label{xresufro}
\ent(\hva)\defi\sup\{|\tr(\hva L)|\colon L\in\bophj, \ \|L\nog=1\} \fin .
\end{equation}
Assume that $\dim(\hh)=\dim(\jj)=\infty$ and $\hva\in\pustahj\setminus\hpustahj$
(hence, $\sra=\srank(\va)=\infty$). We claim that, in this case, $\ent(\hva)=+\infty$.
In fact, by relation~\eqref{xresufro}, and by the fact that $E_\enne$, $\enne\in\nat$,
is a positive operator in $\frh\altp\frj$ such that $\|E_\enne\nog=1$, we obtain
the estimate
\begin{equation} \label{hvaesti}
\ent(\hva)\ge\tr(\hva\sei E_\enne)=|\langle\va,\vc_\enne\rangle|^2
=\langle\va,\vc_\enne\rangle^2 ={\textstyle \big(\sum_{l=1}^\enne a_l\big)^2} \fin ,
\quad \forall\cinque\enne\in\nat \fin .
\end{equation}
At this point, since
$\hva\in\pustahj\setminus\hpustahj\implies\lim_\enne\sum_{l=1}^\enne a_l=
\sum_{l=1}^\infty a_l=\infty$, we conclude that $\ent(\hva)=+\infty$, coherently with
the property~\ref{entc} in Theorem~\ref{theoentfu}.

Finally, observe that, if the pure state $\hva\in\pustahj$ is entangled ---
i.e., if $\ent(\hva)>1$ (or, equivalently, if $\sra=\srank(\va)>1$) --- then, by
the estimates~\eqref{xesthva} and~\eqref{hvaesti}, for $\enne>1$ sufficiently large,
the rank-one positive operator
$E_\enne\in\big(\frh\altp\frj\big)\cap\big(\erreps\sei\pustahj\big)$, with $\|E_\enne\nog=1$,
defined by~\eqref{xropop}, is an \emph{entanglement witness} for $\hva$, because
$\tr(\hva\sei E_\enne)>1$ (recall the sufficient condition~\ref{enwia} in
Proposition~\ref{proentwi}).
\end{example}

In Example~\ref{xexchapu}, we have actually proved the following result, which completely
characterizes the projective norm and the entanglement of pure states on $\hhjj$:

\begin{proposition} \label{xproveca-bis}
Let $\va\in\hhjj$ be a nonzero vector --- with $\|\va\|=1$ --- and let
\begin{equation} \label{xschdeca}
\va=\sum_{l=1}^\sra a_l\tre(\neel\otimes\nffl) \fin , \ \
1\le\sra=\srank(\va)\le\min\{\dim(\hh),\dim(\jj)\} \fin , \
\{a_l\}_{l=1}^\sra=\scfs(\va)\subset\erreps ,
\end{equation}
be the Schmidt decomposition of $\va$, where $\sra=\srank(\va)$ denotes the Schmidt rank of $\va$.
Let us set $\hva\equiv\hvaa\defi{\vaa}$. Then, $\hva\in\hpustahj\defi\pustahj\cap\hstahj$ iff
$\sum_{l=1}^\sra a_l<\infty$.

In the case where $\emme\equiv\min\{\dim(\hh),\dim(\jj)\}<\infty$, the projective norm and the
entanglement of the pure state $\hva\in\pustahj=\hpustahj$ are provided by the expression
\begin{equation} \label{xpnhva}
\ent(\hva)=\|\hva\ptrn={\textstyle \big(\sum_{l=1}^\sra a_l\big)^2
=1+ \sum_{k\neq l} a_k\tre a_l} \fin .
\end{equation}
Hence, in this case, the projective norm and the entanglement of any pure state $\hva\in\pustahj$
satisfy the relations
\begin{equation} \label{xrelveca}
1\le\ent(\hva)=\|\hva\ptrn\le\sra=\srank(\va)\le\emme \fin ,
\end{equation}
where all inequalities can be saturated. In particular, $\hva$ is separable iff
$\ent(\hva)=\|\hva\ptrn=\sra=1$, whereas it is entangled iff
$1<\ent(\hva)=\|\hva\ptrn\le\sra=\srank(\va)\le\emme$.

Suppose now that $\dim(\hh)=\dim(\jj)=\infty$. Then, the set
\begin{equation}
\hpustahj\subsetneq\pustahj
\end{equation}
of those pure states on $\hhjj$ that are cross states consists of those states of the form
$\hva\defi\vaa$, where $\va\in\hhjj$ --- $\|\va\|=1$ --- is such that either $\srank(\va)<\infty$,
or $\srank(\va)=\infty$ and the sequence of its Schmidt coefficients $\{a_l\}_{l\in\nat}$ is contained
in $\elluno\subsetneq\elldue$ {\rm (i.e., $\sum_{l=1}^\infty a_l<\infty$ and $\sum_{l=1}^\infty a_l^2=1$)}.
Moreover, for every $\hva\in\hpustahj$, and given the Schmidt decomposition~\eqref{xschdeca},
we have that
\begin{equation} \label{xrenohva}
{\textstyle 1\le\ent(\hva)=\|\hva\ptrn=\big(\sum_{l=1}^\sra a_l\big)^2}<\infty  \fin ,
\end{equation}
whereas, for every $\hva\in\pustahj\setminus\hpustahj\neq\varnothing$, $\ent(\hva)=+\infty$.
In this case, $\hva$ is separable iff $\ent(\hva)=\|\hva\ptrn=\sra=1$, whereas it is entangled
iff $1<\ent(\hva)\le\sra=\srank(\va)\le\infty$.

Finally, every element of $\hpustahj$, for any dimension of the Hilbert spaces $\hh$
and $\jj$, is $\ptrnor$-optimally decomposable.
\end{proposition}

The preceding results allow us to prove the following important facts:

\begin{lemma} \label{leminfen}
Supposing that $\dim(\hh)=\dim(\jj)=\infty$, let $P\in\pustahj$ be such that $\ent(P)=+\infty$; i.e.,
let $P\in\pustahj\setminus\hpustahj$. Then, for every $D_0\in\stahj$ and every $p\in(0,1]$, the density
operator $D=p\tre P+(1-p)\tre D_0$ is such that $\ent(D)=+\infty$; i.e., $D\in\stahj\setminus\hstahj$.
\end{lemma}

\begin{proof}
Suppose that $\dim(\hh)=\dim(\jj)=\infty$ and the rank-one projection $P=\hva\equiv\vaa$ ---
for some normalized nonzero vector $\va\in\hhjj$ --- \emph{does not} belong to $\hpustahj$, and let
\begin{equation} \label{xysdeva}
\va=\sum_{l=1}^\sra a_l\tre(\neel\otimes\nffl) \fin , \quad
\sra=\srank(\va) \fin , \ \{a_l\}_{l=1}^\sra=\scfs(\va)\subset\erreps ,
\end{equation}
be the Schmidt decomposition of $\va$. Then, by Proposition~\ref{xproveca-bis}, we must have
that $\sra=\srank(\va)=\infty$ and, moreover,
$\hva\not\in\hpustahj \implies \sum_{l=1}^\infty a_l=\infty$.
Now, by property~\ref{entf} in Theorem~\ref{theoentfu}, we have:
\begin{equation} \label{xyzresufro}
\ent(D)=\sup\{\tr(DF)\colon F\in\frh\altp\frj, \ F=F^\ast, \ \|F\nog=1\} \fin .
\end{equation}
To exploit this useful expression of $\|D\ptrn$, let us consider, as in Example~\ref{xexchapu},
the vector $\vc_\enne=\sum_{l=1}^\enne \neel\otimes\nffl\in\hhjj$, $\enne\in\nat$, and the
associated rank-one positive operator
\begin{equation} \label{deranko}
E_\enne\defi\vccn=\sum_{k,\tre l=1}^\enne \nhekl\otimes\nhfkl\in
\big(\frh\altp\frj\big)\cap\big(\erreps\sei\pustahj\big) .
\end{equation}
As argued in Example~\ref{xexchapu} (see~\eqref{estinorg}), $\|E_\enne\nog=1$. Therefore,
by relation~\eqref{xyzresufro}, and by the fact that $E_\enne$ is a positive operator in
$\frh\altp\frj$ such that $\|E_\enne\nog=1$, for the density operator $D=p\tre P+(1-p)\tre D_0$
--- with $D_0\in\stahj$ and $p\in(0,1]$ --- we obtain the estimate
\begin{align}
\ent(D)\ge\tr(D\tre E_\enne)
& =
p\sei\tr(P\tre E_\enne)+(1-p)\tre\tr(D_0\tre E_\enne)
\nonumber \\ \label{xestenne}
& \ge
p\sei\tr(P\tre E_\enne)=p\sei|\langle\va,\vc_\enne\rangle|^2
=p\sei\langle\va,\vc_\enne\rangle^2
={\textstyle p\sei\big(\sum_{l=1}^\enne a_l\big)^2} \fin , \quad
\forall\cinque\enne\in\nat \fin .
\end{align}
At this point, since $\lim_\enne\sum_{l=1}^\enne a_l=\sum_{l=1}^\infty a_l=\infty$ and
$p>0$, by the previous estimate we see that $\ent(D)=+\infty$; i.e., by property~\ref{entc}
in Theorem~\ref{theoentfu}, $D\in\stahj\setminus\hstahj$.
\end{proof}

\begin{theorem} \label{xthespec}
Let $D=\sum_j p_j\tre P_j$ --- with $p_j>0$, $\sum_j p_j=1$ and $P_j\in\pustahj$ --- be the spectral
decomposition of a cross state $D\in\hstahj$, expressed in terms of minimal projections (and converging
wrt the trace norm $\ttrnor$). Then, we have that
\begin{equation}
\{P_j\}\subset\hpustahj\defi\pustahj\cap\hstahj .
\end{equation}
\end{theorem}

\begin{proof}
Our first claim is obvious in the case where $\min\{\dim(\hh),\dim(\jj)\}<\infty$, because,
in this case, $\stahj=\hstahj$ and $\pustahj=\hpustahj$.

Let us then suppose that $\dim(\hh)=\dim(\jj)=\infty$, and let $D=\sum_j p_j\tre P_j\in\hstahj
\subsetneq\stahj$, where the set $\{P_j\}$ consists of minimal --- i.e., rank-one --- spectral
projections of $D$. Reasoning by contradiction, assume that there is some rank-one projection
$P_m=\hva\equiv\vaa\in\{P_j\}$ (for some normalized nonzero vector $\va\in\hhjj$), that
\emph{does not} belong to $\hpustahj$. Then, we have that
$D=p_m\tre P_m + (1-p_m)\tre D_0$, where $D_0=\sum_{j\neq m} \frac{p_j}{1-p_m}\tre P_j\in\stahj$.
By Lemma~\ref{leminfen}, it follows that $\ent(D)=+\infty$; i.e., $D\in\stahj\setminus\hstahj$.
Thus, we are led to a contradiction wrt our initial assumption that $D\in\hstahj$ and there exists
some element $P_m$ of the set of rank-one projections $\{P_j\}$ such that $P_m\not\in\hpustahj$. Thus,
actually, $\{P_j\}\subset\hpustahj$.
\end{proof}

We are now able to characterize the extreme points of the convex set $\hstahj$.

\begin{corollary}
$\hpustahj\defi\pustahj\cap\hstahj=\ext(\hstahj)$.
\end{corollary}

\begin{proof}
First note that, since $\hstahj\subset\stahj$, then, by Fact~\ref{factext}, we have:
$\hpustahj\defi\pustahj\cap\hstahj=\ext(\stahj)\cap\hstahj\subset\ext(\hstahj)$.
Then, to prove our claim, we need to prove the reverse containment relation, i.e.,
that $\hpustahj\supset\ext(\hstahj)$, as well. In fact, let $D\in\hstahj\setminus\hpustahj$,
and let $D=\sum_j p_j\tre P_j$ --- with $p_j>0$, $\sum_j p_j=1$ and $P_j\in\pustahj$ --- be
the spectral decomposition of the cross state $D$, expressed in terms of rank-one projections.
Note that, here, the set $\{P_j\}$ must contain at least two elements, because, otherwise, we
would have that $D=P_1\in\hpustahj$, which would contradict our preceding assumption. Thus, we
can write $D=p_1\tre P_1 + (1-p_1)\tre D_1$, where $D_1=\sum_{j> 1} \frac{p_j}{1-p_1}\tre P_j\in\stahj$.
But actually, since $D$ is a cross state, then, by Theorem~\ref{xthespec}, $\{P_j\}\subset\hpustahj$.
Thus, in particular, $P_1\in\hstahj$, and hence
\begin{equation}
D_1=\frac{1}{1-p_1}\sei D+\frac{p_1}{p_1-1}\sei P_1\in\hstahj=\stahj\cap\thptj \fin ,
\end{equation}
as well. Thus, $D\not\in\ext(\hstahj)$, because it can be expressed as a nontrivial convex
combination of the cross states $P_1$ and $D_1$, and hence $\ext(\hstahj)\subset\hpustahj$.
\end{proof}


\end{document}